\documentclass[a4paper,12pt]{book}

\usepackage[a4paper, total={6in, 8in}]{geometry}
\usepackage{thmtools}
\usepackage{thm-restate}
\usepackage{mathrsfs}
\usepackage{amsmath}
\usepackage{aligned-overset}
\usepackage[backend=bibtex, maxbibnames=9]{biblatex}
\usepackage{amsthm}
\usepackage[toc,page]{appendix}
\usepackage{adjustbox}
\usepackage{amssymb}
\usepackage{bbm}
\usepackage[utf8]{inputenc}
\usepackage{tikz}
\usepackage{tikz-cd}
\usepackage{caption}
\usepackage{subcaption}
\usepackage{algorithm}
\usepackage{algpseudocode}
\usepackage{stmaryrd}
\usepackage{pgfplots}
\usetikzlibrary{automata,positioning,arrows}
\PassOptionsToPackage{hyphens}{url}
\usepackage{hyperref}
\usepackage{url}
\usepackage{cleveref}
\usepackage{ctable}
\usepackage{setspace}
\usepackage{footmisc}

\usepackage{fancyhdr}

\preto\fullcite{\AtNextCite{\defcounter{maxnames}{99}}}

\newtheorem*{theorem*}{Theorem}
\newtheorem{theorem}{Theorem}
\numberwithin{theorem}{subsection}
\newtheorem*{lemma*}{Lemma}
\newtheorem{lemma}[theorem]{Lemma}
\newtheorem*{proposition*}{Proposition}
\newtheorem{proposition}[theorem]{Proposition}
\newtheorem{corollary}[theorem]{Corollary}
\theoremstyle{remark}

\theoremstyle{definition}
\newtheorem{definition}[theorem]{Definition}

\newtheorem{example}[theorem]{Example}
\newtheorem{examples}[theorem]{Examples}

\newcommand{\suff}{\textnormal{\textsf{suf}}}
\newcommand{\Expr}{\textnormal{\textsf{Expr}}}
\newcommand{\BA}{\textnormal{\textsf{BA}}}
\newcommand{\BExpr}{\textnormal{\textsf{BExpr}}}
\newcommand{\At}{\textnormal{\textsf{At}}}
\newcommand{\GS}{\textnormal{\textsf{GS}}}
\newcommand{\GSM}{\textnormal{\textsf{GS}}^{-}}
\newcommand{\false}{\texttt{false}}
\newcommand{\true}{\texttt{true}}
\newcommand{\assert}{\texttt{assert}}
\newcommand{\aand}{\texttt{and}}
\newcommand{\oor}{\texttt{or}}
\newcommand{\nnot}{\texttt{not}}
\newcommand{\ddo}{\texttt{do}}
\newcommand{\iif}{\texttt{if}}
\newcommand{\tthen}{\texttt{then}}
\newcommand{\eelse}{\texttt{else}}
\newcommand{\wwhile}{\texttt{while}}
\newcommand{\row}{\textnormal{\textsf{row}}}
\newcommand{\GLStar}{\textnormal{\textsf{GL}}^{*}}
\newcommand{\LStar}{\textnormal{\textsf{L}}^{*}}
\newcommand{\Lsharp}{\textnormal{\textsf{L}}^{\sharp}}
\newcommand{\C}{\mathscr{C}}
\newcommand{\EM}{{\C}^T}
\newcommand{\KL}{{\C}_T}

\newcommand{\CSL}{\textnormal{\textsf{CSL}}}
\newcommand{\CABA}{\textnormal{\textsf{CABA}}}
\newcommand{\CDL}{\textnormal{\textsf{CDL}}}
\newcommand{\Set}{\textnormal{\textsf{Set}}}
\newcommand{\obs}{\textnormal{\textsf{obs}}}
\newcommand{\Bialg}{\textnormal{\textsf{Bialg}}}
\newcommand{\Coalg}{\textnormal{\textsf{Coalg}}}
\newcommand{\im}{\textnormal{\textsf{im}}}
\newcommand{\Vect}{\textnormal{-\textsf{Vect}}}
\newcommand{\Der}{\textnormal{\textsf{Der}}}
\newcommand{\dimension}{\textnormal{\textsf{dim}}}

\newcommand{\img}{\im}
\newcommand{\gen}{\textnormal{\textsf{gen}}}
\newcommand{\ext}{\textnormal{\textsf{ext}}}
\newcommand{\expa}{\textnormal{\textsf{exp}}}
\newcommand{\free}{\textnormal{\textsf{free}}}

\newcommand{\Sub}{\textnormal{\textsf{Sub}}}
\newcommand{\minn}{\textnormal{\textsf{min}}}

\newcommand{\M}{\textnormal{\textsf{M}}}

\newcommand{\Cat}{\textnormal{\textsf{Cat}}}

\newcommand{\ev}{\textnormal{\textsf{ev}}}

\newcommand{\Pow}{\mathcal{P}}
\newcommand{\lang}{L}

\newcommand{\id}{\textnormal{\textsf{id}}}
\newcommand{\st}{\textnormal{\textsf{st}}}

\newcommand{\supp}{\textnormal{\textsf{supp}}}

\newcommand{\Perm}{\textnormal{\textsf{Perm}}}

\newcommand{\dom}{\textnormal{\textsf{dom}}}
\newcommand{\cod}{\textnormal{\textsf{cod}}}

\newcommand{\Hom}{\textnormal{\textsf{Hom}}}
\newcommand{\Galg}{\textnormal{\textsf{GAlg}}}
\newcommand{\Balg}{\textnormal{\textsf{BAlg}}}
\newcommand{\Mat}{\textnormal{\textsf{Mat}}}
\newcommand{\Alg}{\textnormal{\textsf{Alg}}}
\newcommand{\Algb}{\textnormal{\textsf{Alg}}_{\textnormal{\textsf{B}}}}
\newcommand{\Klb}{\textnormal{\textsf{Kl}}_{\textnormal{\textsf{B}}}}

\newcommand{\Nom}[1]{#1\textnormal{\textsf{-Nom}}}

\begin{document}

\begin{titlepage}
\centering

\Huge
\textbf{
Canonical Algebraic Generators \\
\huge in Automata Learning}
\vfill
\vfill
\vfill
\vfill
\Large
Stefan Jens Zetzsche
\vfill
\vfill
\vfill
\vfill
\Large
A dissertation submitted in partial fulfillment of the requirements for the degree of \\
Doctor of Philosophy \\
at \\
University College London\\
(UCL)
\vfill
\vfill
\vfill
\vfill
Department of Computer Science
\vfill
August, 2023
\end{titlepage}
\cleardoublepage
\addcontentsline{toc}{chapter}{Declaration}
\chapter*{Declaration}

I, Stefan Jens Zetzsche confirm that the work presented in this thesis is my own. Where information has been derived from other sources, I confirm that this has been indicated in the thesis.	
\cleardoublepage
\addcontentsline{toc}{chapter}{Abstract}
\chapter*{Abstract}

Many methods for the verification of complex computer systems require the existence of a tractable mathematical abstraction of the system, often in the form of an automaton. In reality, however, such a model is hard to come up with, in particular manually. Automata learning is a technique that can \emph{automatically} infer an automaton model from a system -- by observing its behaviour. The majority of automata learning algorithms is based on the so-called $\LStar$ algorithm. The acceptor learned by $\LStar$ has an important property: it is \emph{canonical}, in the sense that, it is, up to isomorphism, the \emph{unique} deterministic finite automaton of \emph{minimal size} accepting a given regular language. Establishing a similar result for other classes of acceptors, often with side-effects, is of great practical importance. Non-deterministic finite automata, for instance, can be exponentially more succinct than deterministic ones, allowing verification to scale. Unfortunately, identifying a canonical size-minimal \emph{non-deterministic} acceptor of a given regular language is in general not possible: it can happen that a regular language is accepted by two \emph{non-isomorphic} non-deterministic finite automata of minimal size. In particular, it thus is unclear which one of the automata should be targeted by a learning algorithm. In this thesis, we further explore the issue and identify (sub-)classes of acceptors that admit canonical size-minimal representatives.

In more detail, the contributions of this thesis are three-fold. 

First, we expand the automata (learning) theory of Guarded Kleene Algebra with Tests (GKAT), an efficiently decidable logic expressive enough to model simple imperative programs. In particular, we present $\GLStar$, an algorithm that learns the unique size-minimal GKAT automaton for a given deterministic language, and prove that $\GLStar$ is more efficient than an existing variation of $\LStar$. We implement both algorithms in OCaml, and compare them on example programs.

Second, we present a category-theoretical framework based on generators, bialgebras, and distributive laws, which identifies, for a wide class of automata with side-effects in a monad, canonical target models for automata learning. Apart from recovering examples from the literature, we discover a new canonical acceptor of regular languages, and present a unifying minimality result.
 
Finally, we show that the construction underlying our framework is an instance of a more general theory. First, we see that deriving a minimal bialgebra from a minimal coalgebra can be realized by applying a monad on a category of subobjects with respect to an epi-mono factorisation system. Second, we explore the abstract theory of generators and bases for algebras over a monad: we discuss bases for bialgebras, the product of bases, generalise the representation theory of linear maps, and compare our ideas to a coalgebra-based approach.

\cleardoublepage
\addcontentsline{toc}{chapter}{Impact}
\chapter*{Impact}

\paragraph{Outside Academia}

As hardware and software systems continue to grow in complexity, methods for their verification become increasingly important. Classical model checking approaches to verification require the prior existence of a rich model of the system of interest, able to express all its relevant behaviour. In reality such a model is often unavailable, for instance, when the system comes in the form of a black-box with no access to the source code, or the system is simply too complex for manual processing. The black-box automata learning technique addresses this issue and has been successfully applied in a wide range of use cases, from finding bugs in network protocols \cite{de2015protocol}, reverse engineering smartcard reader for internet banking \cite{chalupar2014automated}, and other industrial applications \cite{hagerer2002model}. A comprehensive survey can be found in \cite{vaandrager2017model}. 

One of the bottle-necks for learning algorithms in an industrial setting is scalability. Identifying canonical target acceptors of minimal size is thus of great practical importance. As it happens, establishing uniqueness and minimality results for classes of acceptors with side-effects, which are often exponentially more succinct than their deterministic counterparts, can be surprisingly difficult \cite{denis2001residual}. In this thesis, we provide a general categorical framework that incorporates existing constructions of canonical minimal automata with side-effects and unveils new ones, and present a unifying minimality result. Another variable that impacts scalability is the maximum number of observation and equivalence queries an algorithm requires to learn a model. In this thesis, we present $\GLStar$, which deduces automata representations of simple imperative programs, and generally requires less queries than existing approaches. This is particularly interesting in view of potential applications to network verification \cite{smolka2019scalable}.

\paragraph{Inside Academia}

The lack of a canonical target acceptor when learning non-deterministic models has lead to a number of independent approaches, for different variants of non-determinism \cite{berndt2017learning, esposito2002learning}. More recently, there have been efforts to give a unifying perspective on those approaches. In this thesis, we present a category-theoretic framework that generalises ideas of Van Heerdt \cite{van2016master, van2020phd, van2020learning}, the notion of a scoop by Arbib and Manes \cite{arbib1975fuzzy}, and the universal-algebraic treatment of Myers et al. \cite{MyersAMU15}. We present the first general algorithm for the construction of succinct automata and give a unifying minimality result. In particular, we discover a previously unknown canonical acceptor of regular languages. By using bialgebras and distributive laws, we are able to clarify the central role of algebraic generators and bases, which has previously been underappreciated. We further explore this direction by developing the independent abstract theory of generators and bases. We expect those results to be also valuable outside of the automata learning community.

Guarded Kleene Algebra with Tests is a variation of Kleene Algebra with Tests that one obtains by restricting the union and iteration operations to guarded versions \cite{smolka2019guarded}. Recently, this variation has become the subject of increasing interest \cite{schmid2021guarded}. We contribute to the development of its automata (learning) theory, for example, by establishing the existence of a unique size-minimal acceptor. Our treatment leads to numerous directions that could be further explored; for instance improving efficiency through more compact data-structures, or an adaption to a probabilistic extension. 

The results of this thesis have been published at the 37th and 38th International Conference on Mathematical Foundations of Programming Semantics \cite{zetzsche2021, zetzsche2022guarded}, and at the 10th Conference on Algebra and Coalgebra in Computer Science \cite{zetzsche2020generators}. Additionally, some of the results have been presented at the 8th Symposium on Compositional Structures \cite{zetzsche2021}.
\cleardoublepage
\chapter*{Acknowledgements}%
\addcontentsline{toc}{chapter}{Acknowledgements}

First of all, my thanks go to my two supervisors, Alexandra and Matteo. 
I am grateful for you giving me the chance to pursue a PhD in the first place. Studying abroad, in London, has been something I have always wanted to do. At all times, you two have been patient and encouraging. Thank you for supporting and guiding me through good and bad times. Alexandra, I have learnt a lot from you. Your sharp intellect, good heart, and work ethic will stay with me. I will also remember your delicious cooking. Matteo, you have always had the right advice for me. Your writing has been as beautiful and simple as the food of Italy. I will miss you two.

I would also like to thank all the peers that have taken the time to review my work and provide valuable feedback. Over time, I received many insightful comments on my submissions, often from anonymous reviewers. My coauthors Alex, Gerco, and Matteo have given me fruitful guidance during numerous discussions. I am particularly grateful to Mehrnoosh and Stefan, for agreeing to be my examiners, and for studying my thesis in thorough detail. I am also thankful to Fabio and James, who chaired my transfer and first year viva, respectively. Both experiences were challenging, but helped me a lot with progressing towards an independent researcher.

I am very fortunate to have been a part of the Programming Principles, Logic, and Verification group. All the numerous seminars, reading groups, and informal conversations were a welcome change to my daily routine. Over the years I shared my office with many inspiring people. My first desk has been in a room with Gerco, Jana, Paul, and Tobias, all of which are greatly missed. Later, I moved to the basement, which I shared with Jas, Leo, Linpeng, Louis, Mateo, Robin, Tiago, Will, and Wojciech. The collaborative spirit in our department allowed me to converse with researchers of all experiences. Thank you, Bas, Benjamin, Christoph, Diana, Fredrik, Jurriaan, Lachlan, Maria, Paul, Simon, Sonia, Tao, Todd,  and everyone else that made our group so unique.

During my PhD I was able to attend, volunteer, and mentor at various conferences and schools, first in person, and later online. Among others, I am grateful to the organisers of the following events: Syco Birmingham 2018, Syco Strathclyde 2018, VeTSS Workshop on Formal Methods and Tools for Security 2018, Calco 2019, Facebook Proofs for Bugs 2019, Mfps 2019, Scottish Programming Languages and Verification Summer School 2019, VeTSS Verified Software Workshop 2019, Pldi 2020, Popl 2020, Splash 2020, Cav 2021,  Icalp 2021, Mfps 2021, Popl 2021, Syco Tallinn 2021, Cav 2022, Mfps 2022. One of the highlights has been the trip to Popl in New Orleans, in early 2020, just before the first Covid cases became public. I remember exploring the town with Jana and Tobias, and dancing to live music, late at night. 

In addition to my own studies, I have been working as a Teaching Assistant for a number of courses. I would recommend everyone to do the same during their PhD; it is a humbling experience. Over the years, I was lucky to meet many students, often bright, always unique in character. For some courses, I collaborated with other PhD students, which was particularly fun. Thank you, Tao, Thomas, and Todd!

I was fortunate enough to experience academia not only from the inside, but also from the outside. Towards the end of my PhD, I completed two twelve-week internships in the tech industry: first at Amazon, then at Meta. Both experiences have been very rewarding and have greatly  broadened my horizon. It would be a hopeless attempt to name all of the people I met along the way. There are, however, a few people I would like to thank particularly. 

First of all, there is Byron, who, when I hesitantly reached out to him with the idea of doing an Amazon internship, replied, to my surprise, instantly, and encouraged me in his unique way. Thank you, Byron, for forwarding me to the right places. 

During the internship I was supervised by Rustan, who was based in Seattle, while I was working from London. I could not have asked for a better match. Rustan was very generous. He took the time to meet me virtually on a daily basis and has answered my endless questions in impressive clarity and depth, and with great spirit. On top of all, Rustan turned out to be also a tremendous chef and host. 

During my time in the London office I met many great people. Thank you Caroline, Claudia, Daniel, Ilina, Sacha, and everyone else. You have made working remotely so much easier for me. 

As Covid cases were coming down, I was able to complete my second internship at Meta in person, in London. The Hack language team has been more than welcoming. Thank you Andrew, Frank, Henri, Max, Michael, Mistral, Scott, and everyone else I met along the way. Our trips to Cambridge and Menlo Park have been highlights I'll cherish. In particular I would like to thank Mistral, my supervisor, who has been a great mentor and friend. I miss our coffee runs and trips to the climbing hall. 

Around the end of 2020, I applied to the SIGPLAN long-term mentorship program. I am happy to say that since then I have met my mentor on an almost bi-weekly basis. Thank you, Ravi, your support and sincerity mean a lot to me. I hope we manage to see each other in person in the near future.

Thanks also to my friends and family for their continuous support during this long journey. In particular, to my parents, Frank and Gabriele, for their unconditional love over all the years. Without you two this thesis would not have been possible. 

Finally, I'd like to say thanks to Aurora. Thank you for your love, advice, and understanding. For helping me through the lows, and for celebrating the highs. You have played a great part in the success of this thesis.

\footnotetext{My research has been supported by GCHQ via the VeTSS grant \emph{Automated Black-Box Verification of Networking Systems} (4207703/RFA 15845) and by the ERC via the Consolidator Grant \emph{AutoProbe} (101002697).}

\tableofcontents
\chapter{Introduction}

As hardware and software systems continue to grow in size and complexity, methods for their analysis become increasingly important. To study the characteristics of a system, classical approaches require the existence of a simplified mathematical model that captures the relevant behaviour of the system. One of the simplest models is the \emph{automaton}, which can be thought of as a diagram generated by the states of the system and transitions between them. Unfortunately, in reality a complex model is only rarely available, for instance, when the system comes in the form of a black-box with no access to the source code, or the system is too complex for manual processing.

The aim of \emph{automata learning} is to automatically infer a minimal sized automata representation of a system by observing its behaviour. The incremental approach has been successfully applied to a wide range of verification tasks from finding bugs in network protocols \cite{de2015protocol}, reverse engineering smartcard reader for internet banking \cite{chalupar2014automated}, and industrial applications \cite{hagerer2002model}. A comprehensive survey of the key developments in automata learning can be found in \cite{vaandrager2017model}. As a result of its success, automata learning has inspired numerous adaptions which target more expressive models. 

This thesis develops along two orthogonal axes. On the one hand, we contribute to the branches of automata learning by developing an algorithm that efficiently learns a model that captures the behaviour of a simple imperative program by applying domain-specific optimisations. On the other hand, we further develop the abstract perspective on automata learning by presenting a mathematical framework that unifies numerous constructions of minimal target models, and unveils new ones.
 
The introduction is structured as follows.
 In \Cref{intro_sec_automata_learning} we present the seminal $\LStar$ automata learning algorithm. We continue with an example run of $\LStar$ in \Cref{intro-examplerun}. The content of  \Cref{intro_sec_othertypes} is an overview of adaptions of $\LStar$ to other types of target models. In \Cref{gkat-sec-intro} we present Guarded Kleene Algebra with Tests, and discuss its ramifications with automata learning. We then discuss the concept of size-minimality for target models in  \Cref{intro_sec_minimality}. The unifying perspective of category theory on automata learning is explored in \Cref{intro_sec_categorical}. 
 We continue with a presentation of the main objectives of the thesis in \Cref{intro_sec_main_objectives}, and conclude with an overview of the structure and the contributions of the thesis in \Cref{intro_sec_contributions}.

\section{Automata Learning}

\label{intro_sec_automata_learning}

Automata learning can be broadly divided into \emph{active} and \emph{passive} learning algorithms.
Active learning algorithms infer a model from a system by interacting with (or \emph{querying}) the system during the execution of the algorithm. In contrast, during passive learning, there is no direct access to a system for the duration of the run of an algorithm. Instead, passive learning algorithms try to model a system based on potentially insufficient stored data obtained through prior observation. In this thesis, we will focus on active automata learning.

\begin{figure}
\tiny
\centering
\begin{tikzpicture}[node distance=5.5em]
	\node[state, initial, initial text=, minimum size=2.5em] (q1) {};
	\node[state, right of=q1, above of=q1, minimum size=2.5em] (q2) {};
		\node[state, below of=q2, minimum size=2.5em] (q3) {};
		\node[state, below of=q3, minimum size=2.5em] (q4) {};
		\node[state, right of=q3, minimum size=2.5em, accepting] (q5) {};
		\node[state, right of=q5, minimum size=2.5em] (q6) {};
	    \path[->]
	(q1) edge[above] node{$a$} (q2)
		(q1) edge[above] node{$b$} (q3)
		(q1) edge[above] node{$c$} (q4)
		(q2) edge[left] node{$b,c$} (q5)
		(q2) edge[bend left, above] node{$a$} (q6)
		(q3) edge[above] node{$a,c$} (q5)
		(q3) edge[bend right, below	] node{$b$} (q6)
		(q4) edge[bend right, below] node{$c$} (q6)
		(q4) edge[left] node{$a,b$} (q5)
		(q5) edge[above] node{$a,b,c$} (q6)
		(q6) edge[loop right] node{$a,b,c$} (q6)
        ;
\end{tikzpicture}
\caption{Up to isomorphism, the unique size-minimal DFA accepting the language $\lbrace ab,ac, ba,bc,ca,cb \rbrace \subseteq \lbrace a, b, c \rbrace^*$ \cite{arnold1992note}}
\label{minimaldfaexample}
\end{figure}
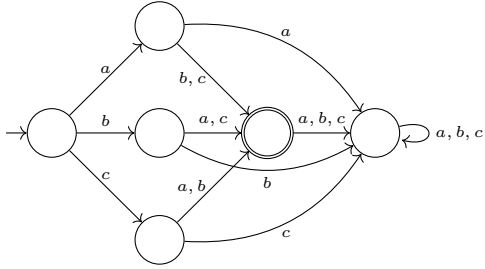

A \emph{non-deterministic finite automaton} (NFA) over a fixed input alphabet set of characters consists of a set of states, partitioned into \emph{accepting} and \emph{rejecting} states, a distinguished initial state, and a transition function that assigns to each state and input character a \emph{set} of next states. If at every transition the set of next states consists of a single state, then we speak of a \emph{deterministic finite automaton} (DFA). Any NFA can be depicted as a directed graph of nodes, which represent its states, and arrows, which represent transitions, thus are labelled by input characters. The node representing the initial state is annotated by an arrow without domain. Nodes that correspond to accepting states are indicated by a double circle. A simple example of a DFA over the input alphabet $\lbrace a, b, c \rbrace$ with six states of which one is accepting is given in \Cref{minimaldfaexample}.

A finite sequence of characters in the input alphabet is referred to as \emph{word}. A set of words is called a \emph{language} (over the input alphabet). A word is \emph{accepted} by a DFA, if consecutively reading in its characters leads to a transition from the initial state to an accepting state. A DFA \emph{accepts} the language that consists of the words it accepts. A language that is accepted by a DFA is called \emph{regular}. For any regular language there is a unique size-minimal DFA accepting it, defined up to a structure-preserving bijection. For example, the size-minimal DFA accepting the regular language of words over the alphabet $\lbrace a,b,c \rbrace$ that have length two and start and end in different characters is depicted in \Cref{minimaldfaexample}. 

In active automata learning it is usually assumed that the behaviour of the system one tries to model  is given in terms of an unknown regular language. In consequence, there exists, in principle, an (unknown) size-minimal DFA that accurately models the behaviour of the system. Under this assumption, the most simple form of interacting with a system is through a \emph{membership query}. During such an idealised interaction, the learning algorithm \emph{submits} a word to the system and observes its response. Either, the word is \emph{accepted}, that is, it is an element of the regular language that represents the behaviour of the system, or, it is \emph{rejected}.

In \cite{moore1956gedanken} Moore showed that membership queries alone are insufficient to deduce, in finite time, a DFA that correctly models the behaviour of a system. Later, Angluin \cite{angluin1981note} proved that with the auxiliary knowledge of the number $n$ of states of the minimal DFA accepting the behaviour of the system, there exists a correct algorithm that converges in finite time, but performs a number of membership queries that is exponential in $n$ in the worst case. Consequently, Angluin described and studied several other types of queries, in particular, the \emph{equivalance query} \cite{angluin1988queries}. During such a query, a hypothesis model is submitted, and either accepted, if it correctly models the target behaviour, or rejected. In the case of rejection, a \emph{counterexample} is provided, that is, a word that is either incorrectly accepted or incorrectly rejected by the hypothesis.

In her seminal work \cite{angluin1987learning} Angluin presented $\LStar$, an algorithm that learns the minimal DFA accepting a target regular language, by assuming a \emph{minimally adequate teacher}, which can answer membership and equivalence queries. The algorithm runs in time polynomial in the number of states of the minimal DFA and the maximum length of any counterexample. 

The minimal adequate teacher framework should be understood as an elegant mathematical abstraction, rather than a precise implementation guideline.  For instance, verifying whether the behaviour of the hypothesis matches the target language can in practice easily become unfeasible.  
Equivalence queries are thus often approximated via testing \cite{1702519}. If a bound to the size of the target model is known, or the target language is represented by a known automaton for experimental purposes, equivalence queries can be performed exactly. In the latter case, for example, one can utilise efficient bisimulation algorithms (see e.g. \cite{hopcroft1971linear, bonchi2013checking}).

\section{An Example Run of $\LStar$}

\label{intro-examplerun}

\begin{figure*}
\centering
\begin{subfigure}[b]{.15\textwidth}
	\centering
		\resizebox{0.65 \textwidth}{!}{
		\begin{tabular}{ c|c}
		 & $\varepsilon$ \\
		 \hline $\varepsilon$ & $1$   \\
		 \hline
		 \hline $a$ & $0$
	\end{tabular}
	}
	\caption{}
	\label{lstar_example_a}
	\end{subfigure}
	\begin{subfigure}[b]{.15\textwidth}
	\centering
		\resizebox{0.65 \textwidth}{!}{
		\begin{tabular}{ c|c}
		 & $\varepsilon$  \\
		 \hline $\varepsilon$ & $1$   \\
		 \hline $a$ & $0$  \\
		 \hline
		 \hline $aa$ & $1$
	\end{tabular}
	}
	\caption{}
		\label{lstar_example_b}
	\end{subfigure}
	\begin{subfigure}[b]{.3\textwidth}
	\centering
			\resizebox{0.8 \textwidth}{!}{
				\begin{tikzpicture}[node distance=5em]
	\node[state, shape=circle, initial, accepting, initial text=] (x) {$\varepsilon$};
		\node[state, shape=circle, right of=x] (y) {$a$};
	    \path[->]
	(x) edge[above, bend left] node{$a$} (y)
	(y) edge[below, bend left] node{$a$} (x)
	;
	\end{tikzpicture}
	}
	\caption{}
		\label{lstar_example_c}
\end{subfigure}
	\begin{subfigure}[b]{.3\textwidth}
	\centering
	\resizebox{0.8 \textwidth}{!}{
		\begin{tabular}{ c|c|c|c|c }
		 & $\varepsilon$ & $a$ & $aa$ & $aaa$   \\
		 \hline $\varepsilon$ & $1$ & $0$ & $1$ & $1$   \\
		 \hline $a$ & $0$ & $1$ & $1$ & $1$ \\
		 \hline
		 \hline $aa$ & $1$ & $1$ & $1$ & $1$  \\ 
	\end{tabular}
	}
	\caption{}
		\label{lstar_example_d}
	\end{subfigure}
		\begin{subfigure}[b]{.3\textwidth}
	\centering
	\resizebox{0.8 \textwidth}{!}{
		\begin{tabular}{ c|c|c|c|c }
		 & $\varepsilon$ & $a$ & $aa$ & $aaa$   \\
		 \hline $\varepsilon$ & $1$ & $0$ & $1$ & $1$   \\
		 \hline $a$ & $0$ & $1$ & $1$ & $1$ \\
		 \hline $aa$ & $1$ & $1$ & $1$ & $1$  \\ 
		 		 \hline
		 		 		 \hline $aaa$ & $1$ & $1$ & $1$ & $1$  \\ 
	\end{tabular}
	}
	\caption{}
	\label{lstar_example_e}
	\end{subfigure}
		\begin{subfigure}[b]{.3\textwidth}
	\centering
		\resizebox{1.2 \textwidth}{!}{
				\begin{tikzpicture}[node distance=5em]
	\node[state, shape=circle, initial, accepting, initial text=] (x) {$\varepsilon$};
		\node[state, shape=circle, right of=x] (y) {$a$};
		\node[state, shape=circle, right of=y, accepting] (z) {$aa$};
	    \path[->]
	(x) edge[above] node{$a$} (y)
	(y) edge[above] node{$a$} (z)
	(z) edge[loop above] node{$a$} (z)
	;
	\end{tikzpicture}
	}
	\caption{}
	\label{lstar_example_f}
\end{subfigure}
\caption{An example run of (a variation of) Angluin's $\LStar$ algorithm for the target regular language $1 + a \cdot a \cdot a^* = \lbrace \varepsilon, aa, aaa, ... \rbrace \subseteq \lbrace a \rbrace^*$}
\label{lstarexamplerun_intro}
\end{figure*}
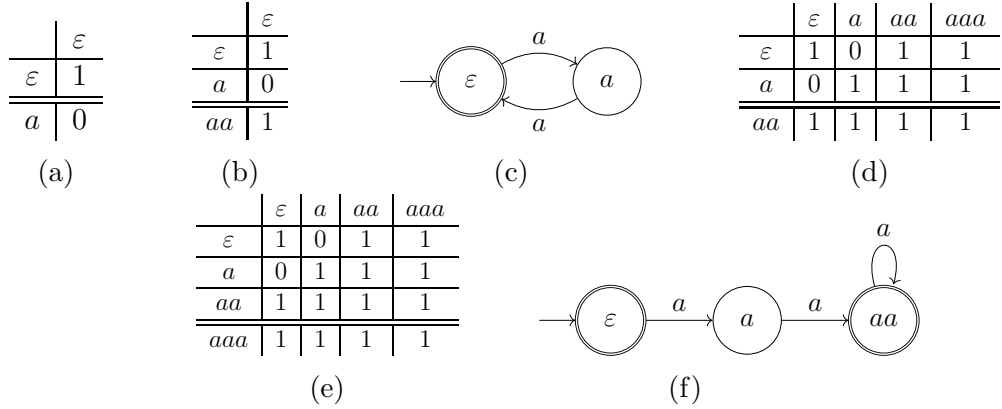

We will now give an intuitive account of how Angluin's $\LStar$ algorithm can be used to learn via membership and equivalence queries the minimal DFA accepting a given regular language. In our example, we fix the regular target language $L = 1 + a \cdot a \cdot a^*$, which consists of words over the singleton input alphabet $\lbrace a \rbrace$, either of length $0$ or of length at least $2$.  The complete run of $\LStar$ (more precisely, a slight variation\footnote{In contrast to the original presentation of $\LStar$, which adds rows for all prefixes of a counterexample, we use a variation by Maler and Pnueli \cite{maler1995learnability}, which adds columns for all suffixes of a counterexample. This has the advantage that \emph{consistency} checks become redundant.} of it) for $L$ can be found in \Cref{lstarexamplerun_intro}. Its result, the minimal DFA accepting $L$, has $3$ states, and is depicted in \Cref{lstar_example_f}.

 At the heart of $\LStar$ is a data-structure called \emph{observation table}, which consists of partial information about $L$, gathered by performing membership queries. The binary value of an observation table at row $i$ and column $j$ specifies whether the word $w_i \cdot w_j$, obtained by concatenating the words $w_i$ and $w_j$ at respective indices, is contained in $L$, or not. In the former case, the table contains a one, and in the latter case a zero. An observation table consists of two disjoint parts: an upper part, and a lower part, visually distinguished by two horizontal lines. Every row in the upper part of the table denotes the state of a hypothesis automaton that is not necessarily well-defined yet. The lower part of the table is used to establish the transitional structure of the automaton. During the run of $\LStar$, the observation table is extended until a well-defined hypothesis can be constructed. The hypothesis is then checked for equivalence, leading to termination, if positive, or further refinement of the table via a counterexample, if negative.
 
 Initially, the observation table for $L$ is initiated with one upper row and one column, both indexed by the empty word $\varepsilon$. Since the concatenation $\varepsilon = \varepsilon \cdot \varepsilon$ is accepted by $L$, the upper left entry of the table in \Cref{lstar_example_a} contains a one. To deduce from the table an automaton, we need to determine the state that is reached from the state indexed by $\varepsilon$, when reading in the character $a$. To do so, we construct the lower row indexed by $a = \varepsilon \cdot a$. The resulting table is depicted in \Cref{lstar_example_a}. 
 
 Since the row indexed by $a$ differs from all rows in the upper part of the table (that is, the one indexed by $\varepsilon$), the automaton potentially corresponding to the table in \Cref{lstar_example_a} contains at least two states. Formally, we move the row indexed by $a$ to the upper part of the table. In consequence, we also need to determine the state that is reached from the state indexed by $a$, when reading in the character $a$. The new table, with a row indexed by $a \cdot a$ in its lower part, is depicted in \Cref{lstar_example_b}. As every row in the lower part of the table coincides with one upper row, we now can deduce a well-defined automaton.
   
  The well-defined hypothesis corresponding to the observation table in \Cref{lstar_example_b} is given in \Cref{lstar_example_c}. Its initial state corresponds to the row indexed by $\varepsilon$. The initial state is also accepting, because the corresponding row contains a one at the column indexed by $\varepsilon$. Is the behaviour of the hypothesis given by the target language $L$? There exists at least one counterexample that witnesses that this is not the case. For example, the word $aaa$ is accepted by $L$, but rejected by the hypothesis. 
  
  To account for the imprecise behaviour, we add to the observation table in \Cref{lstar_example_b} one column for each suffix of the counterexample $aaa$. The updated table is depicted in \Cref{lstar_example_d}. While the rows indexed by $\varepsilon$ and $aa$ have previously contained identical values, they now differ, as is witnessed by their entry at the column indexed by $a$. As the lower row indexed by $aa$ now differs from all upper rows, we move it to the upper part of the table.
 To derive a well-defined transitional structure, we add a new row indexed by $aaa$ to the lower part of the table. The resulting structure is depicted in \Cref{lstar_example_e}. 
 
 Since every row in the lower part of the table corresponds to one row in the upper part of the table, we can derive the well-defined hypothesis \Cref{lstar_example_f}. As one verifies, the language accepted by the hypothesis is precisely the target language $L$, which completes the execution of $\LStar$.

 \section{Other Types of Models}

\label{intro_sec_othertypes}

Since its publication, Angluin's seminal $\LStar$ algorithm \cite{angluin1987learning} for learning the minimal DFA accepting a given regular language has inspired numerous variations. On the one hand, authors have adjusted $\LStar$ by using more efficient data structures  \cite{kearns1994introduction, isberner2014ttt} and handling counterexamples differently \cite{rivest1993inference}. On the other hand, $\LStar$ has been extended to output models other than DFAs, often either equally expressive, but more succinct, or more expressive, but still efficiently learnable. In this section, we will focus on the latter type, and give a few examples of such target models.

Some authors have remarked that DFAs are inappropriate to capture the behaviour of many complex systems \cite{shahbaz2009inferring}. This is due to the behaviour of such systems being often naturally characterized in terms of complex input-output pairs: the system receives an input from the environment, transitions, and produces an output to the environment. DFAs lack such general input-output behaviour, being classifiers, which either output $0$ (reject) or $1$ (accept).  

A more natural model for systems exhibiting general input-output behaviour are \emph{Mealy machines} \cite{mealy1955method}, which often are also more concise than DFAs. 
	Learning algorithms for Mealy machines based on $\LStar$ have appeared in \cite{pena1998new, shahbaz2009inferring}. Practical introductions to the general development of active learning, with a focus on Mealy machines, were given in \cite{steffen2011introduction, steffen2012active}.
	
	Mealy machines are as expressive as \emph{Moore automata} \cite{moore1956gedanken}, which generalise DFA from the two-element Boolean set to an arbitrary output set. Angluin's $\LStar$ algorithm can naturally be extended to Moore automata. An optimized version for learning the products of Moore automata has been presented in \cite{moerman2019learning}.
	A passive learning algorithm for Moore automata is the subject of \cite{giantamidis2021learning}.
	
Another class of automata exhibiting complex input-output behaviour are \emph{weighted automata}, which generalise NFAs from the two-element Boolean semiring to arbitrary semirings. Weighted automata have been used in text processing \cite{mohri2008speech}, character recognition \cite{breuel2008ocropus}, image processing \cite{albert2009digital, culik1993image}, bioinformatics \cite{allauzen2008sequence}, and formal verification \cite{aminof2011formal}.
The characteristics of weighted automata over arbitrary semirings have been extensively studied \cite{salomaa2012automata, mohri2009weighted}. An active-learning algorithm for weighted automata over the field of rationals inspired by $\LStar$ appeared first in \cite{bergadano1994learning}, and has later been generalized to weighted automata over arbitrary fields \cite{bergadano1996learning}.
A passive learning algorithm is the subject of \cite{balle2012spectral}. A survey of the developments until 2015 can be found in  \cite{balle2015learning}. The active learning of weighted automata over general semirings is explored in \cite{van2020learningweighted}. 

In many situations automata dealing with infinite instead of finite structures are a more natural and realistic model of the behaviour of a system. We would like to emphasise two particular classes of such models. First, automata that allow infinite input alphabets (typically equipped with rich additional structure), while characterising words of finite length. Second, automata that require finite input alphabets, but characterise words of infinite length. 

Three examples of the former type are \emph{register automata}, \emph{symbolic automata}, and \emph{nominal automata}.

Register automata (originally called \emph{finite-memory automata} \cite{kaminski1994finite}) extend deterministic automata to infinite input alphabets by introducing a  register that enables the storage of data for future comparison \cite{d2019symbolic}. 
Active learning algorithms for register automata have been the subject of numerous publications \cite{howar2012inferring, bollig2013fresh, aarts2015learning, cassel2016active}. A review of the developments was given in \cite{isberner2014learning, aarts2014algorithms}. 

Symbolic automata are automata with transitions labelled by predicates over an input alphabet that is a Boolean algebra of possibly infinite size. The theoretical aspects of symbolic automata have been extensively studied \cite{tamm2018theoretical}. A minimisation procedure, for example, appears in \cite{d2014minimization}. The learnability of symbolic automata \emph{in the limit} is the content of \cite{fisman2022inferring}. Angluin-style learning algorithms for symbolic automata appear in \cite{maler2014learning, drews2017learning}. 

 \emph{Nominal sets} are sets that are finitely supported with respect to a group action of permutations on a countably infinite set \cite{pitts2013nominal}. Originally introduced as an alternative to set theory, they have been later rediscovered for name binding in the context of programming languages. Nominal automata appear, among others, in \cite{lasota2014automata, schroder2017nominal}. Essentially, they generalise DFA to orbit-finite nominal sets and equivariant functions. An adaptation of $\LStar$ to nominal automata is the content of  \cite{moerman2017learning,moerman2019residual}. 
 
An example of the latter type are \emph{Büchi automata}, which accept the $\omega$-regular languages of words of infinite length. Büchi automata are central to the automata-based model checking  approach to verification \cite{vardi1986automata}. As such, they have been used to describe properties of distributed systems \cite{alpern1987recognizing}, are relevant to the synthesis of reactive systems \cite{pnueli1989synthesis}, and appear in termination proofs for programs \cite{lee2001size}. The first learning algorithm for Büchi automata accepting a strict subclass of $\omega$-regular languages has been introduced by Maler and Pnueli \cite{maler1995learnability}. The first learning algorithm accepting the full class of $\omega$-regular languages has appeared in \cite{Farzan2008Extending}, and is based on ideas in \cite{angluin1987learning} and \cite{calbrix1993ultimately}. Another learning algorithm accepting the complete class of $\omega$-regular languages appeared in \cite{li2021novel, li2017novel}. Among others, it introduced a more efficient data structure \cite{kearns1994introduction,isberner2014ttt} to the work of Angluin and Fisman \cite{angluin2016learning}.

Apart of above cases, $\LStar$ has also been extended to non-deterministic finite automata \cite{bollig2009angluin}, universal finite automata \cite{angluin2015learning}, and alternating finite automata \cite{angluin2015learning, berndt2017learning}, all of which are as expressive as DFAs.

\section{Guarded Kleene Algebra with Tests}

\label{gkat-sec-intro}

\begin{figure}
\begin{subfigure}[b]{.23 \columnwidth}
\centering
							\adjustbox{valign=m}{
$(\wwhile\ b\ \ddo\ p); q$
}
	\label{}
		\caption{$e$}
\end{subfigure}
\begin{subfigure}[b]{.4 \columnwidth}
	\centering
	\footnotesize
								\adjustbox{valign=m}{
\begin{tikzpicture}[node distance=6em]
	\node[state, shape=circle, initial, initial text=] (x) {};
		\node[state, shape=circle, right of=x, label=above:{$\Rightarrow  b, \overline{b} \mid 1$}] (y) {};
	    \path[->]
	(x) edge[loop above] node{$b \mid p$} (x)
	(x) edge[above] node{$\overline{b} \mid q$} (y)
	;
	\end{tikzpicture}}
		\caption{$\mathscr{X}_e: X \rightarrow (2 + \Sigma \times X)^{\At}$}	
					\label{triangle-gkatautomaton}
\end{subfigure}
\begin{subfigure}[b]{.33 \columnwidth}
\centering
							\adjustbox{valign=m}{
$\lbrace \overline{b}qb, \overline{b}q \overline{b}, b p \overline{b} q b, b p \overline{b} q \overline{b}, ... \rbrace$}
		\caption{$\llbracket e \rrbracket = \llbracket \mathscr{X}_e \rrbracket \subseteq (\At \cdot \Sigma)^* \cdot \At$}
			\label{triangle-intro-lang}
\end{subfigure}
\caption{The interplay between expressions, automata, and languages in GKAT, for $\Sigma = \lbrace p, q \rbrace$ and $\At = \lbrace b, \overline{b} \rbrace$}
\label{trianglegkatintro_automaton}
\end{figure}

Even though there are already numerous extensions of Angluin's seminal $\LStar$ algorithm, for many different types of target models, one can still find interesting unexplored domains for its potential application. One such domain is \emph{Guarded Kleene Algebra with Tests} (GKAT), a logic that has recently become the subject of increasing interest \cite{smolka2019guarded, schmid2021guarded}. This attention stems from a few remarkable characteristics. 

First, it is based on $\emph{Kleene Algebra with Tests}$, a well-studied logic with sound mathematical foundations \cite{kozen1996kleene, kozen2001automata}. Through this relationship many constructions for it are either directly induced, or at least hinted at. Second, the behaviour of GKAT expressions $e,f,...$ can be identified with the one of simple imperative programs (cf. \Cref{gkatexpressionsasprograms} on \cpageref{gkatexpressionsasprograms}). It is, for instance, possible to iterate expressions, $e; f$, or build new expression via program-flow constructions such as $\iif\ b\ \tthen\ e\ \eelse\ f$ and $\wwhile\ b\ \ddo\ e$. Third, while remaining overall sufficiently expressive, equivalence of expressions in GKAT is more efficiently decidable than in its foundations. This makes the logic particular attractive for scale-sensitive applications such as network verification \cite{smolka2019scalable}. 

For black-box learning in the spirit of $\LStar$, GKAT is particularly well suited, because of its well-behaved interplay of expressions, automata, and languages, that closely resembles the one of regular expressions, DFAs, and regular languages. In \Cref{triangle-intro-lang}, for example, is depicted the language which models the behaviour of both the GKAT expression $(\wwhile\ b\ \ddo\ p); q$, and the GKAT automaton in \Cref{triangle-gkatautomaton}. More generally, one can, for any GKAT expression, efficiently construct a GKAT automaton that accepts the same language, and, reversely, for every automaton of a particular kind, one can find a language equivalent expression (for details see \cite{smolka2019guarded}). 

In this thesis we investigate how the ideas behind $\LStar$ can be used to derive, from program traces, behavioural equivalent GKAT automata representations, that is, simple abstractions of imperative programs. Having GKAT's potential applications in mind, we are particularly interested in domain-specific optimisations that reduce the number of involved membership queries, and ways to construct automata representations that are as compact as possible.

\section{Size-Minimality}

\label{intro_sec_minimality}

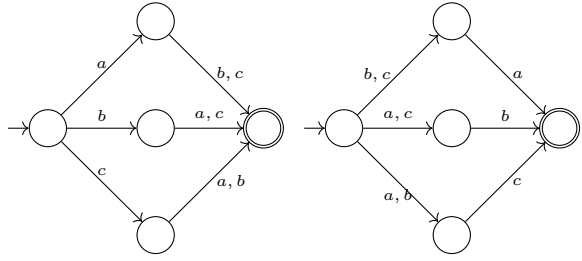
\begin{figure}
\tiny
\centering
	\begin{tikzpicture}[node distance=5.5em]
	\node[state, initial, initial text=, minimum size=1.9em] (q1) {};
	\node[state, right of=q1, above of=q1, minimum size=1.9em] (q2) {};
		\node[state, below of=q2, minimum size=1.9em] (q3) {};
		\node[state, below of=q3, minimum size=1.9em] (q4) {};
		\node[state, right of=q3, minimum size=1.9em, accepting] (q5) {};
	    \path[->]
	(q1) edge[above] node{$a$} (q2)
		(q1) edge[above] node{$b$} (q3)
		(q1) edge[above] node{$c$} (q4)
		(q2) edge[right] node{$b,c$} (q5)
		(q3) edge[above] node{$a,c$} (q5)
		(q4) edge[right] node{$a,b$} (q5)
        ;
\end{tikzpicture}
\begin{tikzpicture}[node distance=5.5em]
	\node[state, initial, initial text=, minimum size=1.9em] (q1) {};
	\node[state, right of=q1, above of=q1, minimum size=1.9em] (q2) {};
		\node[state, below of=q2, minimum size=1.9em] (q3) {};
		\node[state, below of=q3, minimum size=1.9em] (q4) {};
		\node[state, right of=q3, minimum size=1.9em, accepting] (q5) {};
	    \path[->]
	(q1) edge[left] node{$b,c$} (q2)
		(q1) edge[above] node{$a,c$} (q3)
		(q1) edge[below] node{$a,b$} (q4)
		(q2) edge[right] node{$a$} (q5)
		(q3) edge[above] node{$b$} (q5)
		(q4) edge[right] node{$c$} (q5)
        ;
\end{tikzpicture}
\caption{Two non-isomorphic size-minimal NFA accepting the language $\lbrace ab,ac,$ $ba,bc,ca,cb \rbrace \subseteq \lbrace a, b, c \rbrace^*$ \cite{arnold1992note}}
\label{example_nonisomorphic_nfa}
\end{figure}

The deterministic finite-state automaton for a given regular target language derived by $\LStar$ has a remarkable property: first, it accepts the target language, and second, every other deterministic finite-state automaton accepting the target language has either a state space of greater size, or is equivalent up to structure-preserving bijection. In other words, $\LStar$ learns a \emph{canonical} representation: the unique size-minimal DFA. 

As is well-known, the canonical representation of a regular language as DFA can be explicitly constructed, via the \emph{Myhill-Nerode relation} \cite{nerode1958linear}. Under this construction, states of the minimal DFA for a language $L \subseteq A^*$ over an input alphabet $A$ can be identified with equivalence classes of \emph{residual languages} of the type $w^{-1} L = \lbrace w \cdot u \mid u \in A^* \rbrace$, for $w \in A^*$. It is not hard to see that the observation table data-structure of $\LStar$ closely resembles this construction. Indeed, at every step, rows indexed by a word $w \in A^*$ approximate the residual language $w^{-1} L$, since their entry at a column indexed by $u \in A^*$ coincides with the evaluation $u \in w^{-1}L$. This illustrates that the design of a learning algorithm should start with the identification and explicit construction of a \emph{canonical} target model. The algorithm itself is then merely a derivation.

Unfortunately, not all classes of acceptors admit a canonical representative. For example, while there exists, up to isomorphism, precisely one size-minimal DFA accepting the regular language over the alphabet $\lbrace a, b, c \rbrace$ of words of length two, starting and ending in different characters (\Cref{minimaldfaexample}), there exist at least two size-minimal NFAs that are non-isomorphic (\Cref{example_nonisomorphic_nfa}). This immediately leads to the question: what is a canonical NFA for a given regular language? Answering this question is of great practical importance, as NFAs can be exponentially more succinct than their deterministic counterparts. In this case, the advantage is little, with the smallest DFA being of size $6$, whereas a NFA can be constructed with $5$ states. Generally, however, the difference increases with the size of the state-spaces. 
	The problem has been approached independently, for various types of acceptors with side-effects. Most approaches have in common their restriction to a subclass that admits a canonical representative. Some of the better known examples are: the \emph{\'atomaton}~\cite{BrzozowskiT14}, the \emph{canonical residual finite-state automaton} (short \emph{canonical RFSA} and also known as \emph{jiromaton})~\cite{denis2001residual}, the \emph{minimal xor automaton}~\cite{VuilleminG210}, and the \emph{distromaton}~\cite{MyersAMU15}. The canonical RFSA for the regular language $L = (a+b)^*a$ over the input alphabet $\lbrace a, b \rbrace$, for instance, is depicted in \Cref{introjiromaton}. It is minimal in the subclass of those NFAs accepting $L$, for which each state accepts a join of residuals of $L$. Its similarity to the minimal deterministic finite automaton for $L$ in \Cref{introm(l)} is striking.

	\begin{figure}[t]
	\centering
	\begin{subfigure}[b]{.45 \columnwidth}
		\footnotesize
		\center
	\begin{tikzpicture}[node distance=6em]
	\node[state, initial, initial text=] (x) {};	
		\node[state, right of=x, accepting] (y) {};
	    \path[->]
	(x) edge[loop above] node{$b$} (x)
	(y) edge[loop right] node{$a$} (y)
	(x) edge[above, bend left] node{$a$} (y)
	(y) edge[below, bend left] node{$b$} (x)
	;
	\end{tikzpicture}.
		\caption{The size-minimal DFA}
	\label{introm(l)}
		\end{subfigure}
	\begin{subfigure}[b]{.45 \columnwidth}
		\centering
				\footnotesize
							\adjustbox{valign=m}{
		\begin{tikzpicture}[node distance=6em]
			\node[state, initial, initial text=] (x) {};
				\node[state, right of=x, accepting] (y) {};	
		    \path[->]
		(x) edge[loop above] node{$a,b$} (x)
		(x) edge[above, bend left] node{$a$} (y)
		(y) edge[below, bend left] node{$a,b$} (x)
		(y) edge[loop right] node{$a$} (y)
		;
		\end{tikzpicture}
		}
	\caption{The canonical RFSA}
		\label{introjiromaton}
		\end{subfigure}
		\caption{Two canonical acceptors for $(a+b)^*a$}
\end{figure}
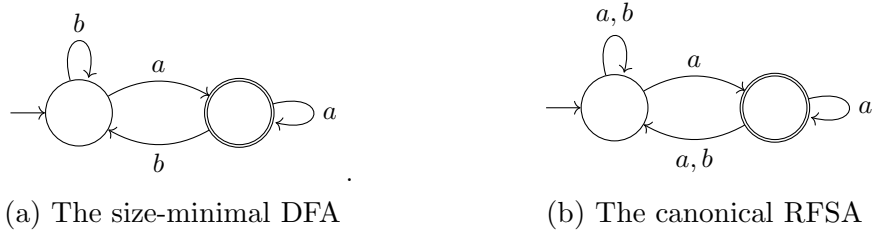
	 
	 Somewhat surprisingly, the state spaces of all the above canonical minimal representatives consist of \emph{generators}, for different algebraic structures. For example, the state-space of the \'atomaton is given by the \emph{atoms}\footnote{A non-zero element $a$ in a Boolean algebra B is called an \emph{atom}, if for all $x \in B$ with $x \leq a$ it follows $x = 0$ or $x = a$. A Boolean algebra $B$ is \emph{atomic}, if for all $x \in B$ there exists a decomposition $x = \vee_{I} a_i$, where $\lbrace a_i \mid i \in I \rbrace$ is some set of atoms.} for a complete atomic Boolean algebra, the states of the canonical RFSA are the \emph{join-irreducibles}\footnote{A non-zero element $a$ in a lattice $L$ is called \emph{join-irreducible}, if for all $y,z \in L$ with $a=y \vee z$ it follows $a = y $ or $a = z$. For any $x$ in a finite lattice $L$ there exists a decomposition $x = \vee_{I} a_i$, where $\lbrace a_i \mid i \in I \rbrace$ is some set of join-irreducibles.} of a complete semi-lattice (for example, the states in \Cref{introjiromaton} are the irreducibles of the lattice in \Cref{overlineml}), and the state-space of the minimal xor automaton consists of a \emph{basis} for a vector-space. All these subsets of algebraic structures have in common that they contain the \emph{minimal} amount of information to generate from it, by performing a closure with respect to algebraic operations, the full structure. 
For example, a subset of a vector space is called a basis for the former if every vector can be uniquely written as a finite linear combination of basis elements. Part of the importance of bases stems from the convenient consequences that follow from their existence. For instance, linear transformations between vector
spaces admit matrix representations relative to pairs of bases, which can be used for efficient numerical calculations.
	 
	 This observation immediately leads to numerous questions.	Is the size-minimality of the generators underlying the state-spaces related to the size-minimality of the models? How are the underlying algebraic structures related to the type of side-effects of the models? What is the connection to the Myhill-Nerode construction for the minimal deterministic finite automaton? Is there a general procedure for the construction of minimal representatives?  
	
	Some of above questions were answered by Myers et al. \cite{MyersAMU15}, whose approach is based on an equivalence between finite algebras in a locally finite variety and finite structured sets and relations. Their construction, however, is restricted to non-deterministic automata, and does not provide a general algorithm to construct a succinct automaton. A different unifying perspective was given by van Heerdt \cite{van2020learning, van2016master, van2020phd}. One of the central notions in van Heerdt's work is the concept of a \emph{scoop}, originally introduced by Arbib and Manes \cite{arbib1975fuzzy}, and essentially a simple \emph{category-theoretic} generalisation of algebraic generators. In this thesis, we refer to scoops as generators, and further develop their abstract theory.
		
\section{Categorical Perspective}

\label{intro_sec_categorical}

Category theory is a mathematical framework that provides a unifying birds-eye perspective on different mathematical structures. The theory's central subjects are objects and the relations between them. Relations are treated as first class objects: they are \emph{data}, rather than just a \emph{property}. Examples of categories occur in all areas of mathematics and computer science. Some of the simplest cases are the categories sets and functions, and the category of vector spaces and linear maps.  One of the attractive characteristics of category theory is that it allows simple unifying characterisations of constructions that deduce new mathematical objects from existing ones. For instance, one can show that the cartesian product is for sets and functions what the direct product is for groups and group homomorphisms. For this thesis, it will be sufficient to work with a relatively simple subset of category theory. 
We are particularly interested in \emph{algebras} and \emph{coalgebras}, and their combination into \emph{bialgebras}.

\paragraph{Algebras}

\begin{figure}
\centering
\begin{tabular}{c|c}
${\Set}^T$ & $TX$ \\
\hline
\hline
complete semi-lattices & $\lbrace f: X \rightarrow 2 \rbrace$ \\
\hline	
$\mathbb{K}$-vector spaces & $\lbrace f: X \rightarrow \mathbb{K} \mid \supp(f)\ \textnormal{is finite} \rbrace$ \\
\hline
complete atomic Boolean algebras & $\lbrace f: (X \rightarrow 2) \rightarrow 2 \rbrace$ \\
\hline 
complete distributive lattices & $\lbrace f: ((X \rightarrow 2), \subseteq) \rightarrow (2, \leq) \mid f\ \textnormal{is monotone} \rbrace$\end{tabular}
	\caption{Algebraic structures as algebras over a monad on the category of sets}
	\label{intro-algebras}
\end{figure}

	In the category theoretic approach to universal algebra, algebraic structures are typically captured as \emph{algebras} over a \emph{monad} \cite{eilenberg1965adjoint, linton1966some}. 
	
	  In the context of computer science, monads have been introduced by Moggi, as a general perspective on exceptions, side-effects, and continuations \cite{moggi1988computational, moggi1990abstract, moggi1991notions}. Intuitively, they are a categorification of \emph{closure operators}\footnote{Closure operators are monotone functions $T: P \rightarrow P$ that satisfy $x \leq T(x)$ and $T^2(x) = T(x)$ for all $x \in P$.} on partially ordered sets. A simple example of a monad $T$ on the category of sets is the $\emph{free}$ $\mathbb{K}\emph{-vector space monad}$ $\mathcal{V}_{\mathbb{K}}$, for any field $\mathbb{K}$. It assigns to a set $X$ the set $\mathcal{V}_{\mathbb{K}}(X)$ of finitely-supported\footnote{\label{support}The \emph{support} of $\varphi: X \rightarrow \mathbb{K}$ is defined by $\supp(\varphi) = \lbrace x \in X \mid \varphi(x) \not = 0 \rbrace$. If the set $\supp(f)$ is finite, then we say $f$ has \emph{finite support}, or is \emph{finitely-supported}.} functions $\varphi: X \rightarrow \mathbb{K}$; maps an element $x \in X$ to $\eta(x) \in \mathcal{V}_{\mathbb{K}}(X)$, the \emph{Dirac measure}\footnote{The Dirac measure $\eta(x): X \rightarrow \mathbb{K}$ for $x \in X$ satisfies $\eta(x)(y) = 1$, if $x = y$, and $0$ otherwise.}; and flattens $\Phi \in \mathcal{V}_{\mathbb{K}}^2(X)$ to $\mu(\Phi) \in \mathcal{V}_{\mathbb{K}}(X)$ in the usual manner: if we write $\Phi$ as formal linear combination $\sum_{\varphi} \Phi_{\varphi} \cdot \varphi$, where $\Phi_{\varphi}$ is short for $\Phi(\varphi)$, then $\mu(\Phi)(x) = \sum_{\varphi} \Phi_{\varphi} \cdot \varphi(x)$. Most algebraic theories admit a monad on the category of sets by assigning to a set the set underlying the algebraic structure it freely generates. For example, the free vector space monad above is of this type.
	
	An algebra over a monad $T$ on the category of sets consists of a set $X$ with a function $h: TX \rightarrow X$ that  \emph{interprets} elements in $TX$ in a way that is coherent with the monad structure.
	A $\mathbb{K}$-vector space, for instance, is an algebra for the free $\mathbb{K}$-vector space monad $\mathcal{V}_{\mathbb{K}}$. It is given by a set $X$ with a function $h: \mathcal{V}_{\mathbb{K}}(X) \rightarrow X$ that coherently interprets a finitely-supported function $\lambda: X \rightarrow \mathbb{K}$ (a \emph{formal} linear combination) as an \emph{actual} linear combination $h(\lambda) = \sum_x \lambda_x \cdot x \in X$ \cite{coumans2010scalars}. Any set $X$ induces a \emph{free} algebra $TX$ over a set monad, by making use of the monad structure. A brief list of monads and the theories their algebras correspond to is given in \Cref{intro-algebras}.
	
It is straightforward to see that under the above perspective a basis for a $\mathbb{K}$-vector space consists of a subset $Y \subseteq X$ and a function $d$ that assigns to a vector $x \in X$ a finitely-supported function $d(x) \in \mathcal{V}_{\mathbb{K}}(Y)$ such that $h(d(x)) = x$ for all $x \in X$ and $d(h(\lambda)) = \lambda$ for all finitely-supported functions $\lambda : Y \rightarrow \mathbb{K}$. In other words, the restriction of $h$ to finitely-supported functions with domain $Y$ is a bijection with inverse $d$, and surjectivity corresponds to the fact that the subset $Y$ generates the vector space, while injectivity captures that $Y$ does so uniquely. The concept easily generalises to algebras over arbitrary monads on arbitrary categories by making the subset relation explicit. The generality also makes formal the intuitive observation that any set is a generator for the free algebra it induces.

\begin{figure*}[t]
\centering
	\begin{tikzcd}[ ampersand replacement=\&]
		X \ar{d}[swap]{k} \ar{r}{\{-\}} \&
			\mathcal{P}(X) \ar{dl}{k^\sharp} \ar[dashed]{r}{} \&
			2^{A^*} \ar{d}{} \\
		2 \times \mathcal{P}(X)^A \ar[dashed]{rr}[below]{} \&
			\&
			2 \times (2^{A^*})^A
	\end{tikzcd}
	\qquad
			\begin{tikzcd}[ampersand replacement=\&]
			X \ar{d}[swap]{k} \ar{r}{\eta} \&
			TX \ar{dl}{k^\sharp} \ar[dashed]{r}{} \&
				\Omega \ar{d}{} \\
			FTX \ar[dashed]{rr}[below]{} \&
				\&
				F\Omega
		\end{tikzcd}
\caption{Generalised determinisation of automata with side-effects in a monad}
\label{gen-det-diagrams-intro}
\end{figure*}

\paragraph{Coalgebras}

In the category theoretic approach to state-based systems, systems are typically captured as \emph{coalgebras} over an \emph{endofunctor} \cite{jacobs2017introduction, rutten2019method, rutten2000universal}.

Endofunctors are categorifications of monotone functions $F: P \rightarrow P$ on a partially ordered set. They assign to any object $X$ an object $FX$, and to every morphism $f: X \rightarrow Y$ a morphism $Ff: FX \rightarrow FY$. In particular, every monad thus is an endofunctor. A coalgebra over an endofunctor $F$ on the category of sets consists of a set $X$ with a function $k: X \rightarrow FX$. While algebras describe the \emph{construction} of states, their dual notion, coalgebras, capture the \emph{deconstruction} of states. A simple example of an endofunctor on the category of sets is the functor $F$ that assigns to a set $X$ the set $2 \times X^A$, where $A$ is any fixed input set, and operates on functions as one expects. A coalgebra for it is simply an \emph{unpointed} (i.e. without a specified initial state) deterministic automaton: the function $k$ pairs the final state function and the transition function assigning a next state to each letter $a \in A$. In general, the definition of the functor $F$ describes the \emph{type}, or the \emph{dynamics} of a system. By varying the underlying category and the type of $F$, one can recover many different types of transition systems. 

There are many advantages to using the coalgebraic abstraction of state-based systems. Among others, it allows one to set aside irrelevant specifics of concrete instantiations, and instead work with elegant, universal properties. For instance, a central idea in the theory of systems is the notion of observable \emph{behaviour}. In the coalgebraic formalism, the semantics of a system is conveniently captured as a unique structure-preserving function into a \emph{final} coalgebra $\Omega$. For example, for the above functor $F$ with $FX = 2 \times X^A$, the final coalgebra is carried by the set of all languages $A^* \rightarrow 2$, and the final coalgebra homomorphism assigns to a state $x$ of an unpointed deterministic automaton the language it accepts, when given the initial state $x$. This simplicity has lead to, among others, coalgebra successfully serving as a framework for the generalisation of automata learning algorithms in the style of $\LStar$ \cite{jacobs2014automata, van2016master, van2020phd}. 
 
\paragraph{Bialgebras}

\begin{figure}
		\centering
							\adjustbox{valign=m}{
			\begin{tikzpicture}[node distance=6em]
				\node[state] (0) {$x_0$};
				\node[state,  right of=0, initial, initial text=] (x) {$x_1$};
					\node[state,  right of=x, accepting] (y) {$x_2$};	
			    \path[->]
			(0) edge[loop above] node{$a,b$} (0)
			(x) edge[loop above] node{$b$} (x)
			(x) edge[above, bend left] node{$a$} (y)
			(y) edge[below, bend left] node{$b$} (x)
			(y) edge[loop right] node{$a$} (y)
			;
			\end{tikzpicture}
			}
			\qquad
					\adjustbox{valign=m}{
					\tiny
			\resizebox{0.3 \columnwidth}{!}{%
				\begin{tabular}{ c|c|c|c } 
 $\vee$ & $x_1$ & $x_2$ &  $x_0$ \\
 \hline 
$x_1$ & $x_1$ & $x_2$ & $x_1$ \\ 
 \hline
 $x_2$ & $x_2$ & $x_2$ & $x_2$ \\
 \hline
 $x_0$ & $x_1$ & $x_2$ & $x_0$
\end{tabular}
}}
					\caption{The minimal CSL-structured DFA accepting $(a+b)^*a \subseteq \lbrace a, b \rbrace^*$}
		\label{overlineml-intro}
		\end{figure}
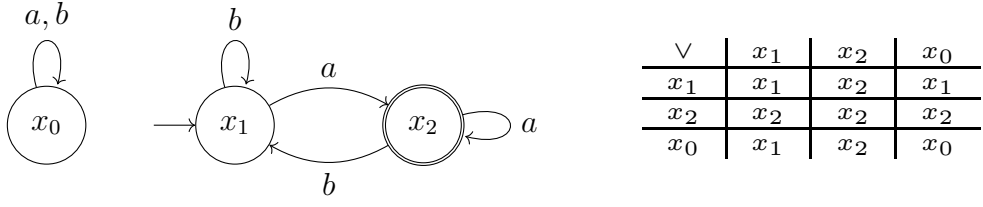

Of particular interest for us are systems that have both an algebraic and a coalgebraic component, interacting with each other in a well-defined way. 

One simple example of such a system is the unpointed deterministic automaton $k^{\sharp} : \mathcal{P}(X) \rightarrow 2 \times \mathcal{P}(X)^A$ one obtains from determinising an unpointed non-deterministic automaton $k: X \rightarrow 2 \times \mathcal{P}(X)^A$ via the classical powerset-construction\footnote{The powerset-construction assigns to an unpointed non-deterministic automaton $\langle \varepsilon, \delta \rangle: X \rightarrow 2 \times \mathcal{P}(X)^A$ the unpointed deterministic automaton $\langle \varepsilon^{\sharp}, \delta^{\sharp} \rangle:  \mathcal{P}(X) \rightarrow 2 \times \mathcal{P}(X)^A$ defined by $\varepsilon^{\sharp}(U) = \vee_{u \in U} \varepsilon(u)$ and $\delta^{\sharp}(U)(a) = \cup_{u \in U} \delta(u)(a)$.}. As one verifies, the coalgebraic structure of the lifting $k^{\sharp}$ has an additional property: it preserves the join semi-lattice structure $\cup: \mathcal{P}^2(X) \rightarrow \mathcal{P}(X)$ of its state-space. Indeed, a union $U \cup V$ is accepted if and only if either $U$ or $V$ is accepted; and reading in $a$ at a state $U \cup V$ leads to a state $(U \cup V)_a$ that can equivalently be reached by taking the union $U_a \cup V_a$ of states reached by reading in $a$ at $U$ and $V$ separately, respectively. 

As seen on the right of  \Cref{gen-det-diagrams-intro}, the classical powerset-construction is an instance of a more general procedure which is parametric in an endofunctor $F$ and a monad $T$, that interact via a \emph{distributive law}\footnote{A distributive law $\lambda$ consists of a family of functions $\lambda_X: TFX \rightarrow FTX$, one for each set $X$, satisfying certain coherence conditions.} $\lambda$ \cite{silva2010generalizing, rutten2013generalizing}. Under this perspective, a \emph{succinct} coalgebra $k: X \rightarrow FTX$ with side-effects in a monad $T$ is transformed into a deterministic coalgebra $k^{\sharp} : TX \rightarrow FTX$ that interacts well with the freely generated algebra $\mu: T^2X \rightarrow TX$, and therefore is called a $\lambda$-\emph{bialgebra}. An example of a bialgebra that intertwines the coalgebraic structure of a DFA with the algebraic structure of a complete semi-lattice (CSL) is depicted in \Cref{overlineml-intro}. Cases of succinct coalgebras are numerous, and include, aside of non-deterministic automata, probabilistic automata, nominal automata, and weighted automata. Active learning algorithms for succinct coalgebras have been studied by van Heerdt \cite{HeerdtMSS19, van2020learning, van2020learning}. Distributive laws have originally been used to compose monads \cite{beck1969distributive}, but have since been generalised in a wide range of ways \cite{Street2009}. Bialgebras occur, among others, in a category-theoretic perspective on Structural Operational Semantics (SOS) \cite{turi1997towards, klin2011bialgebras, lenisa2000distributivity, jacobs2006bialgebraic}.  

In this thesis, we study a construction that is reverse to the generalized determinisation procedure. That is, we are looking for answers to the following questions. Can we transform a given bialgebra into a language equivalent succinct coalgebra that is size-minimal among all solutions? For instance, by exploiting the additional algebraic structure of the state-space of a bialgebra through an identification of algebraic generators? Is it possible to recover canonical automata from this transformation? Potentially, by identifying minimal generators for a minimal bialgebra?

\section{Main Objectives}

\label{intro_sec_main_objectives}

The objectives of this thesis are two-fold.

On the one hand, we would like to add to the automata (learning) theory of Guarded Kleene Algebra with Tests. It still being a relatively new logic, there are numerous central but open questions we plan to address, such as the existence of a minimal acceptor for a given language. Having the logic's potential applications in mind, a main goal of ours is to investigate variations and improve the efficiency of existing black-box learning algorithms that could be applied.

On the other hand, we aim to give a category-theoretic framework that unifies existing ways to construct minimal automata of different types. Canonical models of minimal size are the targets of automata learning: the design of algorithms often closely follows their construction. We expect this abstract approach to clarify and emphasise the currently under-appreciated role of algebraic generators, and lead to the discovery of new canonical acceptors. Part of our objective is to work generally enough such that our results are of value also outside the context of automata theory.

\section{Overview and Contributions}

\label{intro_sec_contributions}

This thesis consists of five parts. Chapter 1 and 2 contain an introduction and the preliminaries, respectively. The main contributions are given in chapter 3 to 5. Their content is as follows.

\paragraph{Chapter 3}

In this chapter, we explore the automata theory of Guarded Kleene Algebra with Tests (GKAT). In particular, we present $\GLStar$ (\Cref{GlStaralgorithm}), an active learning algorithm that derives a GKAT automaton
representation of a black-box, by observing its behaviour through queries.

For any GKAT automaton, we define a second automaton, which we call its \emph{minimisation} (\Cref{minimdef}). In a series of results, we prove central properties about the minimisation of an automaton (\Cref{sizeminimal}, \Cref{minimalbisim}, \Cref{minimalunique}, \Cref{minimizationcoequation}).
We show that if $\GLStar$ is instantiated with the language accepted by a particular type of GKAT automaton, then the algorithm terminates with its minimisation in finite time (\Cref{correctnesstheorem}). We show that the semantics of GKAT automata can be reduced to the well-known semantics of Moore automata (\Cref{embeddinglanguage}, \Cref{minimalembeddingiso}).
 A complexity analysis (\Cref{complexity}) shows that it is more efficient to learn a representation with $\GLStar$ than with an existing variation of $\LStar$ for Moore automata. We implement $\GLStar$ and $\LStar$ in OCaml \cite{Ocaml} and compare their performances on example programs (\Cref{comparisongraph}).

The contributions of this chapter are based on the following publication (published at the 38th International Conference on Mathematical Foundations of Programming Semantics), of which the author of this thesis is the main author:

\fullcite{zetzsche2022guarded}

\paragraph{Chapter 4}

In this chapter, we provide a general categorical framework based on bialgebras and distributive law homomorphisms that unifies constructions of canonical non-deterministic automata and unveils new ones.

We strictly improve the expressivity of previous work \cite{HeerdtMSS19, arbib1975fuzzy} by including the \'atomaton (\Cref{atomatonexample}) and the distromaton (\Cref{distromatonexample}), which were previously excluded. While other frameworks restrict themselves to the category of sets \cite{HeerdtMSS19}, we are able to include canonical acceptors in other categories, such as the canonical \emph{nominal} RFSA (\Cref{nominalexample}). By relating vector spaces over the unique two element field with complete atomic Boolean algebras, we discover a previously unknown canonical mod-2 weighted acceptor for regular languages (\Cref{minimalxorcabaexample}). Finally, we show that every regular language satisfying a suitable property parametric in two monads admits a size-minimal succinct acceptor (\Cref{minimalitytheorem}) and establish a size comparison between different acceptors (\Cref{minimalityimplications}).

The contributions of this chapter are based on the following publication (published at the 37th International Conference on Mathematical Foundations of Programming Semantics), of which the author of this thesis is the main author:

\fullcite{zetzsche2021}

\paragraph{Chapter 5}

In this chapter, we show that the construction of canonical automata underlying our categorical framework is an instance of a more general theory.

First, we see (\Cref{closure_of_coalg}) that deriving a minimal bialgebra from a minimal coalgebra can be realized by applying a monad (\Cref{inducedmonad}) on a category of subobjects with respect to an epi-mono factorisation system (\Cref{subobject_def}). We then explore the abstract theory of generators and bases for algebras over a monad. We define a category of algebras with generators (\Cref{generatorcategory}), which, we show, is in adjunction with the category of Eilenberg-Moore algebras (\Cref{galgfreeforgetful}). We discuss products of generators and bases, and see that, under certain assumptions, the category of algebras with generators is monoidal (\Cref{monoidalproduct}). The content of \Cref{representationtheorysec} is a generalisation of the representation theory of vector spaces. In \Cref{basesforbialgebrasec} we discuss bases for bialgebras, which are algebras over a particular monad. A comparison of our ideas with an alternative approach that generalises bases as coalgebras is the content of \Cref{basesascoaglebrassec}. Signatures, equations, and finitary monads are discussed in \Cref{varietiessec}. Finally, in \Cref{finitelygeneratedsec}, we relate our work to the theory of locally finitely presentable categories.

The contributions of this chapter are based on the following publication (published at the 10th Conference on Algebra and Coalgebra in Computer Science), of which the author of this thesis is the main author:

\fullcite{zetzsche2020generators}

\chapter{Preliminaries}

In this chapter, we will review the basic mathematical tools necessary to follow the more advanced constructions in this thesis. Readers familiar with the foundations of the theories of automata and categories may want to skip these pages.

\section{Automata and Behaviour}

The main subjects of this thesis are generalisations of automata and ways to construct them by observing the behaviour of a black-box system. In this section, we briefly recall the very basic definitions of (classical) automata theory. The presentation is entirely standard. For texts that present the full theory we recommend e.g. \cite{kozen1997automata}.

We begin with the definition of a non-deterministic automaton. Whether the non-deterministic or deterministic case is more elementary, thus should be given first, is a matter of taste, since the two notions are equivalent, as is well-known.  

\begin{definition}[Non-Deterministic Automaton]
	A \emph{non-deterministic automaton} (NA) is a tuple $\mathscr{X} = (X,\Sigma, \delta, \varepsilon, x_0)$ consisting of:
	\begin{itemize}
		\item  a set $X$ called \emph{state space};
		\item  a finite non-empty set $\Sigma$ called \emph{alphabet};
		\item  a transition function $\delta: X \rightarrow \mathcal{P}(X)^{\Sigma}$;
		\item  a function $\varepsilon: X \rightarrow 2$ characterising \emph{accepting states};
		\item an \emph{initial state} $x_0 \in X$.
	\end{itemize}
	If the state space $X$ is finite, we speak of a \emph{non-deterministic finite automaton} (NFA).
\end{definition}

A non-deterministic automaton for which the set $\delta(x)(a)$ consists of only one state, for all $x \in X$ and $a \in \Sigma$, is called a \emph{deterministic automaton} (DA). A deterministic automaton with finite state space is called a \emph{deterministic finite automaton} (DFA).

Every non-deterministic automaton $\mathscr{X}$ induces a deterministic automaton $\mathscr{X}^{\sharp} = (\mathcal{P}(X), \Sigma, \delta^{\sharp}, \varepsilon^{\sharp}, \lbrace x_0 \rbrace)$, via the \emph{powerset construction}, defined by:
\[
	\delta^{\sharp}(U)(a) = \bigcup_{u \in U} \delta(u)(a), \qquad
	\varepsilon^{\sharp}(U) = \bigvee_{u \in U} \varepsilon(u).
\]

 Clearly $\mathscr{X}$ is finite if and only if $\mathscr{X}^{\sharp}$ is finite. 

We say that a state $x \in X$ in a DA \emph{transitions} to a state $y \in X$ via an input character $a \in \Sigma$, if $\delta(x)(a) = y$, in which case we write $x \overset{a}{\rightarrow} y$. A state $x \in X$ in a DA is \emph{accepting}, if $\varepsilon(x) = 1$, and \emph{rejecting} otherwise.
The transition function $\delta: X \rightarrow X^{\Sigma}$ of a DA can be inductively extended to a function $\widehat{\delta}: X \rightarrow X^{\Sigma^*}$ that operates on words as follows:
\[
\widehat{\delta}(x)(\varepsilon) = x, \qquad \widehat{\delta}(x)(av) = \widehat{\delta}(\delta(x)(a))(v).
\] 
For any $w \in \Sigma^*$, we write $\widehat{\delta}_w: X \rightarrow X$ for the function defined by $\widehat{\delta}_w(x) = \widehat{\delta}(x)(w)$.
It is not hard to see that $\widehat{\delta}$ is a right-action of the free monoid $\Sigma^*$ on the set $X$, that is, it satisfies $\widehat{\delta}_{\varepsilon} = \id_X$ and $\widehat{\delta}_{v \cdot w} = \widehat{\delta}_w \circ \widehat{\delta}_v$. A state $x \in X$ in a DA is \emph{reachable} (from the initial state $x_0$), if there exists a word $w \in \Sigma^*$, such that $\widehat{\delta}(x_0)(w) = x$. A DA is \emph{reachable}, if all of its states are reachable. Composing $\widehat{\delta}$ with the characterising function $\varepsilon$ yields a function 
\[ \llbracket - \rrbracket= \varepsilon^{\Sigma^*} \circ \widehat{\delta}
: X \rightarrow 2^{\Sigma^*}\]
 that assigns to any state $x \in X$ its \emph{behaviour}, or \emph{semantics}, $\llbracket x \rrbracket \in 2^{\Sigma^*}$. A deterministic automaton is \emph{observable}, if $\llbracket - \rrbracket$ is injective, that is, states can be distinguished by observing their behaviour. A word $w \in \Sigma^*$ is \emph{accepted} by a state $x \in X$ of a DA, if it induces a transition from $x$ to an accepting state, $\llbracket x \rrbracket(w) = 1$. The language accepted by a deterministic automaton $\mathscr{X}$ consists of the words accepted by its initial state, \[ \llbracket \mathscr{X} \rrbracket = \llbracket x_0 \rrbracket: \Sigma^* \rightarrow 2. \]
Note that we sometimes implicitly identify characteristic functions with subsets, that is, in this case, we may speak of the \emph{set} of accepted words $\llbracket \mathscr{X} \rrbracket \subseteq \Sigma^*$. We call two deterministic automata $\mathscr{X}$ and $\mathscr{Y}$ \emph{language equivalent}, if they accept the same languages, $\llbracket \mathscr{X} \rrbracket = \llbracket \mathscr{Y} \rrbracket$.

The semantics of a NA $\mathscr{X}$ is defined in terms of its determinisation, that is, $\llbracket \mathscr{X} \rrbracket = \llbracket \mathscr{X}^{\sharp} \rrbracket$. It is not hard to see that NFA and DFA accept the same class of languages: the \emph{regular} languages.

Regular languages are unique in that they have a particularly nice property: for any regular language, there exists a (uniquely defined up-to isomorphism) deterministic finite automaton that accepts it and has the smallest state-space among all deterministic finite automata with the same behaviour. To this end, we make the following definition.

\begin{definition}[Minimal Deterministic Automaton]
	The \emph{minimal} deterministic automaton for a language $L \subseteq \Sigma^*$ is the deterministic automaton $\M_L = (\Der(L), \delta, \varepsilon,L)$, where:
	\begin{itemize}
		\item $\Der(L) = \lbrace w^{-1}L \mid w \in \Sigma^* \rbrace$;
		\item $w^{-1}L = \lbrace v \in \Sigma^* \mid wv \in L \rbrace$;
		\item $\delta(w^{-1}L)(a) = (wa)^{-1}L$;
		\item $\varepsilon(w^{-1}L) =  \begin{cases}
 1 & w \in L \\
 0 & \textnormal{else}
 \end{cases}.$
	\end{itemize}
\end{definition}

The states of the minimisation are called \emph{derivatives} or \emph{residuals} (of $L$) and correspond to the equivalence classes of the \emph{Myhill-Nerode relation} $\simeq_{L} \mathord{\subseteq}\  \Sigma^* \times \Sigma^*$:
\[
w_1 \simeq_{L} w_2 :\Leftrightarrow w_1^{-1}L = w_2^{-1}L.
\]

A central result of classical automata theory says that the Myhill-Nerode relation $\simeq_{L}$ admits finitely many equivalence classes if and only if $L$ is regular.

\begin{lemma}[\cite{nerode1958linear}]
	$\M_L$ is finite if(f) $L$ is regular.
\end{lemma} 

The minimal DFA $\M_L$ for a regular language $L$ is reachable and observable, and it accepts $L$. It deserves its name since it is minimal in the following sense: For any DFA $\mathscr{X}$ accepting $L$, it holds $\vert \M_L \vert \leq \vert \mathscr{X} \vert$ with $\vert \M_L \vert = \vert \mathscr{X} \vert$ if and only if there is an isomorphism $\M_L \cong \mathscr{X}$. Note that this implies that all DFA of minimal-size accepting $L$ are isomorphic to $\M_L$. As we will see later, the existence of a uniquely defined size-minimal acceptor in this sense is somewhat special to DFA.

	\section{Categories}

In this section we give a brief introduction to the foundations of category theory. Our presentation is standard and may be skipped by readers familiar with categories, functors, natural transformations, adjoints, and monads. The scope of the presentation is limited by the applications of category theory in this thesis. Some elementary notions such as limits and Kan extensions are omitted. There are many great introductory texts that go into more depth, for instance  \cite{mac2013categories,awodey2010category,leinster2014basic}.

\paragraph{Objects and Morphisms} 	The central subjects of category theory are objects and the relations between them. Most notably, relations are not just a \emph{property}, but are treated as \emph{data}, witnessed by morphisms.

\begin{definition}[Category]
	A \emph{category} $\mathscr{C}$ consists of the following data:
	\begin{itemize}
		\item A class of \emph{objects} $A,B,C,D,..., X, Y, Z$
		\item A class of \emph{morphisms} or \emph{arrows} $f, g, h, ...$
		\item A binary operation that assigns to each morphism $f$ two objects 
			$
			\dom(f)$, and $\cod(f)
			$
			called the \emph{domain} and \emph{codomain} of $f$, respectively. The expression
			$
			f: X \rightarrow Y
			$
			indicates that $X = \dom(f)$ and $Y = \cod(f)$. 
		\item A binary operation that assigns to any two morphisms $f: X \rightarrow Y$ and $g: Y \rightarrow Z$ a morphism
			$
			g \circ f: X \rightarrow Z
			$			
			called the \emph{composition} of $f$ and $g$.
		\item For any object $X$, there is a morphism
			$
			1_X = \id_X: X \rightarrow X
			$
			called the \emph{identitiy morphism} of $X$.
	\end{itemize}
	The data is required to satisfy the following constraints:
	\begin{itemize}
		\item If $f: A \rightarrow B$, $g: B \rightarrow C$ and $h: C \rightarrow D$, then
			$
			h \circ (g \circ f) = (h \circ g) \circ f.
			$
		\item If $f: X \rightarrow Y$, then 
			$
			f \circ 1_X = f = 1_Y \circ f.
			$
	\end{itemize}
\end{definition}

The class of morphisms with domain $X$ and codomain $Y$ is denoted by $\Hom_{\mathscr{C}}(X,Y)$ or $\mathscr{C}(X,Y)$. If convenient, we omit parentheses and write $fg$ for $f \circ g$. The expressions $X \in \mathscr{C}$ and $f \in \mathscr{C}$ indicate that $X$ is an element of the class of objects of $\mathscr{C}$ and $f$ is an element of the class of morphisms of $\mathscr{C}$, respectively. 

A category is \emph{locally small}, if its class of morphisms $\mathscr{C}(X,Y)$ is a set for any choice of objects, and \emph{small}, if it is locally small and its class of objects is a set. 

The canonical example of a category is $\Set$, which has sets as objects and functions as morphisms.  Many other examples stem from algebraic theories. In this thesis, we are mainly interested in the following cases:
\begin{itemize}
	 \item The category $\mathbb{K}\Vect$ has as objects vector spaces over a field $\mathbb{K}$ and as morphisms $\mathbb{K}$-linear maps. 
 \item The category $\CSL$ has as objects complete join-semi lattices, and as morphisms functions that preserve all joins. 
	\item The category $\CABA$ has as objects complete atomic Boolean algebras, and as morphisms Boolean algebra homomorphisms that preserve all meets and joins. 
	\item The category $\CDL$ has as objects completely distributive lattices, and as morphisms functions that preserve all meets and all joins.
	\item The category $\Nom{\mathbb{A}}$ has as objects finitely supported nominal\footnote{\label{nomdef}Let $\Perm(\mathbb{A})$ be the set of \emph{permutations} on $\mathbb{A}$, i.e. the bijective functions $\pi: \mathbb{A} \rightarrow \mathbb{A}$. A nominal $\mathbb{A}$-set is a set $X$ equipped with an action of the permutation group $(\Perm(\mathbb{A}), \id_{\mathbb{A}}, \circ)$. We say that a nominal set has \emph{finite support}, if for each $x \in X$, there exists a finite set $A_x \subseteq \mathbb{A}$ such that for all $\pi \in \Perm(\mathbb{A})$ with $\pi.a = a$ for all $a \in A_x$, we have $\pi.x = x$. A function $f: X \rightarrow Y$ between nominal sets is \emph{equivariant} if $f(\pi.x) = \pi.f(x)$ for all $\pi \in \Perm(\mathbb{A}), x \in X$.} $\mathbb{A}$-sets, and as morphisms equivariant functions, for any countable set $\mathbb{A}$.
	\end{itemize}
	
	In some sense, categories are generalisations of monoids. Indeed, any monoid $M$ induces a category that consists of one object, $\star$, and a family of morphisms $(f_m: \star \rightarrow \star)_{m \in M}$, composed by $f_m f_n = f_{m n}$. The category induced by the trivial monoid is called the \emph{final category} and is denoted by $1_{\Cat}$. It consists of one object $\star$ and one morphism, the identity $1_{\star}$.	
	
Categories induce, and can be composed to, new categories. For example, the \emph{opposite} or \emph{dual} category $\mathscr{C}^{\textnormal{op}}$ of a category $\mathscr{C}$ has the same objects as $\mathscr{C}$, but an arrow $f: X \rightarrow Y$ in $\mathscr{C}^{\textnormal{op}}$ is an arrow $\tilde{f}: Y \rightarrow X$ in $\mathscr{C}$. Another example is the \emph{product category} $\mathscr{C} \times \mathscr{D}$ of categories $\mathscr{C}, \mathscr{D}$, which consists of pairs of objects and pairs of morphisms, and component-wise defined composition. 

\paragraph{Diagrammatical Proofs}	
Many proofs in this thesis are given diagrammatically, via a method called \emph{diagram chasing}. To this end, assume we are given the following objects and morphisms between them:
\[
\begin{tikzcd}
	A \arrow{r}{f} \arrow{d}[left]{i} & B \arrow{r}{g} \arrow{d}{l} & C \arrow{d}{h} \\
	D \arrow{r}[below]{j} & E \arrow{r}[below]{k} & F
\end{tikzcd}
\]
We say that the outer diagram \emph{commutes}, if $hgf = kji$. To prove that the outer diagram commutes, it is sufficient to show that the two inner diagrams commute, that is, $lf = ji$ and $hg = kl$. Diagrammatical proofs \emph{divide and conquer}: they slice larger diagrams into smaller diagrams whose commutativity is known.

\paragraph{Universal Properties}

Many constructions in category theory are given in terms of \emph{universal properties}, that define an object, if it exists, uniquely up to unique isomorphism. Below we give a few basic examples of such characterisations:

\begin{itemize}
	\item A \emph{product} of two objects $X,Y$ consists of an object $X \times Y$ and two morphisms $\pi_X: X \times Y \rightarrow X$ and $\pi_Y: X \times Y \rightarrow Y$ that satisfy the following universal property: For every object $Z$ and two morphisms $f_X: Z \rightarrow X$ and $f_Y: Z \rightarrow Y$, there exists a unique morphism $\langle f_X, f_Y \rangle: Z \rightarrow X \times Y$, such that $f_X = \pi_X \circ \langle f_X, f_Y \rangle$ and $f_Y = \pi_Y \circ \langle f_X, f_Y \rangle$:
	\[
	\begin{tikzcd}[row sep = 2.5em, column sep =5em]
		Z \arrow{r}{f_X} \arrow{d}[left]{f_Y} \arrow[dashed]{dr}{\exists! \langle f_X, f_Y \rangle} & X \\
		Y & X \times Y \arrow{u}[right]{\pi_X} \arrow{l}[below]{\pi_Y}
	\end{tikzcd}.
	\] The morphisms $\pi_X,\pi_Y$ are referred to as the \emph{projections} of the product. In the category of sets and functions, all binary products exist: they are given by the cartesian product $X \times Y = \lbrace (x,y) \mid x \in X,\ y \in Y \rbrace$ with its usual projections.
	\item An \emph{exponential} of objects $Y, Z$ in a category with binary products consists of an object $Z^Y$ and a morphism $\varepsilon: Z^Y \times Y \rightarrow Z$, called \emph{evaluation}, such that for any object $X$ and morphism $f: X \times Y \rightarrow Z$ there is a unique arrow $f^{\dagger}: X \rightarrow Z^Y$ such that $\varepsilon \circ (f^{\dagger} \times 1_Y) = f$: 
	\[
		\begin{tikzcd}
		X \times Y \arrow{d}{f} \\
		Z
	\end{tikzcd}
	\qquad
	\begin{tikzcd}
		X \arrow[dashed]{d}{\exists! f^{\dagger}} \\
		Z^Y
	\end{tikzcd}
	\qquad 
	\begin{tikzcd}
		X \times Y \arrow{dr}{f} \arrow{d}[left]{f^{\dagger} \times 1_Y}  \\
		Z^Y \times Y \arrow{r}[below]{\varepsilon} & Z
	\end{tikzcd}
	\] A category that has all finite products and exponentials is called \emph{cartesian closed}. The category of sets and functions is cartesian closed: an exponential $X^Y$ consists of all functions $f: Y \rightarrow X$, and $\varepsilon$ satisfies $\varepsilon(f)(y) = f(y)$.
	\item An object $0$ is \emph{initial}, if for any object $X$ there is a unique morphism $!_X: 0 \rightarrow X$. In the category of sets and functions, there is an initial object, the empty-set $\emptyset$.
\end{itemize}

\paragraph{Duality}

The theory of categories has built in a notion of \emph{duality}, which admits for any construction a second one, of symmetric importance. Many well-known mathematical entities can be shown to come in such duality pairs. 

One simple example of such a pair is given by the cartesian product and the disjoint union of sets. Indeed, by defining the \emph{coproduct} $X + Y$ in $\mathscr{C}$ as the product $X \times Y$ in $\mathscr{C}^{\textnormal{op}}$, one can show that in $\Set$ the coproduct of two sets is given by their disjoint union with its obvious embeddings.
 
  Another example follows from defining an object $1$ in $\mathscr{C}$ as $\emph{final}$, if it is initial in $\mathscr{C}^{\textnormal{op}}$. That is, for any object $X$, there exists a unique morphism $!_X: X \rightarrow 1$. In the category $\Set$, any singleton set $\lbrace \star \rbrace$ is a final object.

\paragraph{Functors}

At the heart of category theory lies the idea that relationships between entities are equally as important as the entities themselves. That is, for any type of object, there should exist a corresponding type of morphism that preserves their structure. The natural type of structure-preserving morphism that corresponds to categories is called a \emph{functor} and defined below.

\begin{definition}[Functor]
	A \emph{functor} $F: \mathscr{C} \rightarrow \mathscr{D}$ between categories $\mathscr{C}$ and $\mathscr{D}$ consists of the following data:
	\begin{itemize}
		\item For each object $X$ in $\mathscr{C}$ there is an object $F(X)$ in $\mathscr{D}$.
		\item For each morphism $f: X \rightarrow Y$ in $\mathscr{C}$ there is a morphism $F(f): F(X) \rightarrow F(Y)$ in $\mathscr{D}$.
	\end{itemize} 
	The data is subject to the following constraints:
	\begin{itemize}
		\item If $X$ is an object in $\mathscr{C}$, then
		$
		F(1_X) = 1_{F(X)}
		$.
		\item If $f: X \rightarrow Y$ and $g: Y \rightarrow Z$ are morphisms in $\mathscr{C}$, then
		$
		F(g \circ f) = F(g) \circ F(f)$.
	\end{itemize}
\end{definition}

If convenient, we omit parentheses and write $FX$ for $F(X)$, and $Ff$ for $F(f)$.

 For every category $\mathscr{C}$ there exists an \emph{identity functor} $1_{\mathscr{C}}: \mathscr{C} \rightarrow \mathscr{C}$, which maps objects and morphisms to themselves. The \emph{composition} of two functors $F: \mathscr{C} \rightarrow \mathscr{D}$ and $G: \mathscr{D} \rightarrow \mathscr{E}$ is the functor $G \circ F: \mathscr{C} \rightarrow \mathscr{E}$ defined by $G \circ F(X) = G(F(X))$ and $G \circ F(f) = G(F(f))$. There exists a category, denoted by $\Cat$, that has as objects small\footnote{For the same reason one cannot have a set of all sets, one can not construct a category that contains \emph{all} categories as objects. A standard way to deal with the issue is to allow only \emph{small} categories as objects. There are other ways (e.g. by using the language of higher category-theory) to avoid running into paradoxes. We will not encounter $\Cat$ again, thus refrain from elaborating.} categories and as morphisms functors between them. Many familiar categories can be recovered as universal constructions within this category.
For example, the product category $\mathscr{C} \times \mathscr{D}$ can be recognised as the product of $\mathscr{C}, \mathscr{D}$ in $\Cat$.

There are numerous examples of functors, in every corner of mathematics. Simple generic cases are given by the \emph{constant}\footnote{For any object $X \in \mathscr{D}$, the constant functor $F_X: \mathscr{C} \rightarrow \mathscr{D}$ satisfies $F_X(Y) = X$ and $F_X(f) = 1_X$. } functor $F_X: \mathscr{C} \rightarrow \mathscr{D}$ (for any $X \in \mathscr{D}$), the \emph{diagonal}\footnote{The diagonal functor $\Delta: \mathscr{C} \rightarrow \mathscr{C} \times \mathscr{C}$ is defined by $\Delta(X) = (X,X)$ and $\Delta(f) = (f,f)$.} functor $\Delta: \mathscr{C} \rightarrow \mathscr{C} \times \mathscr{C}$, and the  \emph{product}\footnote{For any category $\mathscr{C}$ with binary products, the  product functor $\Pi: \mathscr{C} \times \mathscr{C} \rightarrow \mathscr{C}$ is defined on objects by $\Pi(X,Y) = X \times Y$, and on morphisms $f: X \rightarrow X'$ and $g: Y \rightarrow Y'$ by $\Pi(f,g) = \langle f \pi_{X}, g \pi_{Y} \rangle: X \times Y \rightarrow X' \times Y'$.} functor $\Pi: \mathscr{C} \times \mathscr{C} \rightarrow \mathscr{C}$ (for any category $\mathscr{C}$ with binary products). A few more concrete examples are:
\begin{itemize}
\item The \emph{dual vector space} functor $(-)^*: (\mathbb{K}\Vect)^{\textnormal{op}} \rightarrow \mathbb{K}\Vect$ is defined on vector spaces over a field $\mathbb{K}$ by $V^* = \mathbb{K}\Vect(V, \mathbb{K})$, and on linear maps $f$ by $f^*(g) = gf$. 
\item The \emph{free vector space} functor $\mathcal{V}_{\mathbb{K}}: \Set \rightarrow  \mathbb{K}\Vect$ over a field $(\mathbb{K}, +, \cdot)$ is defined on objects by $\mathcal{V}_{\mathbb{K}}(X) = \lbrace \varphi: X \rightarrow \mathbb{K} \mid \supp(\varphi)\ \textnormal{is finite}\footref{support} \rbrace$, equipped with the vector space structure induced by $\mathbb{K}$, and on morphisms $f: X \rightarrow Y$ by $\mathcal{V}_{\mathbb{K}}(f)(\varphi)(y) = \sum_{x \in f^{-1}(y)} \varphi(x) \in \mathbb{K}$. The \emph{forgetful} functor $U: \mathbb{K}\Vect \rightarrow \Set$ maps a vector space to its set of states, and a linear map to itself, viewed as function.
\item The \emph{free complete join-semi lattice} functor $\mathcal{P}: \Set \rightarrow \CSL$ is defined on sets $X$ by $\mathcal{P}X = 2^X$, equipped with the join-semi lattice structure induced by $2$, and on morphisms $f: X \rightarrow Y$ by $2^f(\varphi)(y) = \vee_{x \in f^{-1}(y)} \varphi(x) \in 2$. Equivalently, if we identify subsets with their characteristic functions, $\mathcal{P}X = \lbrace U \mid U \subseteq X \rbrace$ and 
$\mathcal{P}f(U) = \lbrace f(u) \mid u \in U \rbrace$. As before, the \emph{forgetful} functor $U: \CSL \rightarrow \Set$ maps objects and morphisms to their underlying sets and functions.
 \end{itemize}

\paragraph{Natural Transformations}

Morphisms relate objects, functors relate categories, and \emph{natural transformations} relate functors. (Historically the interest in these notions has in fact been inverse to what the progression suggests. That is, Mac Lane had an interest in natural transformations in the context of homology that pre-dates the formal introduction of a functor \cite{mac2013categories}.)

\begin{definition}[Natural Transformation]
	A \emph{natural transformation} $\eta: F \Rightarrow G$ between functors $F, G: \mathscr{C} \rightarrow \mathscr{D}$ on categories $\mathscr{C}, \mathscr{D}$ consists of a family of morphisms $(\eta_X: FX \rightarrow GX)_{X \in \mathscr{C}}$
		in $\mathscr{D}$, subject to the following \emph{naturality} constraint: If $f: X \rightarrow Y$ is a morphism in $\mathscr{C}$, then the following diagram commutes:
		\[
			\begin{tikzcd}
		FX \arrow{d}[left]{Ff} \arrow{r}{\eta_X} & GX \arrow{d}{Gf} \\
		FY \arrow{r}[below]{\eta_Y} & GY
	\end{tikzcd}
		\]
\end{definition}
The class of natural transformations between $F$ and $G$ is denoted by $\textnormal{Nat}(F,G)$. For any functor $F$ there exists an \emph{identity transformation} $1_F \in \textnormal{Nat}(F,F)$ that is defined by $(1_F)_X = 1_{FX}$.
A \emph{natural isomorphism} is a natural transformation for which the morphism $\eta_X$ is an isomorphism for every object $X$. 

The classical example of a natural transformation is $\eta: 1_{\mathbb{K}\Vect} \Rightarrow (-)^{**}$, given component-wise as the linear map $\eta_V: V \rightarrow V^{**}$ defined by $\eta_V(x) = \ev_x: V^* \rightarrow \mathbb{K}$, where $\ev_x(f) = f(x)$. If $V$ is finite-dimensional, the embedding $\eta_V$ is an isomorphism, $V \cong V^{**}$. Since the definition of $\eta_V$ does not require the choice of a basis for $V$, it is canonical or \emph{natural}. (On the other hand, any isomorphism witnessing $V \cong V^*$ for finite-dimensional $V$ requires the choice of a basis, thus is \emph{not} natural.)

 Natural transformations compose with functors. That is, for any transformation $\eta: F \Rightarrow G$ between functors $F, G: \mathscr{C} \rightarrow \mathscr{D}$, and functor $H: \mathscr{D} \rightarrow \mathscr{E}$, there is a natural transformation $H \eta: HF \Rightarrow HG$ defined component-wise by $(H \eta)_X = H\eta_X$, and for any functor $K: \mathscr{B} \rightarrow \mathscr{C}$, there is a natural transformation $\eta_K: FK \Rightarrow GK$ defined by $(\eta_K)_X = \eta_{KX}$.
Natural transformations can also be composed with each other, both \emph{vertically}\footnote{If $\eta: F \Rightarrow G$ and $\varepsilon: G \Rightarrow H$ are natural transformation between functors $F,G,H: \mathscr{C} \rightarrow \mathscr{D}$, then their vertical composition $\varepsilon \eta: F \Rightarrow H$ is defined component-wise by $(\varepsilon \eta)_X = \varepsilon_X \circ \eta_X$.} and \emph{horizontally}\footnote{If $\eta: F \Rightarrow G$ is a natural transformation between functors $F,G: \mathscr{C} \rightarrow \mathscr{D}$ and $\varepsilon: J \Rightarrow K$ is a natural transformation between functors $J,K: \mathscr{D} \rightarrow \mathscr{E}$, then their horizontal composition $\varepsilon \star \eta: JF \Rightarrow KG$ is defined as the composition $\varepsilon \star \eta := \varepsilon_G J\eta$.}.
Vertical and horizontal composition satisfy the exchange law $(\varepsilon' \eta') \star (\varepsilon \eta) = (\varepsilon' \star \varepsilon)(\eta' \star \eta)$.

If $\mathscr{C}$ is a small category and $\mathscr{D}$ is any category, one can form the \emph{functor category} $\mathscr{D}^{\mathscr{C}}$. Its objects are functors $F,G: \mathscr{C} \rightarrow \mathscr{D}$, and morphisms are natural transformations $\eta: F \Rightarrow G$, composed vertically. A natural transformation is an isomorphism in a functor category if and only if it is a natural isomorphism.

\paragraph{Adjoints}

In many situations, two categories are not isomorphic, but still closely related to each other. Often this weaker form equivalence can be captured by a notion called \emph{adjointness}. Below, we give three different definitions, which all can be shown to be equivalent\footnote{For instance, \Cref{adjunctiondef1} and \Cref{adjunctiondef2} are equivalent via the relations 
$ \phi_{X,Y}(f) = G(f) \circ \eta_Y$ and  $\eta_Y = \phi_{FY, Y}(1_{FY})$.
Similarly, the equivalence of \Cref{adjunctiondef1} and \Cref{adjunctiondef3} follows from
$\phi^{-1}_{X,Y}(g) = \varepsilon_X \circ F(g)$ and $\varepsilon_X = \phi^{-1}_{X, GX}(1_{GX})$.} to each other. Having different formulations at hand will make it easier for us to deduce structure relevant for later purposes. 

\begin{definition}[Adjunction 1]
\label{adjunctiondef1}
We call a functor $F: \mathscr{D} \rightarrow \mathscr{C}$ \emph{left adjoint} to a functor $G: \mathscr{C} \rightarrow \mathscr{D}$, and write $F \dashv G$, if there exists a natural isomorphism

\[
\phi_{X,Y}: \mathscr{C}(FY, X) \cong \mathscr{D}(Y, GX).
\]
between functors of type $\mathscr{D}^{\textnormal{op}} \times \mathscr{C} \rightarrow \Set$.
\end{definition}

\begin{definition}[Adjunction 2]
\label{adjunctiondef2}
We call a functor $G: \mathscr{C} \rightarrow \mathscr{D}$ a \emph{right adjoint} functor, if for each object $Y$ in $\mathscr{D}$ there exists an object $FY$ in $\mathscr{C}$ and a morphism 
$
\eta_Y: Y \rightarrow GFY
$
such that for every object $X$ in $\mathscr{C}$ and every morphism $g: Y \rightarrow GX$, there exists a unique morphism $g^{\sharp}: FY \rightarrow X$ making the following diagram commute:
	\[
	\begin{tikzcd}
		GFY \arrow{r}{Gg^{\sharp}} & GX \\
		Y \arrow{u}{\eta_Y} \arrow{ur}[below]{g} & 
	\end{tikzcd}
	\]
\end{definition}
In the above situation, one can show that $F$ extends to a functor $F: \mathscr{D} \rightarrow \mathscr{C}$ by defining it on morphisms $g: Y \rightarrow X$ as $Fg := (\eta_{X} \circ g)^{\sharp}$, and $\eta$ extends to a natural transformation $\eta: 1_{\mathscr{D}} \Rightarrow GF$. The functor $F$ is called a \emph{left adjoint} to $G$.

\begin{definition}[Adjunction 3]
\label{adjunctiondef3}
We call a functor $F: \mathscr{D} \rightarrow \mathscr{C}$ a \emph{left adjoint functor}, if for each object $X \in \mathscr{C}$ there exists an object $GX$ in $\mathscr{D}$ and a morphism
$
\varepsilon_X: FGX  \rightarrow  X
$
such that for every object $Y$ in $\mathscr{D}$ and every morphism $f: FY \rightarrow X$, there exists a unique morphism $f^{\dagger}: Y \rightarrow GX$ making the following diagram commute:
	\[
		\begin{tikzcd}
		FGX \arrow{d}[left]{\varepsilon_X}  & FY  \arrow{l}[above]{Ff^{\dagger}}  \arrow{dl}[below]{f}   \\
		 X  &
	\end{tikzcd}
	\]	
\end{definition}
In the above situation, one can show that $G$ extends to a functor $G: \mathscr{C} \rightarrow \mathscr{D}$ by defining it on morphisms $f: Y \rightarrow X$ as $Gf := (f \circ \varepsilon_Y)^{\dag}$, and $\varepsilon$ extends to a natural transformation $\varepsilon: FG \Rightarrow 1_{\mathscr{C}}$. The functor $G$ is called a \emph{right adjoint} to $F$.

The natural transformations $\eta$ and $\varepsilon$ are referred to as the \emph{unit} and \emph{counit} of the adjunction, respectively. We further define the natural transformations $\mu := G \varepsilon_F$ and  $\delta := F\eta_G$. The unit $\eta$ and the counit $\varepsilon$ of an adjunction satisfy the \emph{triangle identities} $(G\varepsilon) \eta_G = 1_G$ and $\varepsilon_F (F\eta) = 1_F$.
Adjoints are unique up to isomorphism, that is, if $F \dashv G$ and $F \dashv H$, then there exists a natural isomorphism $G \cong H$. Adjunctions can be composed: if $F \dashv G$ is an adjunction between $\mathscr{C}$ and $\mathscr{D}$, and $J \dashv K$ an adjunction between $\mathscr{D}$ and $\mathscr{E}$, then one can show that $JF \dashv GK$ is an adjunction between $\mathscr{C}$ and $\mathscr{E}$. We conclude with a brief list of adjunctions between previously defined functors:

\begin{itemize}
	\item Assume $\mathscr{C}$ has all binary products. Then $\Delta \dashv  \Pi$, that is, the diagonal functor $\Delta: \mathscr{C} \rightarrow \mathscr{C} \times \mathscr{C}$ is left adjoint to the product functor $\Pi: \mathscr{C} \times \mathscr{C} \rightarrow \mathscr{C}$.
	\item Assume $\mathscr{C}$ has all binary products and exponentials. Then $- \times Y \dashv (-)^{Y}$ for any object $Y \in \mathscr{C}$, i.e. the partial product functor $- \times Y: \mathscr{C} \rightarrow \mathscr{C}$ is left adjoint to the partial exponential $(-)^Y: \mathscr{C} \rightarrow \mathscr{C}$.
	\item The free vector space functor $\mathcal{V}_{\mathbb{K}}: \Set \rightarrow \mathbb{K}\textnormal{-Vect}$ over a field $(\mathbb{K}, +, \cdot)$ is left adjoint to the forgetful functor $U:  \mathbb{K}\textnormal{-Vect} \rightarrow \Set$. The unit of the adjunction $\mathcal{V}_{\mathbb{K}} \dashv U$ is given by $\eta_X(x)(y) = \lbrack x = y \rbrack$ and $\mu$ satisfies $\mu_X(\Phi)(x) = \sum_{\varphi \in \mathcal{V}_{\mathbb{K}}(X)} \Phi(\varphi) \cdot \varphi(x)$ for $\Phi \in \mathcal{V}_{\mathbb{K}}^2(X)$\end{itemize}

\paragraph{Monads}

Given an adjunction $F \dashv G$, it is not unnatural to ask what structure the endofunctors arising as the compositions $G F$ and $F G$ can be equipped with. Two particularly interesting choices are the structures $(G F, \eta, \mu)$ and $(F G, \varepsilon, \delta)$. Axiomatising the characteristics of the former choice leads to a structure called \emph{monad}, whereas the latter choice leads to a structure called \emph{comonad}. While the two notions are dual, thus symmetric, we are particularly interested in the former case.

\begin{definition}[Monad]
\label{def:monad}
A \emph{monad} on a category $\mathscr{C}$ is a tuple $(T, \eta, \mu)$ consisting of an endofunctor $T: \mathscr{C} \rightarrow \mathscr{C}$ and natural transformations
	$
	\eta: 1_{\mathscr{C}} \Rightarrow T$ and $\mu: T^2 \Rightarrow T$
	satisfying the following two commutative diagrams:
	\begin{equation}
	\label{monadlaw}
		\begin{tikzcd}
		T^3 \arrow{r}{\mu_T} \arrow{d}[left]{T\mu} & T^2 \arrow{d}{\mu} \\
		T^2 \arrow{r}[below]{\mu} & T
	\end{tikzcd}
	\qquad
	\begin{tikzcd}
		T \arrow{r}{\eta_T} \arrow{dr}{1_T} \arrow{d}[left]{T\eta} & T^2 \arrow{d}{\mu} \\
		T^2 \arrow{r}[below]{\mu} & T
	\end{tikzcd}.
	\end{equation}
\end{definition}

By an abuse of notation we will refer to a monad by its underlying endofunctor. 

A morphism $(F, \alpha): (\mathscr{C}, S) \rightarrow (\mathscr{D}, T)$  between a monad $S$ on a category $\mathscr{C}$ and a monad $T$ on a category $\mathscr{D}$ consists of a functor $F: \mathscr{C} \rightarrow \mathscr{D}$ and a natural transformation $\alpha: TF \Rightarrow FS$ satisfying $\alpha \circ \eta^T F = F\eta^S$ and $F\mu^S \circ \alpha S \circ T \alpha = \alpha \circ \mu^T F$ \cite{STREET1972149}. 
		The composition of monad morphisms $(F, \alpha): (\mathscr{C}, S) \rightarrow (\mathscr{D}, T)$ and $(G, \beta): (\mathscr{D}, T) \rightarrow (\mathscr{E}, U)$ is the monad morphism $(GF, G\alpha \circ \beta F): (\mathscr{C}, S) \rightarrow (\mathscr{E}, U)$ \cite{STREET1972149}.

A convenient way to view a monad on $\mathscr{C}$ is as categorified monoid in the category of endofunctors on $\mathscr{C}$, equipped with functor composition as monoidal product and the identity functor $1_{\mathscr{C}}$ as unit. Under this perspective, the constraints \eqref{monadlaw} are precisely the associativity and unit laws of a monoid, respectively.

We continue with a list of monads that will be relevant in the course of this thesis:

\begin{examples}
\label{exampleofmonads}	
\begin{itemize}
\item The \emph{powerset} monad $\mathcal{P}$ on $\Set$ assigns to a set $X$ the set  
	$\mathcal{P}X = X \rightarrow 2$, and to a function $f$ the function $\mathcal{P}f$ defined by $\mathcal{P}f(\varphi)(y) = \vee_{x \in f^{-1}(y)} \varphi(x)$. Its unit satisfies $\eta_X(x)(y) = \lbrack x = y \rbrack$, where $\lbrack x = y \rbrack = 1$, if $x = y$, and $0$ otherwise, and its multiplication $\mu_X(\Phi)(x) = \vee_{\varphi \in 2^X} \Phi(\varphi) \wedge \varphi(x)$. Equivalently, if one identifies characteristic functions with subsets, $\mathcal{P}X = \lbrace U \mid U \subseteq X \rbrace$, $\mathcal{P}f(U) = \lbrace f(u) \mid u \in U \rbrace$, $\eta_X(x) = \lbrace x \rbrace$, and $\mu_X(\Phi) = \cup_{U \in \Phi} U$.
	\item The \emph{free vector space} monad $\mathcal{V}_{\mathbb{K}}$ over a field $(\mathbb{K}, +, \cdot)$ on $\Set$ is defined by $\mathcal{V}_{\mathbb{K}}(X) = \lbrace f: X \rightarrow \mathbb{K} \mid \supp(f)\ \textnormal{is finite} \rbrace$ and  $\mathcal{V}_{\mathbb{K}}(\varphi)(y) = \sum_{x \in f^{-1}(y)} \varphi(x)$. Its unit is given by $\eta_X(x)(y) = \lbrack x = y \rbrack$ and its multiplication by $\mu_X(\Phi)(x) = \sum_{\varphi \in \mathcal{V}_{\mathbb{K}}(X)} \Phi(\varphi) \cdot \varphi(x)$ for $\Phi \in \mathcal{V}_{\mathbb{K}}^2(X)$.
The free vector space monad over the unique two element field $(\mathbb{Z}_2, \oplus, \wedge)$ -- with the \emph{exclusive} disjunction $\oplus$  as addition, and the normal conjunction $\wedge$ as multiplication -- is denoted by $\mathcal{R}$.
	\item The \emph{neighbourhood} monad $\mathcal{H}$ on $\Set$ assigns to a set $X$ the double dual function space
	$\mathcal{H}X = (X \rightarrow 2) \rightarrow 2$, and to a function $f$ the function $\mathcal{H}f$ defined by $\mathcal{H}f(\Phi)(\varphi) = \Phi(\varphi \circ f)$. Its unit is given by $\eta_X^{\mathcal{H}}(x)(\varphi) = \varphi(x)$, and its multiplication satisfies $\mu^{\mathcal{H}}_X(\Psi)(\varphi) = \Psi(\eta^{\mathcal{H}}_{2^X}(\varphi))$.
	\item  The \emph{monotone neighbourhood} monad $\mathcal{A}$ on $\Set$ assigns to a set $X$ the set of monotone functions $\mathcal{A}X =  (X \rightarrow 2, \subseteq) \rightarrow (2, \leq)$, and otherwise coincides with the neighbourhood monad.
	\item The \emph{nominal powerset} monad $\mathcal{P}_{\textnormal{n}}$ on $\Nom{\mathbb{A}}$ assigns to a nominal set $X$ the nominal set $\mathcal{P}_{\textnormal{n}}X = \lbrace B \subseteq X \mid B \textnormal{ finitely supported} \rbrace$, where $\pi.B := \lbrace \pi.b \mid b \in B \rbrace$, and otherwise coincides with the classical powerset monad $\mathcal{P}$ \cite{pitts2013nominal}.
\end{itemize}
\end{examples}

\paragraph{Algebras}

Every adjunction yields a monad. As it turns out, the inverse holds too. That is, given a monad, there is always an adjunction that yields it. In fact, there often is more than one adjunction giving rise to a monad. Below we construct the extreme solutions, using the Eilenberg-Moore and Kleisli categories.

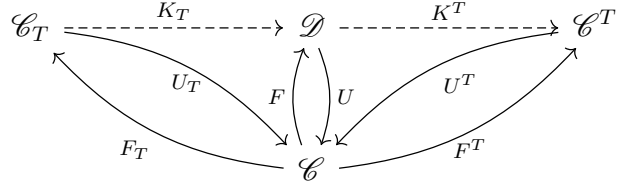
\begin{figure}[t]
\centering
\begin{tikzcd}[column sep=7em, row sep =3em]
\KL  \arrow[dashed]{r}{K_T} \arrow[bend left=20]{dr}[below]{U_T} & \mathscr{D} \arrow[bend left=20]{d}{U}  \arrow[dashed]{r}{K^T} & \EM  \arrow[bend right=20]{dl}{U^T} \\
& \mathscr{C} \arrow[bend right=20]{ur}[below]{F^T} \arrow[bend left=20]{u}{F} \arrow[bend left=20]{ul}{F_T} &	
\end{tikzcd}	
\caption{The category of adjunctions for a monad $T$ on $\mathscr{C}$}
\label{comparisonfunctors}
\end{figure}

\begin{definition}[Eilenberg-Moore Category]
\label{def:algebraovermonad}
	An \emph{algebra over a monad} $T$ on $\mathscr{C}$ is a pair $(X, h)$ consisting of an object $X$ and a morphism $h: TX \rightarrow X$ in $\mathscr{C}$ satisfying the following two commutative diagrams:
	\[
	\begin{tikzcd}
	T^2X \arrow{r}{Th} \arrow{d}[left]{\mu_X} & TX \arrow{d}{h} \\
	TX \arrow{r}[below]{h} & X	
	\end{tikzcd}
	\qquad
	\begin{tikzcd}
	X \arrow{d}[left]{\eta_X} \arrow{dr}{\id_X} & \\
	TX \arrow{r}[below]{h} & X
	\end{tikzcd}.
	\]
A homomorphism $f: (X, h_X) \rightarrow (Y, h_Y)$ between $T$-algebras is a morphism $f: X \rightarrow Y$ such that $h_Y \circ Tf = f \circ h_X$. The \emph{Eilenberg-Moore category} consists of $T$-algebras and $T$-algebra homomorphisms and is denoted by $\EM$.
\end{definition}

The categories $\mathscr{C}$ and  $\EM$ can be related as follows. On the one hand, there is the forgetful functor $U^T: \EM \rightarrow \mathscr{C}$, defined by $U^T(X,h) = X$ and $U^T(f) = f$. On the other hand, there is the free algebra functor $F^T: \mathscr{C} \rightarrow \EM$, which satisfies $F^T(X) = (TX, \mu_X)$ and $F^T(f) = Tf$. One can show that the pair satisfies an adjoint relation $F^T \dashv U^T$ that gives rise to $T$, and is final among all adjoint pairs inducing $T$. That is, for any adjunction $F \dashv U$ between functors with induced monad $T = UF$, there exists a comparison functor $K^T:  \mathscr{D} \rightarrow \EM $ that has the property $U^T K^T = U$ and $K^T F = F^T$. One can show that $K^T$ is the unique functor with this property. An adjunction is \emph{monadic}, if its comparison functor is an isomorphism.

The adjoint pair at the other end of the spectrum can be constructed as follows.
\begin{definition}[Kleisli Category]
The \emph{Kleisli category} for a monad $T$ on $\mathscr{C}$, denoted by $\KL$, has the same objects as $\mathscr{C}$; a morphism $f: X \nrightarrow Y$ in $\KL$ is a morphism $f: X \rightarrow TY$ in $\mathscr{C}$; and the composition of $f: X \nrightarrow Y$ and $g: Y \nrightarrow Z$ in $\KL$ is the composition $\mu_Z \circ Tg \circ f: X \rightarrow TZ$ in $\mathscr{C}$.	
\end{definition}

  Analogously as before, the categories $\mathscr{C}$ and $\KL$ can be related via two functors. On the one hand, there is the functor $U_T: \KL \rightarrow \mathscr{C}$, defined by $U_T(X) = TX$ and $U_T(f: X \nrightarrow Y) = \mu_Y \circ Tf$. On the other hand, there is the functor $F_T: \mathscr{C} \rightarrow \KL$, defined by $F_T(X) = X$ and $F_T(f: X \rightarrow Y) = \eta_Y \circ f$. The two functors satisfy an adjoint relation $F_T \dashv U_T$ that gives rise to the monad $T = U_TF_T$, and for any adjunction $F \dashv U$ with induced monad $T = UF$, there exists a comparison functor $K_T: \KL \rightarrow \mathscr{D}$ that is uniquely defined by the property $U K_T = U_T$ and $K_T F_T = F$. 
 
 \Cref{comparisonfunctors} summarises the initiality of the Kleisli category $\KL$ and the finality of the Eilenberg-Moore category $\EM$ among adjoint pairs giving rise to $T$.

Often, the Eilenberg-Moore category of algebras over a monad can be recognised as a category of algebras over some theory in the sense of universal algebra. (This observation can be formalised, see e.g. \Cref{varietiessec}.) Below, we list the identifications that are of particular interest for this thesis:

\begin{example}
\label{eilenbergexample}
\begin{itemize}
	\item The category $\Set^{\mathcal{P}}$ is isomorphic to $\CSL$, the category of complete join-semi lattices and functions that preserve all joins (see e.g. \cite{jacobs2011bases}).
	\item The category $\Set^{\mathcal{H}}$ is isomorphic to $\CABA$, the category of complete atomic Boolean algebras and Boolean algebra homomorphisms that preserve all meets and all joins (see e.g. \cite{jacobs2015recipe}).
	\item The category $\Set^{\mathcal{A}}$ is isomorphic to $\CDL$, the category of completely distributive lattices and functions that preserve all meets and all joins (see e.g. \cite{jacobs2015recipe}).
	\item The category $\Set^{\mathcal{R}}$ is isomorphic to $\mathbb{Z}_2\Vect$, the category of vector spaces over the unique two element field and linear maps (see e.g. \cite{jacobs2011bases}).
\end{itemize}
\end{example}

\paragraph{Coalgebras}

We conclude with a few words about \emph{coalgebras}. As the name suggests, coalgebras are dual to algebras. However, one typically refrains from requiring the dual of the Eilenberg-Moore laws, and instead simply works with coalgebras over an endofunctor (instead of a comonad). Over the last few years, the theory of coalgebras has become increasingly popular as a unifying framework for the study of infinite data types and state-based systems \cite{rutten2000universal}.

\begin{definition}[Coalgebra]
\label{def:coalgebra}
	A \emph{coalgebra} for an endofunctor $F$ on a category $\mathscr{C}$ is a pair $(X, k)$ consisting of an object $X$ and a morphism $k\colon X \rightarrow FX$ in $\mathscr{C}$. 
	\end{definition}

If $(X,k)$ is a $F$-coalgebra and $x \in X$, we call the tuple $(X,k,x)$ either a $x$-\emph{pointed} $F$-coalgebra, or an $F$-\emph{automaton}.

One of the most basic examples of coalgebras in the category of sets and functions are unpointed deterministic automata: they are of the type $k \colon X \to FX$, where $FX = 2 \times X^A$ and $k$ pairs the final state function and the transition function assigning a next state to each letter $a\in A$.

Crucial in the theory of coalgebras is the notion of homomorphism, which allows to relate states of coalgebras of the same behaviour. A homomorphism $f\colon (X, k_X)  \rightarrow (Y,k_Y)$ between $F$-coalgebras is a morphism $f\colon X \rightarrow Y$ satisfying  $k_Y \circ f = Ff \circ k_X$. The category of $F$-coalgebras and homomorphisms is denoted by $\Coalg(F)$.

 If it exists, the final object of this category is of particular importance.

\begin{definition}[Final Coalgebra]
\label{def:finalcoalgebra}
	An $F$-coalgebra $(\Omega, k_{\Omega})$ is \emph{final} if every $F$-coalgebra $(X, k)$ admits a unique homomorphism $\obs_{(X, k)}: (X, k) \rightarrow (\Omega, k_{\Omega})$.
\end{definition}

 The unique final coalgebra homomorphism can be understood as the observable behaviour of a system. For example, for the functor $FX = 2 \times X^A$, the final $F$-coalgebra is the set of all languages $\mathcal{P}(A^\star)$ and the final coalgebra homomorphism assigns to a state $x$ of an unpointed deterministic automaton the language in $\mathcal{P}(A^*)$ it accepts\footnote{For a deterministic automaton given by $\varepsilon: X \rightarrow 2$ and $\delta: X \rightarrow X^A$, acceptance is coinductively defined as a function $\obs: X \rightarrow 2^{A^*}$ by $\obs(x)(\varepsilon) = \varepsilon(x)$ and $\obs(x)(av) = \obs(\delta(x)(a))(v)$.} when given the initial state $x$. More generally, for any $F$ with $FX = B \times X^A$, the final $F$-coalgebra exists. Its underlying state-space is the set of generalised languages $A^* \rightarrow B$. We say that a $x$-pointed $F$-coalgebra $(X,k,x)$ \emph{accepts} the generalised language $\obs_{(X, k)}(x) \in B^{A^*}$. 
 
 In the course of this thesis we will encounter situations that require us to to deploy simultaneously both an algebraic and a coalgebraic perspective.

\chapter{Learning Guarded Programs}

\label{gkatsection}

Guarded Kleene Algebra with Tests (GKAT) is the fragment of Kleene Algebra with Tests (KAT) that arises by replacing the union and iteration operations of KAT with predicate-guarded variants. GKAT is more efficiently decidable than KAT and expressive enough to model simple imperative programs, making it attractive for applications to e.g. network verification. In this chapter, we further explore GKAT's automata theory, and present $\GLStar$, an algorithm for learning the GKAT automaton representation of a black-box, by observing its behaviour. A complexity analysis shows that it is more efficient to learn a representation of a GKAT program with $\GLStar$ than with Angluin's existing $\LStar$ algorithm. We implement $\GLStar$ and $\LStar$ in OCaml and compare their performances on example programs.

\section{Introduction}

As hardware and software systems grow in size and complexity, practical and scalable methods for verification tasks become increasingly important. Classical model checking approaches to verification require a rich model of the system of interest, able to express all its relevant behaviour. In reality such a model however is rarely available, for instance, when the system comes in the form of a black-box with no access to the source code, or the system is simply too complex for manual processing. 

\emph{Automata learning}, or regular inference, aims to automatically infer an automata model by observing the behaviour of the system. The incremental approach has been successfully applied to a wide range of verification tasks from finding bugs in network protocols [21], reverse engineering smartcard reader for internet banking \cite{chalupar2014automated}, and industrial applications \cite{hagerer2002model}. A comprehensive survey of the field can be found in \cite{vaandrager2017model}. The majority of modern learning algorithms is based on Angluin's $\LStar$ algorithm \cite{angluin1987learning}, which learns the unique minimal deterministic finite automaton (DFA) accepting a given regular language, or more generally, the unique minimal Moore automaton accepting a weighted language (\Cref{LStaralgorithm}). In many situations, however, targeting a DFA is not feasible, due to an explosion in the size of the state-space. Such cases instead require types of models specifically tailored for their domain-specific purposes.

For instance, modern networking systems can operate on very large data sets, making them very challenging to model. As a result, controlling, reasoning about, or extending networks can be surprisingly difficult. One approach to modernise the field that has recently gained popularity is \textit{Software Defined Networking} (SDN) \cite{feamster2014road}. Modern SDN programming languages, notably \emph{NetKAT} \cite{anderson2014netkat}, allow operators to model their network and dynamically fine tune forwarding behaviour in response to events such as traffic shifts.
Globally, NetKAT is based on \emph{Kleene Algebra} (KA) \cite{kozen1994completeness}, the sound and complete theory of regular expressions \cite{kleene1951representation}. Locally, it incorporates \emph{Boolean algebra}, the theory of predicates. Both logics have been unified in the well developed theory of \textit{Kleene Algebra with Tests} (KAT) \cite{kozen1997kleene}, which subsumes propositional Hoare logic and can be used to model standard imperative programming constructs. The automata theory for NetKAT has been introduced in \cite{foster2015coalgebraic}. 

Verifying properties about realistic networks reduces in NetKAT to deciding the behavioural equivalence of pairs of automata. Unfortunately, NetKAT's decision procedure is PSPACE-complete, mainly due its foundations in KAT.   
As a consequence, more efficiently decidable fragments of KAT have been considered. In \cite{smolka2019scalable} it was hinted that the \textit{guarded fragment} of KAT is notably more efficiently decidable than the full language, while still remaining sufficiently expressive for networking purposes. 
The idea has been taken further in \cite{smolka2019guarded}, which formally introduced \emph{Guarded Kleene Algebra with Tests} (GKAT), a variation on KAT that arises by replacing the union and iteration operations from KAT with guarded variants. In contrast to KAT, the equational theory of GKAT is decidable in (almost) linear time. These properties make GKAT a promising candidate for the foundations of a SDN programming language that is more efficiently decidable than NetKAT.

\begin{algorithm}[t]
	\begin{algorithmic}
		\State {$S,E \gets \lbrace \varepsilon \rbrace$}
		\Repeat 
				\While{$T = (S, E, \row: S \cup S \cdot A \rightarrow B^E)$ is not closed}	
					\State{find $t \in S \cdot A$ with $\row(t) \neq \row(s)$ for all $s \in S$}
					\State{$S \gets S \cup \lbrace t \rbrace $}
				\EndWhile
				\State {construct and submit $m(T)$ to the teacher}
				\If { the teacher replies \emph{no} with a counterexample $z \in A^*$ }
					\State { $E \gets E \cup \suff(z)$ }
				\EndIf
		\Until {the teacher replies \emph{yes}} \\
		\Return $m(T)$
	\end{algorithmic}
\caption{Angluin's $\LStar$ algorithm for Moore automata with input alphabet $A$ and output alphabet $B$}
\label{LStaralgorithm}
\end{algorithm}

In view of the potential applications of GKAT to the field of verification, this chapter further investigates its automata theory. In detail, we make the following contributions:
\begin{itemize}
	\item For any GKAT automaton, we define a second automaton, which we call its minimisation (\Cref{minimdef}). We show that in the class of \emph{normal} automata, the minimisation of an automaton is the unique size-minimal normal automaton accepting the same language (\Cref{sizeminimal}). We show that the minimisation of a normal automaton is isomorphic to the automaton that arises by identifying semantically equivalent pairs among reachable states (\Cref{minimalbisim}), and that the minimisations of two language equivalent normal automata are isomorphic (\Cref{minimalunique}). Finally, we show that minimising a normal GKAT automaton preserves important invariants such as the nesting coequation (\Cref{minimizationcoequation}). 
	\item We present $\GLStar$, an active-learning algorithm (\Cref{GlStaralgorithm}) that incrementally infers a GKAT automaton from a black-box by querying an \emph{oracle} (\Cref{learningpart}). We show that if the oracle is instantiated with the language accepted by a finite normal GKAT automaton, then the algorithm terminates with its minimisation in finite time (\Cref{correctnesstheorem}).
	\item We show that the semantics of GKAT automata \eqref{gkatautomatasemantics} can be reduced to the well-known semantics\footnote{In the language of Coalgebra, the semantics is given by the final coalgebra homomorphism for the functor defined by $FX = B \times X^A$, where $A = \At \cdot \Sigma = \lbrace \alpha \cdot p \mid \alpha \in \At,\ p \in \Sigma \rbrace$ and $B = 2^{\At}$, for finite sets $\Sigma$ and $\At$. The carrier of the final coalgebra for $F$ is $\mathcal{P}((\At \cdot \Sigma)^* \cdot \At)$, the set of \emph{guarded string languages}; the semantics of GKAT automata is given by the subclass of \emph{deterministic} guarded string languages.} of Moore automata \eqref{Mooreautomataacceptance}. That is, there exists a language preserving embedding of GKAT automata into Moore automata (\Cref{embeddinglanguage}), which maps the minimisation of a normal GKAT automaton to the language equivalent minimal Moore automaton (\Cref{minimalembeddingiso}). In consequence, GKAT programs could thus, in principle, be also represented by Moore automata, instead of GKAT automata.
	\item We present a complexity analysis which shows that for GKAT programs it is more efficient to learn a GKAT automaton representation with $\GLStar$ than a Moore automaton representation with the $\LStar$ algorithm (\Cref{complexity}). We implement $\GLStar$ and $\LStar$ in OCaml \cite{Ocaml} and compare their performances on example programs (\Cref{comparisongraph}).
	\end{itemize}

\section{Overview of the Approach}

In this section, we give an overview of this chapter through examples. We begin by presenting \Cref{LStaralgorithm}, a slight variation of Angluin's $\LStar$ algorithm for finite Moore automata. We exemplify the algorithm by executing it for the language semantics of a simple GKAT program. We then propose a new algorithm, which, instead of a Moore automaton, infers a GKAT automaton.
\subsection{$\LStar$ Algorithm}

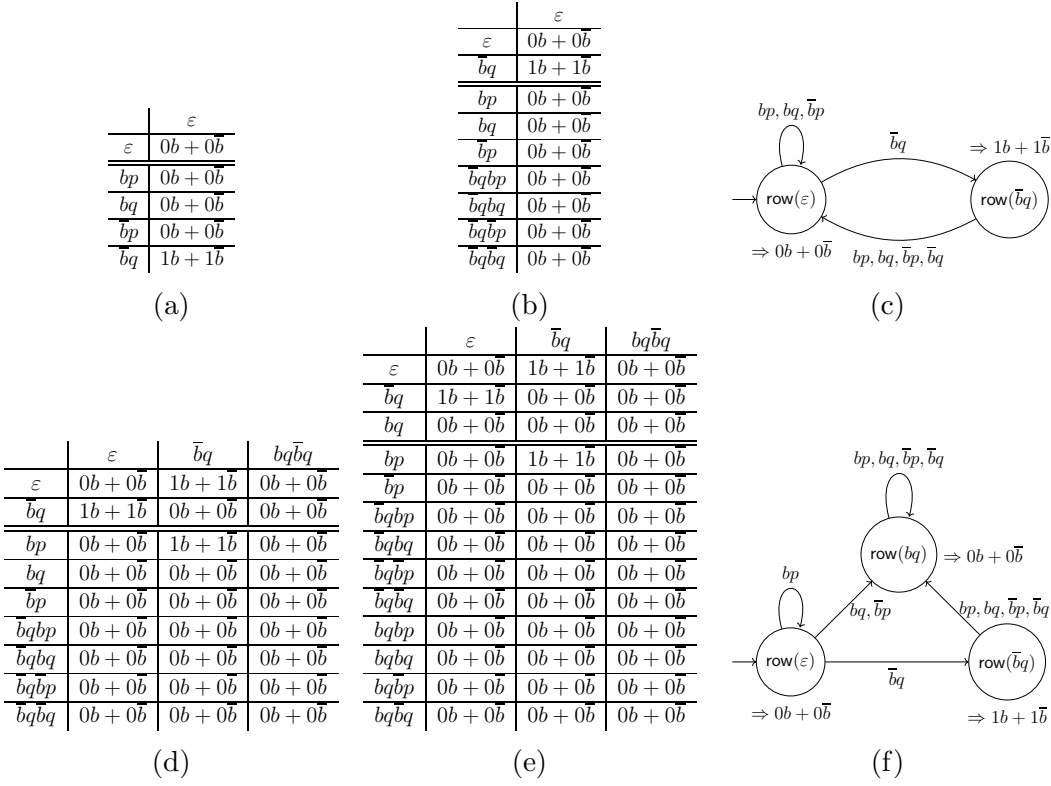
\begin{figure*}[t]
\centering
\begin{subfigure}[b]{.3\textwidth}
	\centering
		\resizebox{0.4 \textwidth}{!}{
		\begin{tabular}{ c|c}
		 & $\varepsilon$ \\
		 \hline $\varepsilon$ & $0b + 0\overline{b}$   \\
		 \hline
		 \hline $b p$ & $0b + 0\overline{b}$  \\
		 \hline $b q$ & $0b + 0\overline{b}$  \\
		 \hline $\overline{b} p$ & $0b + 0\overline{b}$  \\
		 \hline $\overline{b} q$ & $1b + 1\overline{b}$ 
	\end{tabular}
	}
	\caption{}
		\label{LStarT1}
	\end{subfigure}
	\begin{subfigure}[b]{.3\textwidth}
	\centering
		\resizebox{0.45 \textwidth}{!}{
		\begin{tabular}{ c|c}
		 & $\varepsilon$  \\
		 \hline $\varepsilon$ & $0b + 0\overline{b}$   \\
		 \hline $\overline{b} q$ & $1b + 1\overline{b}$  \\
		 \hline
		 \hline $b p$ & $0b + 0\overline{b}$  \\
		 \hline $b q$ & $0b + 0\overline{b}$  \\
		 \hline $\overline{b} p$ & $0b + 0\overline{b}$  \\
		 \hline $\overline{b} q b p$ & $0b + 0\overline{b}$  \\
		 \hline $\overline{b} q b q$ & $0b + 0\overline{b}$  \\
		 \hline $\overline{b} q \overline{b} p$ & $0b + 0\overline{b}$  \\
		 \hline $\overline{b} q \overline{b} q$ & $0b + 0\overline{b}$  \\		 
	\end{tabular}
	}
	\caption{}
			\label{LStarT2}
	\end{subfigure}
	\begin{subfigure}[b]{.3\textwidth}
	\centering
			\resizebox{\textwidth}{!}{
				\begin{tikzpicture}[node distance=12em]
	\node[state, shape=circle, initial, initial text=, label=below:{$\Rightarrow 0b + 0\overline{b}$}] (x) {$\row(\varepsilon)$};
		\node[state, shape=circle, right of=x, label=above:{$\Rightarrow  1b + 1\overline{b}$}] (y) {$\row(\overline{b}q)$};
	    \path[->]
	(x) edge[above, bend left] node{$\overline{b}q$} (y)
	(x) edge[loop above] node{$bp, bq, \overline{b}p$} (x)
	(y) edge[below, bend left] node{$bp, bq, \overline{b}p, \overline{b}q$} (x)
	;
	\end{tikzpicture}
	}
	\caption{}
			\label{LStarmT2}
\end{subfigure}
\\ 
	\begin{subfigure}[b]{.3\textwidth}
	\centering
	\resizebox{ \textwidth}{!}{
		\begin{tabular}{ c|c|c|c }
		 & $\varepsilon$ & $\overline{b} q$ & $bq\overline{b}q$   \\
		 \hline $\varepsilon$ & $0b + 0\overline{b}$ & $1b + 1\overline{b}$ & $0b + 0\overline{b}$   \\
		 \hline $\overline{b} q$ & $1b + 1\overline{b}$ & $0b + 0\overline{b}$ & $0b + 0\overline{b}$ \\
		 \hline
		 \hline $b p$ & $0b + 0\overline{b}$ & $1b + 1\overline{b}$ & $0b + 0\overline{b}$  \\
		 \hline $b q$ & $0b + 0\overline{b}$ & $0b + 0\overline{b}$ & $0b + 0\overline{b}$ \\
		 \hline $\overline{b} p$ & $0b + 0\overline{b}$ & $0b + 0\overline{b}$ & $0b + 0\overline{b}$  \\
		 \hline $\overline{b} q b p$ & $0b + 0\overline{b}$ & $0b + 0\overline{b}$ & $0b + 0\overline{b}$ \\
		 \hline $\overline{b} q b q$ & $0b + 0\overline{b}$ & $0b + 0\overline{b}$ & $0b + 0\overline{b}$ \\
		 \hline $\overline{b} q \overline{b} p$ & $0b + 0\overline{b}$ & $0b + 0\overline{b}$ & $0b + 0\overline{b}$  \\
		 \hline $\overline{b} q \overline{b} q$ & $0b + 0\overline{b}$ & $0b + 0\overline{b}$ & $0b + 0\overline{b}$  \\		 
	\end{tabular}
	}
	\caption{}
				\label{LStarT3}
	\end{subfigure}
	\begin{subfigure}[b]{.3\textwidth}
	\centering
	\resizebox{ \textwidth}{!}{
			\begin{tabular}{ c|c|c|c }
		 & $\varepsilon$ & $\overline{b} q$ & $bq\overline{b}q$   \\
		 \hline $\varepsilon$ & $0b + 0\overline{b}$ & $1b + 1\overline{b}$ & $0b + 0\overline{b}$   \\
		 \hline $\overline{b} q$ & $1b + 1\overline{b}$ & $0b + 0\overline{b}$ & $0b + 0\overline{b}$ \\
		 \hline $b q$ & $0b + 0\overline{b}$ & $0b + 0\overline{b}$ & $0b + 0\overline{b}$ \\
		 \hline
		 \hline $b p$ & $0b + 0\overline{b}$ & $1b + 1\overline{b}$ & $0b + 0\overline{b}$  \\
		 \hline $\overline{b} p$ & $0b + 0\overline{b}$ & $0b + 0\overline{b}$ & $0b + 0\overline{b}$  \\
		 \hline $\overline{b} q b p$ & $0b + 0\overline{b}$ & $0b + 0\overline{b}$ & $0b + 0\overline{b}$ \\
		 \hline $\overline{b} q b q$ & $0b + 0\overline{b}$ & $0b + 0\overline{b}$ & $0b + 0\overline{b}$ \\
		 \hline $\overline{b} q \overline{b} p$ & $0b + 0\overline{b}$ & $0b + 0\overline{b}$ & $0b + 0\overline{b}$  \\
		 \hline $\overline{b} q \overline{b} q$ & $0b + 0\overline{b}$ & $0b + 0\overline{b}$ & $0b + 0\overline{b}$  \\	
		 \hline $b q b p$ & $0b + 0\overline{b}$ & $0b + 0\overline{b}$ & $0b + 0\overline{b}$ \\
		 \hline $b q b q$ & $0b + 0\overline{b}$ & $0b + 0\overline{b}$ & $0b + 0\overline{b}$ \\
		 \hline $b q \overline{b} p$ & $0b + 0\overline{b}$ & $0b + 0\overline{b}$ & $0b + 0\overline{b}$  \\
		 \hline $b q \overline{b} q$ & $0b + 0\overline{b}$ & $0b + 0\overline{b}$ & $0b + 0\overline{b}$	 
	\end{tabular}	
	}
	\caption{}
					\label{LStarT4}
	\end{subfigure}
	\begin{subfigure}[b]{.3\textwidth}
	\centering
\label{minimalmoore}
			\resizebox{\textwidth}{!}{
\begin{tikzpicture}[node distance=6em]
	\node[state, shape=circle, initial, initial text=,  label=below:{$\Rightarrow 0b + 0\overline{b}$}] (x) {$\row(\varepsilon)$};
	\node[state, shape=circle, right of=x, above of=x, label=right:{$\Rightarrow  0b + 0\overline{b}$}] (y) {$\row(bq)$};
	\node[state, shape=circle, right of=y, below of=y, label=below:{$\Rightarrow  1b + 1\overline{b}$}] (z) {$\row(\overline{b}q)$};
	    \path[->]
	(x) edge[loop above] node{$b p$} (x)
	(x) edge[below] node{$\overline{b} q$} (z)
	(y) edge[loop above] node{$bp, bq, \overline{b}p, \overline{b}q$} (y)
	(z) edge[right] node{$bp, bq, \overline{b}p, \overline{b}q$}(y)
	(x) edge[right] node{$bq, \overline{b}p$} (y)
	;
\end{tikzpicture}
}
	\caption{}
				\label{LStarmT4}
\end{subfigure}
\caption{An example run of Angluin's $\LStar$ algorithm for the target language $\llbracket (\wwhile\ b\ \ddo\ p); q \rrbracket$}
\label{lstarexamplerun}
\end{figure*}

Angluin's $\LStar$ algorithm learns the minimal DFA accepting a given regular language \cite{angluin1987learning}. The algorithm has since been modified and generalised for a broad class of transition systems. The variation we present here step-wise infers the minimal Moore automaton accepting a generalised language $L: A^* \rightarrow B$ for an input alphabet $A$ and an output alphabet $B$ \cite{moore1956gedanken}. The algorithm assumes the existence of a \emph{teacher} (or \emph{oracle}), which can respond to two types of queries:
\begin{itemize}
	\item \textbf{Membership queries}, consisting of a word $w \in A^*$, to which the teacher returns the output $L(w) \in B$;
	\item \textbf{Equivalence queries}, consisting of a hypothesis Moore automaton $H$, to which the teacher responds \emph{yes}, if $H$ accepts $L$, and \emph{no} otherwise, providing a counterexample $z \in A^*$ in the symmetric difference of $L$ and the behaviour of $H$.
\end{itemize}
The algorithm incrementally builds an \emph{observation table}, which contains partial information about the language $L$ obtained by performing membership queries.
A table consists of two parts: a top part, with rows indexed by a finite set $S \subseteq A^*$; and a bottom-part, with rows ranging over $S \cdot A$. Columns are indexed by a finite set $E \subseteq A^*$. For any $t \in S \cup S \cdot A$ and $e \in E$, the entry at row $t$ and column $e$, denoted by $\row(t)(e)$, is given by the output $L(te) \in B$.
Note that the sets $S$ and $S \cdot A$ can intersect. In such a case, elements in the intersection are only shown in the top part. We refer to a table as a tuple $T = (S, E, \row)$, leaving the language $L$ implicit. 

	Given a table $T$, one can construct a Moore automaton $m(T) = (X, \delta, \varepsilon, x)$, where $X = \lbrace \row(s) \mid s \in S \rbrace$ is a finite set of states; the transition function $\delta: X \rightarrow X^A$ is given by $\delta(\row(s), a) = \row(sa)$; the output function $\varepsilon: X \rightarrow B$ satisfies $\varepsilon(\row(s)) = \row(s)(\varepsilon)$ (we abuse notation by writing $\varepsilon$ both for the empty string and for the output function); and
		$x = \row(\varepsilon)$ is the initial state.
	For $m(T)$ to be well-defined, the table $T$ has to satisfy $\varepsilon \in S$ and $\varepsilon \in E$, and two properties called closedness and consistency. An observation table is \emph{closed} if for all $t \in S \cdot A$ there exists an $s \in S$ such that $\row(t) = \row(s)$. An observation table is \emph{consistent}, if whenever $s, s' \in S$ satisfy $\row(s) = \row(s')$, then $\row(sa) = \row(s'a)$ for all $a \in A$. A table is consistent in particular if the function $\row$ is injective. 

The algorithm incrementally updates the table to satisfy those properties. If a well-defined hypothesis $m(T)$ can be constructed, the algorithm poses an equivalence query to the teacher, and either terminates, or refines the hypothesis with a counterexample $z \in A^*$. Since we respond to a negative equivalence query by adding the suffixes\footnote{The set $\suff(z)$ of suffixes for $z \in A^*$ is defined by $\suff(\varepsilon) = \lbrace \varepsilon \rbrace$ and $\suff(aw) = \lbrace aw \rbrace \cup \suff(w)$.}  of a counterexample to the set $E$ (opposed to adding the prefixes of a counterexample to the set $S$), rows will always be distinct, rendering consistency trivial\footnote{This variation of $\LStar$ has been introduced by Maler and Pnueli \cite{maler1995learnability}.}. At all times, the set $S$ is prefix-closed and the set $E$ is suffix-closed\footnote{A set $X \subseteq A^*$ is called \emph{suffix-closed}, if $\suff(z) \subseteq X$ for all $z \in X$.}.

\subsubsection{Example of execution}

We now execute Angluin's $\LStar$ (\Cref{LStaralgorithm}) for the target language
\begin{equation}
\label{acceptedlanguageexample}
	L = \llbracket (\wwhile\ b\ \ddo\ p); q \rrbracket =   \lbrace \overline{b}qb, \overline{b}q \overline{b}, b p \overline{b} q b, b p \overline{b} q \overline{b}, ... \rbrace  \subseteq (\At \cdot \Sigma)^* \cdot \At,
	\end{equation}
	where $\At = \lbrace b, \overline{b} \rbrace$ is a finite set of \emph{atoms} and $\Sigma = \lbrace p, q \rbrace$ is a finite set of \emph{actions}. The language $L$ represents the semantics of a program that performs the action $p$  while $b$ is true, and otherwise continues with $q$. It can be viewed as a generalised language $\widehat{L}$ with input alphabet $A = (\At \cdot \Sigma)$ and output alphabet $B = 2^{\At}$ via currying.  We denote functions $f \in B$ as formal sums $\sum_{\alpha \in \At} f(\alpha)\alpha$. A query to $\widehat{L}$ requires $\vert \At \vert$ many queries to $L$.
	
	Initially, the sets $S$ and $E$ are set to the singleton $\lbrace \varepsilon \rbrace$. We build the observation table in \Cref{LStarT1}. Since the row indexed by $\overline{b}q$ does not appear in the upper part, i.e. differs from the row indexed by $\varepsilon$, the table is not closed.
		To resolve the closedness defect we add $\overline{b}q$ to $S$. The observation table (\Cref{LStarT2}) is now closed. We derive from it the hypothesis depicted in \Cref{LStarmT2}. 
		Next, we pose an equivalence query, to which the oracle replies \emph{no} and informs us that the word $z = bq \overline{b}q$ has been falsely classified. Indeed, given $z$, the language accepted by the hypothesis outputs $1b + 1 \overline{b}$, whereas \eqref{acceptedlanguageexample} produces  $0b + 0\overline{b}$.	
To respond to the counterexample $z$, we add its suffixes to $E$. In this case, there are only the two suffixes $\overline{b} q$ and $bq\overline{b}q$. The next observation table (\Cref{LStarT3}) again is not closed: the row indexed by e.g. $bq$ does not equal any of the two upper rows indexed by $\varepsilon$ and $\overline{b}q$.
	To resolve the closedness defect we add $bq$ to $S$, and obtain the table in \Cref{LStarT4}.
The observation table is now closed. We derive from it the automaton in \Cref{LStarmT4}. Next, we pose an equivalence query, to which the oracle replies \emph{yes}.

\subsection{$\GLStar$ Algorithm}

	\begin{algorithm}[t]
	\begin{algorithmic}
		\State {$S \gets \lbrace \varepsilon \rbrace, E \gets \At $}
		\Repeat 
				\While{$T = (S, E, \row: S \cup S \cdot (\At \cdot \Sigma)\rightarrow 2^E)$ is not closed}	
					\State{find $t \in S \cdot (\At \cdot \Sigma)$ with $\row(t)(e) = 1$ for some $e \in E$, but $\row(t) \neq \row(s)$ for all $s \in S$}
					\State{$S \gets S \cup \lbrace t \rbrace $}
				\EndWhile
				\State {construct and submit $m(T)$ to the teacher}
				\If { the teacher replies \emph{no} with a counterexample $z \in (\At \cdot \Sigma)^* \cdot \At$ }
					\State { $E \gets E \cup \suff(z)$ }
				\EndIf
		\Until {the teacher replies \emph{yes}} \\
		\Return $m(T)$
	\end{algorithmic}
\caption{The $\GLStar$ algorithm for GKAT automata}
\label{GlStaralgorithm}
\end{algorithm}

In this section, we propose a new algorithm (\Cref{GlStaralgorithm}) for learning GKAT program representations, which we call $\GLStar$. The new algorithm modifies \Cref{LStaralgorithm} by addressing a number of observations. 

First, we note that the Moore automaton in \Cref{LStarmT4} admits multiple transitions to $\row(bq)$, a \emph{sink-state}, which does not accept any words.
	Second, we observe that languages induced by GKAT programs are \emph{deterministic}\footnote{Deterministic in the sense that, whenever two strings agree on the first $n$ atoms, then they agree on their first $n$ actions (or lack thereof).}. Such languages are naturally represented by GKAT automata, which keep some transitions implicit.
	Third, in some cases\footnote{For instance, the entries of the row indexed by $bq$ in \Cref{LStarT3} must all be zero, since the row indexed by $bp$ admits a non-zero entry.} the deterministic nature of the target language allows us to fill-in parts of the observation table without membership queries.
	Fourth, the cells of the observation table are labelled by functions, each of which requires two membership queries to \eqref{acceptedlanguageexample}; as a consequence, table extensions require an unfeasible amount of queries.

As before, we assume two finite sets, $\At$ and $\Sigma$, and a deterministic language $L \subseteq (\At \cdot \Sigma)^* \cdot \At$. The oracle of $\GLStar$ can answer two types of queries: membership queries consist of a word $w \in (\At \cdot \Sigma)^* \cdot \At$, to which the oracle returns the output $L(w) \in 2$; equivalence queries consist of a hypothesis GKAT automaton $H$, to which the oracle responds \emph{yes}, if $H$ accepts $L$, and \emph{no} otherwise, providing a counterexample $z \in (\At \cdot \Sigma)^* \cdot \At$ in the symmetric difference of $L$ and the language accepted by $H$.

An observation table in $\GLStar$ consists of two parts: a top part, with rows indexed by a finite set $S \subseteq (\At \cdot \Sigma)^*$; and a bottom-part, with rows ranging over $S \cdot \At \cdot \Sigma$. Columns range over a finite set $E \subseteq (\At \cdot \Sigma)^* \cdot \At$. The entry of the observation table at row $t$ and column $e$, denoted by $\row(t)(e)$, is given by $L(te) \in 2$. We refer to a table by $T= (S, E, \row)$ and leave the deterministic language $L$ implicit.

Given a table $T$, we construct an automaton $m(T) = (X, \delta, x)$, where $X = \lbrace \row(s) \mid s \in S \rbrace$ is a set of states; $x = \row(\varepsilon)$ is the initial state; and $\delta: X \rightarrow (2 + \Sigma \times X)^{\At}$ evaluates $\delta(\row(s))(\alpha)$ to $(p, \row(s \alpha p))$, if there exists an $e \in E$ and a $p \in \Sigma$ with $\row(s \alpha p)(e) = 1$; to $1$, if $ \row(s)(\alpha) = 1$; and to $0$, otherwise.

Most of the properties a table needs to satisfy such that the hypothesis $m(T)$ is well-defined are guaranteed by the construction of \Cref{GlStaralgorithm}, since $L$ is deterministic. We only have to verify that the table is \emph{closed}, that is, for all $t \in S \cdot \At \cdot \Sigma$ with $\row(t)(e) = 1$ for some $e \in E$, there exists some $s \in S$ such that $\row(t) = \row(s)$. As in the case of $\LStar$, the algorithm incrementally updates the table until closedness is guaranteed. It then constructs a well-defined hypothesis, and poses an equivalence query to the teacher. If the oracle replies \emph{yes}, the algorithm terminates, and if the response is \emph{no}, it adds the suffixes\footnote{The set $\suff(z)$ of suffixes for $z \in A^* \cdot B$ is defined by $\suff(w b) = \lbrace vb \mid v \in \suff(w) \rbrace$.}  of a counterexample $z \in (\At \cdot \Sigma)^* \cdot \At$ to $E$.

The differences between $\GLStar$ and $\LStar$ (instantiated for $A = \At \cdot \Sigma$ and $B = 2^{\At}$) are essentially a consequence of currying. In the former case, the set $E$ contains elements of type $(\At \cdot \Sigma)^* \cdot \At$, and the table is filled with booleans in $2$; in the latter case, the set $E$ contains elements of type $(\At \cdot \Sigma)^*$, and the table is filled with functions $\At \rightarrow 2$. This, however, does not mean that $\GLStar$ is merely a shift in perspective: its new types induce independent definitions, and termination needs to be established with novel correctness proofs (\Cref{learningpart}). A thorough comparison is given in \Cref{comparisonmoore}.

\subsubsection{Example of execution}

\label{glstarexamplerunsection}

   \begin{figure*}
\centering
\begin{subfigure}[b]{.13\textwidth}
	\centering
		\resizebox{0.8 \textwidth}{!}{
		\begin{tabular}{ c|c|c}
		 & $b$ & $\overline{b}$ \\
		 \hline $\varepsilon$ & 0 & 0  \\
		 \hline
		 \hline $b p$ & 0 & 0 \\
		 \hline $b q$ & 0 & 0 \\
		 \hline $\overline{b} p$ & 0 & 0 \\
		 \hline $\overline{b} q$ & 1 & 1
	\end{tabular}
	}
	\caption{}
	\label{GLStarT1}
	\end{subfigure}
	\hfill
	\begin{subfigure}[b]{.13\textwidth}
	\centering
		\resizebox{0.85 \textwidth}{!}{
		\begin{tabular}{ c|c|c}
		 & $b$ & $\overline{b}$ \\
		 \hline $\varepsilon$ & 0 & 0  \\
		 \hline $\overline{b} q$ & 1 & 1 \\
		 \hline
		 \hline $b p$ & 0 & 0 \\
		 \hline $b q$ & 0 & 0 \\
		 \hline $\overline{b} p$ & 0 & 0 \\
		 \hline $\overline{b} q b p$ & 0 & 0 \\
		 \hline $\overline{b} q b q$ & 0 & 0 \\
		 \hline $\overline{b} q \overline{b} p$ & 0 & 0 \\
		 \hline $\overline{b} q \overline{b} q$ & 0 & 0 \\		 
	\end{tabular}
	}
	\caption{}
		\label{GLStarT2}
	\end{subfigure}
	\hfill
		\begin{subfigure}[b]{.22\textwidth}
	\centering
	\resizebox{\textwidth}{!}{
	\begin{tikzpicture}[node distance=7em]
	\node[state, shape=circle, initial, initial text=, label=above:{$\Rightarrow b \mid 0$}] (x) {$\row(\varepsilon)$};
		\node[state, shape=circle, right of=x, label=above:{$\Rightarrow  b, \overline{b} \mid 1$}] (y) {$\row(\overline{b}q)$};
	    \path[->]
	(x) edge[above] node{$\overline{b} \mid q$} (y)
	;
	\end{tikzpicture}
	}
	\caption{}
		\label{GLStarmT2}
	\end{subfigure}
	\hfill
	\begin{subfigure}[b]{.2\textwidth}
	\centering
	\resizebox{\textwidth}{!}{
		\begin{tabular}{ c|c|c|c|c }
		 & $b$ & $\overline{b}$ & $bp\overline{b}qb$ & $\overline{b}qb$  \\
		 \hline $\varepsilon$ & 0 & 0 & 1 & 1  \\
		 \hline $\overline{b} q$ & 1 & 1 & 0 & 0\\
		 \hline
		 \hline $b p$ & 0 & 0 & 1 & 1 \\
		 \hline $b q$ & 0 & 0 & 0 & 0\\
		 \hline $\overline{b} p$ & 0 & 0 & 0 & 0 \\
		 \hline $\overline{b} q b p$ & 0 & 0 & 0 & 0 \\
		 \hline $\overline{b} q b q$ & 0 & 0 & 0 & 0\\
		 \hline $\overline{b} q \overline{b} p$ & 0 & 0 & 0 & 0 \\
		 \hline $\overline{b} q \overline{b} q$ & 0 & 0 & 0 & 0 \\		 
	\end{tabular}
	}
	\caption{}
		\label{GLStarT3}
	\end{subfigure}
	\hfill
	\begin{subfigure}[b]{.22\textwidth}
	\centering
	\resizebox{\textwidth}{!}{
	\begin{tikzpicture}[node distance=7em]
	\node[state, shape=circle, initial, initial text=] (x) {$\row(\varepsilon)$};
		\node[state, shape=circle, right of=x, label=above:{$\Rightarrow  b, \overline{b} \mid 1$}] (y) {$\row(\overline{b}q)$};
	    \path[->]
	(x) edge[loop above] node{$b \mid p$} (x)
	(x) edge[above] node{$\overline{b} \mid q$} (y)
	;
	\end{tikzpicture}
	}
	\caption{}
		\label{GLStarmT3}
	\end{subfigure}
	
\caption{An example run of $\GLStar$ for the target language $\llbracket (\wwhile\ b\ \ddo\ p); q \rrbracket$}
\label{glstarexamplerun}
\end{figure*}

We now execute \Cref{GlStaralgorithm} for the target language \eqref{acceptedlanguageexample}.
Initially, $S = \lbrace \varepsilon \rbrace$ and $E = \At$.	
		We build the observation table in \Cref{GLStarT1}.  Since the bottom row indexed by $\overline{b}q$ contains a non-zero entry and differs from all upper rows (in this case, only the row indexed by $\varepsilon$), the table is not closed. We resolve the closedness defect by adding  $\overline{b}q$ to $S$. The observation table (\Cref{GLStarT2}) is now closed. Note that the row indexed by $\overline{b}q$ indicates that the words $\overline{b}qb$ and $\overline{b}q\overline{b}$ are accepted. Since we know the target language is deterministic, the last four rows of the table can be filled with zeroes, without performing any membership queries. From \Cref{GLStarT2} we derive the hypothesis depicted in \Cref{GLStarmT2}. Next, we pose an
equivalence query, to which the oracle replies $\emph{no}$ and provides us with the counterexample $z = bp\overline{b}qb$, which is in the language \eqref{acceptedlanguageexample}, but not accepted by the hypothesis.	
	We respond to the counterexample by adding its suffixes $bp\overline{b}qb$, $\overline{b}qb$ and $b$ to $E$. The resulting observation table is depicted in \Cref{GLStarT3}.
	The table is closed, since the only non-zero bottom row is the one indexed by $bp$, which coincides with the upper row indexed by $\varepsilon$. Since the row indexed by $bp$ has a non-zero entry, the row indexed by $bq$ can automatically be filled with zeroes.
	We derive from \Cref{GLStarT3} the automaton in \Cref{GLStarmT3}. Finally, we pose an equivalence query, to which the oracle replies \emph{yes}.

\section{Guarded Kleene Algebra with Tests}
This section recalls the syntax and semantics of Guarded Kleene Algebra with Tests (GKAT). For most parts, we follow the relevant bits of the original presentation in \cite{smolka2019guarded}. We additionally introduce a notion of similarity between GKAT automata. 

\subsection{Syntax}

The syntax of GKAT is inductively built from disjoint non-empty sets of \emph{primitive tests}, $T$, and \emph{actions}, $\Sigma$. In a first step, one generates from $T$ a set of Boolean expressions, $\BExpr$. In a second step, the set is extended with $\Sigma$, to the full set of GKAT expressions, $\Expr$:
\begin{align*}
	b,c,d \in \BExpr &::= 0 \mid 1 \mid t \in T \mid b \cdot c \mid b + c \mid \overline{b} \\
e, f, g \in \Expr &::= p \in \Sigma \mid b \in \BExpr \mid e \cdot f \mid e +_b f \mid e^{(b)}
\end{align*}
By a slight abuse of notation, we will sometimes write $e f$ for $e \cdot f$ and keep parenthesis implicit, e.g. $bc + d$ should be read as $(b \cdot c) + d$.

It is natural to view GKAT expressions as uninterpreted imperative programs. Under this view, one makes the identifications depicted in \Cref{gkatexpressionsasprograms}.

\begin{figure*}
	\begin{gather*}
	0 \equiv \false \qquad 1 \equiv \true \qquad t \equiv t \qquad 
	b \cdot c \equiv b\ \aand\ c \qquad b + c \equiv b\ \oor\ c  \\
	 \overline{b} \equiv \nnot\ b \qquad 	p \equiv \ddo\ p \qquad b \equiv \assert\ b
	 \qquad
	 e \cdot f \equiv e; f \qquad e^{(b)} \equiv \wwhile\ b\ \ddo\ e \\
	  e +_b f \equiv \iif\ b\ \tthen\ e\ \eelse\ f
\end{gather*}
\caption{Identifying GKAT expressions with imperative programs}
\label{gkatexpressionsasprograms}
\end{figure*}

Readers familiar with KAT will notice that the grammar for GKAT is similar to the one of KAT. It differs in that GKAT replaces KAT's union $(+)$ with the guarded union $(+_b)$, and KAT's iteration $(e^*)$ with the guarded iteration $(e^{(b)})$. GKAT's expressions can be encoded within KAT's grammar via the standard embedding that maps a conditional $e +_b f$ to $be + \overline{b}f$, and a while-loop $e^{(b)}$ to $(be)^*\overline{b}$.

\subsection{Semantics: Language Model}

In this section, we recall the language semantics of GKAT, which assigns to a program the traces it could produce once executed. 
Intuitively, an execution trace is a string of the shape $\alpha_0 p_1 \alpha_1 ... p_n \alpha_n$. It can be thought of as a sequence of states $\alpha_i$ a system is in at point $i$ in time, beginning with $\alpha_0$ and ending in $\alpha_n$, intertwined with actions $p_i$ that transition from the state $\alpha_{i-1}$ to the state $\alpha_{i}$.

More formally, let $\equiv_{\BA}$ denote the equivalence relation between Boolean expressions induced by the Boolean algebra axioms. The quotient $\BExpr/_{\equiv_{\BA}}$, that is, the free Boolean algebra on generators $T$, admits a natural preorder $\leq$ defined by $b \leq c \Leftrightarrow b + c \equiv_{\BA} c$. The minimal nonzero elements with respect to this order are called \emph{atoms}, the set of which is denoted by $\At$. If $T = \lbrace t_1, ..., t_n \rbrace$ is finite, an atom $\alpha \in \At$ is of the form $\alpha = c_1 \cdot ... \cdot c_n$ with $c_i \in \lbrace t_i, \overline{t_i} \rbrace$.

A \emph{guarded string} is an element of the set $\GS := \At \cdot (\Sigma \cdot \At)^*$, or equivalently, $(\At \cdot \Sigma)^* \cdot \At$. The set of guarded strings without terminating atom is $\GSM := (\At \cdot \Sigma)^*$. 

A guarded string language $L \subseteq \GS$ is \emph{deterministic} \cite[Def. 2.2]{smolka2019guarded}, if, whenever $\alpha_1 p_1 ... \alpha_{n-1} p_{n-1}\alpha_n v \in L$ and $\alpha_1 q_1 ... \alpha_{n-1}q_{n-1}\alpha_n w \in L$, then $p_i = q_i$ for all $1 \leq i \leq n-1$, and either $v = w = \varepsilon$, or $v = p_n v'$ and $w = q_n w'$ with $p_n = q_n$. The set of deterministic guarded string languages is denoted by $\mathscr{L}$.

Guarded strings can be partially composed via the \emph{fusion product} defined by $v \alpha \diamond \beta w := v \alpha w$, if $\alpha = \beta$, and undefined otherwise. The partial product lifts to a total function on guarded languages by $L \diamond K := \lbrace v \diamond w \mid v \in L, w \in K \rbrace$.
The $n$-th power of a guarded language is inductively defined by $L^0 := \At$ and $L^{n+1} := L^n \diamond L$.
For $B \subseteq \At$ and $\overline{B} := \At \setminus B$, the guarded sum and the guarded iteration of languages are given by
\[
L +_B K := (B \diamond L) \cup (\overline{B} \diamond K) \qquad L^{(B)} := \cup_{n \geq 0} (B \diamond L)^n  \diamond \overline{B}.
\]
The \emph{language model} of GKAT is given by the semantic function $\llbracket - \rrbracket : \Expr \rightarrow \mathscr{P}(\GS)$, which is inductively defined as follows:  
\begin{gather*}
	\llbracket p \rrbracket := \lbrace \alpha p \beta \mid \alpha, \beta \in \At \rbrace \qquad \llbracket b \rrbracket := \lbrace \alpha \in \At \mid \alpha \leq b \rbrace  \\
	\llbracket e \cdot f \rrbracket := \llbracket e \rrbracket \diamond \llbracket f \rrbracket \qquad \llbracket e +_b f \rrbracket := \llbracket e \rrbracket +_{\llbracket b \rrbracket} \llbracket f \rrbracket \qquad
	\llbracket e^{(b)} \rrbracket := \llbracket e \rrbracket^{(\llbracket b \rrbracket)}. 
\end{gather*}
Equivalently, the language semantics of GKAT can be constructed by post-composing the embedding of GKAT expressions into KAT expressions with the semantics of KAT.
The language $\llbracket e \rrbracket$ accepted by a GKAT program $e$ is deterministic.

\begin{example}
\label{examplewhileloop}
Let tests and actions be defined by $T := \lbrace b \rbrace$ and $\Sigma := \lbrace p , q\rbrace$, respectively. Then there exist only two atoms, $\At = \lbrace b, \overline{b} \rbrace$. The language model assigns to 
$ p^{(b)}q \equiv (\wwhile\ b\ \ddo\ p); q$ the guarded deterministic language in \eqref{acceptedlanguageexample}.
\end{example}

\subsection{Semantics: Automata Model}

\begin{figure}
\centering
\adjustbox{scale=0.8}{
		\begin{tikzpicture}[node distance=5em]
	\node[state, shape=circle, initial, initial text=] (x) {$x$};
		\node[state, shape=circle, left of=x, below of=x] (y) {$y$};
		\node[state, shape=circle, right of=x, below of=x, label=right:{$\Rightarrow  b, \overline{b} \mid 1$}] (z) {$z$};
	    \path[->]
	 (x) edge[left] node{$b \mid p$} (y)
	 (x) edge[right] node{$\overline{b} \mid q$} (z)
	(y) edge[loop left] node{$b \mid p$} (y)
	(y) edge[above] node{$\overline{b} \mid q$} (z)
	;
	\end{tikzpicture}
	}
	\caption{The Thompson-automaton $\mathscr{X}_{p^{(b)}q}$ for $T = \lbrace b \rbrace$ and $\Sigma = \lbrace p , q\rbrace$}
			\label{thompsonautomaton}
\end{figure}

In this section, we recall the automata model of GKAT, the central subject of this chapter. As before, we assume two finite sets of tests $T$ and actions $\Sigma$, the former of which induces a finite set of atoms, $\At$.  

Let $G$ be the functor on the category of sets which is defined on objects by $GX = (2 + \Sigma \times X)^{\At}$ (where $2 = \lbrace 0, 1 \rbrace$) and on morphisms in the usual way. 
A \emph{$G$-coalgebra} (cf. \Cref{def:coalgebra}) consists of a pair $\mathscr{X} = (X, \delta)$, where $X$ is a set called \emph{state-space} and $\delta: X \rightarrow GX$ is a function called \emph{transition map}. A $G$-coalgebra \emph{homomorphism} $f: (X, \delta^X) \rightarrow (Y, \delta^Y)$ is a function $f: X \rightarrow Y$  that commutes with the transition maps, $\delta^{Y} \circ f = Gf \circ \delta^X$. More concretely \cite[Def. 5.7.]{smolka2019guarded}, $f$ is a $G$-coalgebra homomorphism, if for all $\alpha \in \At, p \in \Sigma$, and $x, y \in X$, 
\begin{itemize}
	\item if $\delta^X(x)(\alpha) \in 2$, then $\delta^Y(f(x))(\alpha) = \delta^X(x)(\alpha)$; and 
	\item if $\delta^X(x)(\alpha) = (p, y)$, then $\delta^Y(f(x))(\alpha) = (p, f(y))$.
\end{itemize} 
A \emph{$G$-automaton} is a $G$-coalgebra $\mathscr{X}$ with a designated initial state $x \in X$. A homomorphism $f: (\mathscr{X}, x) \rightarrow (\mathscr{Y}, y)$ between $G$-automata is a homomorphism between the underlying $G$-coalgebras that maps initial state to initial state, $f(x) = y$. 

For each state $x \in X$, given an input $\alpha \in \At$, a $G$-coalgebra either i) halts and accepts, that is, satisfies $\delta(x)(\alpha) = 1$; ii) halts and rejects, that is, satisfies $\delta(x)(\alpha) = 0$; or iii) produces an output $p$ and moves to a new state $y$, that is, satisfies $\delta(x)(\alpha) = (p, y)$. Intuitively, for each state $x \in X$, a guarded string $\alpha_0p_1\alpha_1...p_n \alpha_n$ is accepted, if the $G$-coalgebra in state $x$ produces the output $p_1...p_n$, halts and accepts. Formally, one defines a function $\llbracket - \rrbracket: X \rightarrow \mathscr{P}(\GS)$ as follows:
		\begin{align}
		\label{gkatautomatasemantics}
		\begin{split}
			\alpha \in \llbracket x \rrbracket &:\Leftrightarrow \delta(x)(\alpha) = 1; \\
			\alpha p w \in \llbracket x \rrbracket &:\Leftrightarrow \exists y \in X: \delta(x)(\alpha) = (p, y) \textnormal{ and } w \in  \llbracket y \rrbracket.
		\end{split}	
		\end{align}
								A $G$-coalgebra is \emph{observable}, if the function $\llbracket - \rrbracket$ is injective.
								
	A guarded string $w$ is \emph{accepted} by $x$, if $w \in \llbracket x \rrbracket$. The language accepted by a $G$-automaton, $\llbracket \mathscr{X} \rrbracket$, is the language accepted by its initial state $\llbracket x \rrbracket$. 
	Every language accepted by a $G$-automaton is deterministic \cite[Thm. 5.8]{smolka2019guarded}.
	 Conversely, one can equip the set of deterministic languages with a $G$-coalgebra structure $(\mathscr{L}, \delta^{\mathscr{L}})$ defined by:
\begin{align*}
			\delta^{\mathscr{L}}(L)(\alpha) = \begin{cases}
				(p, (\alpha p)^{-1}L) & \textnormal{ if } (\alpha p)^{-1}L \not = \emptyset \\
				1 & \textnormal{ if } \alpha \in L \\
				0 & \textnormal{ otherwise}
			\end{cases},
		\end{align*}
		where $(\alpha p)^{-1}L = \lbrace w \in \GS \mid \alpha p w \in L \rbrace$. Note that $\delta^{\mathscr{L}}(L)$ is well-defined because $L$ is deterministic.
	Since $\llbracket L \rrbracket = L$ for any $L \in \mathscr{L}$ \cite[Thm. 5.8]{smolka2019guarded}, every deterministic language can be recognized by a $G$-automaton with possibly infinitely many states. 
	
	A $G$-coalgebra $(X, \delta)$ is \emph{normal}, if it only transitions to \emph{live} states, that is, $\delta(x)(\alpha) = (p, y)$ implies $\llbracket y \rrbracket \not = \emptyset$. For any $G$-automaton $\mathscr{X}$ one can construct a language equivalent normal $G$-automaton $\widehat{\mathscr{X}}$ \cite[Lem. 5.6]{smolka2019guarded}. If $\mathscr{X}$ is normal, the function $\llbracket - \rrbracket: X \rightarrow \mathscr{P}(\GS)$ is the unique coalgebra homomorphism $\llbracket - \rrbracket: (X, \delta) \rightarrow (\mathscr{L}, \delta^{\mathscr{L}})$ \cite[Thm. 5.8]{smolka2019guarded}. 
	
	Two states $x,y \in X$ of a normal coalgebra accept the same language, $\llbracket x \rrbracket = \llbracket y \rrbracket$, if and only if they are \emph{bisimilar}\footnote{In \Cref{similarityappendix} we introduce the notion of similarity, which is to bisimilarity what a partial order is to equality.}, $x \simeq y$, that is, there exists a binary relation $R \subseteq X \times X$ with $x R y$, such that if $x' R y'$, then the following two implications hold:
	\begin{itemize}
		\item if $\delta(x')(\alpha) \in 2$, then $\delta(y')(\alpha) = \delta(x')(\alpha)$; and
		\item if $\delta(x')(\alpha) = (p, x'')$, then $\delta(y')(\alpha) = (p, y'')$ and $x'' R y''$ for some $y'' \in X$.
	\end{itemize}
	Bisimilarity is a symmetric relation and can be extended to two coalgebras by constructing a coalgebra that has the disjoint union of their state-spaces as state-space.

Using a construction that is reminiscent of Thompson's construction for regular expressions \cite{thompson1968programming}, it is possible to interpret a GKAT expression $e$ as an automaton $\mathscr{X}_e$ that accepts the same language \cite{smolka2019guarded}. Alternatively, one can use a construction \cite{smolka2019guarded} that mirrors Kozen's syntactic form of Brzozowski's derivatives for KAT \cite{kozen2017coalgebraic}. 

\begin{example}
\label{thompsonexample}
	The Thompson-automaton assigned to the expression $p^{(b)}q$ is depicted in \Cref{thompsonautomaton}. It is normal and reachable, but not observable, since the states $x$ and $y$ are bisimilar, $x \simeq y$, thus accept the same language, $\llbracket x \rrbracket = \llbracket y \rrbracket$. It also is equivalent to the expression by which it is generated, that is, it  satisfies $\llbracket \mathscr{X}_{p^{(b)}q} \rrbracket = \llbracket p^{(b)}q \rrbracket$.
\end{example}

 \subsection{A Note on Similarity}
 \label{similarityappendix}
 
 In this section we briefly introduce a notion of similarity that is to bisimulation, what a partial order is to equality. Our construction addresses the coalgebraic side of the proposal to replace the primitive notion of equality (equivalence) of GKAT expressions with a partial order of GKAT expressions \cite{smolka2019guarded}. We acknowledge that similarity has been studied more generally, for arbitrary coalgebras \cite{baltag2000logic,hughes2004simulations,hesselink2000fixpoint,levy2011similarity}.
 
 \begin{definition}
		\label{simulation}
		Let $\mathscr{X}$ be a $G$-coalgebra. A \emph{simulation} is a binary relation $R \subseteq X \times X$, such that if $xRy$, then:
		\begin{itemize}
			\item if $\delta(x)(\alpha) = 1 $, then $\delta(y)(\alpha) = 1$;
			\item if $\delta(x)(\alpha) = (p, x')$, then $\delta(y)(\alpha) = (p, y')$ and $x'Ry'$ for some $y' \in X$.
		\end{itemize}
		States $x$ and $y$ are \emph{similar}, $x \precsim y$, if there exists a simulation relating $x$ to $y$. 
	\end{definition}

The result below shows that similarity is more fundamental than bisimilarity, and that the definition of the latter naturally arises from the former.

\begin{lemma}
\label{bisim-sim}
	$x \simeq y$ if and only if $x \precsim y$ and $y \precsim x$.
\end{lemma}
\begin{proof}
\begin{itemize}
	\item Assume the bisimilarity $x \simeq y$ is witnessed by some relation $R$. We show that $R$ witnesses the similarity $x \precsim y$. Clearly $xRy$ by definition. Let $x'Ry'$ for arbitrary $x', y' \in X$, then we find:
	\begin{itemize}
		\item If $\delta(x')(\alpha) = 1$, then $\delta(y') = 1$ since $R$ is a bisimulation.
		\item If $\delta(x')(\alpha) = (p, x'')$, then $\delta(y')(\alpha) = (p, y'')$ and $x''Ry''$ for some $y'' \in X$, since $R$ is a bisimulation.
	\end{itemize}
	Similarly, we show that the reverse relation $R^r$ witnesses the similarity  $y \precsim x$. Clearly $y R^r x$, since by construction $x R y$. Let $y' R^r x'$, i.e. $x' R y'$, for arbitrary $x', y' \in X$, then we find:
	\begin{itemize}
		\item If $\delta(y')(\alpha) = 1$, then $\delta(x')(\alpha) = 1$, since as $R$ is a bisimulation we could otherwise falsely deduce $\delta(y')(\alpha) = 0$ or $\delta(y')(\alpha) \not \in 2$.
		\item If $\delta(y')(\alpha) = (p, y'')$, then $\delta(x')(\alpha) = (p,x'')$ with $y'' R^r x''$, i.e. $x''Ry''$. Indeed, since $R$ is a bisimulation, if $\delta(x')(\alpha) \in 2$, it falsely follows $\delta(y')(\alpha) \in 2$, and if $\delta(x')(\alpha) = (q,x'')$, it follows $\delta(y')(\alpha) = (q, y''')$ with $x'' R y'''$, as $R$ is a bisimulation. It remains to observe $(p, y'') = \delta(y')(\alpha) = (q, y''')$, which implies $p = q$ and $y'' = y'''$.	\end{itemize}
		
		\item Assume $x \precsim y$ and $y \precsim x$ are witnessed by relations $R_1 \subseteq X \times X$ and $R_2 \subseteq X \times X$, respectively. We define $R := R_1 \cap R_2^r$, and show that $R$ is a bisimulation witnessing $x \simeq y$. Clearly $x R_1 y$ and $x R_2^r y$, i.e. $x R y$. Thus let $x' R y'$ for arbitrary $x', y' \in X$, then we find:
			\begin{itemize}
			\item If $\delta(x')(\alpha) = 0$, then $\delta(y')(\alpha) = 0$. Indeed, if $\delta(y')(\alpha) = 1$ or $\delta(y')(\alpha) \not \in 2$, we could falsely deduce $\delta(x')(\alpha) = 1$ or $\delta(x')(\alpha) \not \in 2$, as $y' R_2 x'$, and $R_2$ is a simulation.	
			\item If $\delta(x')(\alpha) = 1$, then $\delta(y')(\alpha) = 1$, since $x'R_1y'$, and $R_1$ is a simulation. 
			\item If $\delta(x')(\alpha) = (p, x'')$, then (i) $\delta(y')(\alpha) = (p, y'')$ with $x'' R_1 y''$, since $x' R_1 y'$, and $R_1$ is a simulation; and (ii) $y'' R_2 x''$, since $y' R_2 x'$ implies $\delta(x')(\alpha) = (p, x''') = (p, x'')$ for $y'' R_2 x'''$.
			Thus we find by definition of $R$ that $x'' R y''$. \qedhere
			\end{itemize}
			\end{itemize}
\end{proof}

\begin{lemma}
\label{simimpliesinclusion}
	If $x \precsim y$ then $\llbracket x \rrbracket \subseteq \llbracket y \rrbracket$.
\end{lemma}	
\begin{proof}
The proof is similar to the one of its bisimilar counterpart \cite[Lemma 5.2]{smolka2019guarded}. 

We prove $w \in \llbracket x \rrbracket$ implies $w \in \llbracket y \rrbracket$ for all $w \in \GS$ by induction on the length of $w$.
\begin{itemize}
	\item For the induction base, let $w = \alpha$, then:
	\begin{align*}
		\alpha \in \llbracket x \rrbracket \Leftrightarrow\  & \delta(x)(\alpha) = 1 && \textnormal{(Definition of } \llbracket - \rrbracket) \\
		\Rightarrow\  & \delta(y)(\alpha) = 1 && (x \precsim y) \\
		\Leftrightarrow\ & \alpha \in \llbracket y \rrbracket && \textnormal{(Definition of } \llbracket - \rrbracket)
	\end{align*}
	\item For the induction step, let $w = \alpha p v$, then we derive: 
	\begin{align*}
		\alpha p v \in \llbracket x \rrbracket & \Leftrightarrow \delta(x)(\alpha) = (p, x'),\ v \in \llbracket x' \rrbracket & & \textnormal{(Definition of } \llbracket - \rrbracket) \\
		&\Rightarrow \delta(y)(\alpha) = (p, y'),\ v \in \llbracket y' \rrbracket && (x \precsim y,\ \textnormal{IH}) \\
		&\Leftrightarrow \alpha p v \in \llbracket y \rrbracket && \textnormal{(Definition of } \llbracket - \rrbracket)
	\end{align*} 
\end{itemize}
\end{proof}

\begin{lemma}
\label{simlanguageequiv}
	Let $L_1, L_2 \in \mathscr{L}$, then $L_1 \subseteq L_2$ iff $L_1 \precsim L_2$ in $(\mathscr{L}, \delta^{\mathscr{L}})$.
\end{lemma}
\begin{proof}
	\Cref{simimpliesinclusion} shows that $L_1 \precsim L_2$  implies $L_1 = \llbracket L_1 \rrbracket \subseteq \llbracket L_2 \rrbracket = L_2$.
	
	 Conversely, we show that $\subseteq$ is a simulation.
	 	Assume $L_1 \subseteq L_2$, then we compute:
	 	\begin{align*}
	 		\delta^{\mathscr{L}}(L_1)(\alpha) = 1 &\Leftrightarrow \alpha \in L_1 && \textnormal{(Definition of } \delta^{\mathscr{L}}) \\
	 		&\Rightarrow \alpha \in L_2 && (L_1 \subseteq L_2) \\
	 		&\Leftrightarrow \delta^{\mathscr{L}}(L_2)(\alpha) = 1 && \textnormal{(Definition of } \delta^{\mathscr{L}})
	 	\end{align*}
	 	Moreover, we find: 
	 	\begin{align*}
	 		\delta^{\mathscr{L}}(L_1)(\alpha) = (p, L^1) 
	 		& \Leftrightarrow \emptyset \not = L^1 =   (\alpha p)^{-1}L_1&& \textnormal{(Definition of } \delta^{\mathscr{L}}) \\
	 		& \Rightarrow \emptyset \not = L^1 = (\alpha p)^{-1}L_1 \subseteq (\alpha p)^{-1}L_2 = L^2 && (L_1 \subseteq L_2) \\
	 		& \Leftrightarrow \delta^{\mathscr{L}}(L_2)(\alpha) = (p, L^2),\ L^1 \subseteq L^2 && \textnormal{(Definition of } \delta^{\mathscr{L}}) 
	 	\end{align*}
	\end{proof}
	
	\begin{corollary}
\label{simnormal}
	Let $\mathscr{X}$ be a normal $G$-coalgebra, then $x \precsim y$ iff $\llbracket x \rrbracket \subseteq \llbracket y \rrbracket$.
\end{corollary}
	\begin{proof}
	The proof is similar to the one of its bisimilar counterpart \cite[Corollary 5.9]{smolka2019guarded}.
	
	From \Cref{simimpliesinclusion} it follows that $x \precsim y$ implies $\llbracket x \rrbracket \subseteq \llbracket y \rrbracket$.	
	
	 Conversely, assume $\llbracket x \rrbracket \subseteq \llbracket y \rrbracket$. We define a relation $R:= \lbrace (s,t) \in X \times X \mid \llbracket s \rrbracket \subseteq \llbracket t \rrbracket \rbrace$. In order to show $x \precsim y$ it is sufficient to prove that $R$ is a simulation. Since $\mathscr{X}$ is normal, $\llbracket - \rrbracket$ is a $G$-coalgebra homomorphism. 
	\begin{itemize}
		\item Suppose $s R t$ and $\delta(s)(\alpha) = 1$. As $\llbracket - \rrbracket$ is a $G$-coalgebra homomorphism it follows $\delta^{\mathscr{L}}(\llbracket s \rrbracket)(\alpha) = 1$. Since $\llbracket s \rrbracket \subseteq \llbracket t \rrbracket$ implies $\llbracket s \rrbracket \precsim \llbracket t \rrbracket$ by \Cref{simlanguageequiv}, 
		we thus can deduce $\delta^{\mathscr{L}}(\llbracket t \rrbracket)(\alpha) = 1$. 
		Since  $\llbracket - \rrbracket$ is a $G$-coalgebra homomorphism, we can conclude $ \delta(t)(\alpha) = 1 $.
		\item Suppose $s R t$ and $\delta(s)(\alpha) = (p,s')$. Since $\llbracket - \rrbracket$ is a $G$-coalgebra homomorphism it follows
$
			\delta^{\mathscr{L}}(\llbracket s \rrbracket)(\alpha) = (p, \llbracket s' \rrbracket) .
$
Since $\llbracket s \rrbracket \subseteq \llbracket t \rrbracket$ implies $\llbracket s \rrbracket \precsim \llbracket t \rrbracket$ by \Cref{simlanguageequiv}, we deduce $\delta^{\mathscr{L}}(\llbracket t \rrbracket)(\alpha) = (p, L)$ for some $L \in \mathscr{L}$ with $\llbracket s' \rrbracket \precsim L$ in $\mathscr{L}$. Since 
$\llbracket - \rrbracket$ is a $G$-coalgebra homomorphism, it follows $L = \llbracket t' \rrbracket$ with $\delta(t)(\alpha) = (p, t')$. Thus we have $\llbracket s' \rrbracket \precsim \llbracket t' \rrbracket$, or equivalently $\llbracket s' \rrbracket \subseteq \llbracket t' \rrbracket$ by \Cref{simlanguageequiv}. The latter implies $s' R t'$ by definition of $R$. Thus we find $\delta(t)(\alpha) = (p, t')$ and $s' R t'$. \qedhere
	\end{itemize} 
\end{proof}

\section{The Minimal Representation $m(\mathscr{X})$}

The automaton $\mathscr{X}_e$ assigned to an expression $e$ by the Thompson construction is not always the most efficient representation of the language $\llbracket e \rrbracket$. For instance, as seen in \Cref{thompsonexample}, the Thompson-automaton $\mathscr{X}_{p^{(b)}q}$ in \Cref{thompsonautomaton} contains redundant structure, since its states $x$ and $y$ exhibit the same behaviour. In this section, we show that any $G$-automaton $\mathscr{X}$ admits an equivalent \emph{minimal} representation, $m(\mathscr{X})$.
\subsection{Reachability}

We begin by formally defining what it means for a state of a $G$-automaton to be reachable, and show that restricting an automaton to its reachable states leaves important properties invariant.
 
\begin{definition}
	Let $(X, \delta)$ be a $G$-coalgebra. We write $\rightarrow\ \subseteq X \times \GSM \times X$ for the smallest relation satisfying:
\begin{gather}
\label{transitiondef}
	\frac{}{x \xrightarrow{\varepsilon} x} \quad
	\frac{\delta(x)(\alpha) = (p, y)}{x \xrightarrow{\alpha p} y} \quad
	\frac{x \xrightarrow{\alpha_1 p_1 ... \alpha_{n-1} p_{n-1}} y \quad , \quad y \xrightarrow{\alpha_n p_n} z}{x \xrightarrow{\alpha_1 p_1 ... \alpha_n p_n} z}.
\end{gather}
The states \emph{reachable} from $x \in X$ are 
$r(x) := \lbrace y \in X \mid \exists w \in \GSM: x \xrightarrow{w} y \rbrace$, and their \emph{witnesses} are $R(x) := \lbrace w \in \GSM \mid \exists x_w \in X : x \xrightarrow{w} x_w \rbrace$.
\end{definition}

The following result shows that a state reached by a word is uniquely defined.

\begin{lemma}
\label{uniquenessstates}
If $x \xrightarrow{w} x_w^1$ and $x \xrightarrow{w} x_w^2$, then $x_w^1 = x_w^2$.
\end{lemma}
\begin{proof}
	We show the statement by induction on the length of $w \in \GSM$:
	\begin{itemize}
		\item The induction base $w = \varepsilon$ follows from the base case of \eqref{transitiondef}: $x \xrightarrow{\varepsilon} x_w^i$ iff $ x_w^i = x$ for $i = 1,2$.
		\item For the induction step let $w = v \alpha p$ for some $v \in \GSM$. By \eqref{transitiondef} there exist $x_v^1, x_v^2 \in X$ such that $x \xrightarrow{v} x_v^1 \xrightarrow{\alpha p} x_w^1$ and $x \xrightarrow{v} x_v^2 \xrightarrow{\alpha p} x_w^2$. From the induction hypothesis it follows $x_v^1 = x_v^2$. Thus, by \eqref{transitiondef}, $(p, x_w^1) =  \delta(x_v^1)(\alpha) = \delta(x_v^2)(\alpha) = (p, x_w^2)$, which yields $x_w^1 = x_w^2$. \qedhere
	\end{itemize}
\end{proof}

Given a  $G$-coalgebra $(X, \delta)$, we call a subset $A \subseteq X$ $\delta$\emph{-invariant}, if $y \in A$ and $\delta(y)(\alpha) = (p, z)$, then $z \in A$. In such a situation, we write $\mathscr{X}^A = (A, \delta^A)$ for the well-defined restriction of $\mathscr{X} = (X, \delta)$ to $A$.

We denote the sub-automaton one obtains by restricting a $G$-automaton $\mathscr{X} = (X, \delta, x)$ to the $\delta$-invariant set of states reachable from the initial state $x \in X$ by $r(\mathscr{X}) := \mathscr{X}^{r(x)}$, and call an automaton \emph{reachable}, if $\mathscr{X} = r(\mathscr{X})$. Following \cite[Def. 15]{van2017calf}, we call a normal, reachable, and observable automaton $\emph{minimal}$.

The set $R(x)$ of words witnessing the reachability of states in $\mathscr{X} = (X, \delta, x)$ can be equipped with a $G$-automaton structure $R(\mathscr{X}) := (R(x), \partial, \varepsilon)$, where $\partial(w)(\alpha) = (p, w\alpha p)$, if $\delta(x_w)(\alpha) = (p, x_{w \alpha p})$ for some $x_{w \alpha p} \in X$, and $\partial(w)(\alpha) = \delta(x_w)(\alpha)$ otherwise. The automaton $r(\mathscr{X})$ can then be recovered as the image of the automata homomorphism $f: R(\mathscr{X}) \rightarrow \mathscr{X}$ defined by $f(w) = x_w$. In other words, there exists an epi-mono factorisation $R(\mathscr{X}) \twoheadrightarrow r(\mathscr{X}) \hookrightarrow \mathscr{X}$.

We conclude with a list of important properties preserved by restricting an automaton to its reachable states. Among those properties are \emph{well-nestedness} \cite{smolka2019guarded} and \emph{satisfing the nesting coequation} \cite{schmid2021guarded}. We refer the reader to \Cref{relatedwork} for a high-level comparison between the two notions. 

The definition of well-nestedness requires the following construction.

\begin{definition}[\cite{smolka2019guarded}]
		Let $\mathscr{X} = (X, \delta)$ be a $G$-coalgebra. The \emph{uniform continuation} of $A \subseteq X$ by $h \in G(X)$  is the $G$-coalgebra $\mathscr{X}\lbrack A, h \rbrack = (X, \delta\lbrack A, h \rbrack)$, where:
		\[
		\delta\lbrack A, h \rbrack(x)(\alpha) = \begin{cases}
			h(\alpha) & \textnormal{if } x \in A, \delta(x)(\alpha) = 1 \\
			\delta(x)(\alpha) & \textnormal{else }
		\end{cases}.
		\]
	\end{definition}
	
	Further let $\mathscr{X} + \mathscr{Y} := (X + Y, \delta^{\mathscr{X}} + \delta^{\mathscr{Y}})$ be the disjoint union of automata, where $\delta^{\mathscr{X}} + \delta^{\mathscr{Y}}(x)(\alpha) = \delta^{\mathscr{X}}(x)(\alpha)$, if $x \in X$, and $\delta^{\mathscr{Y}}(x)(\alpha)$ otherwise.

\begin{definition}[\cite{smolka2019guarded}]
		The class of \emph{well-nested} $G$-coalgebras is defined as follows:
		\begin{itemize}
			\item If $\mathscr{X} = (X, \delta)$ has no transitions, i.e. if $\delta: X \rightarrow 2^{\At}$, then $\mathscr{X}$ is well-nested.
			\item If $\mathscr{X}$ and $\mathscr{Y}$ are well-nested and $h \in G(X + Y)$, then $(\mathscr{X} + \mathscr{Y})\lbrack X, h \rbrack$ is well-nested.
		\end{itemize}
	\end{definition}

While the following results are stated in terms of arbitrary $\delta$-invariant subsets, we are particularly interested in $r(x)$, the subset of states reachable by an initial state $x$.

\begin{lemma}
		\label{invariantsimple}
		The restriction of a well-nested $G$-coalgebra to a $\delta$-invariant subset is well-nested.
			\end{lemma}	
	\begin{proof}
		We show the statement by induction on the well-nested structure of $\mathscr{X}$.
		As before, we write $\mathscr{X}^A = (A, \delta^A)$ for the well-defined restriction of $\mathscr{X} = (X, \delta)$ to a $\delta$-invariant subset $A \subseteq X$.
	\begin{itemize}
		\item For the induction base assume that $\mathscr{X} = (X, \delta)$ satisfies $\delta: X \rightarrow 2^{\At}$, and $A \subseteq X$ is a $\delta$-invariant set. Then clearly the restriction is of type $\delta^A: A \rightarrow 2^{\At}$, i.e. $\mathscr{X}^A = (A, \delta^A)$ is well-nested.
		\item For the induction step let $\mathscr{Y} = (Y, \delta^{\mathscr{Y}})$ and $\mathscr{Z} = (Z, \delta^{\mathscr{Z}})$ be well-nested $G$-coalgebras, $h \in G(Y + Z)$, and $\mathscr{X} = (Y + Z, (\delta^{\mathscr{Y}} + \delta^{\mathscr{Z}})\lbrack Y, h \rbrack)$. Moreover let $A \subseteq Y + Z$ be a 
		$(\delta^{\mathscr{Y}} + \delta^{\mathscr{Z}})\lbrack Y, h \rbrack$-invariant set. We would like to show that $\mathscr{X}^A$ is well-nested.
		The induction hypothesis reads: \begin{itemize}
			\item for all $\delta^{\mathscr{Y}}$-invariant sets $B \subseteq Y$, the subcoalgebra $\mathscr{Y}^B = (B, (\delta^{\mathscr{Y}})^B)$ is well-nested;
			\item for all $\delta^{\mathscr{Z}}$-invariant sets $C \subseteq Z$, the subcoalgebra $\mathscr{Z}^C = (C, (\delta^{\mathscr{Z}})^C)$ is well-nested.
		\end{itemize} 
		We begin by showing that $A \cap Y \subseteq Y$ and $A \cap Z \subseteq Z$ are $\delta^{\mathscr{Y}}$- and $\delta^{\mathscr{Z}}$-invariant sets, respectively. Let $\delta^{\mathscr{Y}}(x)(\alpha) = (p, y)$ for $x \in A \cap Y$ and $y \in Y$. Then by definition \[ (\delta^{\mathscr{Y}} + \delta^{\mathscr{Z}})\lbrack Y, h \rbrack(x)(\alpha) = \delta^{\mathscr{Y}}(x)(\alpha) = (p, y), \] which by $(\delta^{\mathscr{Y}} + \delta^{\mathscr{Z}})\lbrack Y, h \rbrack$-invariance of $A$ implies $y \in A$. It hence follows $y \in A \cap Y$.
			Analogously one dedues that $A \cap Z$ is $\delta^{\mathscr{Z}}$-invariant.
		Thus $\mathscr{Y}^{A \cap Y} = (A \cap Y, (\delta^{\mathscr{Y}})^{A \cap Y})$ and $\mathscr{Z}^{A \cap Z} = (A \cap Z, (\delta^{\mathscr{Z}})^{A \cap Z}) $ are well-defined, and moreover, by the induction hypothesis they are well-nested.
		
		We observe the equality $A = A \cap Y + A \cap Z$, which follows from $A \subseteq Y + Z$, and define $\overline{h} \in G(A \cap Y + A \cap Z) = G(A)$ by:
		\begin{equation}
			\label{newh}
			\overline{h}(\alpha) = \begin{cases}
							1 & \textnormal{if } h(\alpha) = 1 \\
				(p, x) & \textnormal{if } h(\alpha) = (p, x),\ x \in A \\				
				0  & \textnormal{else } 
			\end{cases}
		\end{equation}
		
		It follows 
		$\mathscr{X}^A = (A \cap Y + A \cap Z, ((\delta^{\mathscr{Y}})^{A \cap Y} + (\delta^{\mathscr{Z}})^{A \cap Z})\lbrack A \cap Y, \overline{h} \rbrack)$, since for any $x \in A$ it holds:
		\begin{align*}
			&((\delta^{\mathscr{Y}} + \delta^{\mathscr{Z}})\lbrack Y, h \rbrack)^A(x)(\alpha) \\
			&= \begin{cases}
				\delta^{\mathscr{Y}}(x)(\alpha) & x \in A \cap Y,\ \delta^{\mathscr{Y}}(x)(\alpha) \not = 1 \\
				h(\alpha) & x \in A \cap Y,\ \delta^{\mathscr{Y}}(x)(\alpha) = 1  \\
				\delta^{\mathscr{Z}}(x)(\alpha) & x \in A \cap Z
			\end{cases} \\
			&= \begin{cases}
				\delta^{\mathscr{Y}}(x)(\alpha) & x \in A \cap Y,\ \delta^{\mathscr{Y}}(x)(\alpha) \not = 1 \\
				1 & x \in A \cap Y,\ \delta^{\mathscr{Y}}(x)(\alpha) = 1,\ h(\alpha) = 1  \\
								0 & x \in A \cap Y,\ \delta^{\mathscr{Y}}(x)(\alpha) = 1,\ h(\alpha) = 0  \\
								(p, x') & x \in A \cap Y,\ \delta^{\mathscr{Y}}(x)(\alpha) = 1,\ h(\alpha) = (p, x')  \\					
				\delta^{\mathscr{Z}}(x)(\alpha) & x \in A \cap Z
			\end{cases} \\
			& \overset{(\star)}{=} \begin{cases}
				\delta^{\mathscr{Y}}(x)(\alpha) & x \in A \cap Y,\ \delta^{\mathscr{Y}}(x)(\alpha) \not = 1 \\
				1 & x \in A \cap Y,\ \delta^{\mathscr{Y}}(x)(\alpha) = 1, \overline{h}(\alpha) = 1  \\
								0 & x \in A \cap Y,\  \delta^{\mathscr{Y}}(x)(\alpha) = 1, \overline{h}(\alpha) = 0  \\
								(p, x') & x \in A \cap Y,\ \delta^{\mathscr{Y}}(x)(\alpha) = 1, \overline{h}(\alpha) = (p, x')  \\					
				\delta^{\mathscr{Z}}(x)(\alpha) & x \in A \cap Z
			\end{cases} \\
			& = \begin{cases}
				(\delta^{\mathscr{Y}})^{A \cap Y}(x)(\alpha) & x \in A \cap Y,\ (\delta^{\mathscr{Y}})^{A \cap Y}(x)(\alpha) \not = 1 \\
				\overline{h}(\alpha) & x \in A \cap Y,\ (\delta^{\mathscr{Y}})^{A \cap Y}(x)(\alpha) = 1  \\
				(\delta^{\mathscr{Z}})^{A \cap Z}(x)(\alpha) & x \in A \cap Z
			\end{cases} \\
			&= ((\delta^{\mathscr{Y}})^{A \cap Y} + (\delta^{\mathscr{Z}})^{A \cap Z})\lbrack A \cap Y, \overline{h} \rbrack(x)(\alpha),
		\end{align*} 
		where we use for $(\star)$ that $((\delta^{\mathscr{Y}} + \delta^{\mathscr{Z}})\lbrack Y, h \rbrack)^A(x)(\alpha) = (p, x')$ for $x \in A$, implies $x' \in A$, as $A$ is $(\delta^{\mathscr{Y}} + \delta^{\mathscr{Z}})\lbrack Y, h \rbrack$-invariant. \qedhere
		\end{itemize}
		\end{proof}
	
Restricting to a $\delta$-invariant subset not only preserves well-nestedness, but also language semantics.	
	
\begin{lemma}
	\label{invariantlanguage}
Let $\mathscr{X}^A = (A, \delta^A)$ be the restriction of a $G$-coalgebra $\mathscr{X} = (X, \delta)$ to a $\delta$-invariant subset $A \subseteq X$. Then 
$\llbracket a \rrbracket_{\mathscr{X}} = \llbracket a \rrbracket_{\mathscr{X}^A}$ for all $a \in A$.
	\end{lemma}	
	\begin{proof}
	We show $w \in \llbracket a \rrbracket_{\mathscr{X}}$ iff $w \in \llbracket a \rrbracket_{\mathscr{X}^A}$ for all $a \in A$ and $w \in \GS$ by induction on the length of $w$. 
	\begin{itemize}
		\item For the induction base assume $w = \alpha$, then we deduce:
		\begin{align*}
			\alpha \in \llbracket a \rrbracket_{\mathscr{X}} & \Leftrightarrow \delta(a)(\alpha) = 1 && \textnormal{(Definition of }\llbracket - \rrbracket)  \\
			& \Leftrightarrow \delta^A(a)(\alpha) = 1 && (a \in A) \\
			&\Leftrightarrow  \alpha \in  \llbracket a \rrbracket_{\mathscr{X}^A} && \textnormal{(Definition of }\llbracket - \rrbracket).
		\end{align*}
		\item For the induction step let $w = \alpha p v$, then we find: \begin{align*}
			\alpha p v \in  \llbracket a \rrbracket_{\mathscr{X}} \Leftrightarrow\ & \exists x \in X: \delta(a)(\alpha) = (p, x),\ v \in \llbracket x \rrbracket_{\mathscr{X}} && \textnormal{(Definition of }\llbracket - \rrbracket) \\
			 \Leftrightarrow\ &
			 \exists b \in A: \delta(a)(\alpha) = (p, b),\ v \in \llbracket b \rrbracket_{\mathscr{X}} && (a \in A,\ \delta \textnormal{-inv}) \\
			\Leftrightarrow\ & \exists b \in A: \delta^A(a)(\alpha) = (p, b),\ v \in \llbracket b \rrbracket_{\mathscr{X}^A} && (a, b \in A, \textnormal{ IH)}  \\
			\Leftrightarrow\ & \alpha p v \in \llbracket a \rrbracket_{\mathscr{X}^A} && \textnormal{(Definition of }\llbracket - \rrbracket).
		\end{align*}	
	\end{itemize}
	\end{proof}

In consequence, we immediately obtain that restricting to a $\delta$-invariant subset preserves normality.

\begin{lemma}
	\label{invariantnormal}
		The restriction of a normal $G$-coalgebra to a $\delta$-invariant subset is normal.
	\end{lemma}
		\begin{proof}
		Let $\mathscr{X} = (X, \delta)$ be a normal $G$-coalgebra and $A \subseteq X$ a $\delta$-invariant subset. We write $\mathscr{X}^A = (A, \delta^A)$ for the restriction of $\mathscr{X}$ to $A$. 
			Assume for $a, b \in A$ we have $\delta^A(a)(\alpha) = (p, b)$. Since $a \in A$, we have $\delta(a)(\alpha) = (p,b)$, which by normality of $\mathscr{X}$ implies $\emptyset \not = \llbracket b \rrbracket_{\mathscr{X}}$. From $b \in A$ and \Cref{invariantlanguage} we thus can deduce  $\emptyset \not = \llbracket b \rrbracket_{\mathscr{X}^A}$.
\end{proof}

We will conclude this section with a summarising result. We say that a $G$-automaton $\mathscr{X}$ \emph{satisfies the nesting coequation}, if the final coalgebra homomorphism $\obs_{\mathscr{X}} = \llbracket \cdot \rrbracket: \mathscr{X} \rightarrow \mathcal{L}$ factors through the \emph{coequation} $\lbrace \llbracket e \rrbracket \mid e \in \Expr \rbrace$ -- that is, for any $x \in X$ there exists an expression $e_x \in \Expr$ such that $\llbracket e_x \rrbracket = \llbracket x \rrbracket$. The interested reader will find more details in \cite{dahlqvist2021write, schmid2021guarded}. For our purposes it is sufficient to know that the class of all $G$-automata satisfying the nesting coequation forms a \emph{covariety} \cite{schmid2021guarded}. Covarieties are a categorical dualization of \emph{varieties}, which are well-known from universal algebra (cf. \Cref{varietiessec}). Birkhoff's famous HSP theorem states that varieties are closed under homomorphic images (H), subalgebras (S), and products (P) \cite{birkhoff1935structure}.
 Covarieties enjoy similarly desirable properties: they are closed under homomorphic images, subcoalgebras, and coproducts \cite{dahlqvist2021write}.

\begin{proposition}
\label{reachablesubcoalgebra}
Let $\mathscr{X}$ be a $G$-automaton, then $r(\mathscr{X})$ is well-nested, normal, or satisfies the nesting coequation, whenever $\mathscr{X}$ does. Moreover, $r(\mathscr{X})$ accepts the same language as $\mathscr{X}$.
\end{proposition}
\begin{proof}
	Let us write $\mathscr{X} = (X, \delta, x)$. From \eqref{transitiondef} it is immediate that $r(x) \subseteq X$ is $\delta$-invariant. Thus, \Cref{invariantsimple} and \Cref{invariantnormal}, respectively, imply that $r(\mathscr{X}) = (r(x), \delta^{r(x)}, x)$ is well-nested, or normal, whenever $\mathscr{X}$ is. From \Cref{invariantlanguage} it further follows \[ \llbracket r(\mathscr{X}) \rrbracket = \llbracket x \rrbracket_{r(\mathscr{X})} = \llbracket x \rrbracket_{\mathscr{X}} = \llbracket \mathscr{X} \rrbracket. \]
Since the class of all $G$-automata satisfying the nesting coequation forms a covariety, it is closed under subautomata. As there exists an epi-mono factorisation
 \[ R(\mathscr{X}) \twoheadrightarrow r(\mathscr{X}) \hookrightarrow \mathscr{X} \]
 the automaton $r(\mathscr{X})$ thus satisfies the nesting coequation, whenever $\mathscr{X}$ does. 	
\end{proof}

\subsection{Minimality}

Recall that the state-space of the minimal DFA for a regular language consists of the equivalence classes of the Myhill-Nerode equivalence relation \cite{nerode1958linear}. 

Similarly, we define the state-space of the minimisation of a GKAT automaton $\mathscr{X}$ as the equivalence classes of the equivalence relation $\equiv_{\llbracket \mathscr{X} \rrbracket}$ on $\GSM$, defined for any guarded string language $L \subseteq \GS$ by: 
	\begin{equation}
	\label{equivdef}
		v \equiv_L w :\Leftrightarrow \forall u \in \GS: vu \in L \textnormal{ iff } wu \in L.
	\end{equation}
	Let $v^{-1}L = \lbrace u \in \GS \mid v u \in L \rbrace$ be the derivative of $L$ with respect to $v$. Then two words $v, w$ are equivalent with respect to $\equiv_L$ iff the derivatives $v^{-1}L$ and $w^{-1}L$ coincide. 	

\begin{definition}
\label{minimdef}
	The \emph{minimisation} of a $G$-automaton $\mathscr{X} = (X, \delta, x)$ is $m(\mathscr{X}) := (\lbrace w^{-1}\llbracket \mathscr{X} \rrbracket \mid w \in R(x) \rbrace, \partial, \llbracket \mathscr{X} \rrbracket)$ with:
\begin{align}
\label{minimaltransition}
	\partial(L)(\alpha) := \begin{cases}
		(p, (\alpha p)^{-1}L) & \textnormal{if } (\alpha p)^{-1}L \not= \emptyset \\
		1 & \textnormal{if } \alpha \in L \\
		0 & \textnormal{otherwise}
	\end{cases},
\end{align}
for $L \in \lbrace w^{-1}\llbracket \mathscr{X} \rrbracket \mid w \in R(x) \rbrace$.
\end{definition}

A few remarks on the well-definedness of above definition are in order. The language accepted by a $G$-automaton is deterministic, and taking the derivative of a language preserves its deterministic nature. Thus only one of the three cases in \eqref{minimaltransition} occurs. Since $\varepsilon \in R(x)$ and $\varepsilon^{-1}L = L$, the initial state of the minimisation is well-defined. Transitioning to a new state is well-defined since $v^{-1}(w^{-1}L) = (wv)^{-1}L$.

	It is not hard to see that on a high-level the minimisation can be recovered as the image of the final automata homomorphism $\llbracket - \rrbracket: R(\mathscr{X}) \rightarrow \mathcal{L}$, which, as the result below shows, satisfies $\llbracket w \rrbracket_{R(\mathscr{X})} = w^{-1}\llbracket \mathscr{X} \rrbracket$. 
	
	\begin{lemma}
	Let $\mathscr{X}$ be a $G$-automaton with initial state $x \in X$. Then $\llbracket w \rrbracket_{R(\mathscr{X})} = w^{-1}\llbracket \mathscr{X} \rrbracket$ for all $w \in R(x)$.
\end{lemma}
\begin{proof}
	We prove $u \in 	\llbracket w \rrbracket_{R(\mathscr{X})}$ iff $u \in w^{-1}\llbracket \mathscr{X} \rrbracket$ for all $u \in \GS$ and $w \in R(x)$ by induction on the length of $u$.
	\begin{itemize}
		\item For the induction base assume $u = \alpha$, then  we find:
		\begin{align*}
			\alpha \in 	\llbracket w \rrbracket_{R(\mathscr{X})} &\Leftrightarrow \partial(w)(\alpha) = 1 && \textnormal{(Definition of } \llbracket - \rrbracket)\\
			&\Leftrightarrow \delta(x_w)(\alpha) = 1 &&  \textnormal{(Definition of } \partial) \\
			&\Leftrightarrow \alpha \in \llbracket x_w \rrbracket_{\mathscr{X}} && \textnormal{(Definition of } \llbracket - \rrbracket) \\
			&\Leftrightarrow w\alpha \in \llbracket \mathscr{X} \rrbracket && \textnormal{(Definition of } \llbracket - \rrbracket) \\
			&\Leftrightarrow \alpha \in w^{-1}\llbracket \mathscr{X} \rrbracket && \textnormal{(Definition of } w^{-1}\llbracket \mathscr{X} \rrbracket)
		\end{align*}
	\item For the induction step let $u = \alpha p v$, then it follows:
	\begin{align*}
		&\alpha p v \in \llbracket w \rrbracket_{R(\mathscr{X})} \\
		\Leftrightarrow\ & \partial(w)(\alpha) = (p, w \alpha p ),\ v \in \llbracket w \alpha p \rrbracket_{R(\mathscr{X})} && \textnormal{(Definition of } \llbracket - \rrbracket) \\
		\Leftrightarrow\ & \delta(x_w)(\alpha) = (p, x_{w \alpha p}),\ v \in (w \alpha p)^{-1}\llbracket \mathscr{X} \rrbracket && \textnormal{(Definition of } \partial, \textnormal{ IH}) \\
		\Leftrightarrow\ & \delta(x_w)(\alpha) = (p, x_{w \alpha p}),\ v \in \llbracket x_{w \alpha p}  \rrbracket_{\mathscr{X}} && \textnormal{(Definition of } (w \alpha p)^{-1}\llbracket \mathscr{X} \rrbracket) \\
		\Leftrightarrow\ & \alpha p v \in \llbracket x_w \rrbracket_{{\mathscr{X}}} && \textnormal{(Definition of } \llbracket - \rrbracket) \\
		\Leftrightarrow\ & \alpha p v \in w^{-1}\llbracket \mathscr{X} \rrbracket && \textnormal{(Definition of } w^{-1}\llbracket \mathscr{X} \rrbracket) 	\end{align*}\qedhere
	\end{itemize}
\end{proof}
	
	In other words, there exists an epi-mono factorisation $R(\mathscr{X}) \twoheadrightarrow m(\mathscr{X}) \hookrightarrow \mathcal{L}$.

\subsubsection{Properties of $m(\mathscr{X})$}

\label{properitesofminimal}

In this section we prove properties of $m(\mathscr{X})$, which one would expect to hold by a minimisation construction. We begin by showing that minimising a normal automaton results in a reachable acceptor.

\begin{lemma}
\label{reachableminimal}
Let $\mathscr{X}$ be a normal $G$-automaton with initial state $x \in X$. Then $\llbracket \mathscr{X} \rrbracket \xrightarrow{w} w^{-1}\llbracket \mathscr{X} \rrbracket$ in $m(\mathscr{X})$ for all $w \in R(x)$. In particular, $m(\mathscr{X})$ is reachable. 
\end{lemma}
\begin{proof}
We prove the statement by induction on the length of $w \in R(x)$:
\begin{itemize}
	\item For the induction base, let $w = \varepsilon$, then $\llbracket \mathscr{X} \rrbracket \xrightarrow{\varepsilon} \llbracket \mathscr{X} \rrbracket = \varepsilon^{-1}\llbracket \mathscr{X} \rrbracket$ by the base case of \eqref{transitiondef}.
	\item In the induction step, let $w = v \alpha p$ with $v \in \GSM$. By the definition of reachability, $v \in R(x)$. From the induction hypothesis we deduce $\llbracket \mathscr{X} \rrbracket \xrightarrow{v} v^{-1}\llbracket \mathscr{X} \rrbracket$ in $m(\mathscr{X})$. The normality of $\mathscr{X}$ implies the inequality $(\alpha p)^{-1}(v^{-1}\llbracket \mathscr{X} \rrbracket) \not = \emptyset$. From \eqref{minimaltransition} it thus follows $v^{-1}\llbracket \mathscr{X} \rrbracket \xrightarrow{\alpha p} (\alpha p)^{-1}(v^{-1}\llbracket \mathscr{X} \rrbracket) = w^{-1} \llbracket \mathscr{X} \rrbracket$. We conclude  $\llbracket \mathscr{X} \rrbracket \xrightarrow{w = v \alpha p} w^{-1} \llbracket \mathscr{X} \rrbracket$ by \eqref{transitiondef}. \qedhere
\end{itemize}
\end{proof}

The next result proves that minimisation preserves language semantics.

\begin{lemma}
\label{minimallanguage}
	Let $\mathscr{X}$ be a $G$-automaton, then $\llbracket L \rrbracket = L$ for all $L$ in $m(\mathscr{X})$. In particular, $\llbracket m(\mathscr{X}) \rrbracket = \llbracket \mathscr{X} \rrbracket$.
\end{lemma}
\begin{proof}
We show $v \in \llbracket w^{-1}\llbracket \mathscr{X} \rrbracket \rrbracket$ iff $v \in w^{-1}\llbracket \mathscr{X} \rrbracket$ for all $v \in \GS$, $w \in R(x)$, by induction on the length of $v$:
\begin{itemize}
	\item For the induction base, let $v = \varepsilon$. Then we can compute:
	\begin{align*}
		\alpha \in 	\llbracket w^{-1}\llbracket \mathscr{X} \rrbracket \rrbracket & \Leftrightarrow 
		\partial(w^{-1}\llbracket \mathscr{X} \rrbracket)(\alpha) = 1 && \textnormal{(Definition of } \llbracket - \rrbracket) \\
		& \Leftrightarrow \alpha \in w^{-1}\llbracket \mathscr{X} \rrbracket && \textnormal{(Definition of } \partial)
			\end{align*} 
\item In the induction step, let $v = \alpha p u$. Then we have the following equivalences:
\begin{align*}
	\alpha p u \in 	\llbracket w^{-1}\llbracket \mathscr{X} \rrbracket \rrbracket 
	& \Leftrightarrow u \in \llbracket (w\alpha p)^{-1}\llbracket \mathscr{X} \rrbracket \rrbracket && \textnormal{(Definition of } \llbracket - \rrbracket, \eqref{minimaltransition})\\
	& \Leftrightarrow u \in (w \alpha p)^{-1}\llbracket \mathscr{X} \rrbracket && \textnormal{(IH)} \\
	& \Leftrightarrow w \alpha p u \in \llbracket \mathscr{X} \rrbracket && \textnormal{(Definition of } (-)^{-1}\llbracket \mathscr{X} \rrbracket) \\
	& \Leftrightarrow \alpha p u \in w^{-1}\llbracket \mathscr{X} \rrbracket && \textnormal{(Definition of } (-)^{-1}\llbracket \mathscr{X} \rrbracket)
\end{align*}
		 \end{itemize}
In particular,	$
	\llbracket m(\mathscr{X}) \rrbracket = \llbracket \llbracket \mathscr{X} \rrbracket \rrbracket = \llbracket \varepsilon^{-1}\llbracket \mathscr{X} \rrbracket \rrbracket = \varepsilon^{-1}\llbracket \mathscr{X} \rrbracket = \llbracket \mathscr{X} \rrbracket.
	$	
\end{proof}

An immediate consequence of above statement is that the states of the minimisation can be distinguished by their observable behaviour, that is, different states accept different languages. Another implication of \Cref{minimallanguage} is the normality of the minimisation: all states are \emph{live}.

\begin{corollary}
\label{minimalnormalobservable}
		Let $\mathscr{X}$ be a $G$-automaton, then $m(\mathscr{X})$ is normal and observable.
\end{corollary}
\begin{proof}
\begin{itemize}
	\item By \Cref{minimallanguage}, $\llbracket L_1 \rrbracket = \llbracket L_2 \rrbracket$ implies  $L_1 = L_2$, which shows that $\llbracket - \rrbracket$ is injective. By definition, this proves that $m(\mathscr{X})$ is observable.
	\item Assume $\partial(L_1)(\alpha) = (p, L_2)$, then $L_2 = (\alpha p)^{-1}L_1 \not = \emptyset$ by \eqref{minimaltransition}. It thus follows from 
	 \Cref{minimallanguage} that $\llbracket L_2 \rrbracket = L_2 \not = \emptyset$, which shows that $m(\mathscr{X})$ is normal. \qedhere
\end{itemize}
\end{proof}

Since $m(\mathscr{X})$ is normal, reachable, and observable, if $\mathscr{X}$ is normal, it is, by our definition, \emph{minimal} (cf. \cite[Def. 15]{van2017calf}). Its size-minimality among normal automata language equivalent to $\mathscr{X}$ follows from the abstract definition, cf. \Cref{sizeminimal}.
	
	\subsubsection{Identifying $m(\mathscr{X})$} 
	\label{identifyingsection}

	In this section, we identify the minimisation of a normal $G$-automaton with an alternative, but equivalent, construction. In consequence, we are able to derive that the minimisation of a normal automaton is size-minimal among language equivalent normal automata and preserves the nesting coequation.
	We begin by observing its universality in the following sense.

\begin{figure*}
\centering
\begin{subfigure}[b]{0.4\textwidth}
\centering
\adjustbox{scale=0.8,center}{
			\begin{tikzcd}[row sep=small]
		R(\mathscr{X}) \arrow[twoheadrightarrow]{r}{} \arrow[twoheadrightarrow]{dd}{} & r(\mathscr{X})  \arrow[dashed, twoheadrightarrow]{ddl}{\pi} \\
		& \mathscr{X} \arrow{d} \arrow[hookleftarrow]{u}{} \\
		m(\mathscr{X}) \arrow[hookrightarrow]{r} & \mathcal{L}
		\end{tikzcd}	
	}
		\caption{The morphism $\pi$ as unique diagonal}
				\label{uniquediagonal}
\end{subfigure}
\hfill
\begin{subfigure}[b]{0.55\textwidth}
\centering
\adjustbox{scale=0.8,center}{
			\begin{tikzcd}[row sep=small]
			e \arrow{d} && f \arrow{d}\\
				\mathscr{X}_e \arrow{d} & & \mathscr{X}_e \arrow{d} \\
				\widehat{\mathscr{X}_e} \arrow{d} & & \widehat{\mathscr{X}_f} \arrow{d} \\
				m(\widehat{\mathscr{X}_e}) \arrow{rr}{\cong} & & \arrow{ll} m(\widehat{\mathscr{X}_f})
			\end{tikzcd}
		}
			\caption{$\llbracket e \rrbracket = \llbracket f \rrbracket$ iff  $m(\widehat{\mathscr{X}_e})$ and $m(\widehat{\mathscr{X}_f})$ are isomorphic}
			\label{kozencompletenessdiagram}
\end{subfigure}
\caption{A high-level view of the notions introduced in \Cref{identifyingsection}}
\end{figure*}
	
\begin{proposition}
\label{uniquehomminimal}
Let $\mathscr{X}$ and $\mathscr{Y}$ be normal $G$-automata with $\llbracket \mathscr{X} \rrbracket =  \llbracket \mathscr{Y} \rrbracket$, and $y \in Y$ the initial state of $\mathscr{Y}$. Then $\pi: r(\mathscr{Y}) \rightarrow m(\mathscr{X})$ with $\pi(z) = w_z^{-1}\llbracket \mathscr{X} \rrbracket$, for $y \xrightarrow{w_z} z$ in $\mathscr{Y}$, is a (surjective) $G$-automata homomorphism, uniquely defined.
\end{proposition}
\begin{proof}
We have to show that $\pi$ is well-defined, surjective, preserves initial states, is a $G$-coalgebra homomorphism, and is unique. In this order:
\begin{itemize}
	\item Let $z \in r(y)$, then by definition there exists at least one $w_1 \in R(y)$ such that $y \xrightarrow{w_1} z$ in $\mathscr{Y}$. Since $\mathscr{Y}$  is normal, we have $w_1^{-1} \llbracket \mathscr{X} \rrbracket =  w_1^{-1} \llbracket \mathscr{Y} \rrbracket \not = \emptyset$.  Hence there exists some $z' \in X$, such that $x  \xrightarrow{w_1} z'$ in $\mathscr{X}$, that is, $w_1 \in R(x)$, where $x$ is the initial state of $\mathscr{X}$.  Assume there exists a second $w_2 \in R(y)$, such that $y \xrightarrow{w_2} z$ in $\mathscr{Y}$. Then we have:
	\begin{align*}
		w_1 u \in \llbracket \mathscr{X} \rrbracket & \Leftrightarrow w_1u \in  \llbracket \mathscr{Y} \rrbracket && (\llbracket \mathscr{X} \rrbracket = \llbracket \mathscr{Y} \rrbracket) \\
	 & \Leftrightarrow  u \in \llbracket z \rrbracket && \textnormal{(Definition of } \llbracket - \rrbracket) \\
	& \Leftrightarrow  w_2u \in \llbracket \mathscr{Y} \rrbracket && \textnormal{(Definition of } \llbracket - \rrbracket) \\
	& \Leftrightarrow w_2 u \in \llbracket \mathscr{X} \rrbracket && (\llbracket \mathscr{X} \rrbracket = \llbracket \mathscr{Y} \rrbracket)
	\end{align*}
	for all $u \in \GS$. In other words, $w_1 \equiv_{\llbracket \mathscr{X} \rrbracket} w_2$, or, equivalently, $w_1^{-1}\llbracket \mathscr{X} \rrbracket = w_2^{-1}\llbracket \mathscr{X} \rrbracket$. Thus $\pi$ is a well-defined function.
\item  Let $w \in R(x)$, then by definition there exists $x_w \in X$ with $x \xrightarrow{w} x_w$ in $\mathscr{X}$. Since $\mathscr{X}$ is normal, $w^{-1}\llbracket \mathscr{Y} \rrbracket = w^{-1}\llbracket \mathscr{X} \rrbracket \not = \emptyset$, i.e. $y \xrightarrow{w} y_w$ in $\mathscr{Y}$, for some $y_w \in Y$. Thus, by construction, $\pi(y_w) = w^{-1}\llbracket \mathscr{X} \rrbracket$, which shows that $\pi$ is surjective. 
\item Initial states are preserved since by \eqref{transitiondef} we have $y \xrightarrow{\varepsilon} y$, which by definition of $\pi$ implies $\pi(y) = \varepsilon^{-1}\llbracket \mathscr{X} \rrbracket = \llbracket \mathscr{X} \rrbracket$.
\item $\pi$ is a $G$-coalgebra homomorphism:
\begin{itemize}
	\item Let $\delta^{\mathscr{Y}}(z)(\alpha) = 0$, then $\delta^{m(\mathscr{X})}(\pi(z))(\alpha) \not = 1$, since otherwise $w_z \alpha \in \llbracket \mathscr{X} \rrbracket = \llbracket \mathscr{Y} \rrbracket$ by the definition of $\delta^{m(\mathscr{X})}$, which would imply the contradiction $1 = \delta^{\mathscr{Y}}(z)(\alpha) = 0$. Assume $\delta^{m(\mathscr{X})}(\pi(z))(\alpha) = (p, (w_z \alpha p)^{-1}\llbracket \mathscr{X} \rrbracket)$. By definition of $\delta^{m(\mathscr{X})}$ there exists some $v \in \GS$, such that $w_z \alpha p v \in \llbracket \mathscr{X} \rrbracket = \llbracket \mathscr{Y} \rrbracket$. Hence it follows $\delta^{\mathscr{Y}}(z)(\alpha) \not = 0$ by the definition of $\llbracket - \rrbracket$, which is a contradiction. 
	We can thus conclude $\delta^{m(\mathscr{X})}(\pi(z))(\alpha) = 0$
	\item 
	Let $\delta^{\mathscr{Y}}(z)(\alpha) = 1$, then $w_z \alpha \in \llbracket \mathscr{Y} \rrbracket = \llbracket \mathscr{X} \rrbracket$ by the definition of $\llbracket - \rrbracket$. From the definition of $\delta^{m(\mathscr{X})}$, it follows $\delta^{m(\mathscr{X})}(\pi(z))(\alpha) = 1$.
		\item 
	Let $\delta^{\mathscr{Y}}(z)(\alpha) = (p, z')$, then, by normality of $\mathscr{Y}$, there exists some $v \in \llbracket z' \rrbracket \not = \emptyset$. The latter implies $w_z \alpha p v \in \llbracket \mathscr{Y} \rrbracket = \llbracket \mathscr{X} \rrbracket$. By the definitions of $\delta^{m(\mathscr{X})}$ and $w_{z'}$, it follows $\delta^{m(\mathscr{X})}(\pi(z))(\alpha) = (p, (w_z \alpha p)^{-1} \llbracket \mathscr{X} \rrbracket) = (p, \pi(z'))$.	\end{itemize}
\item Let $g: r(\mathscr{Y}) \rightarrow m(\mathscr{X})$ be any $G$-automata homomorphism.
	Let $z \in r(y)$, then by definition there exists $w_z \in R(y)$, such that $y \xrightarrow{w_z} z$ in $\mathscr{Y}$, and thus in $r(\mathscr{Y})$. Since $g$ is a $G$-automata homomorphism, it follows $\llbracket \mathscr{X} \rrbracket = g(y) \xrightarrow{w_z} g(z)$ in $m(\mathscr{X})$. By \Cref{reachableminimal}, on the other hand, we have $\llbracket \mathscr{X} \rrbracket \xrightarrow{w_z} w_z^{-1}\llbracket \mathscr{X} \rrbracket$ in $m(\mathscr{X})$. From \Cref{uniquenessstates} it thus follows $g(z) = w_z^{-1}\llbracket \mathscr{X} \rrbracket = \pi(z)$. \qedhere
\end{itemize}
\end{proof}	

The next result shows that the minimisation of a normal $G$-automaton is isomorphic to the automaton that arises by identifying semantically equivalent pairs among reachable states. 

\begin{lemma}
\label{minimalbisim}
	Let $\mathscr{X}$ be a normal $G$-automaton with initial state $x \in X$ and $\pi: r(\mathscr{X}) \twoheadrightarrow m(\mathscr{X})$ as in \Cref{uniquehomminimal}, then $y \simeq z$ iff $\pi(y) = \pi(z)$ for all $y, z \in r(x)$. Consequently, $m(\mathscr{X})$ is isomorphic to $r(\mathscr{X})/{\simeq} $.
	\end{lemma}
\begin{proof}
	The statement follows from the following chain of equivalences:
	\begin{align*}
		\pi(y) = \pi(z) &\Leftrightarrow (w_y)^{-1}\llbracket \mathscr{X} \rrbracket = (w_z)^{-1}\llbracket \mathscr{X} \rrbracket && (\textnormal{Definition of } \pi) \\
		& \Leftrightarrow \llbracket y \rrbracket = \llbracket z \rrbracket && \textnormal{(Definition of } (-)^{-1}\llbracket \mathscr{X} \rrbracket) \\
		& \Leftrightarrow y \simeq z && (\mathscr{X}\textnormal{ is normal)}
	\end{align*} \qedhere
\end{proof}
	
On a high level, the automata homomorphism $\pi$ can be recovered as the unique (surjective) diagonal making the diagram in \Cref{uniquediagonal} commute.

In \Cref{reachablesubcoalgebra} it was noted that the reachable subautomaton $r(\mathscr{X})$ satisfies the nesting coequation, whenever $\mathscr{X}$ does. By  \Cref{uniquehomminimal} there exists an epimorphism $\pi: r(\mathscr{X}) \twoheadrightarrow m(\mathscr{X})$, if $\mathscr{X}$ is normal. Since coalgebras satisfying a coequation form a covariety, which is closed under homomorphic images \cite{dahlqvist2021write,schmid2021guarded}, we thus can deduce the following result. 

\begin{corollary}
\label{minimizationcoequation}
		Let $\mathscr{X}$ be a normal $G$-automaton, then $m(\mathscr{X})$ satisfies the nesting coequation, whenever $\mathscr{X}$ does.
\end{corollary}
\begin{proof}
In \Cref{reachablesubcoalgebra} it was noted that the reachable subcoalgebra $r(\mathscr{X})$ satisfies the nesting coequation, whenever $\mathscr{X}$ does. By  \Cref{uniquehomminimal} there exists an epimorphism $\pi: r(\mathscr{X}) \twoheadrightarrow m(\mathscr{X})$ for any normal automaton $\mathscr{X}$. The claim follows since coalgebras satisfying a coequation form a covariety, which is in particular closed under homomorphic images \cite{dahlqvist2021write,schmid2021guarded}.
\end{proof}

We continue with the observation that two normal $G$-automata are language equivalent if and only if their minimisations are isomorphic. As depicted in \Cref{kozencompletenessdiagram}, this implies that two expressions $e$ and $f$ are language equivalent if and only if the minimisations of their normalised Thompson automata are isomorphic. A similar idea occurs in Kozen's completeness proof for Kleene Algebra \cite[Theorem 19]{kozen1994completeness}.   
	
\begin{corollary}
\label{minimalunique}
	Let $\mathscr{X}$ and $\mathscr{Y}$ be normal $G$-automata, then $\llbracket \mathscr{X} \rrbracket = \llbracket \mathscr{Y} \rrbracket$ iff $m(\mathscr{X}) \cong m(\mathscr{Y})$. 
\end{corollary}
\begin{proof}
	 We begin by assuming $\llbracket \mathscr{X} \rrbracket = \llbracket \mathscr{Y} \rrbracket$. From \Cref{minimallanguage} and \Cref{minimalnormalobservable} we know that $m(\mathscr{X})$ and $m(\mathscr{Y})$ are normal and accept the same language as $\mathscr{X}$ and $\mathscr{Y}$. From \Cref{reachableminimal} it follows that $m(\mathscr{X})$ and $m(\mathscr{Y})$ are reachable. \Cref{uniquehomminimal} thus implies that there exist $G$-automata homomorphisms $\pi_1: m(\mathscr{Y}) \rightarrow m(\mathscr{X})$ and $\pi_2: m(\mathscr{X}) \rightarrow m(\mathscr{Y})$. From the uniqueness property in \Cref{uniquehomminimal} we deduce $\pi_1 \pi_2 = \id_{m(\mathscr{X})}$ and $\pi_2 \pi_1 = \id_{m(\mathscr{Y})}$. Hence $\pi_2: m(\mathscr{X}) \rightarrow m(\mathscr{Y})$ is an isomorphism with inverse $\pi_1$.
	 
	 Conversely, assume $m(\mathscr{X})$ is isomorphic to $m(\mathscr{Y})$. Then it immediately follows $\llbracket m(\mathscr{X}) \rrbracket = \llbracket  m(\mathscr{Y}) \rrbracket$, which implies $\llbracket  \mathscr{X} \rrbracket = \llbracket \mathscr{Y} \rrbracket$ by \Cref{minimallanguage}.
\end{proof}	

We conclude with the size-minimality of the minimisation of a normal automaton among language equivalent normal automata.

\begin{corollary}
\label{sizeminimal}
Let $\mathscr{X}$ and $\mathscr{Y}$ be normal $G$-automata with $\llbracket \mathscr{X} \rrbracket = \llbracket \mathscr{Y} \rrbracket$. Then $\vert m(\mathscr{X}) \vert \leq \vert \mathscr{Y} \vert$, where $\vert m(\mathscr{X}) \vert = \vert \mathscr{Y} \vert$ iff $m(\mathscr{X}) \cong \mathscr{Y}$.
\end{corollary}
\begin{proof}
From \Cref{minimalunique} it immediately follows $m(\mathscr{X}) \cong m(\mathscr{Y})$. We additionally observe \Cref{uniquediagonal} to derive
\[
\vert m(\mathscr{X}) \vert = \vert m(\mathscr{Y}) \vert \leq \vert r(\mathscr{Y}) \vert \leq \vert \mathscr{Y} \vert.
\]
We show next $\vert m(\mathscr{X}) \vert = \vert \mathscr{Y} \vert$ iff $m(\mathscr{X}) \cong \mathscr{Y}$:
\begin{itemize}
	\item Assume $m(\mathscr{X}) \cong \mathscr{Y}$, then immediately $\vert m(\mathscr{X}) \vert = \vert \mathscr{Y} \vert$.
	\item Assume $\vert m(\mathscr{X}) \vert = \vert \mathscr{Y} \vert$, then \Cref{uniquehomminimal} and \Cref{uniquediagonal} imply:
	\[
	\vert r(\mathscr{Y}) \vert \geq \vert m(\mathscr{X}) \vert = \vert \mathscr{Y} \vert \geq \vert r(\mathscr{Y}) \vert.
	\]
	It thus follows $\vert m(\mathscr{X}) \vert = \vert r(\mathscr{Y}) \vert = \vert \mathscr{Y} \vert$. From the second equality and the definition of $r(\mathscr{Y})$ it immediately follows $\mathscr{Y} \cong r(\mathscr{Y})$. The first equality implies that the epimorphism $\pi: r(\mathscr{Y}) \twoheadrightarrow m(\mathscr{X})$ in \Cref{uniquehomminimal} is a bijective automata homomorphism. Any bijective coalgebra homomorphism is a coalgebra isomorphism \cite[Prop. 2.3]{rutten2000universal}. It is clear that the inverse of an initial state preserving coalgebra isomorphism preserves initial states as well. Thus $\pi: r(\mathscr{Y}) \cong m(\mathscr{X})$ is an $G$-automata isomorphism. \qedhere
\end{itemize}
\end{proof}

\section{Learning $m(\mathscr{X})$}	 

\label{learningpart}

In this section we formally investigate the correctness of $\GLStar$ (\Cref{GlStaralgorithm}). Our main result is \Cref{correctnesstheorem}, which shows that if the oracle is instantiated with a deterministic language accepted by a finite normal $G$-automaton $\mathscr{X}$, then  $\GLStar$ terminates with a hypothesis isomorphic to $m(\mathscr{X})$.
For calculations, it will be convenient to use the following definition of an observation table.

\begin{definition}
\label{observationtable}
	An \emph{observation table} $T = (S, E, \row)$ consists of subsets $S \subseteq \GSM, E \subseteq \GS$ and a function $\row: S \cup S \cdot (\At \cdot \Sigma) \rightarrow 2^E$, such that:
	\begin{itemize}
		\item $\varepsilon \in S$ and $\At \subseteq E$ 		
		\item $\alpha p e \in E$ implies $e \in E$ (suffix-closed)	
		\item $s \alpha p \in S$ implies $s \in S$ (prefix-closed)
		\item $s \not = t$ implies $\row(s) \not = \row(t)$ for $s,t \in S$
		\item $\varepsilon \not = s \in S$ implies $\row(s)(e) = 1$ for some $e \in E$
		\item $\row(s\alpha p)(e) = \row(s)(\alpha p e)$, if $\alpha p e \in E$
	\end{itemize}
\end{definition}

Not every table induces a well-defined $G$-automaton. To ensure correctness, we have to restrict ourselves to a  subclass of tables that satisfies two important properties.
We call an observation table \emph{deterministic} if the guarded string language $\row(s) \subseteq \GS$ is deterministic for all $s \in S$. 
An observation table is \emph{closed}, if for all $t \in  S \cdot (\At \cdot \Sigma)$ with $\row(t)(e) = 1$ for some $e \in E$, there exists an $s \in S$ such that $\row(s) = \row(t)$.

The result below shows that if the oracle is instantiated with a deterministic language accepted by a finite normal $G$-automaton $\mathscr{X}$, we have a well-defined observation table at every step.

\begin{proposition}
\label{welldefinedeverystep}
	If $\Cref{GlStaralgorithm}$ is instantiated with a deterministic language accepted by a finite normal $G$-automaton $\mathscr{X}$, then $T$ is a well-defined deterministic observation table at every step.
\end{proposition}
\begin{proof}
	\begin{itemize}
		\item Any $G$-automaton accepts a deterministic language. Since $\row(s) \subseteq  s^{-1} \llbracket \mathscr{X} \rrbracket$, and determinacy is preserved under derivatives, the determinacy of $T$ is thus implied by the determinacy of $\llbracket \mathscr{X} \rrbracket$. 
		\item In the initial step we have $S = \lbrace \varepsilon \rbrace$ and $E = \At$. In every step that follows the sets $S$ and $E$ are only extended. We thus have $\varepsilon \in S$ and $\At \subseteq E$ in every step. 
		\item In the initial step $S = \lbrace \varepsilon \rbrace$ and $E = \At$ are clearly prefix and suffix closed, respectively. In the following steps $S$ is only extended with strings of the shape $s\alpha p$ for $s \in S$, and $E$ is only extended with the suffixes of some $z$. Hence $S$ and $E$ are prefix and suffix closed, respectively, in every step.
		\item In the initial step $S = \lbrace \varepsilon \rbrace$, hence all rows indexed by elements in $S$ are trivially disjoint. In the following steps we only add $s \alpha p$ to $S$, if $\row(s \alpha p) \not = \row(t)$ for all $t \in S$. Since disjoint rows do at no point collapse, we can deduce that $s \not = t$ implies $\row(s) \not = \row(t)$ for all $s,t \in S$ in every step.
		\item In the initial step we have $S = \lbrace \varepsilon \rbrace$, thus the observation that for all $s \in S$ with $s \not = \varepsilon$ we have $\row(s) \not = \emptyset$ is trivially true. In the following steps we only add elements $s \alpha p$ with $\row(s \alpha p) \not = \emptyset$ to $S$. 
		\item Since $\row(t)(e) = \llbracket \mathscr{X} \rrbracket(te)$, the identity $\row(s \alpha p)(e) = \row(s)(\alpha p e)$, if $\alpha p e \in E$, follows from the associativity of string concatenation. \qedhere
		\end{itemize}
\end{proof}

Any closed deterministic table induces a $G$-automaton in the following way.

\begin{definition}
\label{hypothesis}
Given a closed deterministic observation table $T = (S, E, \row)$, let $m(T) := (\lbrace \row(s) \mid s \in S \rbrace, \delta, \row(\varepsilon))$ be the $G$-automaton with
	\begin{align}
	\label{hypothesisdef}
	\delta(L)(\alpha) = \begin{cases}
		(p, (\alpha p)^{-1}L) & \textnormal{if } (\alpha p)^{-1}L \not = \emptyset \\
		1 & \textnormal{if } \alpha \in L \\
		0 & \textnormal{otherwise}
	\end{cases},
\end{align}
where $L \in \lbrace \row(s) \mid s \in S \rbrace$ and $(\alpha p)^{-1}\row(s) = \row(s \alpha p)$.
\end{definition}

A few remarks on the well-definedness of the above definition are in order.
	By \Cref{observationtable} the upper-rows of an observation table are disjoint. Since $T$ is deterministic, precisely one of the three cases in \eqref{hypothesisdef} occurs.
	  If $(\alpha p)^{-1}\row(s)$ is non-empty, there exists, because $T$ is closed, some $t \in S$ with $(\alpha p)^{-1}\row(s) = \row(t)$. This shows that $m(T)$ is closed under transitions.
\subsection{Properties of $m(T)$}

In what follows, let $T$ be a closed deterministic observation table, unless states otherwise. We will establish a few basic properties of $m(T)$. 
First, we observe its reachability, which is implied by a slightly stronger statement. 

\begin{lemma}
\label{reachability}
For all $s \in S$ and $t \in \GSM$ such that $st \in S$, we have $\row(s) \xrightarrow{t} \row(st)$ in $m(T)$. In particular, $m(T)$ is reachable.
\end{lemma}
\begin{proof}
We show the statement by induction on the length of $t \in \GSM$:
	\begin{itemize}
		\item If $t = \varepsilon$, the statement follows from the base case of \eqref{transitiondef}, i.e. $\row(s) \xrightarrow{\varepsilon} \row(s)$.
		\item If $t = v \alpha p$ for $v \in \GSM$, we have $s v \in S$, since $s v \alpha p = st \in S$ and $S$ is prefix closed by \Cref{observationtable}. Thus $\row(s) \xrightarrow{v} \row(s v)$ by the induction hypothesis. Since $\varepsilon \not = st \in S$, we have $(\alpha p)^{-1}\row(sv) = \row(st) \not = \emptyset$ by \Cref{observationtable}. Thus it follows $\row(sv) \xrightarrow{\alpha p} \row(sv\alpha p)$ by \eqref{hypothesisdef}.  We conclude $\row(s)  \xrightarrow{v \alpha p = t} \row(s v \alpha p) = \row(st)$ by \eqref{transitiondef}.
	\end{itemize}
	Since $\epsilon \in S$ by \Cref{observationtable}, we in particular obtain $\row(\varepsilon)  \xrightarrow{s} \row(s)$ in $m(T)$ for all $s \in S$, which implies the reachability of $m(T)$.
\end{proof}

We call a $G$-automaton $(\mathscr{Y}, y)$ \emph{consistent with $T$}, if $S \subseteq R(y)$ and  $\llbracket y_s \rrbracket(e) = \row(s)(e)$ for all $s \in S$, $e \in E$, and $y_s \in Y$ with $y  \xrightarrow{s} y_s$.
By \Cref{reachability}, the automaton $m(T)$ is consistent with $T$ if and only if $\llbracket \row(s) \rrbracket(e) = \row(s)(e)$ for all $s \in S$ and $e \in E$. The consistency of $m(T)$ with $T$ should not be confused with the consistency of $T$ itself. Both terminologies appear frequently in the literature \cite{angluin1987learning}. We show that $m(T)$ is not only consistent with $T$, but has in fact the fewest number of states among all automata consistent with $T$.

\begin{lemma}
\label{consistent}
	$m(T)$ is size-minimal among automata consistent with $T$.
\end{lemma}
\begin{proof}
We begin by showing that $m(T)$ is consistent with $T$, that is, it satisfies $\llbracket \row(s) \rrbracket(e) = \row(s)(e)$ for all $s \in S, e \in E$, by induction on the length of $e$:
		\begin{itemize}
		\item For the induction base, let $e = \alpha \in \At$, then it immediately follows that:
		\begin{align*}
			\llbracket \row(s) \rrbracket(e) = 1 			 & \Leftrightarrow  \delta(\row(s))(\alpha) = 1 && \textnormal{(Definition of } \llbracket - \rrbracket) \\
			& \Leftrightarrow \row(s)(\alpha) = 1 && \eqref{hypothesisdef}
		\end{align*}
		\item For the induction step, let $e = \alpha p w$ for $w \in \GS$, then \Cref{observationtable} implies $w \in E$ and $\row(s)(\alpha p w) = \row(s \alpha p)(w)$. Thus we can deduce: 
		\begin{align*}
			& \llbracket \row(s) \rrbracket(\alpha p w ) = 1 && \\
			\Leftrightarrow &\ \exists t \in S: \emptyset \not = \row(s \alpha p) = \row(t) \textnormal{ and } \llbracket \row(t) \rrbracket(w) = 1 && \textnormal{(Definition of } \llbracket - \rrbracket) \\
			\Leftrightarrow &\ \exists t \in S: \row(s \alpha p) = \row(t) \textnormal{ and } \row(t)(w) = 1  && (w \in E,\ \textnormal{IH}) \\
			\Leftrightarrow &\ \row(s \alpha p)(w) = 1 && (T \textnormal{ closed}) \\
			\Leftrightarrow &\ \row(s)(\alpha p w) = 1 && (\textnormal{\Cref{observationtable}})
		\end{align*}
	\end{itemize}
	Let $(\mathscr{Y}, y)$ be any $G$-automaton consistent with $T$, i.e.\ $S \subseteq R(y)$ and $\llbracket y_s \rrbracket(e) = \row(s)(e)$ for all $s \in S, e \in E$, and $y_s \in Y$ with $y  \xrightarrow{s} y_s$. We define a function $f: \lbrace \row(s) \mid s \in S \rbrace \rightarrow Y$ by $f(\row(s)) = y_s$. The function is well-defined, since $S \subseteq R(y)$. Assume $f(\row(s)) = f(\row(t))$, i.e. $y_s = y_t$ for $s,t \in S$. Then we can deduce
	\[
	\row(s)(e) = \llbracket y_s \rrbracket(e) = \llbracket y_t \rrbracket(e) = \row(t)(e)
	\]
	for all $e \in E$. Since by \Cref{observationtable} rows indexed by $S$ are disjoint, it follows $s = t$. This shows that $f$ is injective, which implies the size-minimality of $m(T)$.
\end{proof}

From the consistency of $m(T)$ with $T$ it is straightforward to derive its normality and observability.

\begin{lemma}
\label{normalityandobservable}
	$m(T)$ is normal and observable.
\end{lemma}
\begin{proof}
\begin{itemize}
	\item Assume $\delta(\row(s))(\alpha) = (p, \row(t))$ for $s, t \in S$. Then we have \[ \row(t) = \row(s \alpha p) = (\alpha p)^{-1}\row(s) \not = \emptyset \] by \eqref{hypothesisdef}. From \Cref{consistent} $m(T)$ it follows that $\llbracket \row(t) \rrbracket$ is non-empty, which proves the normality of $m(T)$.
		\item Assume $\llbracket \row(s) \rrbracket = \llbracket \row(t) \rrbracket$ for $s, t \in S$. Then, by \Cref{consistent}, we have
	\[
	\row(s)(e) = \llbracket \row(s) \rrbracket (e) = \llbracket \row(t) \rrbracket (e) = \row(t)(e)
	\]
	for all $e \in E \subseteq \GS$. Thus $\row(s) = \row(t)$, which shows that $\llbracket - \rrbracket$ is injective, that is, $m(T)$ is observable. \qedhere
\end{itemize}
\end{proof}

\subsection{Relationship Between $m(T)$ and $m(\mathscr{X})$}

We will next deduce the correctness of $\GLStar$, that is, its termination with an automaton isomorphic to $m(\mathscr{X})$, if the teacher is instantiated with the language accepted by a finite normal automaton $\mathscr{X}$. 

In a first step we establish that any hypothesis admits an injective function from its state-space into the state-space of $m(\mathscr{X})$. The result below does not necessarily require the observation table to be deterministic or closed.

\begin{lemma}
\label{injectivefunction}
	Let $T = (S, E, \row)$ be an observation table with $\row(t)(e) = \llbracket \mathscr{X} \rrbracket (te)$ for all $t \in S \cup S \cdot (\At \cdot \Sigma),\ e \in E$, and let $x \in X$ be the initial state of $\mathscr{X}$. Then $\pi: \lbrace \row(s) \mid s \in S \rbrace \rightarrow \lbrace w^{-1}\llbracket \mathscr{X} \rrbracket \mid w \in R(x) \rbrace,\ \row(s) \mapsto s^{-1}\llbracket \mathscr{X} \rrbracket$ is a well-defined injective function.
\end{lemma}
\begin{proof}
We first show that $\pi$ is well-defined. To this end, we need to establish that i) $S \subseteq R(x)$; and ii) if $\row(s) = \row(t)$ for $s,t \in S$, then  $s^{-1}\llbracket \mathscr{X} \rrbracket = t^{-1}\llbracket \mathscr{X} \rrbracket$.  

For i) note that if $s = \varepsilon$, then $x  \xrightarrow{s} x$ by the base case of \eqref{transitiondef}, i.e. $s \in R(x)$. If $s \not = \varepsilon$, then \Cref{observationtable} implies the existence of some $e \in E$, such that $\row(s)(e) = 1$. Thus $\llbracket x \rrbracket(se) = \llbracket \mathscr{X} \rrbracket(se) = \row(s)(e) = 1$, which implies $s \in R(x)$ by the definition of $\llbracket - \rrbracket$. 
For ii) it is enough to observe that by \Cref{observationtable} all rows of an observation table are disjoint.

 To show that $\pi$ is injective, assume $\pi(\row(s)) = \pi(\row(t))$, for $s, t \in S$. By definition of $\pi$ we thus have an equivalence $s \equiv_{\llbracket \mathscr{X} \rrbracket} t$. From the definition of $\equiv_{\llbracket \mathscr{X} \rrbracket}$ and the assumptions it thus follows
 \[ e \in \row(s) \Leftrightarrow se \in \llbracket \mathscr{X} \rrbracket \Leftrightarrow te \in \llbracket \mathscr{X} \rrbracket \Leftrightarrow e \in \row(t) \] for all $e \in E$. This proves the equality $\row(s) = \row(t)$.
\end{proof}

If the algorithm terminates with a hypothesis $m(T)$, the latter is, by definition, language equivalent to $\mathscr{X}$, and thus to the minimisation $m(\mathscr{X})$, by \Cref{minimallanguage}. The next result implies a stronger statement: in case of termination, the hypothesis $m(T)$ is \emph{isomorphic} to $m(\mathscr{X})$, via the function $\pi$ of \Cref{injectivefunction}.

\begin{proposition}
\label{isolanguagequiv}
	Let $T = (S, E, \row)$ be a closed deterministic observation table with $\row(t)(e) = \llbracket \mathscr{X} \rrbracket(te)$ for all $t \in S \cup S \cdot (\At \cdot \Sigma),\ e \in E$. Let $\pi$ be the injection of \Cref{injectivefunction}, and $\mathscr{X}$ normal. The following are equivalent:
\begin{enumerate}
	\item $\pi: m(T) \simeq m(\mathscr{X})$ is a $G$-automata isomorphism
	\item $\llbracket m(T) \rrbracket =  \llbracket m(\mathscr{X}) \rrbracket$
\end{enumerate}
\end{proposition}
\begin{proof}
	\begin{itemize}
		\item $1. \rightarrow 2.$: Since $\pi$ is a homomorphism, it follows $\llbracket - \rrbracket_{m(\mathscr{X})} \circ \pi = \llbracket - \rrbracket_{m(T)}$ by uniqueness. In particular we have: 
		\begin{align*}
			\llbracket m(T) \rrbracket_{m(T)} &= \llbracket \row(\varepsilon) \rrbracket_{m(T)} && (\textnormal{Definition of } \llbracket - \rrbracket_{m(T)})  \\
			 &= \llbracket \pi(\row(\varepsilon)) \rrbracket_{m(\mathscr{X})} && (\llbracket - \rrbracket_{m(\mathscr{X})} \circ \pi = \llbracket - \rrbracket_{m(T)}) \\
			 &= \llbracket \varepsilon^{-1}\mathscr{X} \rrbracket_{m(\mathscr{X})} && (\textnormal{Definition of } \pi) \\
			  &= \llbracket \mathscr{X} \rrbracket_{m(\mathscr{X})} && (\textnormal{Definition of } \llbracket - \rrbracket_{m(\mathscr{X})})
		\end{align*}
		\item $2. \rightarrow 1.:$ By \Cref{normalityandobservable} and \Cref{reachability}, $m(T)$ is normal and reachable. From the assumption and \Cref{minimallanguage} it follows $\llbracket m(T) \rrbracket = \llbracket m(\mathscr{X}) \rrbracket = \llbracket \mathscr{X} \rrbracket$. By assumption $\mathscr{X}$ is normal. By \Cref{uniquehomminimal} there thus exists a unique surjective automata homomorphism  $f \colon m(T) = r(m(T)) \rightarrow m(\mathscr{X})$ that satisfies $f(\row(s)) = w_s^{-1}\llbracket \mathscr{X} \rrbracket$ for $\row(\varepsilon) \overset{w_s}{\rightarrow} \row(s)$ in $m(T)$. From \Cref{reachability} it follows that $w_s = s$. Therefore the definitions of $\pi$ and $f$ coincide, $\pi = f$. Since by \Cref{injectivefunction} the function $\pi = f$ is injective, it is a bijective coalgebra homomorphism. Any bijective coalgebra homomorphism is a coalgebra isomorphism \cite[Prop. 2.3]{rutten2000universal}. It is clear that the inverse of an initial state preserving coalgebra isomorphism preserves initial states as well. It thus follows that $\pi = f$ is a $G$-automata isomorphism. \qedhere
	\end{itemize}
\end{proof}

The main argument in the proof of \Cref{correctnesstheorem} is \Cref{ifclosedthen}. To prove the latter, we need the following two results. Both results assume two closed deterministic tables $T$ and $T'$, with the latter extending the former. The first statement, \Cref{mT'impliesT}, relates the transition function of $m(T)$ to the one of $m(T')$.

\begin{lemma}
	\label{mT'impliesT}
	Let $T = (S, E, \row)$ and $T' = (S, E', \row')$ be closed deterministic observation tables with $E \subseteq E'$ and $\row(t)(e) = \row'(t)(e)$ for all $t \in S \cup S \cdot (\At \cdot \Sigma),\ e \in E$. Let $m(T)$ and $m(T')$ have transition functions $\delta$ and $\delta'$, respectively, then for all $s, t \in S$: 
	\begin{itemize}
		\item $\delta'(\row'(s))(\alpha) = 1$ iff $\delta(\row(s))(\alpha) = 1$
		\item $\delta'(\row'(s))(\alpha) = (p, \row'(t))$ implies $\delta(\row(s))(\alpha) = (p, \row(t))$ or \\$\delta(\row(s))(\alpha) = 0$
		\item $\delta'(\row'(s))(\alpha) = 0$ implies $\delta(\row(s))(\alpha) = 0$
	\end{itemize}
\end{lemma}
\begin{proof}
	\begin{itemize}
		\item For the first point we deduce:
		\begin{align*}
		 \delta'(\row'(s))(\alpha) = 1
					& \Leftrightarrow \row'(s)(\alpha) = 1 && (\textnormal{Definition of } \delta') \\
			& \Leftrightarrow \row(s)(\alpha) = 1 && (\alpha \in \At \subseteq E) \\
			& \Leftrightarrow  \delta(\row(s))(\alpha) = 1 && (\textnormal{Definition of } \delta)
		\end{align*} 
		
			\item For the second point, assume $\delta'(\row'(s))(\alpha) = (p, \row'(t))$ for $t \in S$ with $\row'(s \alpha p) = \row'(t)$. Then by the first point \[ \delta(\row(s))(\alpha) = 0, \quad \textnormal{or} \quad \delta(\row(s))(\alpha) = (p, \row(u)) \] for some $u \in S$ with $\row(s \alpha p) = \row(u)$. We further have
			 \[
			\row(t)(e) = \row'(t)(e) = \row'(s \alpha p)(e) = \row(s \alpha p)(e) = \row(u)(e) \]
			 for all $e \in E$. In other words, we have derived $\row(t) = \row(u)$. 
			 		\item For the last point, assume $\delta'(\row'(s))(\alpha) = 0$. Then by the first point \[ \delta(\row(s))(\alpha) = 0, \quad \textnormal{or} \quad \delta(\row(s))(\alpha) = (p, \row(t)) \] for some $t \in S$ with $\row(s \alpha p) = \row(t)$. By the definition of $\delta$, the latter case implies \[ \row'(s \alpha p)(e) = \row(s \alpha p)(e) = 1 \] for some $e \in E \subseteq E'$. It thus follows $\delta'(row'(s))(\alpha) \not \in 2$, which contradicts the assumption $\delta'(\row'(s))(\alpha) = 0$. We thus can conclude $\delta(\row(s))(\alpha) = 0$. \qedhere
			 	\end{itemize}
\end{proof}

The second statement, \Cref{languageinclusion}, establishes an inclusion of the language semantics of $m(T)$ into the language semantics of $m(T')$.

\begin{lemma}
\label{languageinclusion}
	Let $T = (S, E, \row)$ and $T' = (S, E', \row')$ be closed deterministic observation table with $E \subseteq E'$ and $\row(t)(e) = \row'(t)(e)$ for all $t \in S \cup S \cdot (\At \cdot \Sigma),\ e \in E$. Then $\llbracket \row(s) \rrbracket_{m(T)} \subseteq \llbracket \row'(s) \rrbracket_{m(T')}$ for all $s \in S$.
\end{lemma}
\begin{proof}
We show $w \in \llbracket \row(s) \rrbracket_{m(T)} $ implies  $w \in \llbracket \row'(s) \rrbracket_{m(T')}$ for all $s \in S$ and $w \in \GS$ by induction on $w$. We denote the transition functions of $m(T)$ and $m(T')$ by $\delta$ and $\delta'$, respectively.
	\begin{itemize}
		\item For the induction base, assume $w = \alpha$. Then we deduce: 
		\begin{align*}
		\alpha \in \llbracket \row(s) \rrbracket_{m(T)} & \Leftrightarrow 	\delta(\row(s))(\alpha) = 1 && \textnormal{(Definition of } \llbracket - \rrbracket ) \\
		& \Leftrightarrow \delta'(\row'(s))(\alpha) = 1 && \textnormal{(\Cref{mT'impliesT}}) \\
		& \Leftrightarrow \alpha \in \llbracket \row'(s) \rrbracket_{m(T')} && \textnormal{(Definition of } \llbracket - \rrbracket )
		\end{align*}
				\item In the induction step, let $w = \alpha p v$. Then we have:
				\begin{align*}
					& \alpha p v \in \llbracket \row(s) \rrbracket_{m(T)} && \\
					 \Leftrightarrow &\   \textnormal{(Definition of } \llbracket - \rrbracket ) \\
					 & \exists t \in S: \delta(\row(s))(\alpha) = (p, \row(t)),\ v \in \llbracket \row(t) \rrbracket_{m(T)} \\
					 \Rightarrow &\  \textnormal{(\Cref{mT'impliesT}, IH)} \\
					&  \exists t \in S: \delta'(\row'(s))(\alpha) = (p, \row'(t)),\ v \in \llbracket \row'(t) \rrbracket_{m(T')}  \\
					\Leftrightarrow &\ \textnormal{(Definition of } \llbracket - \rrbracket ) \\
					& \alpha p v \in \llbracket \row'(s) \rrbracket_{m(T')} 
				\end{align*} \qedhere
		 \end{itemize}
\end{proof}

The next result shows that, if the oracle replies \emph{no} to an equivalence query and provides us with a counterexample $z$, then the table extended with the suffixes of $z$ can immediately be closed only if it is the first time such a situation occurs.

\begin{proposition}
\label{ifclosedthen}
 	Let $T = (S, E, \row)$ be a closed deterministic observation table with $\row(t)(e) = \llbracket \mathscr{X} \rrbracket(te)$ for all $t \in S \cup S \cdot (\At \cdot \Sigma),\ e \in E$. Let $\llbracket m(T) \rrbracket(z) \not = \llbracket \mathscr{X} \rrbracket(z)$ for some $z \in \GS$, and $T' = (S, E \cup \suff(z), \row')$ with $\row'(t)(e) = \llbracket \mathscr{X} \rrbracket(te)$. If $T'$ is closed, then $\row'(\varepsilon)(e) = 0$ for all $e \in E$, but $\row'(\varepsilon)(z') = 1$ for some $z' \in \suff(z)$. 
 \end{proposition}
\begin{proof}
	\begin{itemize}
		\item Assume $\llbracket \mathscr{X} \rrbracket(z) = 0$ and $\llbracket m(T) \rrbracket(z) = 1$. Since $\llbracket m(T) \rrbracket \subseteq \llbracket m(T') \rrbracket$ by \Cref{languageinclusion}, it follows $\llbracket m(T') \rrbracket(z) = 1$. From \Cref{consistent} and the global assumptions we thus can deduce 
		\[ 1 = \llbracket m(T') \rrbracket(z) = \llbracket \row'(\varepsilon) \rrbracket(z) = \row'(\varepsilon)(z) = \llbracket \mathscr{X} \rrbracket(z), \] which contradicts $0 = \llbracket \mathscr{X} \rrbracket(z)$.
		\item Assume $\llbracket \mathscr{X} \rrbracket(z) = 1$ and $\llbracket m(T) \rrbracket(z) = 0$. From \Cref{consistent} and the global assumptions we can deduce
		\[
		1 = \llbracket \mathscr{X} \rrbracket(z) = \row'(\varepsilon)(z) = \llbracket \row'(\varepsilon) \rrbracket(z) = \llbracket m(T') \rrbracket(z).
		\] By \Cref{mT'impliesT}, there exists some decomposition $z = v\alpha pz'$ for $v \in \GSM, z' \in \suff(z)$, such that for some $t \in S$:
		\begin{enumerate}
		\item $\row'(\varepsilon)  \xrightarrow{v} \row'(t)$ in $m(T')$ and $\row(\varepsilon)  \xrightarrow{v} \row(t)$ in $m(T)$;
		\item $\delta'(\row'(t))(\alpha) = (p, \row'(t'))$ for some $t' \in S$, but $\delta(\row(t))(\alpha) = 0$.
		\end{enumerate}
		For all $e \in E$ it follows \[ 0 = \row(t \alpha p)(e) = \row'(t \alpha p)(e) = \row'(t')(e) = \row(t')(e), \] since otherwise $\delta(\row(t))(\alpha) \not = 0$ by the definition of $\delta$. From \Cref{observationtable} we can deduce $t' = \varepsilon$. Since $(v \alpha p) z' = z \in \llbracket m(T') \rrbracket = \llbracket \row'(\varepsilon) \rrbracket$ and $\row'(\varepsilon)  \xrightarrow{v \alpha p} \row'(t') = \row'(\varepsilon)$, it follows $z' \in \llbracket \row'(\varepsilon) \rrbracket$ by the definition of $\llbracket - \rrbracket$. From \Cref{consistent} we can conclude $z' \in \row'(\varepsilon)$.\qedhere
	\end{itemize}
\end{proof}
 
 In consequence, an infinite chain of negative equivalence queries and immediately closed extended tables is impossible. Since fixing a closedness defect increases the size of $m(T)$, which by \Cref{injectivefunction} is bounded by the finite number of states in $m(\mathscr{X})$, we can deduce the correctness of \Cref{GlStaralgorithm}.

\begin{theorem}
\label{correctnesstheorem}
	If $\Cref{GlStaralgorithm}$ is instantiated with the language accepted by a finite normal automaton $\mathscr{X}$, then it terminates with a hypothesis isomorphic to $m(\mathscr{X})$.
\end{theorem}
\begin{proof}
By \Cref{welldefinedeverystep}, $T$ is a well-defined deterministic observation table at every step.
	We continue by showing that the algorithm yields $m(\mathscr{X})$ in finitely many steps.
		By \Cref{injectivefunction} we have $\vert X \vert \leq \vert Y \vert$ for $X = \lbrace \row(s) \mid s \in S \rbrace$ and $Y = \lbrace w^{-1}\llbracket \mathscr{X} \rrbracket \mid w \in R(x) \rbrace$ at any point of the algorithm. Since $\mathscr{X}$ is finite, the state-space $Y$ of $m(\mathscr{X})$ is finite. At no point of the algorithm does the size of $X$ decrease. Resolving a closedness defect strictly increases the size of $X$. Hence a closedness defect can only occur finitely many times.
		The only way the algorithm could not terminate is thus an infinite chain of negative equivalence queries, for which the subsequent suffix-enriched table is immediately closed again. By applying \Cref{ifclosedthen} twice, one observes that such a case can not occur.
\end{proof}

\section{Comparison with Moore Automata}

\label{comparisonmoore}

How are the minimal GKAT automaton (\Cref{GLStarmT3}) and the minimal Moore automaton (\Cref{LStarmT4}) representing the guarded deterministic language \eqref{acceptedlanguageexample} related? Why should we learn the former, and not the latter? Are there optimizations for $\LStar$ that we could adapt for $\GLStar$? Those are the questions this section seeks to answer. 
 	
\subsection{Embedding of GKAT Automata}	

Comparing the GKAT automaton in \Cref{GLStarmT3} with the Moore automaton (with input alphabet $\At \cdot \Sigma$ and output alphabet $2^{\At}$, short $M$\emph{-automaton}) in \Cref{LStarmT4} suggests that the latter can be recovered from the former by adding a sink-state to make halting transitions explicit. \Cref{embeddinglanguage} below formalises this idea. 

For completeness, we first briefly recall the language semantics of Moore automata. Let $\mathscr{X} = (X, \langle \varepsilon, \delta  \rangle)$ be a $M$-coalgebra, where $MX = B \times X^A$, for an input alphabet $A$ and an output alphabet $B$. Then one can inductively define a function $\llbracket - \rrbracket: X \rightarrow B^{A^*}$ via $\llbracket x \rrbracket(\varepsilon) = \varepsilon(x)$ and $\llbracket x \rrbracket(av) = \llbracket \delta(x)(a) \rrbracket(v)$.
In particular, if $A = (\At \cdot \Sigma)$ and $B = 2^{\At}$ for finite sets $\At$ and $\Sigma$, then the former induces via currying a semantics function $\llbracket - \rrbracket: X \rightarrow \mathcal{P}((\At \cdot \Sigma)^* \cdot \At) = \mathcal{P}(\GS)$ that is defined by: 
\[
\label{Mooreautomataacceptance}
	\alpha \in \llbracket x \rrbracket \Leftrightarrow \varepsilon(x)(\alpha) = 1; \qquad
	\alpha p v \in \llbracket x \rrbracket \Leftrightarrow \delta(x)(\alpha p) = y \textnormal{ and } v \in \llbracket y \rrbracket.
\]

\begin{lemma}
\label{embeddinglanguage}
	Given a $G$-automaton $\mathscr{X} = (X, \delta, x)$, let $f(\mathscr{X}) := (X + \lbrace \star \rbrace, \langle \varepsilon, \partial \rangle, x)$ be the $M$-automaton with: 
	\begin{align*}
		\partial(x)(\alpha p) &:= \begin{cases}
		y & \textnormal{if } x \in X,\ \delta(x)(\alpha) = (p, y) \\
		\star & \textnormal{otherwise}
			\end{cases}\\
		\varepsilon(x)(\alpha) &:= \begin{cases}
		1 & \textnormal{if } x \in X,\ \delta(x)(\alpha) = 1 \\
		0 & \textnormal{otherwise}
	\end{cases}.
	\end{align*}
	Then $\llbracket x \rrbracket_{\mathscr{X}} = \llbracket x \rrbracket_{f(\mathscr{X})}$ for all $x \in X$, and $\llbracket \star \rrbracket_{f(\mathscr{X})} = \emptyset$. In particular, $\llbracket f(\mathscr{X}) \rrbracket_{f(\mathscr{X})} = \llbracket \mathscr{X} \rrbracket_{\mathscr{X}}$.
\end{lemma}
\begin{proof}
	We simultaneously show i) $w \not \in \llbracket \star \rrbracket_{f(\mathscr{X})}$ and ii) $w \in \llbracket x \rrbracket_{f(\mathscr{X})}$ iff $w \in \llbracket x \rrbracket_{\mathscr{X}}$, for all $w \in \GS$ and $x \in X$ by induction on the length of $w$.
	\begin{itemize}
		\item For the induction base assume $w = \varepsilon$, then we find for i):
		\begin{align*}
			\alpha \in \llbracket \star \rrbracket_{f(\mathscr{X})} 
			& \Leftrightarrow
			\varepsilon(\star)(\alpha) = 1 && \textnormal{(Definition of } \llbracket - \rrbracket_{f(\mathscr{X})}) \\
			& \Leftrightarrow 0 = 1 && \textnormal{(Definition of } \varepsilon) \\
			& \Leftrightarrow \textnormal{false} && (0 \not = 1)
		\end{align*}
		Similarly, we derive the following chain for ii): 
		\begin{align*}
	\alpha \in \llbracket x \rrbracket_{f(\mathscr{X})} 
	&\Leftrightarrow 
	\varepsilon(x)(\alpha) = 1 && \textnormal{(Definition of } \llbracket - \rrbracket_{f(\mathscr{X})}) \\
			&\Leftrightarrow \delta(x)(\alpha) = 1 && \textnormal{(Definition of } \varepsilon,\ x \in X) \\
		&\Leftrightarrow \alpha \in \llbracket x \rrbracket_{\mathscr{X}} && \textnormal{(Definition of } \llbracket - \rrbracket_{\mathscr{X}})
		\end{align*}
		\item For the induction step, let $w = \alpha p v$ for $v \in \GS$, then we deduce for i):
		\begin{align*}
			\alpha p v \in \llbracket \star \rrbracket_{f(\mathscr{X})} & \Leftrightarrow \partial(\star)(\alpha p) = y,\ v \in \llbracket y \rrbracket_{f(\mathscr{X})} &&  \textnormal{(Definition of } \llbracket - \rrbracket_{f(\mathscr{X})}) \\
			& \Leftrightarrow  \partial(\star)(\alpha p) = \star,\ v \in \llbracket \star \rrbracket_{f(\mathscr{X})} && \textnormal{(Definition of } \partial) \\
			& \Leftrightarrow \textnormal{false} && \textnormal{(IH)}
		\end{align*} 
		Analogously, we deduce for ii):
		\begin{align*}
		& \alpha p v \in \llbracket x \rrbracket_{f(\mathscr{X})} \\
		\Leftrightarrow &\ \textnormal{(Definition of } \llbracket - \rrbracket_{f(\mathscr{X})}) \\
		& \partial(x)(\alpha p) = y,\ v \in \llbracket y \rrbracket_{f(\mathscr{X})} \\
		\Leftrightarrow &\ \textnormal{(Definition of } \partial) \\
		& \delta(x)(\alpha) = (p,y),\ v \in  \llbracket y \rrbracket_{f(\mathscr{X})} \textnormal{ or } \delta(x)(\alpha) \in 2,\ v \in  \llbracket \star \rrbracket_{f(\mathscr{X})} \\
		\Leftrightarrow &\ \textnormal{(IH)} \\
			 & \delta(x)(\alpha) = (p,y),\ v \in  \llbracket y \rrbracket_{\mathscr{X}} \textnormal{ or false}   \\
			\Leftrightarrow &\ \textnormal{(Definition of } \llbracket - \rrbracket_{\mathscr{X}}) \\
& \alpha p v \in \llbracket x \rrbracket_{\mathscr{X}} 	\end{align*}
	\end{itemize}
In particular, $\llbracket f(\mathscr{X}) \rrbracket_{f(\mathscr{X})} = \llbracket x \rrbracket_{f(\mathscr{X})} = \llbracket x \rrbracket_{\mathscr{X}} = \llbracket \mathscr{X} \rrbracket_{\mathscr{X}}$.
 \end{proof}

As one would hope for, the above construction maps, up to isomorphism, the minimal GKAT automaton $m(\mathscr{X})$ to the minimal Moore automaton accepting the same language as $\mathscr{X}$.
  
\begin{corollary}
\label{minimalembeddingiso}
 	 	Let $\mathscr{X}$ be a normal $G$-automaton, then $f(m(\mathscr{X})) \cong m(f(\mathscr{X}))$ as $M$-automata.
 \end{corollary}
\begin{proof}
From \Cref{embeddinglanguage} and \Cref{minimallanguage} we can deduce that $f(m(\mathscr{X}))$ accepts $\llbracket \mathscr{X} \rrbracket$, which is also accepted by $m(f(\mathscr{X}))$. 

 	By \Cref{minimalnormalobservable} $\llbracket - \rrbracket_{m(\mathscr{X})}$ is injective. As \Cref{minimalnormalobservable} and \Cref{reachableminimal} imply that $m(\mathscr{X})$ is normal and reachable, we know that $\llbracket - \rrbracket_{m(\mathscr{X})}$ never evaluates to the empty set. From \Cref{embeddinglanguage} we thus can deduce that $\llbracket - \rrbracket_{f(m(\mathscr{X}))}$ is injective. 
 	
 	It is not hard to see that if a state is reachable in $\mathscr{Y}$, then it is reachable in $f(\mathscr{Y})$. The element $\star$ is reachable in $f(\mathscr{Y})$ in particular if $\mathscr{Y}$ is reachable and normal. Hence, since $m(\mathscr{X})$ is reachable and normal by \Cref{reachableminimal} and \Cref{minimalnormalobservable}, respectively, $f(m(\mathscr{X}))$ is reachable.
 	
 	The automaton $f(m(\mathscr{X}))$ thus accepts the same language as $m(f(\mathscr{X}))$, is observable, and reachable. By uniqueness, $f(m(\mathscr{X}))$ and $m(f(\mathscr{X}))$ are thus isomorphic.
 \end{proof}
 
\subsection{Complexity Analysis} 

We now compare the worst-case complexities of $\LStar$ (\Cref{LStaralgorithm}) and $\GLStar$ (\Cref{GlStaralgorithm}) for learning automata representations of GKAT programs $e$. We are mainly interested in a bound to the number of membership queries to $\llbracket e \rrbracket$. The example runs in \Cref{lstarexamplerun} and \Cref{glstarexamplerun} seem to indicate that with respect to this aspect, $\GLStar$ performs better than $\LStar$. The result below confirms this intuition.

\begin{proposition}
\label{complexity}
 	\Cref{LStaralgorithm} requires at most $O(a * (\vert \At \vert * b))$ many membership queries to $\llbracket e \rrbracket$ for learning a $M$-automaton representation of $e$, whereas $\Cref{GlStaralgorithm} $ requires at most $O(a * (\vert \At \vert + b))$ many membership queries to $\llbracket e \rrbracket$ for learning a $G$-automaton representation of $e$, for some\footnote{Let $m$ be the maximum length of a counterexample and $n$ the size of the minimal Moore automaton accepting $\llbracket e \rrbracket$, then  $a = n * \vert \At \vert * \vert \Sigma \vert$ and $b = m * n$.} integers $a, b \in \mathbb{N}$. 
 \end{proposition}
  \begin{proof}
 	 	We first derive the maximum number of entries a table indexed by $S$ and $E$ can have during a run of $\LStar$ for generalised languages with input alphabet $A$ and output alphabet $B$. Since we use a suffix-strategy for the handling of counterexamples (opposed to a prefix-strategy), our presentation slightly differs from the one in \cite{angluin1987learning}. Let $k$ denote the cardinality of the alphabet $A$. The number of states of the minimal Moore automaton accepting the target language is referred to by $n$, and the maximum length of a counterexample by $m$. The size of $S$ is bounded by $n$. In the worst case, the sets $S$ and $S \cdot A$ are disjoint. The cardinality of $S \cup S \cdot A$ is thus bounded by $n + n * k$. The maximum number of strings in $E$ is $1 + m * (n-1)$. This is because $E$ is instantiated with $\varepsilon$, and only extended with suffixes of counterexamples. Each counterexample has at most $m$ suffixes, and there can only be $n-1$ counterexamples, since any counterexample leads to a closedness defect, resolving which increases the size of $S$, which is instantiated with $\varepsilon$, and bounded in size by $n$. A table can thus have at most $(n + n * k)*(1 + m * (n-1))$, or $ O(m * n^2 * k)$, many entries.

In the case of $A = (\At \cdot \Sigma)$ and $B = 2^{\At}$, each entry requires $\vert \At \vert$ many membership queries. Overall \Cref{LStaralgorithm} thus requires at most
\[ O(n * \vert \At \vert *  \vert \Sigma \vert * (\vert \At \vert * m * n)) \]
many membership queries to learn a deterministic guarded string language.

We now derive a bound to the number of membership queries $\GLStar$ requires. As before, let $m$ denote the maximum length of a counterexample, and $n$ the number of states of the minimal Moore automaton accepting the target language. The cardinality of the set $S$ is bounded by the number of states in the minimal GKAT automaton accepting the target language, which by \Cref{minimalembeddingiso} is $n-1$. The cardinality of $S \cup S \cdot (\At \cdot \Sigma)$ is thus bounded by $(n-1) + (n-1) * \vert \At \vert * \vert \Sigma \vert$. The maximum number of strings in $E$ is $\vert \At \vert + m * (n-1)$. This is because $E$ is instantiated with $\At$, and only extended with suffixes of counterexamples. Each counterexample has at most $m$ suffixes, and there can only be $n-1$ counterexamples. Indeed, assume there are $u > n-1$ counterexamples. Since resolving a closedness defect increases the size of $S$, which is instantiated with $\varepsilon$, and in size bounded by $n-1$, there can be at most $(n-1)-1=n-2$ counterexamples for which the extended table is not closed. Thus there must be at least $u-(n-2) = u-n+2 \geq n-n+2 = 2$ counterexamples for which the extended table is closed. This is a contradiction, since \Cref{ifclosedthen} implies that this can be the case for at most $1$ counterexample. A table can thus have at most 
$ ((n-1) + (n-1) * \vert \At \vert * \vert \Sigma \vert )* (\vert \At \vert + m * (n-1)) $
entries. Each entry requires one membership query. Overall \Cref{GlStaralgorithm} requires at most
\[
O(n * \vert \At \vert * \vert \Sigma \vert * (\vert \At \vert + m * n))
\]
many membership queries to learn a deterministic guarded string language. The statement follows by setting $a := n * \vert \At \vert * \vert \Sigma \vert$ and $b := m * n$. 
 \end{proof}
 
 The following result shows that for all integers $x, y$ greater than $2$, the product $x * y$ is strictly greater than the sum $x + y$. 
 
   \begin{lemma}
 \label{aba+b}
 	Let $a,b \in \mathbb{N}_{\geq 3}$, then $a * b > a + b$.
 \end{lemma}
 \begin{proof}
 	We prove the statement by induction on $a \in \mathbb{N}_{\geq 3}$. 
 	\begin{itemize}
 	\item For the induction base, assume $a = 3$. We show that $ 3 * b > 3 + b$ for all $b \in \mathbb{N}_{\geq 3}$ by induction on $b$. In the induction base, $b = 3$, the statement immediately is implied by $3 * 3 = 9 > 6 = 3  + 3$. Assume the statement is true for some $b \geq 3$. The induction step follows from
 	\[
 	3 * (b + 1) = 3b + 3 > (3 + b) + 3 = 6 + b > 3 + (b + 1).
 	\]
 	\item Assume the statement is true for $a \geq 3$. The induction step follows from
 	\[
 		(a + 1) * b = (a * b) + b > (a + b) + b > (a + 1) + b. \qedhere
 	\] 
 	\end{itemize}
 \end{proof}
 
In fact, it is also not hard to see that the difference between $a * b$ and $a + b$ increases with the sizes of $a$ and $b$.

\begin{lemma}
	$a + b$ lies in $o(a * b)$
\end{lemma}
\begin{proof}
	We need to show that for any positive real number $c > 0$, there exists a natural number $N$, such that for all natural numbers $a, b \geq N$, we have $a + b < c * (a * b)$, or equivalently, $\frac{a + b}{a * b} < c$.
	
Given $c > 0$, we define $N := \lceil \frac{2}{c} \rceil + 1$. Let $a, b \geq N > \frac{2}{c}$. Then it immediately follows that $\frac{c}{2} > \frac{1}{b}$ and $\frac{c}{2} > \frac{1}{a}$. Thus we can compute
\begin{align*}
	\frac{a + b}{a*b} = \frac{a}{a * b} + \frac{b}{a * b} = \frac{1}{b} + \frac{1}{a} < \frac{c}{2} + \frac{c}{2} = c.
\end{align*}
\end{proof}

The advantage of $\GLStar$ over $\LStar$ for learning deterministic guarded string languages in terms of membership queries thus increases with the number of atoms, which is exponential in the number of primitive tests, $\At \cong 2^{T}$. In applications to network verification, the number of tests, thus atoms, is typically quite large \cite{anderson2014netkat}. The difference between $\GLStar$ and $\LStar$ described in \Cref{complexity} is mainly due to a subtle play with the table indices, based on currying. It can be further increased by avoiding querying certain rows all together, taking into account the deterministic nature of the target language, as indicated in \Cref{glstarexamplerunsection}.

 \subsection{Optimized Counterexamples}
 
 \label{optimizedcounterexamplesec}
 
 In this section we present an optimization of $\GLStar$ that is based on a subtle refinement of \Cref{ifclosedthen}. We show that while \Cref{GlStaralgorithm} reacts to a negative equivalence query with counterexample $z \in \GS$ by adding columns for \emph{all} suffixes in $\suff(z)$, it is in fact enough to add columns for a smaller subset of suffixes $\suff(z') \subseteq \suff(z)$, for some $z' \in \suff(z)$ of minimal length. Our approach is inspired by the optimized counterexample handling method of Rivest and Schapire for $\LStar$ \cite{rivest1993inference}. To state \Cref{optimizedcounterexample}, we need to define the following set:
\begin{align*}
	A_z := \lbrace & z' \in \suff(z) \mid\ z = v \alpha p z',\ \row(\varepsilon)  \xrightarrow{v} \row(s_v),\ x  \xrightarrow{s_v} x_{s_v}, \\
	&\llbracket \row(s_v) \rrbracket(\alpha p z') \not = \llbracket x_{s_v} \rrbracket(\alpha p z') \rbrace
\end{align*}
Intuitively, $A_z$ contains all strictly-shorter suffixes of $z$ that witness a mismatch between the behaviour of the hypothesis and the target language.

 \begin{lemma}
 \label{optimizedcounterexample}
 	Let $T = (S, E, \row)$ be a closed deterministic observation table with $\row(t)(e) = \llbracket \mathscr{X} \rrbracket(te)$ for all $t \in S \cup S \cdot (\At \cdot \Sigma),\ e \in E$. Let $\llbracket m(T) \rrbracket(z) \not = \llbracket \mathscr{X} \rrbracket(z)$ for some $z \in \GS$, and 
 	$z' := \minn(A_z)$. If $T' = (S, E \cup \suff(z'), \row')$ with $\row'(t)(e) = \llbracket \mathscr{X} \rrbracket(te)$ is closed, then $\row'(\varepsilon)(e) = 0$ for all $e \in E$, but $\row'(\varepsilon)(z') = 1$.
 \end{lemma}
 \begin{proof}
We begin by showing that $z \not \in \At$. Let us assume the opposite, $z = \alpha \in \At \subseteq E$. In that case, it follows:
 \begin{align*}
 	\llbracket m(T) \rrbracket(z) &= \llbracket \row(\varepsilon) \rrbracket(\alpha) && \textnormal{(Definition of } \llbracket - \rrbracket,\ z) \\
 	&= \row(\varepsilon)(\alpha) && (\textnormal{\Cref{consistent}}) \\
 	&= \llbracket \mathscr{X} \rrbracket(\alpha) && (\row(t)(e) = \llbracket \mathscr{X} \rrbracket(te)) \\
 	&= \llbracket \mathscr{X} \rrbracket(z) && \textnormal{(Definition of } z)
 \end{align*}
 which is a contradiction. Thus there exists a decomposition $z = \varepsilon \alpha p z'$ for some $z' \in \suff(z)$. The former immediately implies that the set $A_z$ is non-empty. Hence the shortest suffix $z' := \minn(A_z)$ is well-defined. By construction, we have $\row(\varepsilon)  \xrightarrow{v} \row(s_v)$,\ $x  \xrightarrow{s_v} x_{s_v}$, and $\llbracket \row(s_v) \rrbracket(\alpha p z') \not = \llbracket x_{s_v} \rrbracket(\alpha p z')$.
 \begin{itemize}
 	\item Assume $\llbracket x_{s_v} \rrbracket(\alpha p z') = 0$ and $\llbracket \row(s_v) \rrbracket(\alpha p z') = 1$, then there exists $s_{v \alpha p} \in S$ with $\row(\varepsilon) \xrightarrow{v} \row(s_v)  \xrightarrow{\alpha p} \row(s_{v \alpha p})$ such that:
		\begin{align*}
			1 &= \llbracket \row(s_v) \rrbracket(\alpha p z') && (\textnormal{Assumption}) \\
			&= \llbracket \row(s_{v \alpha p}) \rrbracket(z')	&& (\textnormal{Definition of } \llbracket - \rrbracket) \\
		  &= \llbracket \row'(s_{v \alpha p}) \rrbracket(z') && (\textnormal{\Cref{languageinclusion}}) \\	  
		  &= \row'(s_{v \alpha p})(z') && (\textnormal{\Cref{consistent}}) \\
		  &= \row'(s_{v} \alpha p)(z') && (\textnormal{Definition of } \row') \\
		  &=  \llbracket \mathscr{X} \rrbracket(s_{v} \alpha p z') && (\row'(t)(e) = \llbracket \mathscr{X} \rrbracket(te)) \\
		  &= \llbracket x_{s_v} \rrbracket(\alpha p z') && (\textnormal{Definition of } \llbracket - \rrbracket)
		\end{align*} 
		  which contradicts $0 = \llbracket x_{s_v} \rrbracket(\alpha p z')$.
	 \item Assume $\llbracket x_{s_v} \rrbracket(\alpha p z') = 1$ and $\llbracket \row(s_v) \rrbracket(\alpha p z') = 0$, then there exists $s_{v \alpha p} \in S$ with $\row'(\varepsilon)  \xrightarrow{v} \row'(s_v)  \xrightarrow{\alpha p} \row'(s_{v \alpha p})$ such that:
	 \begin{align*}
	 	1 &= \llbracket x_{s_v} \rrbracket(\alpha p z') && (\textnormal{Assumption}) \\
	 	&= \llbracket \mathscr{X} \rrbracket(s_v \alpha p z')  && (\textnormal{Definition of } \llbracket - \rrbracket) \\
	 	&= \row'(s_{v} \alpha p)(z') && (\row'(t)(e) = \llbracket \mathscr{X} \rrbracket(te)) \\
	 	&= \row'(s_{v \alpha p})(z') && (\textnormal{Definition of } \row') \\
	 	&= \llbracket \row'(s_{v \alpha p}) \rrbracket(z') && (\textnormal{\Cref{consistent}}) \\
	 	&= \llbracket  \row'(s_{v}) \rrbracket(\alpha p z') && (\textnormal{Definition of } \llbracket - \rrbracket)
	 \end{align*}
	 By \Cref{mT'impliesT} there thus are two possibilities:
	 \begin{itemize}
	 \item Assume $\row(s_v)  \xrightarrow{\alpha p} \row(s_{v \alpha p})$, then there are two options:
	 	\begin{itemize}
	 	\item If $z' \not \in \At$, we find a contradiction to the minimality of $z'$ in $A_z$. 
	 	\item If $z' \in \At \subseteq E$, then:
	 	\begin{align*}
	 		0 &= \llbracket \row(s_v) \rrbracket(\alpha p z') && (\textnormal{Assumption}) \\
		&= \llbracket \row(s_{v \alpha p}) \rrbracket(z') && (\row(s_v)  \xrightarrow{\alpha p} \row(s_{v \alpha p})) \\
		&=   \row(s_{v \alpha p})(z') && (\textnormal{\Cref{consistent}}) \\
		&= 	\llbracket \mathscr{X} \rrbracket(s_{v \alpha p} z') && (\row(t)(e) = \llbracket \mathscr{X} \rrbracket(te)) \\
		&= \llbracket x_{s_{v \alpha p}} \rrbracket(z') && (\textnormal{Definition of } \llbracket - \rrbracket) \\
		&= \llbracket x_{s_{v}} \rrbracket(\alpha p z')  && (\textnormal{Definition of } \llbracket - \rrbracket)
	 	\end{align*}
	 	which is a contradiction to $\llbracket x_{s_v} \rrbracket(\alpha p z') = 1$.
	 	\end{itemize}
	 \item Assume $\delta(\row(s_v))(\alpha) = 0$, then we have, for all $e \in E$:
	 \begin{align*}
	 	0 &= \row(s_v \alpha p)(e) && (\delta(\row(s_v))(\alpha) = 0) \\
	 	&=  \row'(s_v \alpha p)(e) && (\row(t)(e) = \row'(t)(e)) \\
	 	&= \row'(s_{v \alpha p})(e) && (\textnormal{Definition of } \row') \\
	 	&= \row(s_{v \alpha p})(e) && (\row(t)(e) = \row'(t)(e))
	 \end{align*}
	 From \Cref{observationtable} it follows $s_{v \alpha p} = \varepsilon$, which implies the claim. \qedhere
	 \end{itemize}	 
 \end{itemize}
\end{proof}

  Let $z_0$ be the shortest suffix of $z$ and $z_i$ the suffix of $z$ of length $\vert z_{i-1}\vert + 1$. The suffix $\minn(A_z)$ can be computed in at most $\vert \suff(z) \vert - 1$ steps: verify whether $z_i \in A_z$, beginning with $z_0$; if positive, break and set $\minn(A_z) := z_i$, otherwise loop with $z_{i+1}$.
  
 For example, if $T$ is the closed table in \Cref{GLStarT2} with the corresponding hypothesis $m(T)$ in \Cref{GLStarmT2} and counterexample $z = b p \overline{b} q b$, then $z' = \minn(A_z) = \overline{b}qb$, since $b \not \in A_z$. \Cref{optimizedcounterexample} shows that, instead of adding columns for the two non-present suffixes $b p \overline{b} q b$ and  $\overline{b}qb$ of $z$, it is sufficient to add only one column for the single non-present suffix  $\overline{b}qb$ of $z'$. In this case, the counterexample $z$ is relatively short, thus the number of avoided columns small; in general, however, the advantage can be more significant.

 \section{Implementation}
 
 \begin{figure}
 \centering
 \begin{subfigure}[c]{.48\textwidth}
 \centering	
\begin{subfigure}[c]{.48\textwidth}
\centering
 \adjustbox{scale=0.7,center}{
	 	\begin{tikzpicture}
\begin{axis}[
	width = 22em,
    xlabel={$\vert T \vert = \vert \lbrace t_1,...,t_n \rbrace \vert$},
    ylabel={Membership queries to $\llbracket e \rrbracket$},
    xmin=1, xmax=9,
    ymin=1, ymax=8000000,
    xtick={1,2,3,4,5,6,7,8,9},
    ytick={500000,1000000,2500000,8000000},
    legend pos=north west,
]
 \addplot[
    mark=x,
    ]
    coordinates {
    (1,114)
    (2,444)
    (3,1752)
    (4,6960)
    (5,27744)
    (6,110784)
    (7,442752)
    (8,1770240)
    (9,7079424)
    };
\addplot[
    mark=square,
    ]
    coordinates {
    (1,26)
    (2,100)
    (3,392)
    (4,1552)
    (5,6176)
    (6,24640)
    (7,98432)
    (8,393472)
    (9,1573376)
    };

    \legend{$\LStar$,$\GLStar$}   
    
\end{axis}
\end{tikzpicture}
}
\label{yy}
\end{subfigure}
\caption{$e = \iif\ t_1\ \tthen\ \ddo\ p_1\ \eelse\ \ddo\ p_2$}
\label{performance_ife}
\end{subfigure}
\qquad
\begin{subfigure}[c]{.4\textwidth}
 \centering	
\begin{subfigure}[c]{.4\textwidth}
\centering
 \adjustbox{scale=0.7,center}{
	 	\begin{tikzpicture}
\begin{axis}[
	width = 22em,
    xlabel={$\vert T \vert = \vert \lbrace t_1,...,t_n \rbrace \vert$},
    ylabel={Membership queries to $\llbracket e \rrbracket$},
    xmin=1, xmax=9,
    ymin=1, ymax=8000000,
    xtick={1,2,3,4,5,6,7,8,9},
    ytick={500000,1000000,2500000,5000000, 8000000},
    legend pos=north west,
]
 \addplot[
    mark=x,
    ]
    coordinates {
    (1,78)
    (2,300)
    (3,1176)
    (4,4659)
    (5,18528)
    (6,73920)
    (7,295296)
    (8,1180416)
    (9,4720128)
    };
    
\addplot[
    mark=square,
    ]
    coordinates {
    (1,36)
    (2,102)
    (3,330)
    (4,1170)
    (5,4386)
    (6,16962)
    (7,66690)
    (8,264450)
    (9,1053186)
    };

    \legend{$\LStar$,$\GLStar$}   
    
\end{axis}
\end{tikzpicture}
}
\label{zz}
\end{subfigure}

\caption{$e = (\wwhile\ t_1\ \ddo\ p_1); \ddo\ p_2$}
\label{performance_while}
\end{subfigure}
\caption{A comparison between $\GLStar$ and $\LStar$ with respect to membership queries}
\label{comparisongraph}
 \end{figure}

 \begin{figure*}
 \centering
\begin{subfigure}[c]{.48\textwidth}
\centering
\footnotesize
\begin{tabular}{ c | c | r | r }
 $\vert \Sigma \vert$ & $\vert T \vert$ & $\GLStar$ & $\LStar$ \\
 \hline
3 & 1 & 26 & 114 \\
3 & 2 & 100 & 444 \\ 
3 & 3 & 392 & 1.752 \\
3 & 4 & 1.552 & 6.960 \\
3 & 5 & 6.176 & 27.744 \\
3 & 6 & 24.640 & 110.784 \\
3 & 7 & 98.432 & 442.752 \\
3 & 8 & 393.472 & 1.770.240 \\
3 & 9 & 1.573.376 & 7.079.424	\\
\hline
\end{tabular}
\caption{$e = \iif\ t_1\ \tthen\ \ddo\ p_1\ \eelse\ \ddo\ p_2$}
\end{subfigure}
\begin{subfigure}[c]{.48\textwidth}
\centering
\footnotesize
\begin{tabular}{ c | c | r | r }
 $\vert \Sigma \vert$ & $\vert T \vert$ & $\GLStar$ & $\LStar$ \\
 \hline
2 & 1 & 36 & 78 \\
2 &2 & 102 & 300 \\ 
2 &3 & 330 & 1.176 \\
2 &4 & 1.170 & 4.656 \\
2 &5 & 4.386 & 18.528 \\
2 &6 & 16.962 & 73.920 \\
2 &7 & 66.690 & 295.296 \\
2 &8 & 264.450 & 1.180.416 \\
2 & 9 & 1.053.186 & 4.720.128	\\
\hline
\end{tabular}
\caption{$e = (\wwhile\ t_1\ \ddo\ p_1); \ddo\ p_2$}
\end{subfigure}
\caption{The exact number of membership queries to $\llbracket e \rrbracket$ underlying the comparison between $\GLStar$ and $\LStar$ in \Cref{comparisongraph}}
 \label{ife_comparison_table}
 \end{figure*}

\begin{figure*}[t]
\centering
\adjustbox{scale=0.8}{
	\begin{subfigure}[b]{.35\textwidth}
	\centering
		\begin{tabular}{ l|c|c}
		 & $t_1$ &  $\overline{t_1}$ \\
		 \hline $\varepsilon$ & 0 & 0  \\
		 \hline $\overline{t_1}p_2$ & 1 & 1 \\
		 \hline
		\hline $t_1 p_1$ & 1 & 1 \\
				\hline $t_1 p_2$ & 0 & 0 \\
		\hline $t_1 p_3$ & 0 & 0 \\
		 \hline $\overline{t_1}p_1$ & 0 & 0 \\
		 		 \hline $\overline{t_1}p_3$ & 0 & 0 \\
		 \hline $\overline{t_1}p_2t_1p_1$ & 0 & 0 \\
		 \hline $\overline{t_1}p_2t_1p_2$ & 0 & 0 \\
		 \hline $\overline{t_1}p_2t_1p_3$ & 0 & 0 \\
		 \hline $\overline{t_1}p_2\overline{t_1}p_1$ & 0 & 0 \\
		 \hline $\overline{t_1}p_2\overline{t_1}p_2$ & 0 & 0 \\
		 \hline $\overline{t_1}p_2\overline{t_1}p_3$ & 0 & 0 \\

	\end{tabular}
	\caption{}
	\end{subfigure}
	\begin{subfigure}[b]{.6\textwidth}
	\centering
\begin{tikzpicture}[node distance=6em]
	\node[state, shape=rectangle, initial, initial text=,  label=below:{$\Rightarrow 0t_1 + 0\overline{t_1}$}] (x) {$\row(\varepsilon)$};
	\node[state, shape=rectangle, right of=x, above of=x, label=right:{$\Rightarrow  0t_1 + 0\overline{t_1}$}] (y) {$\star$};
	\node[state, shape=rectangle, right of=y, below of=y, label=below:{$\Rightarrow  1t_1 + 1\overline{t_1}$}] (z) {$\row(\overline{t_1}p_2)$};
	    \path[->]
	(x) edge[below] node{$t_1p_1, \overline{t_1}p_2$} (z)
	(y) edge[loop above] node{$\lbrace t_1, \overline{t_1} \rbrace \cdot \lbrace p_1, p_2, p_3 \rbrace$} (y)
	(z) edge[right] node{$\lbrace t_1, \overline{t_1} \rbrace \cdot \lbrace p_1, p_2, p_3 \rbrace$}(y)
	(x) edge[left] node{$t_1p_2, t_1p_3, \overline{t_1}p_1, \overline{t_1}p_3$} (y)
	;
\end{tikzpicture}
	\caption{}
	\label{concrete_glstar_automaton}
\end{subfigure}
}
\caption{For $e = \iif\ t_1\ \tthen\ \ddo\ p_1\ \eelse\ \ddo\ p_2$, $\vert \Sigma \vert = 3$, and $\vert T \vert  = 2$, our implementation of $\GLStar$ accepts the table in (a), which induces the automaton in (b).}
\label{concrete_run_glstar}
\end{figure*}

\begin{figure*}[t]
\centering
\adjustbox{scale=0.7}{
	\begin{subfigure}[b]{.5 \textwidth}
	\centering
		\begin{tabular}{ l|c|c|c}
		 & $\varepsilon$ & $\overline{t_1}p_1\overline{t_1}p_2$ & $\overline{t_1}p_2$ \\
		 \hline $\varepsilon$ & $0t_1 + 0\overline{t_1}$ & $0t_1 + 0\overline{t_1}$ & $1t_1 + 1\overline{t_1}$  \\
		 \hline $\overline{t_1}p_1$ & $0t_1 + 0\overline{t_1}$ & $0t_1 + 0\overline{t_1}$ & $0t_1 + 0\overline{t_1}$ \\
		 \hline $\overline{t_1}p_2$ & $1t_1 + 1\overline{t_1}$ & $0t_1 + 0\overline{t_1}$ & $0t_1 + 0\overline{t_1}$ \\
		 \hline \hline
		$t_1p_1$ & $1t_1 + 1\overline{t_1}$ & $0t_1 + 0\overline{t_1}$ & $0t_1 + 0\overline{t_1}$ \\
\hline $t_1p_2$ & $0t_1 + 0\overline{t_1}$ & $0t_1 + 0\overline{t_1}$ & $0t_1 + 0\overline{t_1}$ \\
\hline $t_1p_3$ & $0t_1 + 0\overline{t_1}$ & $0t_1 + 0\overline{t_1}$ & $0t_1 + 0\overline{t_1}$ \\
	 \hline $\overline{t_1}p_1 t_1p_1$ & $0t_1 + 0\overline{t_1}$ & $0t_1 + 0\overline{t_1}$ & $0t_1 + 0\overline{t_1}$ \\
		 \hline $\overline{t_1}p_1 t_1p_2$ & $0t_1 + 0\overline{t_1}$ & $0t_1 + 0\overline{t_1}$ & $0t_1 + 0\overline{t_1}$ \\
		 \hline $\overline{t_1}p_1 t_1p_3$ & $0t_1 + 0\overline{t_1}$ & $0t_1 + 0\overline{t_1}$ & $0t_1 + 0\overline{t_1}$ \\
		 \hline $\overline{t_1}p_1 \overline{t_1}p_1$ & $0t_1 + 0\overline{t_1}$ & $0t_1 + 0\overline{t_1}$ & $0t_1 + 0\overline{t_1}$ \\
		 \hline $\overline{t_1}p_1 \overline{t_1}p_2$ & $0t_1 + 0\overline{t_1}$ & $0t_1 + 0\overline{t_1}$ & $0t_1 + 0\overline{t_1}$ \\
		 \hline $\overline{t_1}p_1 \overline{t_1}p_3$ & $0t_1 + 0\overline{t_1}$ & $0t_1 + 0\overline{t_1}$ & $0t_1 + 0\overline{t_1}$ \\
			 \hline $\overline{t_1}p_2 t_1p_1$ & $0t_1 + 0\overline{t_1}$ & $0t_1 + 0\overline{t_1}$ & $0t_1 + 0\overline{t_1}$ \\
		 \hline $\overline{t_1}p_2 t_1p_2$ & $0t_1 + 0\overline{t_1}$ & $0t_1 + 0\overline{t_1}$ & $0t_1 + 0\overline{t_1}$ \\
		 \hline $\overline{t_1}p_2 t_1p_3$ & $0t_1 + 0\overline{t_1}$ & $0t_1 + 0\overline{t_1}$ & $0t_1 + 0\overline{t_1}$ \\
		 \hline $\overline{t_1}p_2 \overline{t_1}p_1$ & $0t_1 + 0\overline{t_1}$ & $0t_1 + 0\overline{t_1}$ & $0t_1 + 0\overline{t_1}$ \\
		 \hline $\overline{t_1}p_2 \overline{t_1}p_2$ & $0t_1 + 0\overline{t_1}$ & $0t_1 + 0\overline{t_1}$ & $0t_1 + 0\overline{t_1}$ \\
		 \hline $\overline{t_1}p_2 \overline{t_1}p_3$ & $0t_1 + 0\overline{t_1}$ & $0t_1 + 0\overline{t_1}$ & $0t_1 + 0\overline{t_1}$
	\end{tabular}
	\caption{}
	\end{subfigure}
	}
	\adjustbox{scale=0.7}{
	\begin{subfigure}[b]{.7 \textwidth}
	\centering
\begin{tikzpicture}[node distance=6em]
	\node[state, shape=rectangle, initial, initial text=,  label=below:{$\Rightarrow 0t_1 + 0\overline{t_1}$}] (x) {$\row(\varepsilon)$};
	\node[state, shape=rectangle, right of=x, above of=x, label=right:{$\Rightarrow  0t_1 + 0\overline{t_1}$}] (y) {$\row(\overline{t_1}p_1)$};
	\node[state, shape=rectangle, right of=y, below of=y, label=below:{$\Rightarrow  1t_1 + 1\overline{t_1}$}] (z) {$\row(\overline{t_1}p_2)$};
	    \path[->]
	(x) edge[below] node{$t_1p_1, \overline{t_1}p_2$} (z)
	(y) edge[loop above] node{$\lbrace t_1, \overline{t_1} \rbrace \cdot \lbrace p_1, p_2, p_3 \rbrace$} (y)
	(z) edge[right] node{$\lbrace t_1, \overline{t_1} \rbrace \cdot \lbrace p_1, p_2, p_3 \rbrace$}(y)
	(x) edge[left] node{$t_1p_2, t_1p_3, \overline{t_1}p_1, \overline{t_1}p_3$} (y)
	;
\end{tikzpicture}
	\caption{}
	\label{concrete_lstar_automaton}
\end{subfigure}
}
\caption{For $e = \iif\ t_1\ \tthen\ \ddo\ p_1\ \eelse\ \ddo\ p_2$, $\vert \Sigma \vert = 3$, and $\vert T \vert  = 2$, our implementation of $\LStar$ accepts the table in (a), which induces the automaton in (b).}
\label{concrete_run_lstar}
\end{figure*}

We have implemented both $\GLStar$ and $\LStar$ in OCaml \cite{Ocaml}; the code is available on GitHub\footnote{\url{https://github.com/zetzschest/gkat-automata-learning}}. The implementation allows one to compare, for any GKAT expression $e \in \Expr_{\Sigma, T}$, the number of membership queries to $\llbracket e \rrbracket$ required by $\GLStar$ for learning a $G$-automaton representation of $e$, with the number of membership queries to $\llbracket e \rrbracket$ required by $\LStar$ for learning a $M$-automaton representation of $e$. For each run, we output, for both algorithms, a trace of the involved hypotheses as tables in the \texttt{.csv} format and graphs in the \texttt{.dot} format, as well as an overview of the involved queries in the \texttt{.csv} format. 

In \Cref{performance_ife} we plot the results for the expression $\iif\ t_1\ \tthen\ \ddo\ p_1\ \eelse\ \ddo\ p_2$, the primitive actions $\Sigma = \lbrace p_1, p_2, p_3 \rbrace$, and primitive tests $T = \lbrace t_1,...,t_n \rbrace$ parametric in $n = 1,...,9$. We find that $\GLStar$ outperforms $\LStar$ for all choices of $n$. The difference in the number of membership queries increases with the size of $n$, as suggested by \Cref{complexity}. For $n = 9$ the number of atoms is $2^9$, resulting in an already relatively large number of queries for both algorithms.
The picture is similar in \Cref{performance_while}, where we choose the expression $(\wwhile\ t_1\ \ddo\ p_1); \ddo\ p_2$, the primitive actions $\Sigma = \lbrace p_1, p_2 \rbrace$, and primitive tests $T = \lbrace t_1,...,t_n \rbrace$ parametric in $n = 1,...,9$. Again, $\GLStar$ requires significantly less queries in all cases of $n$, and the difference increases with the size of $n$. The exact numbers of membership queries underlying the graphs in \Cref{comparisongraph} can be found in \Cref{ife_comparison_table}.

 In  \Cref{concrete_run_glstar} and \Cref{concrete_run_lstar} we give concrete examples of the tables and corresponding automata our implementations of $\GLStar$ and $\LStar$ deduce. It is instructive to recall \Cref{minimalembeddingiso} in this context: the embedding in \Cref{concrete_glstar_automaton} is clearly isomorphic to the minimal Moore automaton in \Cref{concrete_lstar_automaton}.  
  
Our implementation generates an oracle for $\LStar$ from a GKAT expression $e$ in the following way. First, we interpret $e$ as a KAT expression $\iota(e)$ via the standard embedding of GKAT into KAT. Next, we generate from the latter a Moore automaton $\mathscr{X}_{\iota(e)}$ accepting $\llbracket e \rrbracket$, by using Kozen's syntactic Brzozowksi derivatives for KAT \cite{kozen2017coalgebraic}. Finally, we answer an equivalence query from a Moore automaton $\mathscr{Y}$ by running a bisimulation between $\mathscr{X}_{\iota(e)}$ and $\mathscr{Y}$, similarly to \cite[Fig. 1]{pous2015symbolic}, and a membership query from $w \alpha \in \GS$ by returning the value of $\alpha$ at the output of the state in $\mathscr{X}_{\iota(e)}$ reached by $w$, that is, $\llbracket e \rrbracket(w \alpha)$. A membership query from $w \in \GSM$ is answered by querying $w \alpha \in \GS$ for all $\alpha \in \At$.

With the oracle for $\LStar$, we can derive an oracle for $\GLStar$ as follows. Membership queries $w \alpha \in \GS$ are delegated and answered by the oracle for $\LStar$ as explained above. An equivalence query from a GKAT automaton $\mathscr{Y}$ is answered by posing an equivalence query to the oracle for $\LStar$ with the Moore automaton $f(\mathscr{Y})$ obtained via the embedding defined in \Cref{embeddinglanguage}. If the oracle for $\LStar$ replies with a counterexample $z \in \GSM$, we extend $z$ with an $\alpha \in \At$ such that $\llbracket  \mathscr{Y} \rrbracket(z\alpha) \not = \llbracket e \rrbracket(z\alpha)$.

\section{Related Work}
\label{relatedwork}

GKAT is a variation on KAT \cite{kozen1996kleene} that one obtains by restricting the union and iteration operations from KAT to guarded versions. While GKAT is less expressive than KAT, term equivalence is notably more efficiently decidable \cite{smolka2019guarded,kozen1996kleene}, making it a candidate for the foundations of network-programming \cite{smolka2019scalable,anderson2014netkat,foster2015coalgebraic}

GKAT automata appear in the literature already prior to \cite{smolka2019guarded}, e.g. in the work of Kozen \cite{kozen2008bohm} under the name \emph{strictly deterministic automata}. In the latter, Kozen states that GKAT automata correspond to a limited class of \emph{automata with guarded strings (AGS)} \cite{kozen2001automata}, for which he gives determinisation and minimisation constructions. In a different paper \cite{kozen2017coalgebraic} Kozen introduces a second definition of (deterministic) AGS as Moore automata, and states the difference to the definition in \cite{kozen2001automata} is inessential. 

Recently, a new perspective on the semantics and coalgebraic theory of GKAT has been given in terms of coequations \cite{schmid2021guarded,dahlqvist2021write}. Using the Thompson construction, it is possible to construct for every expression $e$ a language equivalent automaton $\mathscr{X}_e$. In \cite{kozen2008bohm} it was shown that in general there is no reverse construction: there exists a GKAT automaton that is inequivalent to $\mathscr{X}_e$ for all expressions $e$. In consequence, \cite{smolka2019guarded} proposed a subclass of \emph{well-nested} automata and showed that every finite well-nested automaton is bisimilar to $\mathscr{X}_e$ for some $e$. In \cite{schmid2021guarded} it was shown that well-nestedness is in fact too restrictive: there exists an automaton that is bisimilar to $\mathscr{X}_e$ for some $e$, but not well-nested. 
To capture the \emph{full} class of automata exhibiting the behaviour of expressions, one has to extend the class of well-nested automata to the class of automata satisfying the \emph{nesting coequation}, which forms a \emph{covariety} \cite{dahlqvist2021write} (cf. \Cref{reachablesubcoalgebra}) . 

Regular inference, or active automata learning, is a technique used for deriving a model from a black-box by interacting with it via observations. The original algorithm  $\LStar$ by Angluin \cite{angluin1987learning} learns deterministic finite automata, but since then has been extended to other classes of automata \cite{angluin1997learning,aarts2010learning,moerman2017learning}, including Moore automata, as we have seen in \Cref{intro_sec_othertypes}. Typically, algorithms such as $\LStar$ are designed to output for a given language a unique minimal acceptor. Not all classes admit a canonical minimal acceptor, for instance, learning non-deterministic models is a challenge \cite{denis2001residual,bollig2009angluin,zetzsche2021,van2020learning}, as is shown in detail in \Cref{canonical_automata_chapter}.

\section{Discussion and Future Work}

We have presented $\GLStar$, an algorithm for learning the GKAT automaton representation of a black-box, by observing its behaviour via queries to an oracle. We have shown that for every normal GKAT automaton there exists a unique size-minimal normal automaton, accepting the same language: its minimisation. We have identified the minimisation with an alternative but equivalent construction, and derived its preservation of the nesting coequation. A central result showed that if the oracle in $\GLStar$ is instantiated with the language accepted by a finite normal automaton, then $\GLStar$ terminates with its minimisation. A complexity analysis showed the advantage of $\GLStar$ over $\LStar$ for learning automata representations of GKAT programs in terms of membership queries. We discussed additional optimizations, and implemented $\GLStar$ and $\LStar$ in OCaml to compare their performances on example programs. 

There are numerous directions in which the present work could be further explored.
In \Cref{optimizedcounterexamplesec} we introduced an optimization for $\GLStar$ which is inspired by Rivest and Schapire's counterexample handling method for $\LStar$ \cite{rivest1993inference}. The \textit{oberservation pack} algorithm for $\LStar$ \cite{howar2012active} has successfully combined Rivest and Schapire's method with an efficient \textit{discrimination tree} data structure \cite{kearns1994introduction}. The state-of-the-art \textit{TTT}-algorithm \cite{isberner2014ttt} for $\LStar$ extends the former with discriminator finalization techniques. It thus is natural to ask whether for $\GLStar$ there exist similar data structures, potentially exploiting the deterministic nature of the languages accepted by GKAT automata.

While $\LStar$ has seen major improvements over the years and has inspired numerous variations for different types of transition systems, all approaches remain in common their focus on the \emph{equivalence} of observations. The recently presented $\Lsharp$ algorithm \cite{vaandrager2021new} takes a different perspective: it instead focuses on \emph{apartness}, a constructive form of inequality. $\Lsharp$ does not require data-structures such as observation tables or discrimination trees, instead operating directly on tree-shaped automata. It remains open whether a similar shift in perspective is feasible for $\GLStar$.

There exist various domain-specific extensions of KAT (e.g. KAT+B! \cite{grathwohl2014kat+}, NetKAT \cite{anderson2014netkat}, ProbNetKAT \cite{foster2016probabilistic}), and similar directions have been proposed for GKAT. It has been noted that GKAT is better fit for probabilistic domains than KAT, as it avoids mixing non-determinism with probabilities \cite{smolka2019scalable}. Due to its foundations in KAT, NetKAT's decision procedure is PSPACE-complete, thus hindering verification to scale.  In contrast to KAT, the equational theory of GKAT is decidable in (almost) linear time, making GKAT an interesting alternative candidate for the foundations of a SDN programming language like NetKAT. While there currently exists no explicit automata learning algorithm for NetKAT in the style of $\LStar$, there is work closely related \cite{smolka2015fast, foster2015coalgebraic}. Generally, we expect that in the future, for any extension of GKAT, there will be interest in developing the corresponding automata (learning) theories.

\chapter{Canonical Automata}

\label{canonical_automata_chapter}

In \Cref{gkatsection} we explored the automata theory of GKAT and presented $\GLStar$, an algorithm for learning a GKAT automaton by observing the behaviour of a black-box. The design of $\GLStar$ started with us assigning to a GKAT automaton a language equivalent second automaton -- its minimisation. The algorithm itself has been a derivation of this construction: for a given language, it iteratively refines approximations of the \emph{canonical} minimal target model we have proven to exist. In this chapter, we will investigate canonical target models of more general type.

The classical powerset construction is a standard method that is used to convert a non-deterministic automaton into a deterministic one recognising the same language. In \cite{silva2010generalizing}, the powerset construction has been lifted to a more general framework that converts an automaton with side-effects, given by a monad, into a deterministic automaton accepting the same language. The resulting automaton has additional algebraic properties, in the state space and transition structure, inherited from the monad. We will study the reverse construction and present a framework in which a deterministic automaton with additional algebraic structure over a given monad can be converted into an equivalent succinct automaton with side-effects. Apart from recovering examples from the literature, such as the canonical residual finite-state automaton and the \'atomaton, we discover a new canonical automaton for a regular language by relating the free vector space monad over the two element field to the neighbourhood monad. Finally, we show that every regular language satisfying a suitable property parametric in two monads admits a size-minimal succinct acceptor.  

\section{Introduction}

The existence of a unique minimal \emph{deterministic} acceptor is an important property of regular languages. Establishing a similar result for \emph{non-deterministic} acceptors is of great practical importance, as non-deterministic automata can be exponentially more succinct than deterministic ones. Unfortunately, a regular language can be accepted by several size-minimal NFAs that are not isomorphic. In \Cref{example_nonisomorphic_nfa_canonical_automata_chapter} we give one example of such behaviour.
A number of sub-classes of non-deterministic automata have been identified in the literature to tackle this issue, which all admit canonical representatives: the \emph{\'atomaton}~\cite{BrzozowskiT14}, the \emph{canonical residual finite-state automaton} (short \emph{canonical RFSA} and also known as \emph{jiromaton})~\cite{denis2001residual}, the \emph{minimal xor automaton}~\cite{VuilleminG210}, and the \emph{distromaton}~\cite{MyersAMU15}. 

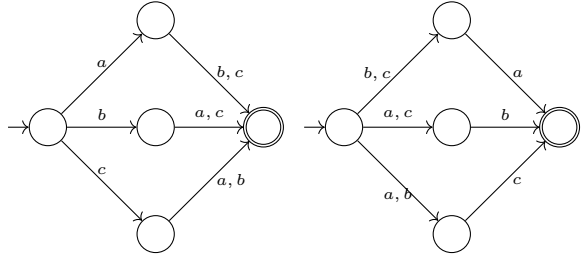
\begin{figure}
\tiny
\centering
	\begin{tikzpicture}[node distance=5.5em]
	\node[state, initial, initial text=, minimum size=1.9em] (q1) {};
	\node[state, right of=q1, above of=q1, minimum size=1.9em] (q2) {};
		\node[state, below of=q2, minimum size=1.9em] (q3) {};
		\node[state, below of=q3, minimum size=1.9em] (q4) {};
		\node[state, right of=q3, minimum size=1.9em, accepting] (q5) {};
	    \path[->]
	(q1) edge[above] node{$a$} (q2)
		(q1) edge[above] node{$b$} (q3)
		(q1) edge[above] node{$c$} (q4)
		(q2) edge[right] node{$b,c$} (q5)
		(q3) edge[above] node{$a,c$} (q5)
		(q4) edge[right] node{$a,b$} (q5)
        ;
\end{tikzpicture}
\begin{tikzpicture}[node distance=5.5em]
	\node[state, initial, initial text=, minimum size=1.9em] (q1) {};
	\node[state, right of=q1, above of=q1, minimum size=1.9em] (q2) {};
		\node[state, below of=q2, minimum size=1.9em] (q3) {};
		\node[state, below of=q3, minimum size=1.9em] (q4) {};
		\node[state, right of=q3, minimum size=1.9em, accepting] (q5) {};
	    \path[->]
	(q1) edge[left] node{$b,c$} (q2)
		(q1) edge[above] node{$a,c$} (q3)
		(q1) edge[below] node{$a,b$} (q4)
		(q2) edge[right] node{$a$} (q5)
		(q3) edge[above] node{$b$} (q5)
		(q4) edge[right] node{$c$} (q5)
        ;
\end{tikzpicture}
\caption{Two non-isomorphic size-minimal NFA accepting the language $\lbrace ab,ac,$ $ba,bc,ca,cb \rbrace \subseteq \lbrace a, b, c \rbrace^*$ \cite{arnold1992note}}
\label{example_nonisomorphic_nfa_canonical_automata_chapter}
\end{figure}

In this chapter we provide a general categorical framework that unifies constructions of canonical non-deterministic automata and unveils new ones. Our framework adopts the well-known representation of side-effects via monads (\Cref{def:monad}) to generalise non-determinism in automata. For instance, an NFA (without initial states) can be represented as a pair $( X,k ) $, where $X$ is the set of states and 
$
	k \colon X \to 2 \times \Pow(X)^A
$
combines the function classifying each state as accepting or rejecting with the function giving the set of next states for each input. The powerset forms a monad $( \Pow, \{-\}, \mu )$, where $\{-\}$ creates singleton sets and $\mu$ takes the union of a set of sets. This allows describing the classical powerset construction, converting an NFA into a DFA, in categorical terms as depicted on the left of \Cref{gen-det-diagrams}, where $k^\sharp \colon \Pow(X) \to 2 \times \Pow(X)^A$ represents an equivalent DFA\footnote{The classical powerset determinisation of $k = \langle \varepsilon, \delta \rangle: X \rightarrow 2 \times \mathcal{P}(X)^A$ is $k^{\sharp} = \langle \varepsilon^{\sharp}, \delta^{\sharp} \rangle: \mathcal{P}(X) \rightarrow 2 \times \mathcal{P}(X)^A$, where $\varepsilon^{\sharp}(U) = \vee_{x \in U} \varepsilon(x)$ and $\delta^{\sharp}(U)(a) = \cup_{x \in U} \delta(x)(a)$.}, obtained by taking the subsets of $X$ as states and $\langle \varepsilon, \delta \rangle : 2^{A^*} \rightarrow 2 \times (2^{A^*})^A$ is the automaton of languages. There then exists a unique automaton homomorphism $\obs$, assigning a language semantics to each set of states.

As seen on the right of \Cref{gen-det-diagrams} this perspective further enables a \emph{generalised determinisation} construction \cite{silva2010generalizing}, where $2 \times (-)^A$ is replaced by any (suitable) functor $F$ describing the automaton structure, and $\mathcal{P}$ by a monad $T$ describing the automaton side-effects.
In this picture, $\Omega \xrightarrow{\omega} F \Omega$ is the final coalgebra (\Cref{def:finalcoalgebra}), providing a semantic universe that generalises the automaton of languages.

Our work starts from the observation that the deterministic automata resulting from this generalised determinisation constructions have \emph{additional algebraic structure}: the state space $\Pow(X)$ of the determinised automaton defines a free complete join-semilattice (CSL) over $X$, and $k^\sharp$ and $\obs$ are CSL homomorphisms. More generally, $TX$ defines a (free) algebra for the monad $T$, and $k^\sharp$ and $\obs$ are $T$-algebra homomorphisms (\Cref{def:algebraovermonad}).

With this observation in mind, our question is: can we exploit the additional algebraic structure to ``reverse'' these constructions? In other words, can we convert a deterministic automaton with additional algebraic structure over a given monad to an equivalent succinct automaton with side-effects, possibly over another monad? To answer this question, we make the following contributions:

\begin{figure*}[t]
\centering
	\begin{tikzcd}[ ampersand replacement=\&]
		X \ar{d}[swap]{k} \ar{r}{\{-\}} \&
			\mathcal{P}(X) \ar{dl}{k^\sharp} \ar[dashed]{r}{\obs} \&
			2^{A^*} \ar{d}{\langle \varepsilon, \delta \rangle} \\
		2 \times \mathcal{P}(X)^A \ar[dashed]{rr}[below]{2 \times \obs^A} \&
			\&
			2 \times (2^{A^*})^A
	\end{tikzcd}
	\qquad
			\begin{tikzcd}[ampersand replacement=\&]
			X \ar{d}[swap]{k} \ar{r}{\eta} \&
			TX \ar{dl}{k^\sharp} \ar[dashed]{r}{\obs} \&
				\Omega \ar{d}{\omega} \\
			FTX \ar[dashed]{rr}[below]{F\obs} \&
				\&
				F\Omega
		\end{tikzcd}.
\caption{Generalised determinisation of automata with side-effects in a monad}
\label{gen-det-diagrams}
\end{figure*}

\begin{itemize}
	\item We present a categorical framework based on bialgebras (\Cref{bialgebradef}) and distributive law homomorphisms (\Cref{distributivelawhomdef}) that allows deriving canonical representatives for a wide class of succinct automata with side-effects in a monad. 		
	\item We strictly improve the expressivity of previous work \cite{HeerdtMSS19, arbib1975fuzzy}: our framework instantiates not only to well-known examples such as the canonical RFSA (\Cref{canonicalrfsaexample}) and the minimal xor automaton (\Cref{minimalxorexmaple}), but also includes the \'atomaton (\Cref{atomatonexample}) and the distromaton (\Cref{distromatonexample}), which were not covered in \cite{HeerdtMSS19, arbib1975fuzzy}.
 While other frameworks restrict themselves to the category of sets \cite{HeerdtMSS19}, we are able to include canonical acceptors in other categories, such as the \textit{canonical nominal RFSA} (\Cref{nominalexample}). 
	\item We relate vector spaces over the unique two element field with complete atomic Boolean algebras and consequently discover a previously unknown canonical mod-2 weighted acceptor for regular languages---the \emph{minimal xor-CABA automaton} (\Cref{minimalxorcabaexample})---that in some sense is to the minimal xor automaton what the \'atomaton is to the canonical RFSA  (\Cref{minimalxorcabadiagram}).
	\item We introduce an abstract notion of \emph{closedness} for succinct automata that is parametric in two monads (\Cref{closedsuccinctdef}), and 
	prove that every regular language satisfying a suitable property admits a canonical size-minimal representative among closed acceptors (\Cref{minimalitytheorem}). By instantiating the latter we subsume known minimality results for canonical automata, prove the xor-CABA automaton minimal, and establish a size comparison between different acceptors (\Cref{minimalityimplications}). 
\end{itemize}

\section{Overview of the Approach}

\label{overview}

In this section, we give an overview of the ideas of this chapter through an example. We show how our methodology allows recovering the construction of the 
\'atomaton for 
the regular language
$
	L = (a+b)^*a
$, which consists of all words over $A = \lbrace a, b \rbrace$ that end in $a$. For each step, we hint at how it is generalised in our framework.

The classical construction of the \'atomaton for $\lang$ consists in closing the \emph{residuals}\footnote{A language is a \textit{residual} or \textit{left quotient} of $L \subseteq A^*$, if it is of the form $v^{-1}L = \lbrace u \in A^* \mid vu \in L \rbrace$ for some $v \in A^*$. } of $\lang$ under all Boolean operations, and then forming a non-deterministic automaton whose states are the atoms\footnote{A non-zero element $a$ in a Boolean algebra B is called an \emph{atom}, if for all $x \in B$ with $x \leq a$ it follows $x = 0$ or $x = a$.}  of the ensuing complete atomic\footnote{A Boolean algebra $B$ is \emph{atomic}, if for all $x \in B$ there exists a decomposition $x = \vee_{I} a_i$, where $\lbrace a_i \mid i \in I \rbrace$ is some set of atoms.} Boolean algebra (CABA)---that is, non-empty intersections of complemented or uncomplemented residuals.
In our categorical setting, this construction is obtained in several steps, which we now describe.

\subsection{Computing Residuals}
We first construct the minimal DFA accepting $\lang$ as a coalgebra of type 
	$
		 \M_\lang \to 2 \times (\M_\lang)^{A} \enspace
	$.
	By the well-known Myhill-Nerode theorem~\cite{nerode1958linear}, $\M_\lang$ is the set of residuals for $\lang$. The automaton is depicted in \Cref{m(l)}.

	\begin{figure*}
		\small
		\center
	\begin{tikzpicture}[node distance=6em]
	\node[state, initial, initial text=] (x) {$x$};	
		\node[state, right of=x, accepting] (y) {$y$};
	    \path[->]
	(x) edge[loop above] node{$b$} (x)
	(y) edge[loop right] node{$a$} (y)
	(x) edge[above, bend left] node{$a$} (y)
	(y) edge[below, bend left] node{$b$} (x)
	;
	\end{tikzpicture}.
		\caption{The minimal DFA for $L = (a+b)^*a$}
	\label{m(l)}
	\end{figure*}
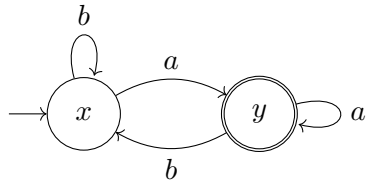
			In our framework, we consider coalgebras over an arbitrary endofunctor $F \colon \C \to \C$ ($F = 2 \times (-)^{A}$ and $\C = \Set$ in this case), as introduced in \Cref{def:coalgebra}.
Minimal realisations, generalising minimal DFAs, exist for a wide class of functors $F$ and categories $\C$, including all the examples in this chapter.

\subsection{Taking the Boolean Closure}
We close the minimal DFA under all Boolean operations, generating an equivalent deterministic automaton that has additional algebraic structure: its state space is a CABA.
This is achieved via a double powerset construction---where sets of sets are interpreted as full disjunctive normal form---and the resulting coalgebra is of type
	$
		\Pow^2(\M_\lang) \to 2 \times (\Pow^2(\M_\lang))^{A}	$.
	Our construction relies on the \emph{neighbourhood monad} $\mathcal{H}$ (\Cref{exampleofmonads}), whose algebras are precisely CABAs, and yields a (free) \emph{bialgebra} capturing both the coalgebraic and the algebraic structure; the interplay of these two structures is captured via a \emph{distributive law}. 
 We then minimise this DFA to identify Boolean expressions evaluating to the same language. As desired, the resulting state space is precisely the Boolean closure of the residuals of $\lang$. Formally, we obtain the minimal bialgebra for $\lang$ depicted in \Cref{overlinem(l)atom}.

This step in our framework is generalised as closure of an $F$-coalgebra w.r.t\ (the algebraic structured induced by) any monad $S$ for which a suitable distributive law $\lambda$ with the coalgebra endofunctor $F$ exists. The first step of the closure yields a free $\lambda$-bialgebra, comprised of both an $F$-coalgebra and an $S$-algebra over the same state space. In a second step, minimisation is carried out in the category of $\lambda$-bialgebras, which guarantees simultaneous preservation of the algebraic structure and of the language semantics.

\begin{figure*}
\center	
	 \tiny
         \adjustbox{valign=m}{\begin{tikzpicture}[node distance= 8.5em]
	\node[state, initial, shape=rectangle, initial text=] (1) {$\lbrack \lbrace \lbrace x \rbrace, \lbrace x, y \rbrace \rbrace \rbrack $};

	\node[state, shape=rectangle, right of = 1] (3) {$\lbrack \emptyset \rbrack$};

	\node[state, shape=rectangle,below of = 1, accepting] (5) {$\lbrack\lbrace \lbrace x \rbrace,  \lbrace y \rbrace, \lbrace x, y \rbrace \rbrace \rbrack$};

	\node[state, shape=rectangle,right of = 5, accepting] (7) {$\lbrack\lbrace \lbrace y \rbrace \rbrace \rbrack$};

	\node[state, shape=rectangle,right of = 3] (9) {$\lbrack\lbrace \emptyset \rbrace \rbrack$};

	\node[state, shape=rectangle,right of = 9] (11) {$\lbrack \lbrace \lbrace x, y \rbrace, \emptyset \rbrace \rbrack$};

	\node[state, shape=rectangle,right of = 7, accepting] (13) {$\lbrack \lbrace \lbrace y \rbrace, \emptyset \rbrace \rbrack$};

	\node[state, shape=rectangle, right of = 13, accepting] (15) {$\lbrack \lbrace \lbrace x \rbrace, \lbrace y \rbrace, \lbrace x, y \rbrace, \emptyset \rbrace \rbrack$};
			
		    \path[->]
	(3) edge[loop above] node{$a,b$} (3)
	(1) edge[loop above] node{$b$} (1)
	(1) edge[right, bend left] node{$a$} (5)
	(7) edge[left] node{$a,b$} (3)
	(5) edge[loop below] node{$a$} (5)
	(5) edge[left, bend left] node{$b$} (1)
	(9) edge[bend left, right] node{$b$} (13)
	(9) edge[loop above] node{$a$} (9)
	(11) edge[left] node{$a,b$} (15)
	(13) edge[left, bend left] node{$a$} (9)
	(13) edge[loop below] node{$b$} (13)
	(15) edge[loop below] node{$a,b$} (15)
;
         \end{tikzpicture}}
	\qquad
                \adjustbox{valign=m}{\resizebox{0.35 \columnwidth}{!}{%
	\begin{tabular}[]{ c|c|c|c|c|c|c|c|c } 
		$\wedge$ & $1$ & $2$ & $3$ & $4$ & $5$ & $6$ & $7$ & $8$ \\
		\hline
		$1$ & $1$ & $2$ & $2$ & $1$ & $1$ & $2$ & $2$ & $1$ \\
		\hline
		$2$ & $2$ & $2$ & $2$ & $2$  &  $2$ & $2$  & $2$  & $2$  \\
		\hline
		$3$ & $2$ & $2$ & $3$ & $3$ & $2$ & $2$ & $3$ & $3$ \\
		\hline
		$4$ & $1$ & $2$ & $3$ & $4$ & $1$ & $2$ & $3$ & $4$ \\
		\hline
		$5$ & $1$ & $2$ & $2$ & $1$ & $5$ & $6$ & $6$ & $5$ \\
		\hline
		$6$ & $2$ & $2$ & $2$ & $2$ & $6$ & $6$ & $6$ & $6$ \\
		\hline
		$7$ & $2$ & $2$ & $3$ & $3$ & $6$ & $6$ & $7$ & $7$ \\
		\hline
		$8$ & $1$ & $2$ & $3$ & $4$ & $5$ & $6$ & $7$ & $8$
	\end{tabular}
	\qquad
		\begin{tabular}[]{ c|c} 
		  & $\neg$\\
		 \hline
		 $1$ & $7$ \\
		 \hline
		 $2$ & $8$ \\
		 \hline
		 $3$ & $5$ \\
		 \hline
		 $4$ & $6$ \\
		 \hline
		 $5$ & $3$ \\
		 \hline
		 $6$ & $4$ \\
		 \hline
		 $7$ & $1$ \\
		 \hline
		 $8$ & $2$ 
		\end{tabular}
        }}
		\caption{The minimal CABA-structured DFA for $L = (a+b)^*a$, where
$1 \equiv \lbrack \lbrace \lbrace x \rbrace, \lbrace x, y \rbrace \rbrace \rbrack$, $2 \equiv \lbrack \emptyset \rbrack$, $3 \equiv \lbrack \lbrace \emptyset \rbrace \rbrack$, $4 \equiv \lbrack \lbrace \lbrace x, y \rbrace, \emptyset \rbrace \rbrack$, $5 \equiv \lbrack \lbrace \lbrace x \rbrace, \lbrace y \rbrace,$ $\lbrace x, y \rbrace \rbrace \rbrack$, $6 \equiv \lbrack \lbrace \lbrace y \rbrace \rbrace \rbrack$, $7 \equiv \lbrack \lbrace \lbrace y \rbrace, \emptyset \rbrace \rbrack$, $8 \equiv \lbrack \lbrace \lbrace x \rbrace, \lbrace y \rbrace, \lbrace x, y \rbrace, \emptyset \rbrace \rbrack$}
	\label{overlinem(l)atom}
	\end{figure*} 

\subsection{Constructing the \'Atomaton}
This step is the key technical result of this chapter. Atoms have the property that their Boolean closure generates the entire CABA. In our framework, this property is generalised via the notion of \emph{generators} for algebras over a monad, which allows one to represent a bialgebra as an equivalent \emph{free} bialgebra over its generators, and hence to obtain succinct canonical representations (\Cref{forgenerator-isharp-is-bialgebra-hom}). In \Cref{succinctbialgebra1} we apply this result to obtain the canonical RFSA, the canonical nominal RFSA, and the minimal xor automaton for a given regular language.

However, to recover the \'atomaton from the minimal CABA-structured DFA of the previous step, in addition a subtle change of perspective is required. In fact, we are still working with the ``wrong'' side-effect: the non-determinism of bialgebras so far is determined by $\mathcal{H}$, whereas we are interested in an NFA, whose non-determinism is captured by the powerset monad $\mathcal{P}$. 
As is well-known, every element of a CABA can be obtained as the join of the atoms below it.
In other words, those atoms are also generators of the underlying CSL, which is an algebra for $\mathcal{P}$. We formally capture this idea as a map between monads $\mathcal{H} \to \mathcal{P}$. Crucially, we show that this map lifts to a \emph{distributive law homomorphism} and allows translating a bialgebra over $\mathcal{H}$ to a bialgebra over $\mathcal{P}$, which can be represented as a free bialgebra over atoms---the \'atomaton for $\lang$, which is shown in \Cref{atomaton}.

In \Cref{succinctbialgebra2} we generalise this idea to the situation of two monads $S$ and $T$ involved in distributive laws with the coalgebra endofunctor $F$. In particular, \Cref{generatorbialgebrahom} is our free representation result, spelling out a condition under which a bialgebra over $S$ can be represented as a free bialgebra over $T$, and hence admits an equivalent succinct representation as an automaton with side-effects in $T$. Besides the \'atomaton and the examples in \Cref{succinctbialgebra1}, this construction allows us to capture the distromaton and a newly discovered canonical acceptor that relates CABAs with vector spaces over the two element field. 

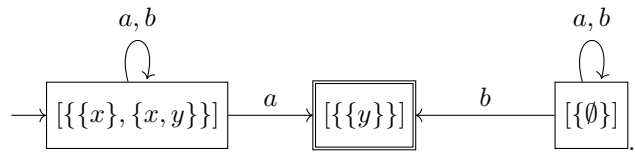
\begin{figure*}
	\center
	\footnotesize
\begin{tikzpicture}[node distance= 8.5em]
	\node[state, initial, shape=rectangle, initial text=] (1) {$\lbrack \lbrace \lbrace x \rbrace, \lbrace x, y \rbrace \rbrace \rbrack $};

	\node[state, shape=rectangle,right of = 1, accepting] (2) {$\lbrack\lbrace \lbrace y \rbrace \rbrace \rbrack$};

	\node[state, shape=rectangle,right of = 2] (3) {$\lbrack\lbrace \emptyset \rbrace \rbrack$};

		    \path[->]
	(1) edge[loop above] node{$a,b$} (1)
	(1) edge[above] node{$a$} (2)
	(3) edge[loop above] node{$a,b$} (3)
	(3) edge[above] node{$b$} (2)	
;
	\end{tikzpicture}.
	\caption{The \'atomaton for $L = (a+b)^*a$}
\label{atomaton}
\end{figure*}

\section{Distributive Laws and Bialgebras}

In this section, we briefly recall distributive laws and bialgebras, which form the technical foundations of our category-theoretical framework. 

Distributive laws have originally occurred as a way to compose monads \cite{beck1969distributive}, but now also exist in a wide range of other forms \cite{Street2009}. For our particular case it is sufficient to consider distributive laws between a monad and an endofunctor, sometimes referred to as \emph{Eilenberg-Moore laws} \cite{jacobs2012trace}.
\begin{definition}[Distributive Law]
\label{distributivelawdef}
	A \emph{distributive law} between a monad $T$ and an endofunctor $F$ on $\C$ is a natural transformation $\lambda: TF \Rightarrow FT$
such that the following two diagrams commute:
\[
\begin{tikzcd}
FX \arrow{r}{F\eta_X} \arrow{d}[left]{TFX} & FTX \\
TFX \arrow{ur}[right]{\lambda_X}
\end{tikzcd}
\qquad
\begin{tikzcd}
T^2FX \arrow{d}[left]{\mu_{FX}} \arrow{r}{T\lambda_X} & TFTX \arrow{r}{\lambda_{TX}} & FT^2X \arrow{d}{F\mu_X} \\
TFX \arrow{rr}[below]{\lambda_X} & & FTX	
\end{tikzcd}.
\]
\end{definition}

We are particularly interested in canonically induced distributive laws involving the endofunctor $F$ with $FX = B \times X^A $. The following statement is well-known and appears for instance in \cite{jacobs2006bialgebraic}. We include a proof of the result, since we were unable to locate one in the literature. Note that the result can easily be generalised to a strong monad on any cartesian closed category. 

\begin{lemma}
\label{induceddistrlaw}
	Every algebra $h: TB \rightarrow B$ for a set monad $T$ induces a distributive law $\lambda^h$ between $T$ and $F$ with $FX = B \times X^A $ defined as the composition
	\begin{equation}
 \label{induceddistrlaweq}
 	\lambda^{h}_X := T(B \times X^A) \overset{\langle T\pi_1, T\pi_2 \rangle}{\longrightarrow} T(B) \times T(X^A) \overset{h \times \st}{\longrightarrow} B \times (TX)^A,
 \end{equation}
 where $\st$ denotes the canonical strength function\footnote{\label{strengthdef}For any two sets $X,A$ the strength function $\st: T(X^A) \rightarrow (TX)^A$ is defined by $ \st(U)(a) = T(\ev_a)(U)$,
	where $\ev_{a}(f) = f(a)$.}.
\end{lemma}
\begin{proof}
\begingroup
\allowdisplaybreaks
The naturality of $\lambda^h$ essentially follows from the naturality of the strength function.
The equation involving the monad unit is a consequence of
\begin{align*}
	& \pi_1 \circ (h \times \st) \circ \langle T\pi_1, T\pi_2 \rangle \circ \eta_{B \times X^A} && \\
=\ & h \circ T\pi_1 \circ \eta_{B \times X^A} && \textnormal{(Definition of } \pi_1) \\
=\ & h \circ \eta_B \circ \pi_1 && \textnormal{(Naturality of } \eta) \\
=\ &  \pi_1 \circ (B \times \eta_X^A) && (\textnormal{Definition of } \pi_1, h \circ \eta_B = \id_B)
\end{align*}
and the chain of equalities
\begin{align*}
	 \pi_2 \circ (h \times \st) \circ \langle T\pi_1, T\pi_2 \rangle \circ \eta_{B \times X^A}  &= \st \circ T\pi_2 \circ \eta_{B \times X^A} && \textnormal{(Definition of } \pi_2) \\
	 &= \st \circ \eta_{X^A} \circ \pi_2 && \textnormal{(Naturality of } \eta) \\
	 &= \eta_X^A \circ \pi_2  && (\st \circ \eta_{X^A} = \eta_X^A)  \\
	&= \pi_2 \circ (B \times \eta_X^A) && \textnormal{(Definition of } \pi_2)
\end{align*}
	Similarly, the equation involving the monad multiplication is a consequence of
		\begin{align*}
		&\pi_1 \circ (B \times \mu_X^A) \circ (h \times \st) \circ \langle T \pi_1, T \pi_2 \rangle 
		 \circ T(h \times \st) \circ T\langle T\pi_1, T\pi_2 \rangle
		   \\
		=&\  \textnormal{(Definition of } \pi_1) \\
		& h \circ Th \circ T^2\pi_1 \\
		=&\  (h \circ Th = h \circ \mu_B) \\
		 & h \circ \mu_B \circ T^2 \pi_1 \\
		 =&\ \textnormal{(Naturality of } \mu) \\
		 & h \circ T \pi_1 \circ \mu_{B \times X^A} \\
		  =&\ \textnormal{(Definition of } \pi_1) \\
		  & \pi_1 \circ (h \times \st) \circ \langle T\pi_1, T\pi_2 \rangle \circ \mu_{B \times X^A} 
	\end{align*}
	and the equalities
	\begin{align*}
		&\pi_2 \circ (B \times \mu_X^A) \circ (h \times \st) \circ \langle T \pi_1, T \pi_2 \rangle
		 \circ T(h \times \st) \circ T( T\pi_1, T\pi_2 )  \\
		 =& \  \textnormal{(Definition of } \pi_2) \\
		 & \mu_X^A \circ \st \circ T(\st) \circ T^2\pi_2  \\
		=& \  (\mu^A_X \circ \st \circ T(\st) = \st \circ \mu_{X^A}) \\
		 & \st \circ \mu_{X^A} \circ T^2\pi_2 \\
		 =\  & \textnormal{(Naturality of } \mu) \\
		& \st \circ T\pi_2 \circ \mu_{B \times X^A} \\
		=& \  \textnormal{(Definition of } \pi_2) \\
		 & \pi_2 \circ (h \times \st) \circ \langle T\pi_1, T\pi_2 \rangle \circ \mu_{B \times X^A}
	\end{align*}
	\endgroup
\end{proof}

Of particular importance for us are the canonical algebra structures for the output set $B = 2$. For instance, the algebra structures defined by $h^{\mathcal{P}}(\varphi) = h^{\mathcal{R}}(\varphi) =  \varphi(1)$ and $h^{\mathcal{H}}(\Phi) = h^{\mathcal{A}}(\Phi) = \Phi(\id_2)$, where we identify subsets with their characteristic functions, and $\mathcal{R}$ and $\mathcal{A}$ denote the free vector space monad over the unique two element field $\mathbb{Z}_2$, and the monotone neighbourhood monad, respectively (cf. \Cref{exampleofmonads}). In these cases we will abuse notation and write $\lambda^T$ instead of $\lambda^{h^T}$. The soundness of the definitions is witnessed by the following four short results.

\begin{lemma}
\label{powersetalgebra}
	The morphism $h^{\mathcal{P}}: \mathcal{P}2 \rightarrow 2$ satisfying $\varphi \mapsto \varphi(1)$ defines a $\mathcal{P}$-algebra.
\end{lemma}
\begin{proof}
On the one hand we find
\begingroup
\allowdisplaybreaks
	\begin{align*}
		& h^{\mathcal{P}} \circ \mu^{\mathcal{P}}_2(\Phi) && \\
		=\ & \mu^{\mathcal{P}}_2(\Phi)(1) && \textnormal{(Definition of } h^{\mathcal{P}}) \\
		=\ & \vee_{\varphi \in 2^2} \Phi(\varphi) \wedge \varphi(1) && \textnormal{(Definition of } \mu^{\mathcal{P}}_2) \\
		=\ & (\vee_{\varphi \in (h^{\mathcal{P}})^{-1}(1)} \Phi(\varphi) \wedge \varphi(1)) \vee (\vee_{\varphi \in (h^{\mathcal{P}})^{-1}(0)} \Phi(\varphi) \wedge \varphi(1)) && \textnormal{(Case split)} \\
		=\ & (\vee_{\varphi \in (h^{\mathcal{P}})^{-1}(1)} \Phi(\varphi) \wedge \varphi(1)) \vee (\vee_{\varphi \in (h^{\mathcal{P}})^{-1}(0)} \Phi(\varphi) \wedge 0) && (\varphi \in (h^{\mathcal{P}})^{-1}(0)) \\
		=\ & (\vee_{\varphi \in (h^{\mathcal{P}})^{-1}(1)} \Phi(\varphi) \wedge \varphi(1)) \vee (\vee_{\varphi \in (h^{\mathcal{P}})^{-1}(0)} 0) && (a \wedge 0 = 0) \\
		=\ & (\vee_{\varphi \in (h^{\mathcal{P}})^{-1}(1)} \Phi(\varphi) \wedge \varphi(1)) \vee 0 && (0 \vee 0 = 0)\\
		=\ & \vee_{\varphi \in (h^{\mathcal{P}})^{-1}(1)} \Phi(\varphi) \wedge \varphi(1) && (a \vee 0 = a)  \\
		=\ & \vee_{\varphi \in (h^{\mathcal{P}})^{-1}(1)} \Phi(\varphi) \wedge 1 && (\varphi \in (h^{\mathcal{P}})^{-1}(1)) \\
		=\ & \vee_{\varphi \in (h^{\mathcal{P}})^{-1}(1)} \Phi(\varphi) && (a \wedge 1 = a) \\
		=\ & \mathcal{P}(h^{\mathcal{P}})(\Phi)(1) && \textnormal{(Definition of } \mathcal{P}(h^{\mathcal{P}})) \\
		=\ & h^{\mathcal{P}} \circ \mathcal{P}(h^{\mathcal{P}})(\Phi) && \textnormal{(Definition of } h^{\mathcal{P}})
	\end{align*}
	\endgroup
	and on the other hand we can deduce
	\begin{align*}
		h^{\mathcal{P}} \circ \eta^{\mathcal{P}}_2(x) &= \eta^{\mathcal{P}}_2(x)(1) && \textnormal{(Definition of } h^{\mathcal{P}}) \\
		&= \lbrack x = 1 \rbrack && \textnormal{(Definition of } \eta^{\mathcal{P}}_2 ) \\
		&= x && \textnormal{(Definition of } \lbrack \cdot \rbrack)
	\end{align*}
where $\lbrack x = y \rbrack$ is $1$, if $x = y$, and $0$ otherwise.	
\end{proof}

As one verifies, under the equivalence $\Set^{\mathcal{P}} \cong \CSL$, the morphism $h^{\mathcal{P}}$ equips $2$ with its canonical complete join-semilattice structure $(2, \vee)$. 

\begin{lemma}
\label{neighbourhoodalgebra}
	The morphism $h^{\mathcal{H}}: \mathcal{H}2 \rightarrow 2$ assigning $\Phi \mapsto \Phi(\id_2)$ defines an $\mathcal{H}$-algebra.
\end{lemma}
\begin{proof}
\begingroup
\allowdisplaybreaks
	Since $\eta^{\mathcal{H}}_{2^2}(\id_2)(\Phi) = \Phi(\id_2) = h^{\mathcal{H}}(\Phi)$ we find
	\begin{align*}
		h^{\mathcal{H}} \circ \mu^{\mathcal{H}}_2(\Psi) &= \mu_2^{\mathcal{H}}(\Psi)(\id_2) && \textnormal{(Definition of } h^{\mathcal{H}}) \\
		&= \Psi(\eta_{2^2}^{\mathcal{H}}(\id_2)) && (\textnormal{Definition of } \mu^{\mathcal{H}}_2) \\
		&= \Psi(\id_2 \circ h^{\mathcal{H}}) && (\eta^{\mathcal{H}}_{2^2}(\id_2) = h^{\mathcal{H}})  \\
		&= \mathcal{H}(h^{\mathcal{H}})(\Psi)(\id_2) && \textnormal{(Definition of } \mathcal{H}(h^{\mathcal{H}})) \\
		&= h^{\mathcal{H}} \circ \mathcal{H}(h^{\mathcal{H}})(\Psi) && (\textnormal{Definition of } h^{\mathcal{H}})
	\end{align*}
	We further can deduce
	\begin{align*}
		h^{\mathcal{H}} \circ \eta_X^{\mathcal{H}}(x) &= \eta_X^{\mathcal{H}}(x)(\id_2) && \textnormal{(Definition of } h^{\mathcal{H}}) \\
		&= \id_2(x) && \textnormal{(Definition of } \eta^{\mathcal{H}}_X) \\
		&= x && \textnormal{(Definition of } \id_2) 
	\end{align*}
	\endgroup
\end{proof}

As one verifies, under the equivalence $\Set^{\mathcal{H}} \cong \CABA$, the morphism $h^{\mathcal{H}}$ equips $2$ with its canonical complete atomic Boolean algebra structure $(2, \vee, \wedge, \neg)$.

\begin{lemma}
\label{monotoneneighbourhoodoutputalgebra}
	The morphism $h^{\mathcal{A}}: \mathcal{A}2 \rightarrow 2$ assigning $\Phi \mapsto \Phi(\id_2)$ defines an $\mathcal{A}$-algebra.
\end{lemma}
 \begin{proof}
 	Analogous to the proof of \Cref{neighbourhoodalgebra}.
 \end{proof}
 
 As one verifies, under the equivalence $\Set^{\mathcal{A}} \cong \CDL$, the morphism $h^{\mathcal{A}}$ equips $2$ with its canonical completely distributive lattice structure $(2, \vee, \wedge)$.

 \begin{lemma}
\label{xoroutputalgebra}
	The morphism $h^{\mathcal{R}}: \mathcal{R}2 \rightarrow 2$ satisfying $\varphi \mapsto \varphi(1)$ defines an $\mathcal{R}$-algebra.
\end{lemma}
\begin{proof}
	Analogous to the proof of \Cref{powersetalgebra}, after observing that $a \wedge 0 = 0$, $0 \oplus 0 = 0$, $a \oplus 0 = a$ and $a \wedge 1 = a$ for all $a \in 2$.
\end{proof}

 As one verifies, under the equivalence $\Set^{\mathcal{R}} \cong \mathbb{Z}_2\Vect$, the morphism $h^{\mathcal{R}}$ equips $2$ with its canonical $\mathbb{Z}_2$-vector space structure $\mathbb{Z}_2 \cong (2, \oplus, \wedge)$.

\begin{example}[Generalised Determinisation \cite{rutten2013generalizing}]
\label{determinisationexample}
Given a distributive law, one can model the determinisation of a system with dynamics in $F$ and side-effects in $T$ (sometimes referred to as \emph{succinct} automaton) by lifting an $FT$-coalgebra $( X, k)$ to the $F$-coalgebra 
	$( TX, k^{\sharp} )$, where $k^{\sharp} := (F \mu_X \circ \lambda_{TX}) \circ Tk$. As one verifies, the latter is in fact a $T$-algebra homomorphism of type $k^{\sharp}: (TX, \mu_X) \rightarrow (FTX, F\mu_X \circ \lambda_{TX})$. For instance, if the distributive law $\lambda$ is induced by the  disjunctive $\mathcal{P}$-algebra $h^{\mathcal{P}}: \mathcal{P}2 \rightarrow 2$ with $h^{\mathcal{P}}(\varphi) = \bigvee_{u \in \varphi} u = \varphi(1)$, the lifting $k^{\sharp}$ is the DFA in $\CSL$ obtained from an NFA $k$ via the classical powerset construction.   
\end{example}	
	
	The example above illustrates the concept of a bialgebra: the algebraic part $(TX, \mu_X)$ and the coalgebraic part $(TX, k^{\sharp})$ of a lifted automaton are compatible along the distributive law $\lambda$.

    \begin{definition}[Bialgebra]
    \label{bialgebradef}
		A $\lambda$\emph{-bialgebra} is a tuple $(X, h, k)$ consisting of a $T$-algebra $(X,h)$ and an $F$-coalgebra $(X, k)$ 
such that the following diagram commutes:
\[
\begin{tikzcd}
TX \arrow{rr}{h} \arrow{d}[left]{Tk} & & X	\arrow{d}{k} \\
TFX \arrow{r}[below]{\lambda_X} & FTX \arrow{r}[below]{Fh} & FX
\end{tikzcd}
\]
	\end{definition}
	
	 A homomorphism between $\lambda$-bialgebras is a morphism between the underlying objects that is simultaneously a $T$-algebra homomorphism and an $F$-coalgebra homomorphism.
	The category of $\lambda$-bialgebras and homomorphisms is denoted by $\Bialg(\lambda)$. 
	
	The existence of a final $F$-coalgebra is equivalent to the existence of a final $\lambda$-bialgebra, as the next result shows.
	
	\begin{lemma}[\cite{jacobs2012trace}]
		Let $(\Omega, k_{\Omega})$ be the final $F$-coalgebra, then $(\Omega, h_{\Omega}, k_{\Omega})$ with $h_{\Omega}:= \obs_{(T\Omega, \lambda_{\Omega} \circ T k_{\Theta})}$ is the final $\lambda$-bialgebra satisfying $\obs_{(X, h, k)} = \obs_{(X, k)}$. Conversely, if $(\Omega, h_{\Omega}, k_{\Omega})$ is the final $\lambda$-bialgebra, then $(\Omega, k_{\Omega})$ is the final $F$-coalgebra.
	\end{lemma}
	For the distributive law in \Cref{determinisationexample},
	the final bialgebra is carried by the final coalgebra $\Pow(A^*)$ with the free $\mathcal{P}$-algebra structure that takes the union of languages.
	
	The generalised determinisation in \Cref{determinisationexample} can be rephrased as a functor $\expa_T$ that \emph{expands} an $F$-coalgebra with side-effects in $T$ into a $\lambda$-bialgebra.
	
	\begin{lemma}[\cite{jacobs2012trace}]
	\label[lemma]{expfunctor}
		Defining $\expa_T(X, k) := (TX, \mu_X, (F \mu_X \circ \lambda_{TX}) \circ Tk )$ and \\$\expa_T(f):= Tf$ yields a functor $\expa_T: \Coalg(FT) \rightarrow \Bialg(\lambda)$. 
	\end{lemma}
	
We will also refer to the functor $\free_T$ that arises from $\expa_T$ by pre-composition with the canonical embedding of $F$-coalgebras into $FT$-coalgebras, therefore assigning to an $F$-coalgebra the $\lambda$-bialgebra it \emph{freely} generates.

	\begin{corollary}
	\label{freeexpfunctor}
		Defining $\free_T(X, k) := (TX, \mu_X, \lambda_X \circ Tk)$ and $\free_T(f) := Tf$ yields a functor $\free_T: \Coalg(F) \rightarrow \Bialg(\lambda)$ with $\free_T(X, k) = \expa_T(X, F\eta_X \circ k)$.
	\end{corollary}

\section{Succinct Automata from Bialgebras}

\label{succinctbialgebra1}

In this section we discuss the foundations of our theoretical contributions.
We begin by recalling the notion of a \textit{generator} \cite{arbib1975fuzzy} and demonstrate how it can be used to translate a bialgebra into an equivalent free bialgebra.
While the treatment is very general, we are particularly interested in the case in which the bialgebra is given by a deterministic automaton that has additional algebraic structure over a given monad, and the translation results in an automaton with side-effects in that monad.
We will demonstrate that the theory in this section instantiates to the canonical RFSA \cite{denis2001residual}, the canonical \emph{nominal} RFSA \cite{moerman2019residual}, and the minimal xor automaton \cite{VuilleminG210}.

\begin{definition}[Generator and Basis]
\label{generatordefinition}
	A \emph{generator} for a $T$-algebra $( X, h )$ is a tuple $( Y, i, d )$  consisting of an object $Y$, a morphism $i \colon Y \rightarrow X$, and a morphism $d \colon X \rightarrow TY$ such that $(h \circ Ti) \circ d = \id_X$.
	A generator is called a \emph{basis} if it additionally satisfies $d \circ (h \circ Ti) = \id_{TY}$, that is, the following two diagrams commute:
	\[
	\begin{tikzcd}
	X \arrow{d}[left]{d} \arrow{r}{\id_X} & X \\
	TY \arrow{r}[below]{Ti} & TX \arrow{u}[right]{h}	
	\end{tikzcd}
	\qquad
	\begin{tikzcd}
	TY \arrow{d}[left]{Ti} \arrow{r}{\id_{TY}} & TY \\
	TX \arrow{r}[below]{h} & X \arrow{u}[right]{d}	
	\end{tikzcd}.
	\]
\end{definition}

A generator for an algebra is called a \textit{scoop} by Arbib and Manes~\cite{arbib1975fuzzy}. Here, we additionally introduce an abstract perspective on the notion of a basis.
Intuitively, one calls a set $Y$ that is embedded into an algebraic structure $X$ a generator for the latter if every element $x \in X$ admits a decomposition $d(x) \in TY$ into a formal combination of elements of $Y$ that evaluates to $x$.
If the decomposition is moreover \emph{unique}, that is, $h \circ Ti$ is not only a \emph{surjection} with right-inverse $d$, but a \emph{bijection} with two-sided inverse $d$, then a generator is called a basis.
Every algebra is generated by itself using the generator $( X, \id_X, \eta_X )$ and every free algebra $(TY, \mu_Y)$ admits the basis $(Y, \eta_Y, \id_{TY})$. In fact, an algebra admits a basis if and only if it is isomorphic to a free algebra (cf. \Cref{basisifffree}).
   We are particularly interested in classes of set-based algebras for which \emph{every} algebra admits a \emph{size-minimal} generator, that is, no generator has a carrier of smaller size. In such a situation we will also speak of \emph{canonical} generators.

\begin{example}
\label{setbasedexamples}
	\begin{itemize}
		\item 
			A tuple $( Y, i, d )$ is a generator for a $\mathcal{P}$-algebra $L = ( X,h ) \simeq ( X, \vee^h )$ iff $x = \vee^h_{y \in d(x)} i(y) $ for all $x \in X$, where we write $\vee^h$ for the complete join-semilattice structure induced by $h$. Note that if $Y \subseteq X$ is a subset, then $i(y) = y$ for all $y \in Y$.
            If $L$ satisfies the descending chain condition\footnote{\label{dcc}A partially ordered set $(P, \leq)$ satisfies the \emph{descending chain condition} if any descending chain $a_0 \geq a_1 \geq a_2 \geq ...$ in $P$ stabilises, that is, there exists some $n \geq 0$, such that $a_n = a_{n + k}$ for all $k \geq 0$.} (DCC), which is in particular the case if $X$ is finite, then defining $i(y) = y$ and $d(x) = \lbrace y \in J(L) \mid y \leq x \rbrace$ turns the set of join-irreducibles\footnote{A non-zero element $a$ in a lattice $L$ is called \emph{join-irreducible}, if for all $y,z \in L$ with $a=y \vee z$ it follows $a = y $ or $a = z$.} $J(L)$ into a size-minimal generator $( J(L), i, d )$ for $L$, cf. \Cref{joinirreducstateminimal} below.
		\item 
			A tuple $( Y, i, d )$ is a generator for an $\mathcal{R}$-algebra $V = ( X,h ) \simeq ( X, +^h, \cdot^h )$ iff $x = \sum^h_{y \in Y} d(x)(y) \cdot^h i(y)$ for all $x \in X$, where we write $+^h$ and $\cdot^h$ for the $\mathbb{Z}_2$-vector space structure induced by $h$.
			As it is well-known that every vector space can be equipped with a basis, every $\mathcal{R}$-algebra $V$ admits a basis. One can show that any finite basis is a size-minimal generator, cf. \Cref{xorbasisstateminimal} below.
	\end{itemize}
\end{example}

\begin{lemma}
\label{joinirreducstateminimal}
	For any finite $\mathcal{P}$-algebra $L = ( X, h )$ the join-irreducibles $( J(L), i, d )$ with $i(y) = y$ and $d(x) = \lbrace y \in J(L) \mid y \leq x \rbrace$ constitute a size-minimal generator.
	 \end{lemma}
\begin{proof}
	Since $L$ is finite, it satisfies the DCC, which in turn can be used to show that the join-irreducibles constitute a generator as follows.

Assume there exists some $x \in X$ with $x \not= i^{\sharp}(d(x))$. We build an infinite sequence $(a_n)$ with $a_i > a_{i+1}$ and $a_i \not = i^{\sharp}(d(a_i))$, which contradicts the DCC. For the base case we define $a_0 := x$. For any $x \in X$, the property $x \in J(L)$ immediately implies $x = i^{\sharp}(d(x))$. Thus we can assume $a_i \not \in J(L)$.
In consequence we have $a_i = y \vee z$ for $y,z \not= a_i$, i.e $a_i > y$ and $a_i > z$. 
Assume $y = i^{\sharp}(d(y))$ and $z = i^{\sharp}(d(z))$. 
Then 
\[ i^{\sharp}(d(a_i)) \leq a_i = y \vee z = i^{\sharp}(d(y)) \vee i^{\sharp}(d(z)) = i^{\sharp}(d(y) \vee d(z)) \leq i^{\sharp}(d(a_i)). \]
It thus follows $a_i = i^{\sharp}(d(a_i))$, which is a contradiction. Hence, w.l.o.g. assume $y \not = i^{\sharp}(d(y))$ and define $a_{i+1} := y$.
	
	Let $( Y, i', d' )$ be an arbitrary generator for $L$. For any $a \in J(L)$ we have $a = \vee^h_{y \in d'(a)} i'(y)$. By the definition of join-irreducibles there exists at least one $y_a \in d'(a)$ such that $i'(y_a) = a$. One can thus define a function $f: J(L) \rightarrow Y$ with $f(a) = y_a$. Assume $f(a) = f(b)$, i.e. $y_a = y_b$. Then, by definition, \[ a = i'(y_a) = i'(y_b) = b \] which shows that $f$ is injective. It thus follows  $\vert J(L) \vert \leq \vert Y \vert$.
\end{proof}

\begin{lemma}
\label{xorbasisstateminimal}
	Any finite basis for an $\mathcal{R}$-algebra is a size-minimal generator.
\end{lemma}
\begin{proof}
Let $( Y, i, d )$ be a basis for an $\mathcal{R}$-algebra $(X, h)$ with $Y$ being finite. It immediately follows that $\mathcal{R}Y = 2^Y$. Let $( Y', i', d' )$ be any other generator for $( X, h )$. We can assume that $Y'$ is finite, as otherwise $Y$ would be immediately of smaller size. It thus follows that $\mathcal{R}Y'= 2^{Y'}$. By definition, $d: X \rightarrow \mathcal{R}Y$ is a bijection with inverse $h \circ \mathcal{R}i$, and $h \circ \mathcal{R}i': \mathcal{R}Y' \rightarrow X$ is a surjection. We thus know that $d \circ h \circ \mathcal{R}i': \mathcal{R}Y' \rightarrow \mathcal{R}Y$ is a surjection. In consequence we find that $\vert \mathcal{R}Y \vert \leq \vert \mathcal{R} Y' \vert$. Assume $\vert Y' \vert < \vert Y \vert$, then we can deduce \[ \vert \mathcal{R}Y' \vert = \vert 2 ^{Y'} \vert = \vert 2 \vert ^{\vert Y' \vert} < \vert 2 \vert ^ { \vert Y \vert} = \vert 2 ^Y \vert = \vert \mathcal{R}Y \vert, \] which contradicts $\vert \mathcal{R}Y \vert \leq \vert \mathcal{R} Y' \vert$. We thus can conclude that $\vert Y' \vert \geq \vert Y \vert$.
\end{proof}

It is enough to find generators for the underlying algebra of a bialgebra to derive an equivalent free bialgebra.
This is because the algebraic and coalgebraic components are tightly intertwined via a distributive law.

\begin{proposition}
\label{forgenerator-isharp-is-bialgebra-hom}
	Let $( X, h, k)$ be a $\lambda$-bialgebra and let $( Y, i, d )$ be a generator for the $T$-algebra $( X,h )$.
	Then $h \circ Ti \colon \expa_T(( Y, Fd \circ k \circ i) ) \rightarrow ( X, h, k )$ is a $\lambda$-bialgebra homomorphism.
\end{proposition}
\begin{proof}
By definition we have the equality \[ \expa_T(( Y, Fd \circ k \circ i )) = ( TY, \mu_Y, F \mu_Y \circ \lambda_{TY} \circ T(Fd \circ k \circ i) ). \] It is well-known that $h \circ Ti$ is a  homomorphism between the underlying $T$-algebra structures. It thus remains to show that it is an $F$-coalgebra homomorphism. The latter follows from the commutativity of the diagram below:
\begin{equation*}
	\begin{tikzcd}[row sep = 1.5em, column sep = 5em]
			TY \arrow{r}{Ti} \arrow{d}[left]{Ti} & TX \arrow{dd}{Tk} \arrow{rr}{h} & & X \arrow{ddddd}{k} \\
		TX \arrow{d}[left]{Tk} \\
		TFX \arrow{d}[left]{TFd} \arrow{r}{\id_{TFX}} & TFX \arrow{ddr}{\lambda_X} \\
		TFTY \arrow{d}[left]{\lambda_{TY}} \arrow{r}{TFTi} & TFTX \arrow{u}{TFh} \\
		FT^2Y \arrow{r}{FT^2i} \arrow{d}[left]{F\mu_Y} & FT^2X \arrow{r}{FTh} \arrow{d}{F \mu_X} & FTX \arrow{d}{Fh} \\
		FTY \arrow{r}[below]{FTi} & FTX \arrow{r}[below]{Fh} & FX \arrow{r}[below]{\id_{FX}} & FX
	\end{tikzcd}.
\end{equation*}
\end{proof}

Intuitively, the bialgebra $( X, h, k )$ is a deterministic automaton with additional algebraic structure in the monad $T$ and say initial state $x \in X$, while the equivalent free bialgebra is the determinisation of the succinct automaton $Fd \circ k \circ i \colon Y \rightarrow FTY$ with side-effects in $T$ and initial state $d(x) \in TY$. 

The statement below establishes that, for bases, the morphism $d$ is an algebra homomorphism, and, intuitively, that elements of a basis are uniquely generated by their image under the monad unit, that is, typically by themselves. We will use this technical result, among others, in \Cref{forbasis-isharp-is-bialgebra-iso}.

\begin{lemma}
\label{forbasis-d-isalgebrahom}
	Let $(Y, i, d)$ be a basis for a $T$-algebra $(X, h)$. Then $\mu_Y \circ Td = d \circ h$ and $d \circ i = \eta_Y$.
\end{lemma}
\begin{proof}
The statement follows from the commutativity of the following two diagrams:
	\begin{equation*}
			\begin{tikzcd}
				TX \arrow{rr}{Td} \arrow{d}[left]{\id_{TX}} & & T^2Y \arrow{dl}{T^2i} \arrow{dd}{\mu_Y} \\
				TX \arrow{dd}[left]{h} & T^2X \arrow{l}{Th} \arrow{d}{\mu_X} \\
				& TX \arrow{dl}{h} & TY \arrow{l}{Ti} \arrow{d}{\id_{TY}} \\
				X \arrow{rr}[below]{d} & & TY
			\end{tikzcd}
			\qquad
		\begin{tikzcd}
			Y \arrow{r}{i} \arrow{d}[left]{\eta_Y} & X \arrow{r}{\id_X} \arrow{d}[left]{\eta_X} & X \arrow[bend left]{dddll}{d} \\
			TY \arrow{dd}[left]{\id_{TY}} \arrow{r}{Ti} &TX \arrow{ur}{h} \\ 
			\\
			TY
		\end{tikzcd}.
	\end{equation*}
Alternatively, the first equality can be deduced from Beck's monadicity theorem: since the forgetful functor $U: \EM \rightarrow \mathscr{C}$ is monadic, it reflects isomorphisms.
\end{proof}

  \begin{lemma}
  \label{basisifffree}
  	A $T$-algebra admits a basis iff it is isomorphic to a free $T$-algebra.
  \end{lemma}
  \begin{proof}
  	Assume $(X,h)$ admits a basis $(Y, i, d)$. Then, by definition, $d: X \rightarrow TY$ and $h \circ Ti: TY \rightarrow X$ are inverse to each other, thus witnessing an isomorphism of $X$ and $TY$ in the base category $\mathscr{C}$. It remains to show that the isomorphism lifts to $\EM$. It immediately follows that $h \circ Ti$ lifts to a $T$-algebra homomorphism $h \circ Ti: (TY, \mu_Y) \rightarrow (X,h)$. The morphism $d$ lifts to a $T$-algebra homomorphism $d: (X,h) \rightarrow (TY, \mu_Y)$ by \Cref{forbasis-d-isalgebrahom}.
  	
  	Conversely, assume $f: (X,h) \rightarrow (TY, \mu_Y)$ is an isomorphism in $\EM$. Let $i := f^{-1} \circ \eta_Y: Y \rightarrow X$ and $d := f: X \rightarrow TY$. Then $(Y, i, d)$ is a basis for $(X,h)$, since
  	\begin{align*}
  		d \circ (h \circ Ti) &= f \circ h \circ T(f^{-1} \circ \eta_Y) && \textnormal{(Definitions of } i, d) \\
  		&= f \circ h \circ T(f^{-1}) \circ T(\eta_Y) && \textnormal{(Functoriality of } T) \\
  		&= f \circ f^{-1} \circ \mu_Y \circ T(\eta_Y) && (f^{-1} \textnormal{ is homomorphism)}\\
  		&= f \circ f^{-1} && (\mu_Y \circ T(\eta_Y) = \id_{TY}) \\
  		&= \id_{TY} && (f \textnormal{ is isomorphism)}
  	\end{align*}
  	and similarly
  	\begin{align*}
  	(h \circ Ti) \circ d &= h \circ T(f^{-1} \circ \eta_Y) \circ f	&& \textnormal{(Definitions of } i, d) \\
  	&= h \circ T(f^{-1}) \circ T(\eta_Y) \circ f && \textnormal{(Functoriality of } T) \\
  	&= f^{-1} \circ \mu_Y \circ T(\eta_Y) \circ f && (f^{-1} \textnormal{ is homomorphism)} \\
  	&= f^{-1} \circ f && (\mu_Y \circ T(\eta_Y) = \id_{TY}) \\
  	&= \id_X && (f \textnormal{ is isomorphism)}
  	\end{align*}
  	  \end{proof}

The following result observes that if the generator for the algebra of a bialgebra in \Cref{forgenerator-isharp-is-bialgebra-hom} is in fact a \emph{basis}, then the equivalence  involved is an isomorphism.

\begin{proposition}
\label{forbasis-isharp-is-bialgebra-iso}
	Let $( X, h, k)$ be a $\lambda$-bialgebra and let $( Y, i, d )$ be a basis for the $T$-algebra $( X,h )$.
	Then $h \circ Ti \colon \expa_T(( Y, Fd \circ k \circ i) ) \rightarrow ( X, h, k )$ is a $\lambda$-bialgebra isomorphism.
\end{proposition}
\begin{proof}
From \Cref{forgenerator-isharp-is-bialgebra-hom} we know that $h\circ Ti$ is a $\lambda$-bialgebra homomorphism. By the definition of a basis, $d$ is a two-sided inverse to $h\circ Ti$ as ordinary morphism. It thus remains to show that $d$ is a $\lambda$-bialgebra homomorphism. By \Cref{forbasis-d-isalgebrahom} it is a $T$-algebra homomorphism, and the diagram below shows that it commutes with $F$-coalgebra structures: 
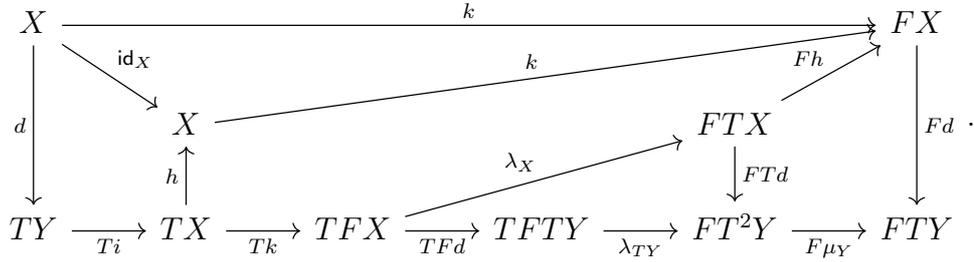
\begin{equation*}
		\begin{tikzcd}[]
			X \arrow{rrrrr}{k} \arrow{dd}[left]{d} \arrow{rd}{\id_X} & & & & & FX \arrow{dd}[right]{Fd} \\
			& X \arrow{rrrru}{k} & & & FTX \arrow{d}{FTd} \arrow{ur}{Fh} & \\
			TY \arrow{r}[below]{Ti} & TX \arrow{u}{h} \arrow{r}[below]{Tk} & TFX \arrow{r}[below]{TFd} \arrow{rru}{\lambda_X} & TFTY \arrow{r}[below]{\lambda_{TY}} & FT^2Y \arrow{r}[below]{F\mu_Y} & FTY
		\end{tikzcd}.
\end{equation*}
\end{proof}

We conclude this section by illustrating how \Cref{forgenerator-isharp-is-bialgebra-hom} can be used to construct the canonical RFSA \cite{denis2001residual}, the canonical nominal RFSA \cite{moerman2019residual}, and the minimal xor automaton \cite{VuilleminG210} for a regular language $L$ over some alphabet $A$. All examples follow three analogous steps:
\begin{enumerate}
	\item We construct the minimal\footnote{Minimal in the sense that every state is reachable by an element of $A^*$ and no two different states observe the same language.} pointed coalgebra $\M_{L}$ for the (nominal) set endofunctor $F = 2 \times (-)^{A}$ accepting $L$. For the case $A = \lbrace a, b \rbrace$ and $L = (a+b)^*a$, the coalgebra $\M_{L}$ is depicted in \Cref{m(l)}.
	\item We equip the former with additional algebraic structure in a monad $T$ (which is related to $F$ via a canonically induced distributive law $\lambda$) by generating the $\lambda$-bialgebra $\free_T(\M_{L})$. By identifying semantically equivalent states we consequently derive the minimal\footnote{Minimal in the sense that every state is reachable by an element of $T(A^*)$ and no two different states observe the same language.} (pointed) $\lambda$-bialgebra $( X, h, k )$  for $L$.
	\item We identify canonical generators $( Y, i, d )$ for $( X, h )$ and use \Cref{forgenerator-isharp-is-bialgebra-hom} to derive an equivalent succinct automaton $( Y, Fd \circ k \circ i )$ with side-effects in $T$.
\end{enumerate} 

\begin{figure}[t]
	\centering
\begin{subfigure}[c]{\columnwidth}
		\centering	
		\scriptsize
		\begin{subfigure}[b]{.68 \columnwidth}
		\centering
		\centering
							\adjustbox{valign=m}{
			\begin{tikzpicture}[node distance=6em]
				\node[state] (0) {$\lbrack \emptyset \rbrack$};
				\node[state, right of=0, initial, initial text=] (x) {$\lbrack \lbrace x \rbrace \rbrack$};
					\node[state, right of=x, accepting] (y) {$\lbrack \lbrace y \rbrace \rbrack$};	
			    \path[->]
			(0) edge[loop above] node{$a,b$} (0)
			(x) edge[loop above] node{$b$} (x)
			(x) edge[above, bend left] node{$a$} (y)
			(y) edge[below, bend left] node{$b$} (x)
			(y) edge[loop right] node{$a$} (y)
			;
			\end{tikzpicture}
			}
					\adjustbox{valign=m}{
			\resizebox{0.4 \columnwidth}{!}{%
				\begin{tabular}{ c|c|c|c } 
 $\vee$ & $\lbrack \lbrace x \rbrace \rbrack$ & $ \lbrack \lbrace y \rbrace \rbrack$ &  $\lbrack \emptyset \rbrack$ \\
 \hline 
$\lbrack \lbrace x \rbrace \rbrack$ & $\lbrack \lbrace x \rbrace \rbrack$ & $\lbrack \lbrace y \rbrace \rbrack$ & $\lbrack \lbrace x \rbrace \rbrack$ \\ 
 \hline
 $\lbrack \lbrace y \rbrace \rbrack$ & $\lbrack \lbrace y \rbrace \rbrack$ & $\lbrack \lbrace y \rbrace \rbrack$ & $\lbrack \lbrace y \rbrace \rbrack$ \\
 \hline
 $\lbrack \emptyset \rbrack$ & $\lbrack \lbrace x \rbrace \rbrack$ & $\lbrack \lbrace y \rbrace \rbrack$ & $\lbrack \emptyset \rbrack$
\end{tabular}
}}
					\caption{}
		\label{overlineml}
		\end{subfigure}
	\begin{subfigure}[b]{.25 \columnwidth}
		\scriptsize
		\centering
							\adjustbox{valign=m}{
		\begin{tikzpicture}[node distance=6em]
			\node[state, initial, initial text=] (x) {$\lbrack \lbrace x \rbrace \rbrack$};
				\node[state, right of=x, accepting] (y) {$\lbrack \lbrace y \rbrace \rbrack$};	
		    \path[->]
		(x) edge[loop above] node{$a,b$} (x)
		(x) edge[above, bend left] node{$a$} (y)
		(y) edge[below, bend left] node{$a,b$} (x)
		(y) edge[loop right] node{$a$} (y)
		;
		\end{tikzpicture}
		}
		\caption{}
	\label{jiromaton}
	\end{subfigure}
	\end{subfigure}	
	\caption{(a) The minimal CSL-structured DFA for $L = (a+b)^*a$; (b) The canonical RFSA for $L = (a+b)^*a$}
	\label{crfsadiag1}
\end{figure}
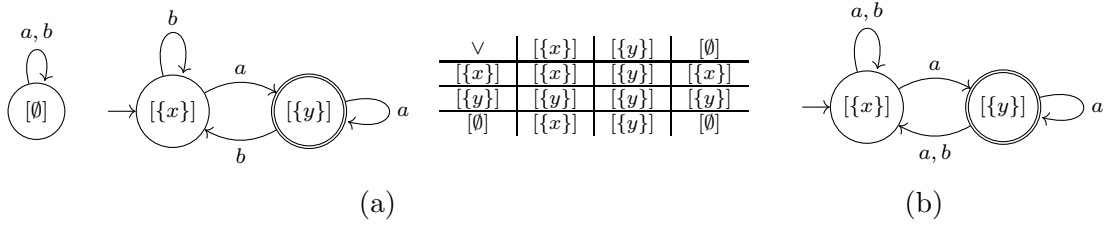

\begin{example}[The Canonical RFSA]
\label{canonicalrfsaexample}
    Using the $\mathcal{P}$-algebra structure $h^{\mathcal{P}}: \mathcal{P}2 \rightarrow 2$ with $h^{\mathcal{P}}(\varphi) = \varphi(1)$, we derive a canonical distributive law $\lambda^{\mathcal{P}}$ between $F$ and the powerset monad $\mathcal{P}$. The minimal pointed $\lambda^{\mathcal{P}}$-bialgebra for $L = (a+b)^*a$ with its underlying CSL structure is depicted in \Cref{overlineml}. The partially ordered state space $L = \lbrace \lbrack \emptyset \rbrack \leq \lbrack \lbrace x \rbrace \rbrack \leq \lbrack \lbrace y \rbrace \rbrack \rbrace$ is necessarily finite, thus satisfies the descending chain condition, which turns the set of join-irreducibles into a size-minimal generator $( J(L), i, d)$ with $i(y) = y$ and $d(x) = \lbrace y \in J(L) \mid y \leq x \rbrace$, cf. \Cref{joinirreducstateminimal}. In this case, the join-irreducibles are given by all non-zero states. The $\mathcal{P}$-succinct automaton consequently induced by \Cref{forgenerator-isharp-is-bialgebra-hom} is depicted in \Cref{jiromaton}; it can be recognised as the canonical RFSA, cf. e.g. \cite{MyersAMU15}.
 \end{example}
 
 The soundness of the construction in \Cref{overlineml} can be verified with the following result, which translates the abstract definition of a free $\lambda^{\mathcal{P}}$-bialgebra to concrete data.
  \begin{lemma}
\label[lemma]{freepowersetbialgebrastructure}
		Let $( \mathcal{P}X, \mu^{\mathcal{P}}_X, \langle \overline{\varepsilon}, \overline{\delta} \rangle ) := \free^{\lambda^{\mathcal{P}}}(( X, \langle \varepsilon, \delta \rangle ))$. Then it holds $\overline{\varepsilon}(\varphi) = \vee_{y \in \varepsilon^{-1}(1)} \varphi(y)$ and $\overline{\delta}_a(\varphi)(x) = \vee_{y \in \delta_a^{-1}(x)} \varphi(y)$.
\end{lemma}
\begin{proof}
\begingroup
\allowdisplaybreaks
	The first equality is a consequence of 
	\begin{align*}
		\overline{\varepsilon}(\varphi) 
		&= \pi_1 \circ (h^{\mathcal{P}} \times \st) \circ ( \mathcal{P}\pi_1, \mathcal{P}\pi_2 ) \circ \mathcal{P}(\langle \varepsilon, \delta \rangle)(\varphi) && \textnormal{(Definition of } \overline{\varepsilon}) \\
		&= h^{\mathcal{P}} \circ \mathcal{P}(\varepsilon)(\varphi) && \textnormal{(Definition of } \pi_1) \\
		&= \mathcal{P}(\varepsilon)(\varphi)(1) && \textnormal{(Definition of } h^{\mathcal{P}}) \\
		&= \vee_{y \in \varepsilon^{-1}(1)} \varphi(y) && \textnormal{(Definition of } \mathcal{P}(\varepsilon))
	\end{align*}
	For the second equality we observe
	\begin{align*}
		\overline{\delta}_a(\varphi)(x) &= 
		\overline{\delta}(\varphi)(a)(x) && \textnormal{(Definition of } \overline{\delta}_a) \\
		&= \pi_2 \circ (h^{\mathcal{P}} \times \st ) \circ \langle \mathcal{P}\pi_1, \mathcal{P}\pi_2 \rangle \circ \mathcal{P}(\langle \varepsilon, \delta \rangle )(\varphi)(a)(x) && \textnormal{(Definition of } \overline{\delta}) \\
		&= \st \circ \mathcal{P}(\delta)(\varphi)(a)(x) && \textnormal{(Definition of } \pi_2) \\
		&= \mathcal{P}(\ev_a)(\mathcal{P}(\delta)(\varphi))(x) && \textnormal{(Definition of } \st) \\
		&= \mathcal{P}(\delta_a)(\varphi)(x) && (\textnormal{Definition of } \delta_a) \\
		&=  \vee_{y \in \delta_a^{-1}(x)} \varphi(y) && \textnormal{(Definition of } \mathcal{P}(\delta_a)) 
	\end{align*}
	\endgroup
\end{proof}
 
 To cover the canonical \emph{nominal} RFSA we need the following result which shows that the canonical strength\footref{strengthdef} function for the powerset monad $\mathcal{P}$ on the category $\Set$ naturally lifts to the nominal powerset monad $\mathcal{P}_{\textnormal{n}}$ (cf. \Cref{exampleofmonads}) on the category $\Nom{\mathbb{A}}$ of finitely supported nominal $\mathbb{A}$-sets and equivariant functions.
 
 \begin{lemma}
 \label{equivariantstrength}
 	The strength function $\st: \mathcal{P}_{\textnormal{n}}(X^A) \rightarrow (\mathcal{P}_{\textnormal{n}}X)^A$ satisfying $ \st(\Phi)(a) = \mathcal{P}_{\textnormal{n}}(\ev_a)(\Phi)$ is equivariant.
 \end{lemma}
 \begin{proof}
  	\begingroup
 	\allowdisplaybreaks
 As before, let $\Perm(\mathbb{A})$ be the set of permutations of $\mathbb{A}$, that is, bijective functions $\pi: \mathbb{A} \rightarrow \mathbb{A}$. We first observe that for any $a \in \mathbb{A}, x \in X$ and $\pi \in \Perm(\mathbb{A})$ the mapping
 \begin{equation}
 	\label{eq:strengthequivariant}
 \lbrace \varphi \in X^{\mathbb{A}} \mid \varphi(a) = x \rbrace \rightarrow \lbrace \varphi \in X^{\mathbb{A}} \mid \varphi(\pi^{-1}.a) = \pi^{-1}.x \rbrace \qquad \pi \mapsto \pi^{-1}.\varphi
 \end{equation}
 defines a bijection with inverse assignment $\varphi \mapsto \pi.\varphi$. Note that the set $2$ is equipped with the trivial action.
 	The statement thus follows from

 	\begin{align*}
 		(\pi.\st(\Phi))(a)(x) 
 		&= \pi.(\st(\Phi)(\pi^{-1}.a))(x) && \textnormal{(Definition of } \pi.\st(\Phi)) \\
 		&= \st(\Phi)(\pi^{-1}.a)(\pi^{-1}.x) && \textnormal{(Definition of } \pi.(\st(\Phi)(\pi^{-1}.a)) ) \\
 		&= \mathcal{P}_{\textnormal{n}}(\ev_{\pi^{-1}.a})(\Phi)(\pi^{-1}.x) && \textnormal{(Definition of } \st) \\
 		&= \vee_{\varphi \in \ev_{\pi^{-1}.a}^{-1}(\pi^{-1}.x)} \Phi(\varphi) && \textnormal{(Definition of } \mathcal{P}_{\textnormal{n}}(\ev_{\pi^{-1}.a})) \\
 		&= \vee_{\varphi \in \ev_{a}^{-1}(x)} \Phi(\pi^{-1}.\varphi) && \textnormal{\eqref{eq:strengthequivariant}} \\
 		&= \vee_{\varphi \in \ev_a^{-1}(x)} (\pi.\Phi)(\varphi) && \textnormal{(Definition of } \pi.\Phi) \\
 		&= \mathcal{P}_{\textnormal{n}}(\ev_a)(\pi.\Phi)(x) && \textnormal{(Definition of } \mathcal{P}_{\textnormal{n}}(\ev_a)) \\
 		&= \st(\pi.\Phi)(a)(x) && \textnormal{(Definition of } \st) 
 		 	\end{align*}
 		 	\endgroup
 \end{proof}

\begin{figure*}
\footnotesize
\centering
	\begin{tikzpicture}[node distance=6em]
			\node[state, initial, initial text=] (L) {$L$};
				\node[state, right of=L] (aL) {$a^{-1}L$};	
\node[state, accepting, right of=aL] (A) {$\mathbb{A}^*$};	
		    \path[->]
		(L) edge[loop above] node{$\mathbb{A}$} (L)
		(L) edge[above, bend left] node{$a$} (aL)
		(aL) edge[above, bend left] node{$a$} (A)
		(aL) edge[loop above] node{$\mathbb{A}$} (aL)
		(aL) edge[above, bend left] node{$\mathbb{A}$} (L)
		(A) edge[loop above] node{$\mathbb{A}$} (A)
		(A) edge[above, bend left=45] node{$\mathbb{A}$} (L)
		(A) edge[above, bend left] node{$\mathbb{A}$} (aL)
		;
		\end{tikzpicture}
		\caption{The orbit-finite representation of the canonical nominal RFSA for $L =  \lbrace v a w a u \mid v, w, u \in \mathbb{A}^*, a \in \mathbb{A} \rbrace$}
				\label{fig:canonialnominalrfsa}
\end{figure*}
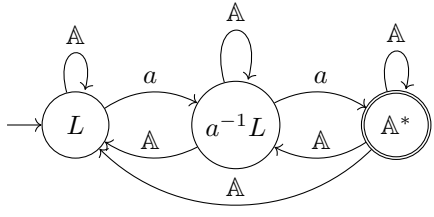

\begin{example}[The Canonical Nominal RFSA]
\label{nominalexample}
    Let $\mathbb{A}$ be a countably infinite set. In \Cref{equivariantstrength} we have established that the canonical strength function in $\Set$ lifts to $\Nom{\mathbb{A}}$. It is also not hard to see that the functor $F$ with $FX = 2 \times X^{\mathbb{A}}$ extends to a functor on $\Nom{\mathbb{A}}$, and $h^{\mathcal{P}_{\textnormal{n}}}: \mathcal{P}_{\textnormal{n}}2 \rightarrow 2$ with $h^{\mathcal{P}_{\textnormal{n}}}(\varphi) = \varphi(1)$ defines a $\mathcal{P}_{\textnormal{n}}$-algebra, which induces a canonical distributive law $\lambda^{\mathcal{P}_{\textnormal{n}}}$ between $F$ and the nominal powerset monad $\mathcal{P}_{\textnormal{n}}$ (cf. \Cref{exampleofmonads}). As in \cite{moerman2019residual}, let $L = \lbrace v a w a u \mid v, w, u \in \mathbb{A}^*, a \in \mathbb{A} \rbrace$, then $a^{-n}L = a^{-2}L = \mathbb{A}^*$ for $n \geq 2$, and $v^{-1}L = \cup_{a \in \mathbb{A}} a^{-\vert v \vert_a} L$, where $\vert v \vert_a$ denotes the number of $a$'s that occur in $v$. In consequence, the nominal CSL underlying the minimal pointed $\lambda^{\mathcal{P}_{\textnormal{n}}}$-bialgebra is generated by the orbit-finite nominal set of join-irreducibles $\lbrace L \rbrace \cup \lbrace a^{-1} L \mid a \in \mathbb{A} \rbrace \cup \lbrace \mathbb{A}^* \rbrace$, which is equipped with the obvious $\Perm(\mathbb{A})$-action and satisfies the inclusion $L \subseteq a^{-1} L \subseteq \mathbb{A}^*$. The orbit-finite representation of the $\mathcal{P}_{\textnormal{n}}$-succinct automaton induced by \Cref{forgenerator-isharp-is-bialgebra-hom} is depicted in \Cref{fig:canonialnominalrfsa}. 
\end{example}

\begin{figure}[t]
\centering
\begin{subfigure}[b]{.68 \columnwidth}
\scriptsize
					\adjustbox{valign=m}{
\begin{tikzpicture}[node distance=6em]
	\node[state] (0) {$\emptyset$};
		\node[state, right of=0, accepting] (xory) {$\lbrace x,y \rbrace$};
	\node[state, below of=0, initial, initial text=] (x) {$\lbrace x \rbrace$};
	\node[state, right of=x, accepting] (y) {$\lbrace y \rbrace$};	
	    \path[->]
	(0) edge[loop above] node{$a,b$} (0)
	(xory) edge[above] node{$a,b$} (0)
	(x) edge[loop above] node{$b$} (x)
	(x) edge[above, bend left] node{$a$} (y)
	(y) edge[below, bend left] node{$b$} (x)
	(y) edge[loop right] node{$a$} (y)
	;
	\end{tikzpicture}
	}
						\adjustbox{valign=m}{
				\resizebox{0.53 \columnwidth}{!}{%
		\begin{tabular}[]{ c|c|c|c|c } 
 $\oplus$ & $ \lbrace x \rbrace $ & $ \lbrace y \rbrace$ & $\lbrace x, y \rbrace$ &   $\emptyset$  \\
 \hline
 $ \lbrace x \rbrace $ & $\emptyset$ & $\lbrace x, y \rbrace$ & $\lbrace y \rbrace$ & $\lbrace x \rbrace$ \\
 \hline
 $ \lbrace y \rbrace$ & $\lbrace x, y \rbrace$ & $\emptyset$ & $\lbrace x \rbrace$ & $\lbrace y \rbrace$\\
\hline
 $\lbrace x, y \rbrace$ & $\lbrace y \rbrace$ & $\lbrace x \rbrace$ & $\emptyset$ & $\lbrace x, y \rbrace$\\
\hline
 $\emptyset$ & $ \lbrace x \rbrace $ & $ \lbrace y \rbrace$ & $\lbrace x, y \rbrace$ &   $\emptyset$
\end{tabular}
}}
	\caption{}
\label{fm(l)xor}
\end{subfigure}
\begin{subfigure}[b]{.23 \columnwidth}
\centering
\scriptsize
					\adjustbox{valign=m}{
\begin{tikzpicture}[node distance=6em]
		\node[state, initial, initial text=] (x) {$\lbrace x \rbrace$};
		\node[state, right of=x, accepting] (xory) {$\lbrace x,y \rbrace$};
	    \path[->]
	(x) edge[loop above] node{$a,b$} (x)
	(x) edge[above] node{$a$} (xory)
	;
	\end{tikzpicture}
	}
	\caption{}
\label{xorautomaton}
\end{subfigure}
\caption{(a) The minimal $\mathbb{Z}_2$-vector space structured DFA for $L = (a+b)^*a$ (freely-generated by the DFA in \Cref{m(l)}); (b) Up to the choice of a basis, the minimal xor automaton for $L = (a+b)^*a$}
\end{figure}
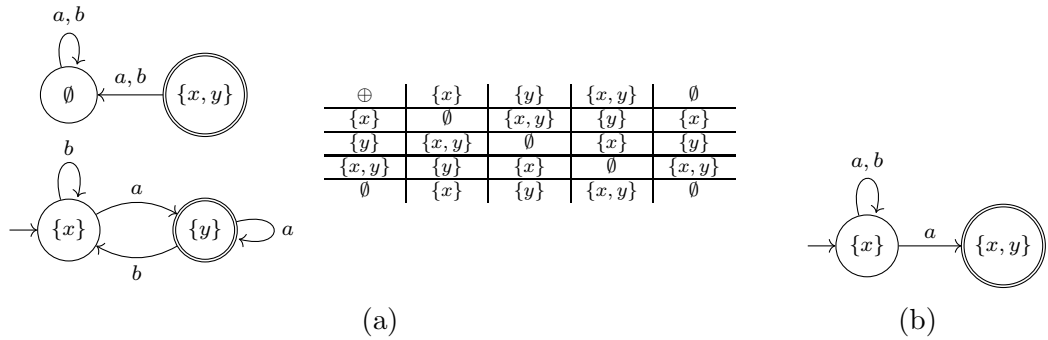

\begin{example}[The Minimal Xor Automaton]
\label{minimalxorexmaple}
    The $\mathcal{R}$-algebra structure $h^{\mathcal{R}}: \mathcal{R}2 \rightarrow 2$ with $h^{\mathcal{R}}(\varphi) = \varphi(1)$ induces a canonical distributive law $\lambda^{\mathcal{R}}$ between $F$ and the free vector space monad $\mathcal{R}$ over the two element field. The minimal pointed $\lambda^{\mathcal{R}}$-bialgebra accepting $L = (a+b)^*a$ is depicted in \Cref{fm(l)xor} and coincides with the bialgebra freely generated by the $F$-coalgebra in \Cref{m(l)}. 
    The underlying vector space structure necessarily has a basis. We choose $( Y, i, d )$ with $Y = \lbrace \lbrace x \rbrace, \lbrace x, y \rbrace \rbrace$, $i(y) = y$, and $d(\emptyset) = \emptyset$, $d(\lbrace x \rbrace) = \lbrace \lbrace x \rbrace \rbrace$, $d(\lbrace y \rbrace) = \lbrace \lbrace x \rbrace$, $\lbrace x, y \rbrace \rbrace$, $d(\lbrace x, y \rbrace) = \lbrace \lbrace x, y \rbrace \rbrace$. The $\mathcal{R}$-succinct automaton induced by \Cref{forgenerator-isharp-is-bialgebra-hom} is depicted in \Cref{xorautomaton}; it can be recognised as the minimal xor automaton, cf. e.g. \cite{MyersAMU15}.  
\end{example}

As before, the soundness of the construction in \Cref{fm(l)xor} can be verified with the help of the following result, which translates the abstract definition of a free $\lambda^{\mathcal{R}}$-bialgebra to concrete data.

 \begin{lemma}
\label[lemma]{freexorbialgebrastructure}
		Let $( \mathcal{R}X, \mu^{\mathcal{R}}_X, \langle \overline{\varepsilon}, \overline{\delta} \rangle ) := \free^{\lambda^{\mathcal{R}}}(( X, \langle \varepsilon, \delta \rangle ))$. Then it holds $\overline{\varepsilon}(\varphi) = \bigoplus_{y \in \varepsilon^{-1}(1)} \varphi(y)$ and $\overline{\delta}_a(\varphi)(x) = \bigoplus_{y \in \delta_a^{-1}(x)} \varphi(y)$.
\end{lemma}
\begin{proof}
	Analogous to the proof of \Cref{freepowersetbialgebrastructure}.
\end{proof}

\section{Changing the Type of Succinct Automata}

\label{succinctbialgebra2}

This section contains a generalisation of the approach in \Cref{succinctbialgebra1}. The extension is based on the observation that in the last section we implicitly considered \textit{two} types of monads: (i) a monad $S$ that describes the additional algebraic structure of a given deterministic automaton; and (ii) a monad $T$ that captures the side-effects of the succinct automaton that is obtained by the generator-based translation.
In \Cref{forgenerator-isharp-is-bialgebra-hom}, the main result of the last section, the monads coincided, but to recover for instance the \'atomaton \cite{BrzozowskiT14} we will have to extend \Cref{forgenerator-isharp-is-bialgebra-hom} to a situation where $S$ and $T$ can differ. 

\subsection{Relating Distributive Laws}

We now recall the main technical ingredient of our extension:  \textit{distributive law homomorphisms}.
As before, we present the theory on the level of arbitrary bialgebras, even though we will later focus on the case where the coalgebraic dynamics are those of deterministic automata. Distributive law homomorphisms will allow us to shift a bialgebra over a monad $S$ to an equivalent bialgebra over a monad $T$, for which we can then find, analogous to \Cref{succinctbialgebra1}, an equivalent succinct representation. The notion we use is an instance of a much more general definition that allows to relate distributive laws on two different categories. We restrict to the case where both distributive laws are given over the same behavioural endofunctor $F$.

\begin{definition}[Distributive Law Homomorphism \cite{watanabe2002well, power2002combining}]
\label{distributivelawhomdef}
Let $\lambda^{S}: SF \rightarrow FS$ and $\lambda^{T}: TF \rightarrow FT$ be distributive laws between monads $S$ and $T$ and an endofunctor $F$, respectively.
	A \emph{distributive law homomorphism} $\alpha: \lambda^{S} \rightarrow \lambda^{T}$ consists of a natural transformation $\alpha: T \Rightarrow S$ such that the following three diagrams commute:
	\[
	\begin{tikzcd}
		T^2 \arrow{rr}{\mu^T} \arrow{d}[left]{T\alpha} & & T \arrow{d}{\alpha} \\
	TS \arrow{r}[below]{\alpha_S} & SS \arrow{r}[below]{\mu^S} & S 
	\end{tikzcd}
	\qquad
	\begin{tikzcd}
		1 \arrow{d}[left]{\eta^T} \arrow{r}{\eta^S} & S \\
		T \arrow{ur}[right]{\alpha} &
	\end{tikzcd}
	\qquad
	\begin{tikzcd}
	TF \arrow{d}[left]{\alpha_F} \arrow{r}{\lambda^T} & FT \arrow{d}{F\alpha}  \\
	SF \arrow{r}[below]{\lambda^S}	& FS 
	\end{tikzcd}.
	\]
\end{definition}

The above definition is such that $\alpha$ induces a functor between the categories of $\lambda^S$- and $\lambda^T$-bialgebras. The statement is well-known \cite{klin2015presenting, bonsangue2013presenting}. Since we were unable to locate a full proof in the literature, we include one by ourselves below.

\begin{lemma}[\cite{klin2015presenting, bonsangue2013presenting}]
\label{inducedbialgebra}
			Let $\alpha: \lambda^S \rightarrow \lambda^T$ be a distributive law homomorphism. Then $\alpha ( X, h, k) := ( X, h \circ \alpha_X, k )$ and  $\alpha(f) := f$ defines a functor $\alpha: \Bialg(\lambda^S) \rightarrow \Bialg(\lambda^T)$.
\end{lemma}
\begin{proof}
 We first show that the construction is well-defined on objects.
The commutativity of the two diagrams below shows that $( X, h \circ \alpha_X )$ is a $T$-algebra:
	\begin{equation*}
		\begin{tikzcd}
			T^2X \arrow{rr}{\mu^T_X} \arrow{d}[left]{T\alpha_X} & & TX \arrow{d}{\alpha_X} \\
			TSX \arrow{d}[left]{Th} \arrow{r}{\alpha_{SX}} & S^2X \arrow{r}{\mu^S_X} \arrow{d}{Sh} & SX  \arrow{d}{h}\\
			TX \arrow{r}[below]{\alpha_X} & SX \arrow{r}[below]{h} & X
		\end{tikzcd}
\qquad
	\begin{tikzcd}
		X \arrow{rr}{1} \arrow{dr}{\eta_X^S} \arrow{dd}[left]{\eta^T_X} & & X \\
		& SX \arrow{ur}[right]{h} & \\
		TX \arrow{ur}[right]{\alpha_X} & & 
	\end{tikzcd}.
	\end{equation*}
	To establish that $( X, h\circ \alpha_X, k )$ is a $\lambda^T$-bialgebra it thus remains to observe the commutativity of the diagram on the left below:
	\begin{equation*}
		\begin{tikzcd}
		TX \arrow{rr}{Tk} \arrow{dd}[left]{\alpha_X} & & TFX \arrow{d}{\lambda^T_X} \arrow{ddl}[left]{\alpha_{FX}} \\
		& & FTX \arrow{d}{F\alpha_X} \\
		SX \arrow{d}[left]{h} \arrow{r}{Sk} & SFX \arrow{r}{\lambda^S_X} & FSX \arrow{d}{Fh} \\
		X \arrow{rr}[below]{k} & & FX	
		\end{tikzcd}
		\qquad
				\begin{tikzcd}
		TX \arrow{d}[left]{\alpha_X} \arrow{r}{Tf} &TY  \arrow{d}{\alpha_Y}	\\
		SX \arrow{d}[left]{h_X} \arrow{r}{Sf} & SY \arrow{d}{h_Y} \\
		X \arrow{r}[below]{f} &Y
		\end{tikzcd}.
	\end{equation*}
	 Well-definedness on morphisms follows from the naturality of $\alpha$, as seen on the right above.
	Compositionality follows immediately from the definition of $\alpha$ on morphisms.
\end{proof}

The next result is a straightforward consequence of \Cref{forgenerator-isharp-is-bialgebra-hom}, and may be strengthened to an isomorphism in case one is given a basis instead of a generator, analogous to \Cref{forbasis-isharp-is-bialgebra-iso}. It can be seen as a road map to the approach we propose in this section.

\begin{corollary}
\label{generatorbialgebrahom}
	Let $\alpha: \lambda^S \rightarrow \lambda^T$ be a homomorphism between distributive laws and $( X,h,k )$ a $\lambda^S$-bialgebra. If $( Y, i, d )$ is a generator for the $T$-algebra $( X, h \circ \alpha_X )$, then $
	(h \circ \alpha_X) \circ Ti: \expa_T((  Y, Fd \circ k \circ i )) \rightarrow ( X, h \circ \alpha_X, k )
	$ is a $\lambda^T$-bialgebra homomorphism.
\end{corollary}
\begin{proof}
	By \Cref{inducedbialgebra} the tuple $( X, h \circ \alpha_X, k )$ constitutes a $\lambda^T$-bialgebra.
	The statement thus follows from \Cref{forgenerator-isharp-is-bialgebra-hom}.
\end{proof}

\subsection{Deriving Distributive Law  Relations}

We now turn to the procedure of deriving a distributive law homomorphism. In practice, coming up with a natural transformation and proving that it lifts to a distributive law homomorphism can be quite cumbersome. 

 Fortunately, for certain cases, there is a way to simplify things significantly. For instance, as the next result shows, if, as in \eqref{induceddistrlaweq}, the involved distributive laws are induced by algebra structures $h^{S}$ and $h^{T}$ for an output set $B$, respectively, then one of the conditions is implied by a less convoluted constraint.

\begin{lemma}
\label{distributivelawaxiomeasier}
	Let $\alpha: T \Rightarrow S$ be a natural transformation satisfying $h^{S} \circ \alpha_B = h^{T}$, then $\lambda^{S} \circ \alpha_F = F \alpha \circ \lambda^{T}$.
\end{lemma}
\begin{proof}
	We need to establish the commutativity of the following diagram:
	\begin{equation*}
	\begin{tikzcd}
		T(X^A \times B) \arrow{r}{\alpha_{X^A \times B}} \arrow{d}[left]{( T \pi_1, T \pi_2 )} & S(X^A \times B) \arrow{d}{( S \pi_1, S \pi_2 )} \\
		T(X^A) \times TB \arrow{r}{\alpha_{X^A} \times \alpha_B} \arrow{d}[left]{\st \times h^{T}} & S(X^A) \times SB \arrow{d}{\st \times h^S}  \\
		(TX)^A \times B \arrow{r}[below]{(\alpha_X)^A \times B} & (SX)^A \times B
	\end{tikzcd}.
	\end{equation*}
The commutativity of the top square is a consequence of the naturality of $\alpha$. Similarly, the commutativity of the bottom square follows from the assumption and the naturality of $\alpha$,
\begin{align*}
	\st \circ \alpha_{X^A}(U)(a) 
	&= S(\ev_a) \circ \alpha_{X^A}(U) & \textnormal{(Definition of } \st)  \\
	&= \alpha_{X} \circ T(\ev_a)(U) & \textnormal{(Naturality of } \alpha) \\
	&= \alpha_X^A \circ \st(U)(a) & \textnormal{(Definition of } \st) 
\end{align*}
\end{proof}

The next result shows that for the neighbourhood monad there exists a family of \textit{canonical} choices of distributive law homomorphisms parametrised by Eilenberg-Moore algebra structures on the output set $B = 2$. While it is well-known that such algebras induce a monad morphism, for instance in the coalgebraic modal logic community \cite{klin2004coalgebraic, schroder2008expressivity, hansen2014strong}, its commutativity with canonical distributive laws has not been observed before. Moreover, we provide a new formalisation in terms of the strength function, which allows the result to be lifted to strong monads and arbitrary output objects on other categories than the one of sets and functions.

\begin{proposition}
\label{algebrainduceddistributivellawhom}
	Any algebra $h: T2 \rightarrow 2$ over a set monad $T$ induces a homomorphism $\alpha^{h}: \lambda^{\mathcal{H}} \rightarrow \lambda^{h}$ between distributive laws by $\alpha^{h}_X := h^{2^X} \circ \st \circ T(\eta^{\mathcal{H}}_X)$.
\end{proposition}
\begin{proof}
It is well-known that the strength operation is natural and satisfies the two equalities
\[
	(\eta_A^T)^B = \st \circ \eta^T_{A^B} \qquad	\st \circ \mu_{A^B} = \mu_{A}^B \circ \st \circ T\st. \]
It is also not hard to see that for functions $f: A \rightarrow B$ and $g: C \rightarrow D$ it holds $f^D \circ A^{g} = B^g \circ f^C$. We write $f^*$ for $A^f$ and $f_*$ for $f^A$, and omit components of natural transformations for readability. 
The naturality of $\alpha^{h^{T}}$ is a consequence of:
\[
			\begin{tikzcd}
				TX \arrow{d}[left]{Tf} \arrow{r}{T\eta^{\mathcal{H}}} & T(2^{2^X}) \arrow{d}{T((f^*)^*)} \arrow{r}{\st} & T(2)^{2^X} \arrow{r}{h_*} \arrow{d}{(f^*)^*}  & 2^{2^X} \arrow{d}{(f^*)^*} \\
				TY \arrow{r}[below]{T\eta^{\mathcal{H}}} & T(2^{2^Y}) \arrow{r}[below]{\st}  & T(2)^{2^Y} \arrow{r}[below]{h_*} & 2^{2^Y}
			\end{tikzcd}.
			\]
Next, note that $2^{\eta^{\mathcal{H}}_{2^X}} \circ \eta^{\mathcal{H}}_{2^{2^X}} = \id_{2^{2^X}}$, as for all $\Phi \in 2^{2^X}, \varphi \in 2^X$ we have:
\begingroup
\allowdisplaybreaks
\begin{align*}
	(2^{\eta^{\mathcal{H}}_{2^X}} \circ \eta^{\mathcal{H}}_{2^{2^X}})(\Phi)(\varphi) &=  (\eta^{\mathcal{H}}_{2^{2^X}}(\Phi) \circ \eta^{\mathcal{H}}_{2^X})(\varphi) && \textnormal{(Definition of } 2^f)  \\
	&= \eta^{\mathcal{H}}_{2^X}(\varphi)(\Phi) && \textnormal{(Definition of } \eta^{\mathcal{H}}_{2^{2^X}})  \\
	&= \Phi(\varphi) && \textnormal{(Definition of } \eta^{\mathcal{H}}_{2^X}) \\
	&= \id_{2^{2^X}}(\Phi)(\varphi) && \textnormal{(Definition of } \id_{2^{2^X}}).
\end{align*}
\endgroup			
	Using the above equality, the equation involving the monad multiplications is seen from the diagram on the left below. The equation involving the monad units is established by the diagram on the right below:
		
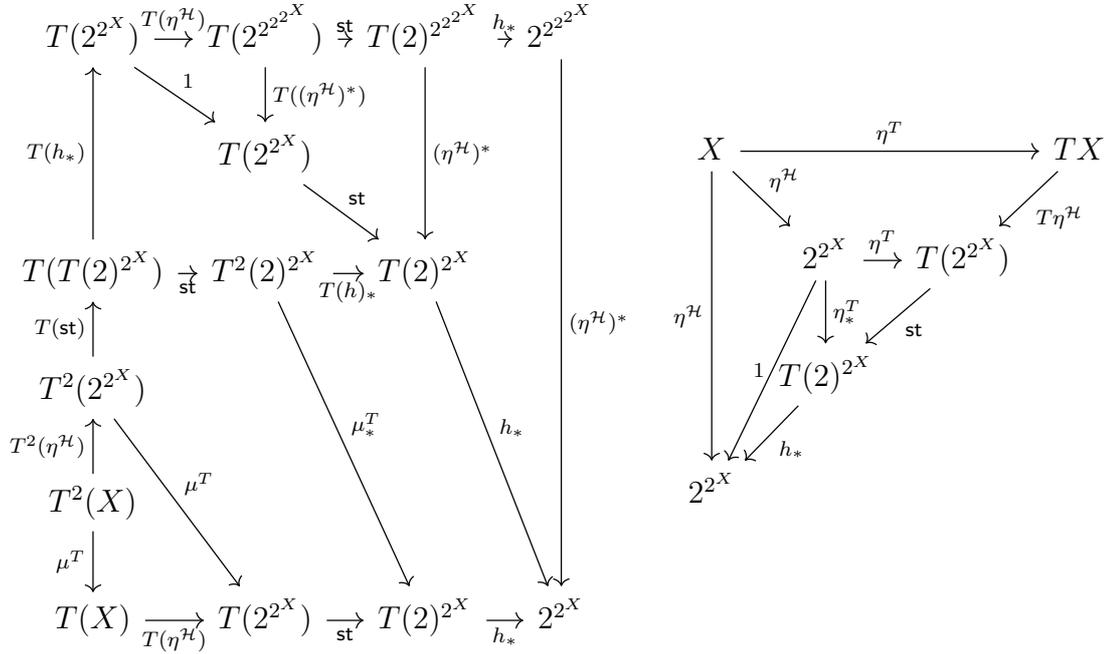
\begin{equation*}
			\begin{tikzcd}[column sep = 0.5em]
				 T(2^{2^X}) \arrow{r}{T(\eta^{\mathcal{H}})}  \arrow{dr}{1} & T(2^{2^{2^{2^X}}}) \arrow{r}{\st} \arrow{d}{T((\eta^{\mathcal{H}})^*)} & T(2)^{2^{2^{2^X}}} \arrow{r}{h_*} \arrow{dd}{(\eta^{\mathcal{H}})^*} & 2^{2^{2^{2^X}}} \arrow{ddddd}[right]{(\eta^{\mathcal{H}})^*} \\
				& T(2^{2^X}) \arrow{dr}{\st} & & \\
				T(T(2)^{2^X}) \arrow{uu}[left]{T(h_{*})} \arrow{r}[below]{\st} & T^2(2)^{2^X} \arrow{dddr}{\mu^T_*} \arrow{r}[below]{T(h)_*} & T(2)^{2^X} \arrow{dddr}{h_*} & \\
				T^2(2^{2^X}) \arrow{ddr}{\mu^T} \arrow{u}[left]{T(\st)} & & & \\
				T^2(X) \arrow{u}[left]{T^2(\eta^{\mathcal{H}})} \arrow{d}[left]{\mu^T} & & & \\
				T(X) \arrow{r}[below]{T(\eta^{\mathcal{H}})} & T(2^{2^X}) \arrow{r}[below]{\st} & T(2)^{2^X} \arrow{r}[below]{h_*} & 2^{2^X} 
			\end{tikzcd}
			\quad
			\begin{tikzcd}[column sep = 0.5em]
				X \arrow{rrr}{\eta^T} \arrow{ddd}[left]{\eta^{\mathcal{H}}} \arrow{dr}{{\eta^{\mathcal{H}}}} & & & TX \arrow{dl}{T\eta^{\mathcal{H}}} \\
				& 2^{2^X} \arrow{d}{\eta^{T}_*} \arrow{r}{\eta^T} \arrow{ddl}[left]{1} & T(2^{2^X}) \arrow{dl}{\st} & \\
				& T(2)^{2^X} \arrow{dl}{h_*} & & \\
				2^{2^{X}} & & & 
			\end{tikzcd}.
		\end{equation*}
	To show that the equation involving the distributive laws holds, we use \Cref{distributivelawaxiomeasier}.  That is, we note that for any $f$ it holds $f \circ \ev_a = \ev_a \circ f_*$ and  $h^{\mathcal{H}} = \ev_{\id_2}$, before establishing the following commutative diagram:
	
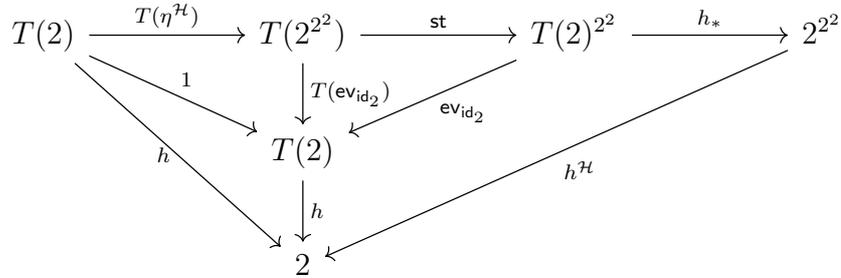
\begin{equation*}
					\begin{tikzcd}[row sep=2em, column sep = 5em]
				T(2) \arrow{ddr}[left]{h} \arrow{r}{T(\eta^{\mathcal{H}})} \arrow{dr}{1} & T(2^{2^{2}}) \arrow{d}{T(\ev_{\id_2})} \arrow{r}{\st} & T(2)^{2^{2}} \arrow{r}{h_*} \arrow{dl}{\ev_{\id_2}} & 2^{2^2} \arrow{ddll}{h^{\mathcal{H}}} \\
				& T(2) \arrow{d}{h} & & \\
				& 2 & & 		
					\end{tikzcd}.
	\end{equation*}
\end{proof}

The rest of the section is concerned with using \Cref{algebrainduceddistributivellawhom} and \Cref{generatorbialgebrahom} to derive canonical acceptors through distributive law homomorphisms.

\subsection{Example: The \'Atomaton}

\label{atomatonexample}
We will now justify the previous informal construction of the \'atomaton. 
As hinted before, the \'atomaton can be recovered by relating the neighbourhood monad $\mathcal{H}$---whose algebras are complete \emph{atomic} Boolean algebras (CABAs)---and the powerset monad $\mathcal{P}$. Formally, as a consequence of \Cref{algebrainduceddistributivellawhom} we obtain the following. 

\begin{corollary}
\label{alphapowersetneighbourhooddistrlaw}
	Let $\alpha_X: \mathcal{P}X \rightarrow \mathcal{H}X$ satisfy $\alpha_X(\varphi)(\psi) = \vee_{x \in X} \varphi(x) \wedge \psi(x)$, then $\alpha$ constitutes a distributive law homomorphism $\alpha: \lambda^{\mathcal{H}} \rightarrow \lambda^{\mathcal{P}}$.
\end{corollary}
\begin{proof}
	We show that $\alpha^{h^{\mathcal{P}}} = \alpha$, the statement then follows from \Cref{algebrainduceddistributivellawhom}:
	\begin{align*}
		\alpha^{h^{\mathcal{P}}}_X(\varphi)(\psi) &= (h^{\mathcal{P}})^{2^X} \circ \st  \circ \mathcal{P}(\eta^{\mathcal{H}}_X)(\varphi)(\psi) && \textnormal{(Definition of } \alpha^{h^{\mathcal{P}}}_X)  \\
		&= \st \circ \mathcal{P}(\eta^{\mathcal{H}}_X)(\varphi)(\psi)(1) && \textnormal{(Definition of } h^{\mathcal{P}}) \\
		&= \mathcal{P}(\ev_{\psi})(\mathcal{P}(\eta^{\mathcal{H}}_X)(\varphi))(1) && \textnormal{(Definition of $\st$)} \\
		&= \mathcal{P}(\ev_{\psi} \circ \eta^{\mathcal{H}}_X)(\varphi)(1) && (\mathcal{P}(f) \circ \mathcal{P}(g) = \mathcal{P}(f \circ g)) \\
		&= \mathcal{P}(\psi)(\varphi)(1) && (\textnormal{Definition of $\ev$, }\eta^{\mathcal{H}}_X) \\
		&= \vee_{x \in \psi^{-1}(1)} \varphi(x) && (\textnormal{Definition of } \mathcal{P}(\psi)) \\
		&= \vee_{x \in X} \varphi(x) \wedge \psi(x) && (x \in \psi^{-1}(1))  \\
		&= \alpha_X(\varphi)(\psi) && \textnormal{(Definition of } \alpha_X)
	\end{align*}
\end{proof}

The next statement follows from a well-known Stone-type duality representation theorem for CABAs \cite{taylor2002subspaces}.

\begin{lemma}
\label{basisshiftecabapowerset}
	Let $\alpha_X: \mathcal{P}X \rightarrow \mathcal{H}X$ satisfy $\alpha_X(\varphi)(\psi) = \vee_{x \in X} \varphi(x) \wedge \psi(x)$. If $B = ( X, h )$ is an $\mathcal{H}$-algebra, then $( \At(B), i, d )$ with $i(a) = a$ and $d(x) = \lbrace a \in \At(B) \mid a \leq x \rbrace$ is a basis for the $\mathcal{P}$-algebra $( X, h \circ \alpha_X )$.
\end{lemma}
\begin{proof}
	Let $K: \Set^{\textnormal{op}} \rightarrow \Set^{\mathcal{H}}$ denote the comparison functor with $K(X) = ( \mathcal{P}X, 2^{\eta^{\mathcal{H}}_X} )$ induced by the self-dual contravariant powerset adjunction. It is well-known that $K$ has a quasi-inverse, namely the functor $\At: \Set^{\mathcal{H}} \rightarrow \Set^{\textnormal{op}}$ assigning to a complete atomic Boolean algebra $B$ its atoms $\At(B)$ \cite{taylor2002subspaces}. The equivalence $d: B \simeq K \circ \At(B)$ is given by $d(x) = \lbrace a \in \At(B) \mid a \leq x \rbrace$.
	The calculation below 
	\begin{align*}
		2^{\eta^{\mathcal{H}}_X} \circ \alpha_{\mathcal{P}X}(\Phi)(x) &=
		\alpha_{\mathcal{P}X}(\Phi)(\eta^{\mathcal{H}}_X(x)) && \textnormal{(Definition of } 2^{\eta^{\mathcal{H}}_X}) \\
		&= \vee_{\varphi \in 2^X} \Phi(\varphi) \wedge \eta^{\mathcal{H}}_X(x)(\varphi) &&\textnormal{(Definition of } \alpha_{\mathcal{P}X}) \\
		&= \vee_{\varphi \in 2^X} \Phi(\varphi) \wedge \varphi(x) && \textnormal{(Definition of } \eta_X^{\mathcal{H}} ) \\
		&= \mu_{X}^{\mathcal{P}}(\Phi)(x) && \textnormal{(Definition of } \mu^{\mathcal{P}}_X)
	\end{align*}
	shows that $2^{\eta^{\mathcal{H}}_X} \circ \alpha_{\mathcal{P}X} = \mu^{\mathcal{P}}_X$. By \Cref{alphapowersetneighbourhooddistrlaw} the definition $\alpha(( X, h )) = ( X, h \circ \alpha_X )$ yields a functor $\alpha: \Set^{\mathcal{H}} \rightarrow \Set^{\mathcal{P}}$. 
	We can thus deduce the following equivalence between $\mathcal{P}$-algebras:
	\begin{align*}
		( X, h \circ \alpha_X ) &= \alpha(B) && \textnormal{(Definition of } \alpha) \\
		&\simeq \alpha \circ K \circ \At(B) && (\id \simeq K \circ \At) \\
		&= ( \mathcal{P}(\At(B)), 2^{\eta^{\mathcal{H}}_{\At(B)}} \circ \alpha_{\mathcal{P}(\At(B))}  ) && \textnormal{(Definition of } \alpha \circ K \circ \At) \\
		&=  ( \mathcal{P}(\At(B)), \mu^{\mathcal{P}}_{\At(B)} ) && (2^{\eta^{\mathcal{H}}_X} \circ \alpha_{\mathcal{P}X} = \mu^{\mathcal{P}}_X)
	\end{align*}
	Using the definition of a basis, the former immediately implies the claim.
	\end{proof}

The \'atomaton for the regular language $L = (a+b)^*a$, for example, can now be obtained as follows.
First, we construct the minimal pointed $\lambda^{\mathcal{H}}$-bialgebra accepting $L$, which is depicted in \Cref{overlinem(l)atom} together with its underlying CABA structure $B$. The construction can be verified with the help of the following result.
\begin{lemma}
\label[lemma]{freeneighbourhoodbialgebrastructure}
	Let $( \mathcal{H}X, \mu^{\mathcal{H}}_X, \langle  \overline{\varepsilon}, \overline{\delta} \rangle ) := \free^{\lambda^{\mathcal{H}}}(( X, \langle \varepsilon, \delta \rangle ))$. Then it holds $\overline{\varepsilon}(\Phi) = \Phi(\varepsilon)$ and $\overline{\delta}_a(\Phi)(\varphi) = \Phi(\varphi \circ \delta_a)$.
\end{lemma}
\begin{proof}
The proof is analogous to the one of \Cref{freepowersetbialgebrastructure}.
	The first equality is seen as follows: 
	\begin{align*}
		\overline{\varepsilon}(\Phi) 
		&= \mathcal{H}(\varepsilon)(\Phi)(\id_2)  && \textnormal{(Cf. proof of \Cref{freepowersetbialgebrastructure})} \\
		&= \Phi(	\id_2 \circ \varepsilon) && \textnormal{(Definition of } \mathcal{H}(\varepsilon)) \\
		&= \Phi(\varepsilon) && (\id_2 \circ \varepsilon = \varepsilon)
	\end{align*}
		For the second equality we observe:
	\begin{align*}
		\overline{\delta}_a(\Phi)(\varphi) 
		&= \mathcal{H}(\delta_a)(\Phi)(\varphi) && \textnormal{(Cf. proof of \Cref{freepowersetbialgebrastructure})} \\
		&= \Phi(\varphi \circ \delta_a) && \textnormal{(Definition of } \mathcal{H}(\delta_a)) 
	\end{align*}
\end{proof}
 Using the distributive law homomorphism $\alpha$ of \Cref{alphapowersetneighbourhooddistrlaw}, the minimal pointed $\lambda^{\mathcal{H}}$-bialgebra can then be translated into an equivalent pointed $\lambda^{\mathcal{P}}$-bialgebra with underlying CSL-structure $\alpha(B)$. By \Cref{basisshiftecabapowerset} the atoms $\At(B)$ of $B$ form a basis for $\alpha(B)$. In this case the atoms are given by $\lbrack \lbrace \lbrace x \rbrace, \lbrace x, y \rbrace \rbrace \rbrack, \lbrack \lbrace \lbrace y \rbrace \rbrace \rbrack$ and $\lbrack \lbrace \emptyset \rbrace \rbrack$. The $\mathcal{P}$-succinct automaton consequently induced by \Cref{generatorbialgebrahom} is depicted in \Cref{atomaton}; it can be recognised as the \'atomaton, cf. e.g. \cite{MyersAMU15}.

\subsection{Example: The Distromaton}

\label{distromatonexample}

\begin{figure}[t]
\centering
\begin{subfigure}[c]{\columnwidth} 
\centering
\begin{subfigure}[b]{0.5 \columnwidth}
\center
\scriptsize
\begin{tikzpicture}[node distance= 2em, ]
	\node[state, initial, shape=rectangle, initial text=] (1) {$\lbrack \lbrace \lbrace x \rbrace, \lbrace x, y \rbrace \rbrace \rbrack$};

	\node[state, shape=rectangle, right = of 1] (3) {$\lbrack \emptyset \rbrack$};

	\node[state, shape=rectangle, below = of 1, accepting] (5) {$\lbrack \lbrace \lbrace x \rbrace,  \lbrace y \rbrace, \lbrace x, y \rbrace \rbrace \rbrack$};

	\node[state, shape=rectangle, right = of 5, accepting] (15) {$\lbrack \lbrace \lbrace x \rbrace, \lbrace y \rbrace, \lbrace x, y \rbrace, \emptyset \rbrace \rbrack$};
		    \path[->]
	(3) edge[loop above] node{$a,b$} (3)
	(1) edge[loop above] node{$b$} (1)
	(1) edge[right, bend left] node{$a$} (5)
	(5) edge[loop below] node{$a$} (5)
	(5) edge[left, bend left] node{$b$} (1)
	(15) edge[loop below] node{$a,b$} (15)
;
	\end{tikzpicture} \\
		\resizebox{0.4 \columnwidth}{!}{%
		\begin{tabular}[]{ c|c|c|c|c} 
	$\vee$ & $1$ & $2$ & $3$ & $4$ \\
	\hline
	$1$ & $1$ & $1$ & $3$ & $4$ \\
	\hline
	$2$ & $1$ & $2$ & $3$ & $4$ \\
	\hline
	$3$ & $3$ & $3$ & $3$ & $4$ \\
	\hline
	$4$ & $4$ & $4$ & $4$ & $4$\\
	\end{tabular}
	}
			\resizebox{0.4 \columnwidth}{!}{%
			\begin{tabular}[]{ c|c|c|c|c} 
	$\wedge$ & $1$ & $2$ & $3$ & $4$ \\
	\hline
	$1$ & $1$ & $2$ & $1$ & $1$ \\
	\hline
	$2$ & $2$ & $2$ & $2$ & $2$ \\
	\hline
	$3$ & $1$ & $2$ & $3$ & $3$ \\
	\hline
	$4$ & $1$ & $2$ & $3$ & $4$
	\end{tabular}
	}

\caption{}
\label{overlinem(l)distro}
\end{subfigure} 
\begin{subfigure}[b]{0.4 \columnwidth}
\centering
\scriptsize
\begin{tikzpicture}[node distance=2em]
	\node[state, initial, shape=rectangle, initial text=] (1) {$\lbrack \lbrace \lbrace x \rbrace, \lbrace x, y \rbrace \rbrace \rbrack$};

	\node[state, shape=rectangle,below= of  1, accepting] (5) {$\lbrack \lbrace \lbrace x \rbrace,  \lbrace y \rbrace, \lbrace x, y \rbrace \rbrace \rbrack$};

	\node[state, shape=rectangle, right= of  5, accepting] (15) {$\lbrack \lbrace \lbrace x \rbrace, \lbrace y \rbrace, \lbrace x, y \rbrace, \emptyset \rbrace \rbrack$};
		    \path[->]
	(1) edge[loop above] node{$a,b$} (1)
	(1) edge[right, bend left] node{$a$} (5)
	(5) edge[loop below] node{$a$} (5)
	(5) edge[left, bend left] node{$a,b$} (1)
	(15) edge[loop below] node{$a,b$} (15)
	(15) edge[bend right, right] node{$a,b$} (1)
	(15) edge[below] node{$a,b$} (5)
;
	\end{tikzpicture}
\caption{}
\label{distromaton}
\end{subfigure} 
\end{subfigure}
\caption{(a) The minimal CDL-structured DFA for $L = (a+b)^*a$, where
$1 \equiv \lbrack \lbrace \lbrace x \rbrace, \lbrace x, y \rbrace \rbrace \rbrack$, $2 \equiv \lbrack \emptyset \rbrack$, $3 \equiv \lbrack \lbrace \lbrace x \rbrace, \lbrace y \rbrace, \lbrace x, y \rbrace \rbrace \rbrack$, $4 \equiv \lbrack \lbrace \lbrace x \rbrace, \lbrace y \rbrace, \lbrace x, y \rbrace,$ $\emptyset \rbrace \rbrack$; (b) The distromaton for $L = (a+b)^*a$}
\end{figure}

We shall now use our framework to recover another canonical non-deterministic acceptor: the \textit{distromaton} \cite{MyersAMU15}. As the name suggests, it can be constructed by relating the monotone neighbourhood monad $\mathcal{A}$ -- whose algebras are completely \textit{distributive} lattices -- and the powerset monad $\mathcal{P}$. Formally, the relationship can be established by the same natural transformation we used for the \'atomaton. 

\begin{corollary}
\label{neighbourhoodpowersetmorphism}
	Let $\alpha_X: \mathcal{P}X \rightarrow \mathcal{A}X$ satisfy $\alpha_X(\varphi)(\psi) = \vee_{x \in X} \varphi(x) \wedge \psi(x)$, then $\alpha$ constitutes a distributive law homomorphism $\alpha: \lambda^{\mathcal{A}} \rightarrow \lambda^{\mathcal{P}}$.
\end{corollary}
\begin{proof}
	We observe that $\alpha_X(\varphi): ( 2^X, \subseteq ) \rightarrow ( 2, \leq )$ is monotone for all $\varphi \in 2^X$. Since the monotone neighbourhood monad $\mathcal{A}$ and the neighbourhood monad $\mathcal{H}$ only differ on objects, the result follows from \Cref{alphapowersetneighbourhooddistrlaw}.  
\end{proof}

The distromaton for the regular language $L = (a+b)^*a$, for example, can now be obtained as follows.
First, we construct the minimal pointed $\lambda^{\mathcal{A}}$-bialgebra for $L$, depicted in \Cref{overlinem(l)distro} with its underlying CDL structure $h$. The construction can be verified with the help of the result below. 

\begin{lemma}
\label[lemma]{freealternatingbialgebrastructure}
	Let $( \mathcal{A}X, \mu^{\mathcal{A}}_X, \langle \overline{\varepsilon}, \overline{\delta}  \rangle ) := \free^{\lambda^{\mathcal{A}}}(( X, \langle \varepsilon, \delta \rangle ))$. Then it holds $\overline{\varepsilon}(\Phi) = \Phi(\varepsilon)$ and $\overline{\delta}_a(\Phi)(\varphi) = \Phi(\varphi \circ \delta_a)$.
\end{lemma}
\begin{proof}
	Analogous to the proof of \Cref{freeneighbourhoodbialgebrastructure}.
\end{proof}

Using the distributive law homomorphism $\alpha$ in \Cref{neighbourhoodpowersetmorphism}, the minimal pointed $\lambda^{\mathcal{A}}$-bialgebra can be translated into an equivalent pointed $\lambda^{\mathcal{P}}$-bialgebra with underlying CSL structure $L = h \circ \alpha_X$. Its partially ordered state space \[ \lbrack \emptyset \rbrack \leq \lbrack \lbrace \lbrace x \rbrace, \lbrace x, y \rbrace \rbrace \rbrack \leq \lbrack \lbrace \lbrace x \rbrace, \lbrace y \rbrace, \lbrace x, y \rbrace \rbrace \rbrack \leq \lbrack \lbrace \lbrace x \rbrace, \lbrace y \rbrace, \lbrace x, y \rbrace, \emptyset \rbrace \rbrack \] is necessarily finite, which turns the set of join-irreducibles into a size-minimal generator $( J(L), i, d )$ for $L$, where $i(y) = y$ and $d(x) = \lbrace y \in J(L) \mid y \leq x \rbrace$. In this case, the join-irreducibles are given by all non-zero states.   
	The $\mathcal{P}$-succinct automaton consequently induced by \Cref{generatorbialgebrahom} is depicted in \Cref{distromaton} and can be recognised as the distromaton, cf. \cite{MyersAMU15}.

\subsection{Example: The Minimal Xor-CABA Automaton}

\label{minimalxorcabaexample}

We conclude this section by relating the neighbourhood monad $\mathcal{H}$ with the free vector space monad $\mathcal{R}$ over the unique two element field $\mathbb{Z}_2$. In particular, we derive a new canonical succinct acceptor for regular languages, which we call the \emph{minimal xor-CABA automaton}. 

The next result shows that every CABA can be equipped with a symmetric difference operation that turns it into a vector space over the two element field. 

\begin{corollary}
\label{alphaxorneighbourhooddistrlaw}
	Let $\alpha_X: \mathcal{R}X \rightarrow \mathcal{H}X$ satisfy $\alpha_X(\varphi)(\psi) = \bigoplus_{x \in X} \varphi(x) \cdot \psi(x)$, then $\alpha$ constitutes a distributive law homomorphism $\alpha: \lambda^{\mathcal{H}} \rightarrow \lambda^{\mathcal{R}}$.
\end{corollary}
\begin{proof}
	Analogous to the proof of \Cref{alphapowersetneighbourhooddistrlaw}.
\end{proof}

Since every vector space admits a basis, the above result leads to the definition of a new acceptor of regular languages.
In what follows, let $\alpha$ denote the homomorphism in \Cref{alphaxorneighbourhooddistrlaw} and $F$ the endofunctor given by $FX = 2 \times X^{A}$.

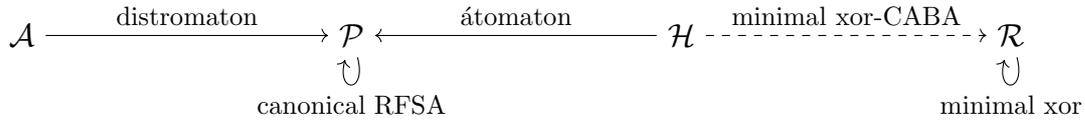
\begin{figure*}[t]
\centering
\[
	\begin{tikzpicture}[node distance=10.5em]
		\node[] (H) {$\mathcal{H}$};
		\node[left of = H] (P) {$\mathcal{P}$};
		\node[right of = H] (X) {$\mathcal{R}$};
		\node[left of = P] (A) {$\mathcal{A}$};
		    \path[->]
	(H) edge[above] node{\footnotesize \'atomaton} (P)
	(H) edge[above, dashed] node{\footnotesize
 minimal xor-CABA} (X)
	(P) edge[loop below] node{\footnotesize
 canonical RFSA} (P)
	(X) edge[loop below] node{\footnotesize
minimal xor} (X)
	(A) edge[above] node{\footnotesize
 distromaton} (P)
;
	\end{tikzpicture}
	\]
	\caption{The minimal xor-CABA automaton is to the minimal xor automaton what the \'atomaton is to the canonical RFSA}
	\label{minimalxorcabadiagram}
\end{figure*}

\begin{definition}[Minimal Xor-CABA Automaton]
	Let $( X, h, k )$ be the minimal $x$-pointed $\lambda^{\mathcal{H}}$-bialgebra accepting a regular language $L \subseteq A^*$, and $B = ( Y, i, d )$ a basis for the $\mathcal{R}$-algebra $( X, h \circ \alpha_X )$. The \emph{minimal xor-CABA automaton} for $L$ with respect to $B$ is the $d(x)$-pointed $\mathbb{Z}_2$-weighted automaton $Fd \circ k \circ i$.
	\end{definition}
	
	In \Cref{minimalxorcabadiagram} it is indicated how the canonical acceptors in this chapter, including the minimal xor-CABA automaton, are based on relations between pairs of monads.
		
	For the regular language $L = (a+b)^*a$ the above definition instantiates as follows. First, as for the \'atomaton, we construct the minimal pointed $\lambda^{\mathcal{H}}$-bialgebra $( X, h, k)$ for $L$; it is depicted in \Cref{overlinem(l)atom}. As one easily verifies, the $\mathbb{Z}_2$-vector space $(X, h \circ \alpha_X)$ is induced by the symmetric difference operation $\oplus$ on subsets. Using the notation in \Cref{overlinem(l)atom}, we choose the basis $( Y, i, d )$ with $Y = \lbrace  4,6,7,8 \rbrace$; $i(y) = y$; and $d(1) = 7 \oplus 8$, $d(2) = \emptyset$, $d(3) = 6 \oplus 7$, $d(4) = 4$, $d(5) = 6 \oplus 7 \oplus 8$, $d(6) = 6$, $d(7) = 7$, $d(8) = 8$. The induced $d(1) = 7 \oplus 8$-pointed $\mathcal{R}$-succinct automaton accepting $L$, i.e. the minimal xor-CABA automaton, is depicted in \Cref{minimalxorcaba}.
\begin{figure}[t]
\footnotesize
\centering
\[
	\begin{tikzpicture}[node distance=3em]
		\node[state, initial, shape=rectangle, initial text=, accepting] (7) {$\lbrack \lbrace \lbrace y \rbrace, \emptyset \rbrace \rbrack$};
		
	\node[state, shape=rectangle, right = of 7, accepting] (6) {$\lbrack \lbrace \lbrace y \rbrace \rbrace \rbrack$};

	\node[state, shape=rectangle, initial, right = of 6, accepting, initial text=] (8) {$\lbrack \lbrace \lbrace x \rbrace, \lbrace y \rbrace, \lbrace x, y \rbrace, \emptyset \rbrace \rbrack$};
	
		\node[state, shape=rectangle, right = of 8] (4) {$\lbrack \lbrace \lbrace x, y \rbrace, \emptyset \rbrace \rbrack$};
		
		    \path[->]
	(7) edge[loop above] node{$a,b$} (7)
	(7) edge[above] node{$a$} (6)
	(8) edge[loop above] node{$a,b$} (8)
	(4) edge[above] node{$a,b$} (8)
;
	\end{tikzpicture}
	\]
	\caption{The minimal xor-CABA automaton for $L = (a+b)^*a$}
		\label{minimalxorcaba}
\end{figure}
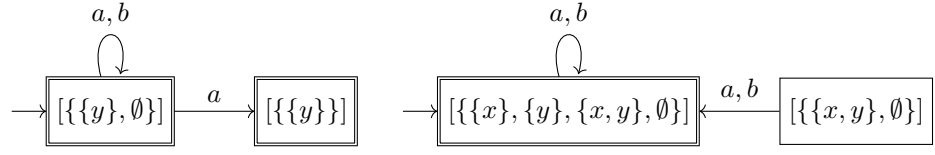

\section{Minimality}

\label{minimality}

 In this section we restrict ourselves to the category of (nominal) sets. We show that every language satisfying a suitable property parametric in monads $S$ and $T$ admits a size-minimal succinct automaton of type $T$ accepting it. As a main result we obtain \Cref{minimalitytheorem}. In \Cref{minimalityimplications} we instantiate the former to subsume known minimality results for canonical automata, to prove the xor-CABA automaton minimal, and to establish a size-comparison between different acceptors.

Given a distributive law homomorphism $\alpha: \lambda^S \rightarrow \lambda^T$, let $\ext: \Coalg(FT) \rightarrow \Coalg(FS)$ be the functor given by $\ext(( X, k )) = ( X, F\alpha_X \circ k )$ and $\ext(f) = f$. Moreover, let $\expa_U: \Coalg(FU) \rightarrow \Bialg(\lambda^U)$ for $U \in \lbrace S, T \rbrace$ denote the functor introduced in \Cref{expfunctor}.

\begin{proposition}
\label{alphaunderlies}
Let $\alpha: \lambda^{S} \rightarrow \lambda^T$ be a distributive law homomorphism. Then $\alpha_X: TX \rightarrow SX$ underlies a natural transformation $\alpha: \expa_T \Rightarrow \alpha \circ \expa_S \circ \ext$ between functors of type $\Coalg(FT) \rightarrow \Bialg(\lambda^T)$.
\end{proposition}
\begin{proof}
	Given a $T$-succinct automaton $\mathcal{X} = ( X, k )$ the definitions imply:
	\begin{align*}
		\expa_T(\mathcal{X}) &= ( TX, \mu^T_X, F \mu^T_X \circ \lambda^T_{TX} \circ Tk ) \\
		\alpha \circ \expa_S \circ \ext(\mathcal{X}) &= ( SX, \mu^S_X \circ \alpha_{SX}, F\mu^S_X \circ \lambda^S_{SX} \circ SF\alpha_X \circ Sk ).
	\end{align*}
	By the definition of distributive law homomorphisms, the morphism $\alpha_X$ commutes with the underlying $T$-algebra structures. Its commutativity with the underlying $F$-coalgebra structures follows from:
	\[
	\begin{tikzcd}
		TX \arrow{dd}[left]{\alpha_X} \arrow{r}{Tk} & TFTX \arrow{dd}{\alpha_{FTX}} \arrow{rr}{\lambda^T_{TX}} \arrow{dr}{TF\alpha_X} & & FT^2X \arrow{r}{F\mu^T_X} \arrow{d}{FT\alpha_X} & FTX \arrow{dd}{F\alpha_X} \\
		&& TFSX \arrow{r}{\lambda^T_{SX}} \arrow{d}{\alpha_{FSX}} & FTSX \arrow{d}{F\alpha_{SX}} &  \\
		SX \arrow{r}{Sk} & SFTX \arrow{r}{SF \alpha_X} & SFSX \arrow{r}{\lambda^S_{SX}} & FS^2X \arrow{r}{F\mu^S_X} & FSX
	\end{tikzcd}.
	\]
	Above we use the naturality of $\alpha$ and $\lambda^T$, and the definition of a distributivity law homomorphism. The naturality of $\alpha$ as natural transformation $\alpha: \expa_T \Rightarrow \alpha \circ \expa_S \circ \ext$ follows from the naturality of $\alpha$ as natural transformation $\alpha: T \Rightarrow S$.
\end{proof}

In the above situation a $T$-succinct automaton admits \emph{two} semantics, by lifting the former either to a bialgebra over $\lambda^S$ or $\lambda^T$.
The next definition introduces a notion of \emph{closedness} that captures the cases in which the image of both semantics coincides.

\begin{definition}[$\alpha$-Closed Succinct Automaton]
\label[definition]{closedsuccinctdef}
Let $\alpha: \lambda^{S} \rightarrow \lambda^T$ be a distributive law homomorphism.
	We say that a $T$-succinct automaton $\mathcal{X}$ is \emph{$\alpha$-closed} if the unique diagonal below is an isomorphism:
	\[
	\begin{tikzcd}
		\expa_T(\mathcal{X}) \arrow[twoheadrightarrow]{r}{\obs} \arrow{d}[left]{\obs \circ \alpha_X} & \img(\obs_{\expa_T(\mathcal{X})}) \arrow[dashed]{dl}{}   \\
		\img(\obs_{\alpha(\expa_S(\ext(\mathcal{X})))}) \arrow[hookrightarrow]{r}{} & \Omega \arrow[hookleftarrow]{u}{}
	\end{tikzcd}.
	\]
	\end{definition}
	
It immediately follows that the unique diagonal in \Cref{closedsuccinctdef} is injective.

The result below shows that instead of defining the $\alpha$-closedness of $\mathcal{X}$ as true if the unique diagonal is an isomorphism, we could have equivalently defined it to be true if their exists \emph{any} isomorphism of such type.

 \begin{lemma}
\label{alphaclosednessequivalentform}
	$\mathcal{X}$ is $\alpha$-closed iff $\img(\obs_{\expa_T(\mathcal{X})}) \cong \img(\obs_{\alpha(\expa_S(\ext(\mathcal{X})))})$.
\end{lemma}
\begin{proof}
	Clearly, if $\mathcal{X}$ is $\alpha$-closed, then, by definition, the unique diagonal in \Cref{closedsuccinctdef} witnesses the desired isomorphism.
	
	Conversely, assume there is \emph{any} isomorphism $\img(\obs_{\expa_T(\mathcal{X})}) \overset{\varphi}{\cong} \img(\obs_{\alpha(\expa_S(\ext(\mathcal{X})))})$. 
	As the unique diagonal $d$ in \Cref{closedsuccinctdef} is necessarily injective, it remains to show that it is also surjective, that is $\im(d) \cong \img(\obs_{\alpha(\expa_S(\ext(\mathcal{X})))})$. The latter follows immediately from the uniqueness of epi-mono factorisations:
	\[
			\begin{tikzcd}
			\img(\obs_{\expa_T(\mathcal{X})}) \arrow[twoheadrightarrow]{d}[left]{\varphi} \arrow[twoheadrightarrow]{rr}{d} & & \im(d)  \arrow[dashed]{dll}{\cong} \\
			\img(\obs_{\alpha(\expa_S(\ext(\mathcal{X})))}) \arrow[hookrightarrow]{r}[below]{\varphi^{-1}} & \img(\obs_{\expa_T(\mathcal{X})}) \arrow[hookrightarrow]{r}[below]{d} & \img(\obs_{\alpha(\expa_S(\ext(\mathcal{X})))}) \arrow[hookleftarrow]{u}
		\end{tikzcd}.
	\]
	\end{proof}	

Succinct automata obtained from generators for bialgebras are $\alpha$-closed.

\begin{lemma}
\label{generatorclosed}
Let $\alpha: \lambda^{S} \rightarrow \lambda^T$ be a distributive law homomorphism and $( X, h, k )$ a $\lambda^S$-bialgebra.
	If $( Y, i, d )$ is a generator for $( X, h \circ \alpha_X )$, then $( Y,  Fd \circ k \circ i)$ is $\alpha$-closed.
\end{lemma}
\begin{proof}
	We write $\mathbb{X} := ( X, h, k )$, $\mathbb{G} := ( Y, i, d )$, and $\gen(\alpha(\mathbb{X}), \mathbb{G}) := ( Y,Fd \circ k \circ i)$.  The definitions imply
	\begin{align*}
		\expa_T(\gen(\alpha(\mathbb{X}), \mathbb{G})) &= ( TY, \mu^T_Y, (Fd \circ k \circ i)^{\sharp} ) \\
		\alpha \circ \expa_S \circ \ext(\gen(\alpha(\mathbb{X}), \mathbb{G})) &= ( SY, \mu^S_Y \circ \alpha_{SY}, (F(\alpha_Y \circ d) \circ k \circ i)^{\sharp} ).
	\end{align*}
	Since $\mathbb{G}$ is a generator for $( X, h \circ \alpha_X )$, \Cref{forgenerator-isharp-is-bialgebra-hom} implies that $(h \circ \alpha_X) \circ Ti: \expa_T(\gen(\alpha(\mathbb{X}), \mathbb{G})) \rightarrow \alpha(\mathbb{X})$ is a $\lambda^T$-bialgebra homomorphism. By the definition of a generator, $(h \circ \alpha_X) \circ Ti$ has a right-inverse, which implies its surjectivity.
	Naturality of $\alpha$ shows that $\overline{G} = ( Y, i, \alpha_Y \circ d )$ is a generator for $( X, h )$. Thus \Cref{forgenerator-isharp-is-bialgebra-hom} implies that $h \circ Si: \expa_S(\gen(\mathbb{X}, \overline{\mathbb{G}})) \rightarrow \mathbb{X}$ is a $\lambda^S$-bialgebra homomorphism. By the definition of a generator, $h \circ Si$ has a right-inverse, which implies its surjectivity.
 Applying $\alpha$ shows that $h \circ Si: \alpha \circ \expa_S \circ \ext(\gen(\alpha(\mathbb{X}), \mathbb{G})) \rightarrow \alpha(\mathbb{X})$ is a surjective $\lambda^T$-bialgebra homomorphism.
	The statement follows from the uniqueness of epi-mono factorisations:
	\[
	\begin{tikzcd}[column sep = 1.5em]
		\expa_T(\gen(\alpha(\mathbb{X}), \mathbb{G})) \arrow[twoheadrightarrow]{rrr}{\obs} \arrow{d}[left]{\alpha_Y} \arrow[twoheadrightarrow]{dr}{(h \circ \alpha_X) \circ Ti} & & & \img(\obs_{\expa_T(\gen(\alpha(\mathbb{X}), \mathbb{G}))})  \arrow[dashed]{dl}{\simeq} \\
		\alpha \circ \expa_S \circ \ext(\gen(\alpha(\mathbb{X}), \mathbb{G})) \arrow[twoheadrightarrow]{d}[left]{\obs}\arrow[twoheadrightarrow]{r}[below]{h \circ Si} & \alpha(\mathbb{X}) \arrow[twoheadrightarrow]{r}{\obs} &\img(\obs_{\alpha(\mathbb{X})}) \arrow[dashed]{dll}{\simeq} \arrow[hookrightarrow]{dr}{} & \\
		\img(\obs_{\alpha \circ \expa_S \circ \ext(\gen(\alpha(\mathbb{X}), \mathbb{G}))}) \arrow[hookrightarrow]{rrr}{} & & & \Omega \arrow[hookleftarrow]{uu}{}
	\end{tikzcd}.
	\]
\end{proof}

We are now able to state our main result.

In \cite[Theorem 4.8]{MyersAMU15} it is shown that i) the canonical RFSA for $L$ accepts the least number of languages among all NFAs accepting $L$; and ii) among NFAs with the same number of accepted languages, the canonical RFSA is state-minimal. Note how the two bullet points of \Cref{minimalitytheorem} below resemble the ones in \cite[Theorem 4.8]{MyersAMU15}: if $S = T = \mathcal{P}$ and $\alpha$ is trivial, the theorem implies that i) the set of languages accepted by the canonical RFSA for $L$ is included in any other set of languages accepted by NFAs for $L$; and ii) among NFAs for $L$ that accept the same set of languages (i.e. $\overline{\Der(L)}^{\CSL}$), the canonical RFSA is size-minimal (cf. \Cref{canonicalrfsaminimal}).

\begin{theorem}[Minimal Succinct Automata]
\label{minimalitytheorem}
	Given a language $L \in \Omega$ such that there exists a minimal pointed $\lambda^S$-bialgebra $\mathbb{M}$ accepting $L$ and the underlying algebra of $\alpha(\mathbb{M})$ admits a size-minimal generator, there exists a pointed $\alpha$-closed $T$-succinct automaton $\mathcal{X}$ accepting $L$ such that: \begin{itemize}
		\item for any pointed $\alpha$-closed $T$-succinct automaton $\mathcal{Y}$ accepting $L$ we have that $\img(\obs_{\expa_T(\mathcal{X})}) \subseteq \img(\obs_{\expa_T(\mathcal{Y})})$;
		\item if $\img(\obs_{\expa_T(\mathcal{X})}) = \img(\obs_{\expa_T(\mathcal{Y})})$, then $\vert X \vert \leq \vert Y \vert$, where $X$ and $Y$ are the carriers of $\mathcal{X}$ and $\mathcal{Y}$, respectively.
	\end{itemize} 
 \end{theorem}
 \begin{proof}
We use a similar notation as in the proof of \Cref{generatorclosed}.
	Let $\mathbb{G} = ( X, i, d )$ be the size-minimal generator for the underlying algebra of $\alpha(\mathbb{M})$, which we assume to be $x$-pointed.
	We define a $d(x)$-pointed $T$-succinct automaton $\mathcal{X} := \gen(\alpha(\mathbb{M}), \mathbb{G})$. By \Cref{generatorbialgebrahom} there exists a $\lambda^T$-bialgebra homomorphism $i^\sharp \colon \expa_T(\mathcal{X}) \to \alpha(\mathbb{M})$. By the definition of a generator, $i^{\sharp}$ has a right-inverse, which implies its surjectivity. From the uniqueness of morphisms into a final coalgebra we can deduce that $\mathcal{X}$ accepts the language accepted by $\alpha(\mathbb{M})$.
	Because $\alpha$ only modifies the algebraic part of a bialgebra and the final bialgebra homomorphism is induced by the underlying final coalgebra homomorphism, the language accepted by $\alpha(\mathbb{M})$ is the language $L$ accepted by $\mathbb{M}$.
	From \Cref{generatorclosed} it follows that $\mathcal{X}$ is $\alpha$-closed.

	Consider any pointed $\alpha$-closed $T$-succinct automaton $\mathcal{Y}$ accepting $L$.
	By minimality of $\mathbb{M}$ there is an injective $\lambda^S$-bialgebra homomorphism $j \colon \mathbb{M} \hookrightarrow \img(\obs_{\expa_S(\ext(\mathcal{Y}))})$, which is also an injective $\lambda^T$-bialgebra homomorphism $j \colon \alpha(\mathbb{M}) \hookrightarrow \img(\obs_{\alpha(\expa_S(\ext(\mathcal{Y}))}))$, because the functor $\alpha$ is the identity on morphisms, and only modifies the algebraic part of a bialgebra.
	Since $\mathcal{Y}$ is $\alpha$-closed, we have an isomorphism $\img(\obs_{\alpha(\expa_S(\ext(\mathcal{Y}))})) \cong \img(\obs_{\expa_T(\mathcal{Y})})$. By composing the previous functions, we thus obtain an injective $\lambda^T$-bialgebra homomorphism $k: \alpha(\mathbb{M}) \hookrightarrow \img(\obs_{\expa_T(\mathcal{Y})})$. Since any final bialgebra homomorphism is induced by the underlying final $F$-coalgebra homomorphism, we have $\obs_{\alpha(\mathbb{M})} = \obs_{\mathbb{M}}$. By assumption $\mathbb{M}$ is minimal, that is, $\obs_{\mathbb{M}}$ is injective. By the uniqueness of morphisms into a final object the outer square of the diagram below commutes: 	\[
		\begin{tikzcd}
			\expa_T(\mathcal{X}) \arrow[twoheadrightarrow]{r}{i^{\sharp}} \arrow[twoheadrightarrow]{d}[left]{\obs_{\expa_T(\mathcal{X})}} & \alpha(\mathbb{M})  \arrow[dashed]{dl}{\cong} \\
			\im(\obs_{\expa_T(\mathcal{X})}) \arrow[hookrightarrow]{r} & \Omega \arrow[hookleftarrow]{u}[right]{\obs_{\alpha(\mathbb{M})} = \obs_{\mathbb{M}}}
		\end{tikzcd}
	\]
	By the uniqueness of epi-mono factorisations, there thus exists a diagonal isomorphism $\alpha(\mathbb{M}) \cong \img(\obs_{\expa_T(\mathcal{X})})$. By composing the diagonal isomorphism with $k$, we obtain a $\lambda^T$-bialgebra homomorphism $\img(\obs_{\expa_T(\mathcal{X})}) \to \img(\obs_{\expa_T(\mathcal{Y})})$. The latter commutes with observability maps and thus is an inclusion, so $\img(\obs_{\expa_T(\mathcal{X})}) \subseteq \img(\obs_{\expa_T(\mathcal{Y})})$.

	Suppose $\img(\obs_{\expa_T(\mathcal{X})}) = \img(\obs_{\expa_T(\mathcal{Y})})$, which implies that $j$ is an isomorphism.
	Then there exists a surjective $\lambda^T$-bialgebra homomorphism $\expa_T(\mathcal{Y}) \to \alpha(\mathbb{M})$, which means that $Y$ forms the carrier of a generator for the underlying algebra of $\alpha(\mathbb{M})$.
	By the size-minimality of $\mathbb{G}$ we thus obtain $\vert X \vert \leq \vert Y \vert$.
\end{proof}

For a $T$-succinct automaton $\mathcal{X}$ let us write $\obs^{\dag}_{\mathcal{X}} := \obs_{\expa_T(\mathcal{X})} \circ \eta^T_X: X \rightarrow \Omega$ for a generalisation of the semantics of non-deterministic automata. \Cref{imgobssuccinctsem} provides an equivalent characterisation of $\alpha$-closedness in terms of $\obs^{\dag}$ that will be particularly useful in \Cref{minimalityimplications}. To prove \Cref{imgobssuccinctsem}, we need the following technical results.

\begin{lemma}
\label[lemma]{alphapreservesfinal}
	Let $\alpha: \lambda^S \rightarrow \lambda^T$ be a distributive law homomorphism. If $( \Omega, h, k )$ is the final $\lambda^S$-bialgebra, then $( \Omega, h \circ \alpha_{\Omega}, k )$ is the final $\lambda^T$-bialgebra.
\end{lemma}
\begin{proof}
	It is well-known that if $( \Omega, h, k )$ is the final $\lambda^S$-bialgebra, then $( \Omega, k )$ is the final $F$-coalgebra and $h: S\Omega \rightarrow \Omega$ is the unique homomorphism satisfying $k \circ h = Fh \circ \lambda^S_{\Omega} \circ Sk$. Similarly, it is well-known that $( \Omega, \overline{h}, k )$ is the final $\lambda^T$-bialgebra, where $\overline{h}: T\Omega \rightarrow \Omega$ is the unique homomorphism satisfying $k \circ \overline{h} = F\overline{h} \circ \lambda^T_{\Omega} \circ Tk$. The statement thus follows from uniqueness:
	\[
	\begin{tikzcd}
		T\Omega \arrow{d}[left]{Tk} \arrow{r}{\alpha_{\Omega}} & S\Omega \arrow{d}{Sk} \arrow{r}{h} & \Omega \arrow{dd}{k} \\
		TF\Omega \arrow{r}{\alpha_{F\Omega}} \arrow{d}[left]{\lambda^T_{\Omega}} & SF\Omega \arrow{d}{\lambda^S_{\Omega}} & \\
		FT\Omega \arrow{r}[below]{F\alpha_{\Omega}} & FS\Omega \arrow{r}[below]{Fh} & F\Omega
	\end{tikzcd}.
	\]
\end{proof} 

\begin{lemma}
\label{obsdagext}
Let $\mathcal{X}$ be a $T$-succinct automaton, then $\obs^{\dag}_{\mathcal{X}} = \obs^{\dag}_{\ext(\mathcal{X})}$.
\end{lemma}
\begin{proof}
	Since by \Cref{alphaunderlies} the morphism $\alpha_X: \expa_T(\mathcal{X}) \rightarrow \alpha(\expa_S(\ext(\mathcal{X})))$ is a $\lambda^T$-bialgebra homomorphism, we have by uniqueness $\obs_{\alpha(\expa_S(\ext(\mathcal{X})))} \circ \alpha_X = \obs_{\expa_T(\mathcal{X})}$. Since any final bialgebra homomorphism is induced by the underlying final $F$-coalgebra homomorphism it holds $\obs_{\alpha(\expa_S(\ext(\mathcal{X})))} = \obs_{\expa(\ext_S(\mathcal{X}))}$, which thus implies 
	\begin{equation}
	\label{obsdagproofeq}
		\obs_{\expa_S(\ext(\mathcal{X}))} \circ \alpha_X = \obs_{\expa_T(\mathcal{X})}.
	\end{equation}
	 The statement follows from
	 \begingroup
	 \allowdisplaybreaks
	\begin{align*}
		\obs^{\dag}_{\mathcal{X}} &= \obs_{\expa_T(\mathcal{X})} \circ \eta^T_X && \textnormal{(Definition of } \obs^{\dag}_{\mathcal{X}}) \\
		&= \obs_{\expa_S(\ext(\mathcal{X}))} \circ \alpha_X \circ \eta^T_X && \eqref{obsdagproofeq} \\
		&= \obs_{\expa_S(\ext(\mathcal{X}))} \circ \eta^S_X && (\alpha \textnormal{ is distributive law hom.)}\\
		&= \obs^{\dag}_{\ext(\mathcal{X})} && \textnormal{(Definition of } \obs^{\dag}_{\ext(\mathcal{X})})
	\end{align*}
	\endgroup
\end{proof}

\begingroup
\allowdisplaybreaks
\begin{lemma}
\label{imgobssuccinctsem}
	Let $\alpha: \lambda^{S} \rightarrow \lambda^T$ be a distributive law homomorphism. For any $T$-succinct automaton $\mathcal{X}$ it holds that $\img(\obs_{\expa_T(\mathcal{X})}) = \img(h \circ \alpha_{\Omega} \circ T(\obs^{\dag}_{\mathcal{X}}))$ and 
		$\img(\obs_{\alpha(\expa_S(\ext(\mathcal{X})))}) = \img(h \circ S(\obs^{\dag}_{\mathcal{X}}))$, where $( \Omega, h, k )$ is the final $\lambda^S$-bialgebra.
	\end{lemma}
\begin{proof}
	By \Cref{alphapreservesfinal} $( \Omega, h \circ \alpha_{\Omega}, k )$ is the final $\lambda^T$-bialgebra. The first statement follows from
	\begin{align*}
		\obs_{\expa_T(\mathcal{X})} 
		&= \obs_{\expa_T(\mathcal{X})} \circ \mu^T_X  \circ T(\eta^T_X) && \textnormal{(Definition of } T) \\
		&= h \circ \alpha_{\Omega} \circ T(\obs_{\expa_T(\mathcal{X})}) \circ T(\eta^T_X) && (\obs_{\expa_T(\mathcal{X})} \textnormal{ is } T\textnormal{-algebra hom}) \\
		&= h \circ \alpha_{\Omega} \circ T(\obs^{\dag}_{\mathcal{X}}) && \textnormal{(Definition of } \obs^{\dag}_{\mathcal{X}})
	\end{align*}
	Similarly one shows that $	\obs_{\expa_S(\ext(\mathcal{X}))} =  h \circ S(\obs^{\dag}_{\ext(\mathcal{X})})$.
 Since any final bialgebra homomorphism is induced by the underlying final $F$-coalgebra homomorphism, it thus follows
	\begin{align*}
		\img(\obs_{\alpha(\expa_S(\ext(\mathcal{X})))}) &= \img(\obs_{\expa_S(\ext(\mathcal{X}))}) && (\obs_{\alpha(\expa_S(\ext(\mathcal{X})))} = \obs_{\expa_S(\ext(\mathcal{X}))}) \\
		&= \img(h \circ S(\obs^{\dag}_{\ext(\mathcal{X})})) && (\obs_{\expa_S(\ext(\mathcal{X}))} =  h \circ S(\obs^{\dag}_{\ext(\mathcal{X})})) \\
		&= \img(h \circ S(\obs^{\dag}_{\mathcal{X}})) && \textnormal{(\Cref{obsdagext})}
	\end{align*}
\end{proof}
\endgroup
In light of \Cref{alphaclosednessequivalentform}, the above \Cref{imgobssuccinctsem} implies that a $T$-succinct automaton $\mathcal{X}$ is $\alpha$-closed iff there exists an isomorphism $\img(h \circ \alpha_{\Omega} \circ T(\obs^{\dag}_{\mathcal{X}})) \cong \img(h \circ S(\obs^{\dag}_{\mathcal{X}}))$.

\subsection{Applications to Canonical Automata}

\label{minimalityimplications}

In this section we instantiate \Cref{minimalitytheorem} to characterise a variety of canonical acceptors from the literature as size-minimal representatives among subclasses of $\alpha$-closed succinct automata, i.e. those automata whose images of the two semantics induced by $\alpha$ coincide.
We begin with the canonical RFSA and the minimal xor automaton, for which $\alpha$ is the identitity and $\alpha$-closedness therefore is trivial. 

In \cite{denis2001residual} the canonical RFSA for $L$ has been characterised as size-minimal among those NFAs accepting $L$ for which states accept a residual of $L$. More recently, it was shown that the class in fact can be extended to those NFAs accepting $L$ for which states accept a \emph{union} of residuals of $L$ \cite{MyersAMU15}. In \Cref{canonicalrfsaminimal} we recover the latter as a consequence of the second point in \Cref{minimalitytheorem}. To prove \Cref{canonicalrfsaminimal}, we need the following technical result.

\begin{lemma}
\label{obsgenerated}
	Let $\mathbb{X} = ( X, h, k )$ be an observable $\lambda^S$-bialgebra and $\mathbb{G}$ a generator for $( X, h \circ \alpha_X )$, then $\img(\obs_{\expa_T(\gen(\alpha(\mathbb{X}), \mathbb{G}))}) \simeq X$.
\end{lemma}
\begin{proof}
By \Cref{forgenerator-isharp-is-bialgebra-hom} there exists a surjective $\lambda^T$-bialgebra homomorphism \\$\expa_T(\gen(\alpha(\mathbb{X}), \mathbb{G})) \rightarrow \alpha(\mathbb{X})$. Since the final $\lambda^T$-bialgebra homomorphism is induced by the underlying final $F$-coalgebra homomorphism and  $\alpha(\mathbb{X}) = ( X, h \circ \alpha_X, k )$, it holds $\obs_{\alpha(\mathbb{X})} = \obs_{\mathbb{X}}$. The statement follows from the uniqueness of epi-mono factorisations and the definition of $\alpha(\mathbb{X})$:
	\[
	\begin{tikzcd}
		\expa_T(\gen(\alpha(\mathbb{X}), \mathbb{G})) \arrow[twoheadrightarrow]{r}{} \arrow[twoheadrightarrow]{d}{}  & \alpha(\mathbb{X})  \arrow[dashed]{dl}{\simeq} \\
		\img(\obs_{\expa_T(\gen(\alpha(\mathbb{X}), \mathbb{G}))}) \arrow[hookrightarrow]{r}{} & \Omega \arrow[hookleftarrow]{u}[right]{\obs_{\alpha(\mathbb{X})} = \obs_{\mathbb{X}}}
	\end{tikzcd}.
	\]
\end{proof}

We write $\overline{Y}^{\mathbb{A}}$ for the algebraic closure of a subset $Y \subseteq A$ of some $T$-algebra $\mathbb{A}$ (for details see \Cref{closure_sec}). 
For example, if $Y = \im(f)$ for some $f$ with codomain $\mathbb{A} = (A,h)$, the closure is given by the induced $T$-algebra structure on $\im(h \circ Tf)$. Recall that the set $2 = \lbrace 0, 1 \rbrace$ -- and consequently the set $2^{A^*}$ underlying the final coalgebra for the functor $FX = 2 \times X^A$ -- can be turned into a $\mathcal{P}$-algebra (\Cref{powersetalgebra}), an $\mathcal{H}$-algebra (\Cref{neighbourhoodalgebra}), an $\mathcal{A}$-algebra (\Cref{monotoneneighbourhoodoutputalgebra}), and an $\mathcal{R}$-algebra (\Cref{xoroutputalgebra}). We thus can take the closure of a subset $Y \subseteq 2^{A^*}$ with respect to a $T$-algebra structure $(2^{A^*}, h_T)$, for any $T \in \lbrace \mathcal{P}, \mathcal{H}, \mathcal{A}, \mathcal{R} \rbrace$. In such a situation, we keep $h_T$ implicit and abbreviate $\overline{Y}^{\EM} := \overline{Y}^{(2^{A^*}, h_T)}$. For example, $\overline{\Der(L)}^{\CSL}$ denotes the $\mathcal{P}$-closure of the subset $\Der(L) \subseteq 2^{A^*}$. For more cases recall \Cref{eilenbergexample}. 

\begin{corollary}
\label{imgcabanfa}
	Let $\alpha: \mathcal{P}X \rightarrow \mathcal{H}X$ satisfy $\alpha_X(\varphi)(\psi) = \vee_{x \in X} \varphi(x) \wedge \psi(x)$. For any unpointed non-deterministic automaton $\mathcal{X}$ it holds:
	\begin{itemize}
		\item $\img(\obs_{\expa_{\mathcal{P}}(\mathcal{X})}) = \overline{\im(\obs^{\dag}_{\mathcal{X}})}^{\CSL}$;
		\item $\img(\obs_{\alpha(\expa_{\mathcal{H}}(\ext(\mathcal{X})))}) = \overline{\im(\obs^{\dag}_{\mathcal{X}})}^{\CABA}$.
	\end{itemize}
\end{corollary}
\begin{proof}
	The final $\lambda^{\mathcal{H}}$-bialgebra is given by $( 2^{A^*}, 2^{\eta^{\mathcal{H}}_{A^*}}, ( \varepsilon, \delta ) )$. In the proof of \Cref{basisshiftecabapowerset} it was shown that $2^{\eta^{\mathcal{H}}_X} \circ \alpha_{\mathcal{P}X} = \mu^{\mathcal{P}}_X$. Thus it follows
	\begin{align*}
		\img(\obs_{\expa_{\mathcal{P}}(\mathcal{X})}) &= \img(2^{\eta^{\mathcal{H}}_{A^*}} \circ \alpha_{2^{A^*}} \circ \mathcal{P}(\obs^{\dag}_{\mathcal{X}})) && \textnormal{(\Cref{imgobssuccinctsem})} \\
		&= \img(\mu^{\mathcal{P}}_{A^*} \circ \mathcal{P}(\obs^{\dag}_{\mathcal{X}})) && (2^{\eta^{\mathcal{H}}_X} \circ \alpha_{\mathcal{P}X} = \mu^{\mathcal{P}}_X) \\
		&= \lbrace \cup_{u \in U} \obs^{\dag}_{\mathcal{X}}(u) \mid U \subseteq X \rbrace && \textnormal{(Definitions of } \mathcal{P}(-),\ \mu^{\mathcal{P}}) \\
		&= \overline{\lbrace \obs^{\dag}_{\mathcal{X}}(x) \mid x \in X \rbrace}^{\CSL} && \textnormal{(Definition of } \overline{(-)}^{\CSL})
	\end{align*}
	Similarly one computes
	\begin{align*}
&\img(\obs_{\alpha(\expa_{\mathcal{H}}(\ext(\mathcal{X})))}) \\
=& \ \textnormal{(\Cref{imgobssuccinctsem})} \\
& \img(2^{\eta^{\mathcal{H}}_{A^*}} \circ \mathcal{H}(\obs^{\dag}_{\mathcal{X}})) \\
=& \ \textnormal{(Definitions of } 2^{\eta^{\mathcal{H}}_{A^*}}, \mathcal{H}(-)) \\
& \lbrace \cup_{\varphi \in \Phi} \cap_{x \in \varphi} \obs^{\dag}_{\mathcal{X}}(x) \cap \cap_{x \not \in \varphi} \obs^{\dag}_{\mathcal{X}}(x)^c  \mid \Phi \subseteq 2^X \rbrace \\
=&\ \textnormal{(Set equality)}  \\
& \lbrace \lbrace w \in A^* \mid \lbrace x \in X \mid \obs^{\dag}_{\mathcal{X}}(x)(w) = 1 \rbrace \in \Phi \rbrace \mid \Phi \subseteq 2^X \rbrace \\
	=&\ \textnormal{(Definition of } \overline{(-)}^{\CABA})	\\
		& \overline{\lbrace \obs^{\dag}_{\mathcal{X}}(x) \mid x \in X \rbrace}^{\CABA} 
	\end{align*}
\end{proof}

\begin{corollary}
\label{canonicalrfsaminimal}
	The canonical RFSA for $L$ is size-minimal among NFAs $\mathcal{Y}$ accepting $L$ with 
	$
	 \overline{\im(\obs^{\dag}_{\mathcal{Y}})}^{\CSL} \subseteq \overline{\Der(L)}^{\CSL} 
	$.
\end{corollary}
\begin{proof}
	By \Cref{powersetalgebra} the morphism $h^{\mathcal{P}}: \mathcal{P}2 \rightarrow 2$ with $h^{\mathcal{P}}(\varphi) = \varphi(1)$ is a $\mathcal{P}$-algebra. As shown in \Cref{induceddistrlaw}, it can used to derive a canonical distributive law $\lambda^{\mathcal{P}}$. It is not hard to see that the minimal pointed $\lambda^{\mathcal{P}}$-bialgebra $\mathbb{M}$ accepting $L$ exists and that its underlying state space is given by the finite complete join-semi lattice $\overline{\Der(L)}^{\CSL}$. By \Cref{joinirreducstateminimal} the join-irreducibles for $\mathbb{M}$ constitute a size-minimal generator $\mathbb{G}$. By definition, the canonical RFSA for $L$ is given by $\mathcal{X} := \gen(\mathbb{M}, \mathbb{G})$. From \Cref{obsgenerated} it follows that $\img(\obs_{\expa_{\mathcal{P}}(\mathcal{X})}) \simeq \overline{\Der(L)}^{\CSL}$. As seen in e.g. \Cref{imgcabanfa}, one has $\img(\obs_{\expa_{\mathcal{P}}(\mathcal{Y})}) =\overline{\im(\obs^{\dagger}_{\mathcal{Y}})}^{\CSL}$ for any NFA $\mathcal{Y}$. By choosing $\alpha$ as the identity, which implies $\alpha$-closedness for any NFA, the statement thus follows from \Cref{minimalitytheorem}.
\end{proof}

The second condition in \Cref{minimalitytheorem} is always satisfied for a \emph{reachable} succinct automaton $\mathcal{Y}$. Since for $\mathbb{Z}_2$-weighted automata it is possible to find an equivalent reachable $\mathbb{Z}_2$-weighted automaton with less or equally many states (which for NFA is not necessarily the case), the minimal xor automaton is minimal among \emph{all} $\mathbb{Z}_2$-weighted automata, as was already known from for instance \cite{VuilleminG210}.

\begin{corollary}
\label{minimalxor} 
		The minimal xor automaton for $L$ is size-minimal among $\mathbb{Z}_2$-weighted automata accepting $L$.
\end{corollary}
\begin{proof}
	Analogous to \Cref{canonicalrfsaminimal} one can show that the minimal xor automaton for $L \subseteq A^*$ is size-minimal among all $\mathbb{Z}_2$-weighted automata $\mathcal{Y}$ accepting $L$ such that
$ \overline{\im(\obs^{\dagger}_{\mathcal{Y}})}^{\mathbb{Z}_2\Vect} \subseteq \overline{\Der(L)}^{\mathbb{Z}_2\Vect}$.
	Specific to this case are \Cref{xoroutputalgebra} and \Cref{xorbasisstateminimal}. It remains to observe that for any $\mathbb{Z}_2$-weighted automaton $\mathcal{X}$, one can find an equivalent $\mathbb{Z}_2$-weighted automaton $\mathcal{Y}$ with a state space of size not greater than the one of $\mathcal{X}$, such that above inclusion holds. The state space of $\mathcal{Y}$ can be chosen as a basis for the underlying vector space of the epi-mono factorisation of the reachability map $\mathcal{X}(A^*) \rightarrow \overline{\im(\obs^{\dagger}_{\mathcal{X}})}^{\mathbb{Z}_2\Vect}$.
	\end{proof}

For the \'atomaton, the distromaton, and the minimal xor-CABA automaton the distributive law homomorphism $\alpha$ in play is non-trivial; $\alpha$-closedness translates to the below equalities between closures. In all three cases it is possible to waive the inclusion induced by the second point in \Cref{minimalitytheorem}. The proof of the minimality result for the  \'atomaton (\Cref{minimalityatomaton}) requires the following three technical results (\Cref{atomsstateminimalgenerator}, \Cref{imgcabanfa}, \Cref{atomscaba}).

\begin{corollary}
\label{atomsstateminimalgenerator}
	Let $\alpha_X: \mathcal{P}X \rightarrow \mathcal{H}X$ satisfy $\alpha_X(\varphi)(\psi) = \vee_{x \in X} \varphi(x) \wedge \psi(x)$. If $B = ( X, h )$ is a finite $\mathcal{H}$-algebra, then $( \At(B), i, d )$ with $i(a) = a$ and $d(x) = \lbrace a \in \At(B) \mid a \leq x \rbrace$ is a size-minimal generator for $( X, h \circ \alpha_X )$. 
	\end{corollary}
\begin{proof}
	By \Cref{basisshiftecabapowerset} $( \At(B), i, d )$ is a basis for $( X, h \circ \alpha_X )$. Analogously to \Cref{xorbasisstateminimal}, it follows that any finite basis for a $\mathcal{P}$-algebra is a size-minimal generator.
\end{proof}

\begin{lemma}[\cite{stackex}]
\label{join-irred-dlattice}
	Let $A$ be a sub-lattice of a finite distributive lattice $B$, then $\vert J(A) \vert \leq \vert J(B) \vert$.   
\end{lemma}
\begin{proof}
	For $x \in J(B)$ define $\hat{x} := \bigwedge \lbrace y \in A \mid x \leq y \rbrace \geq x$. To see that $\hat{x} \in J(A)$, assume $\hat{x} = y \vee z$ for $y,z \in A$. By distributivity we have	$
	x = \hat{x} \wedge x = (y \vee z) \wedge x = (y \wedge x ) \vee (z \wedge x)$.
	Since $x \in J(B)$, it thus follows w.l.o.g. $x = y \wedge x$, which implies $x \leq y$. Consequently $\hat{x} \leq y \leq \hat{x}$, i.e. $\hat{x} = y$. Let $z \in J(A)$, then the join-density of join-irreducibles implies
	\[
	z = \vee \lbrace x \in J(B) \mid x \leq z \rbrace = \vee \lbrace \hat{x} \in J(A) \mid x \in J(B) : x \leq z \rbrace .
	\]
	Since $z$ is join-irreducible it follows $z = \widehat{x_z}$ for some $x_z \in J(B)$ with $x_z \leq z$. We thus find
	$
	J(A) = \lbrace \hat{x} \mid x \in J(B) \rbrace$, which implies the claim $\vert J(A) \vert \leq \vert J(B) \vert$.
\end{proof}

\begin{corollary}
\label{atomscaba}
	Let $A$ be a sub-algebra of a finite atomic Boolean algebra $B$. Then $\vert \At(A) \vert \leq  \vert \At(B) \vert$.
\end{corollary}
\begin{proof}
	For atomic Boolean algebras, join-irreducibles and atoms coincide. Every Boolean algebra is in particular a distributive lattice. The claim thus follows from \Cref{join-irred-dlattice}.
\end{proof}

\begin{corollary}
\label{minimalityatomaton}
		The \'atomaton for $L$ is size-minimal among non-deterministic automata $\mathcal{Y}$ accepting $L$ with $\overline{\im(\obs^{\dag}_{\mathcal{Y}})}^{\CSL} = \overline{\im(\obs^{\dag}_{\mathcal{Y}})}^{\CABA}$.
\end{corollary}
\begin{proof}
Let $\alpha: \lambda^{\mathcal{H}} \rightarrow \lambda^{\mathcal{P}}$ be the distributive law homomorphism in \Cref{alphapowersetneighbourhooddistrlaw}.
\begin{itemize}
	\item By \Cref{imgcabanfa} and the definition of $\alpha$-closedness, any non-deterministic automaton $\mathcal{Y}$ satisfies $\overline{\im(\obs^{\dag}_{\mathcal{Y}})}^{\CSL} = \overline{\im(\obs^{\dag}_{\mathcal{Y}})}^{\CABA}$ iff it is $\alpha$-closed.
	\item Let $\mathbb{M}$ be the minimal pointed $\lambda^{\mathcal{H}}$-bialgebra accepting $L$. The state-space of $\mathbb{M}$ is the finite set $\overline{\Der(L)}^{\CABA}$. By \Cref{atomsstateminimalgenerator}, the set $\At(\overline{\Der(L)}^{\CABA})$ underlies a size-minimal generator $\mathbb{G}$ for the algebraic part of $\alpha(\mathbb{M})$. The \'atomaton for $L$ can thus be recovered as the $\alpha$-closed non-deterministic automaton $\mathcal{X} := \gen(\alpha(\mathbb{M}), \mathbb{G})$ accepting $L$ in \Cref{minimalitytheorem}. In particular, by the first bullet point above, the \'atomaton thus lives in the class of non-deterministic automata $\mathcal{Y}$ accepting $L$ with $\overline{\im(\obs^{\dag}_{\mathcal{Y}})}^{\CSL} = \overline{\im(\obs^{\dag}_{\mathcal{Y}})}^{\CABA}$.
	\item Let $\mathcal{Y}$ be any non-deterministic automaton accepting $L$ with $\overline{\im(\obs^{\dag}_{\mathcal{Y}})}^{\CSL} = \overline{\im(\obs^{\dag}_{\mathcal{Y}})}^{\CABA}$ and state-space $Y$. By construction, there exists an epimorphism $\obs_{\expa_{\mathcal{P}}(\mathcal{Y})}: \mathcal{P}Y \twoheadrightarrow \overline{\im(\obs^{\dag}_{\mathcal{Y}})}^{\CSL}$, which turns $Y$ into a generator for the finite $\mathcal{P}$-algebra $B := \overline{\im(\obs^{\dag}_{\mathcal{Y}})}^{\CSL} = \overline{\im(\obs^{\dag}_{\mathcal{Y}})}^{\CABA}$. As for CABAs join-irreducibles and atoms coincide, the size-minimality of join-irreducibles in \Cref{joinirreducstateminimal} thus implies
	$\vert \At(B) \vert \leq \vert Y \vert$. By \Cref{minimalitytheorem}, we have that $\im(\obs_{\exp_T(\mathcal{X})}) \subseteq B$, where $\mathcal{X}$ denotes the \'atomaton. From \Cref{obsgenerated} and the definition of $\mathcal{X}$ it follows that $\im(\obs_{\exp_T(\mathcal{X})}) \cong \overline{\Der(L)}^{\CABA}$. We thus have $\overline{\Der(L)}^{\CABA} \subseteq B$. By \Cref{atomscaba}, it follows $\vert \At(\overline{\Der(L)}^{\CABA}) \vert \leq \vert \At(B) \vert$. Consequently we can deduce $\vert  \At(\overline{\Der(L)}^{\CABA}) \vert \leq \vert Y \vert$, which shows that the \'atomaton is size-minimal. \qedhere
	\end{itemize}
	  \end{proof}

The above result can be shown to be similar to \cite[Theorem 4.9]{MyersAMU15}, which characterises the \'atomaton as size-minimal among non-deterministic automata whose accepted languages are \emph{closed under complement}. 

\Cref{distromatonminimal} below is very similar to a characterisation of the distromaton as size-minimal among non-deterministic automata whose accepted languages are \emph{closed under intersection} \cite[Theorem 4.13]{MyersAMU15}. The proof of \Cref{distromatonminimal} requires the following technical result.

\begin{corollary}
\label{imgdistromaton}
	Let $\alpha: \mathcal{P}X \rightarrow \mathcal{A}X$ satisfy $\alpha_X(\varphi)(\psi) = \vee_{x \in X} \varphi(x) \wedge \psi(x)$. For any unpointed non-deterministic automaton $\mathcal{X}$ it holds:
	\begin{itemize}
		\item $\img(\obs_{\expa_{\mathcal{P}}(\mathcal{X})}) = \overline{\im(\obs^{\dag}_{\mathcal{X}})}^{\CSL}$;
		\item $\img(\obs_{\alpha(\expa_{\mathcal{A}}(\ext(\mathcal{X})))}) = \overline{\im(\obs^{\dag}_{\mathcal{X}})}^{\CDL}$.
	\end{itemize}
\end{corollary}
\begin{proof}
	Analogous to the proof of \Cref{imgcabanfa}.
\end{proof}

\begin{corollary}
\label{distromatonminimal}
		The distromaton for $L$ is size-minimal among non-deterministic automata $\mathcal{Y}$ accepting $L$ with $ \overline{\im(\obs^{\dag}_{\mathcal{Y}})}^{\CSL} = \overline{\im(\obs^{\dag}_{\mathcal{Y}})}^{\CDL}$.
\end{corollary}
\begin{proof}
	Analogous to the proof of \Cref{minimalityatomaton}. Specific to this case are the results  \Cref{powersetalgebra}, \Cref{monotoneneighbourhoodoutputalgebra}, \Cref{neighbourhoodpowersetmorphism}, \Cref{joinirreducstateminimal}, \Cref{join-irred-dlattice}, and \Cref{imgdistromaton}.
\end{proof}

The size-minimality result (\Cref{corminimalxorcaba}) for the newly discovered minimal xor-CABA automaton is analogous to the ones for the \'atomaton and the distromaton. It requires the following technical result.

\begin{corollary}
\label{imgxorcaba}
	Let $\alpha: \mathcal{R}X \rightarrow \mathcal{H}X$ satisfy $\alpha_X(\varphi)(\psi) = \bigoplus_{x \in X} \varphi(x) \cdot \psi(x)$. For any unpointed $\mathbb{Z}_2$-weighted automaton $\mathcal{X}$ it holds:
	\begin{itemize}
		\item $\img(\obs_{\expa_{\mathcal{R}}(\mathcal{X})}) = \overline{\im(\obs^{\dag}_{\mathcal{X}})}^{\mathbb{Z}_2\Vect}$;
		\item $\img(\obs_{\alpha(\expa_{\mathcal{H}}(\ext(\mathcal{X})))}) = \overline{\im(\obs^{\dag}_{\mathcal{X}})}^{\CABA}$.
	\end{itemize}
\end{corollary}
\begin{proof}
	Analogous to the proof of \Cref{imgcabanfa}.
\end{proof}

\begin{corollary}
\label{corminimalxorcaba}
		The minimal xor-CABA automaton for $L$ is size-minimal among $\mathbb{Z}_2$-weighted automata $\mathcal{Y}$ accepting $L$ with $ \overline{\im(\obs^{\dag}_{\mathcal{Y}})}^{\mathbb{Z}_2\Vect} = \overline{\im(\obs^{\dag}_{\mathcal{Y}})}^{\CABA}$.
\end{corollary}
\begin{proof}
	Analogous to the proof of \Cref{minimalityatomaton}. Specific to this case are \Cref{neighbourhoodalgebra}, \Cref{xoroutputalgebra}, \Cref{alphaxorneighbourhooddistrlaw}, \Cref{xorbasisstateminimal}, \Cref{imgxorcaba} and the observation that if $A \subseteq B$ is a subvector space of a finite vector space $B$, then $\dimension(A) \leq \dimension(B)$.
\end{proof}

We conclude with a size-comparison between acceptors that is parametric in the closure of derivatives.

\begin{corollary}
\label{sizecomparison}
	\begin{itemize}
		\item If $\overline{\Der(L)}^{\mathbb{Z}_2\Vect} = \overline{\Der(L)}^{\CABA}$, then the minimal xor automaton and the minimal xor-CABA automaton for $L$ are of the same size.
		\item If $\overline{\Der(L)}^{\CSL} = \overline{\Der(L)}^{\CDL}$, then the canonical RFSA and the distromaton for $L$ are of the same size.
		\item If $\overline{\Der(L)}^{\CSL} = \overline{\Der(L)}^{\CABA}$, then the canonical RFSA and the \'atomaton for $L$ are of the same size. 
	\end{itemize}
\end{corollary}
\begin{proof}
	\begin{itemize}
		\item By \Cref{minimalxor} the minimal xor automaton $\mathcal{X}$ is of size not greater than the minimal xor-CABA automaton $\mathcal{Y}$. Conversely, we find
		\begin{align*}
			\overline{\Der(L)}^{\CABA} &= \overline{\Der(L)}^{\mathbb{Z}_2\Vect} && \textnormal{(Assumption)} \\
			&= \im(\obs_{\expa_{\mathcal{R}}(\mathcal{X})}) && \textnormal{(\Cref{obsgenerated})}  \\
			&= \overline{\im(\obs^{\dagger}_{\mathcal{X}})}^{\mathbb{Z}_2\Vect} && \textnormal{(\Cref{imgxorcaba})}
		\end{align*}
		which can be used to show $\overline{\im(\obs^{\dagger}_{\mathcal{X}})}^{\mathbb{Z}_2\Vect} = \overline{\im(\obs^{\dagger}_{\mathcal{X}})}^{\CABA}$. By \Cref{corminimalxorcaba} the latter implies that $\mathcal{Y}$ is of size not greater than $\mathcal{X}$, which shows the claim.
	\item Let $\mathcal{X}$ denote the canonical RFSA and $\mathcal{Y}$ the distromaton. On the one hand we find
	\begin{align*}
			\overline{\Der(L)}^{\CSL} &= \overline{\Der(L)}^{\CDL} && \textnormal{(Assumption)} \\
			&= \im(\obs_{\expa_{\mathcal{P}}(\mathcal{Y})}) && \textnormal{(\Cref{obsgenerated})}  \\
			&= \overline{\im(\obs^{\dagger}_{\mathcal{Y}})}^{\CSL} && \textnormal{(\Cref{imgdistromaton})}
		\end{align*}	
		which by \Cref{canonicalrfsaminimal} implies that $\mathcal{X}$ is of size not greater than $\mathcal{Y}$. Conversely, we establish the equality
		\begin{align*}
			\overline{\Der(L)}^{\CDL} &= \overline{\Der(L)}^{\CSL} && \textnormal{(Assumption)} \\
			&= \im(\obs_{\expa_{\mathcal{P}}(\mathcal{X})}) && \textnormal{(\Cref{obsgenerated})}  \\
			&= \overline{\im(\obs^{\dagger}_{\mathcal{X}})}^{\CSL} && \textnormal{(\Cref{imgdistromaton})}
		\end{align*}	
		which can be used to show $\overline{\im(\obs^{\dagger}_{\mathcal{X}})}^{\CSL} = \overline{\im(\obs^{\dagger}_{\mathcal{X}})}^{\CDL}$. By \Cref{distromatonminimal} the latter implies that $\mathcal{Y}$ is of size not greater than $\mathcal{X}$, which shows the claim.
		\item The proof for the \'atomaton is analogous to the proof for the distromaton in the previous point. \qedhere
	\end{itemize}
\end{proof}

\section{Related Work}

\label{relatedwork}

One of the motivations for our work are active learning algorithms for the derivation of succinct state-based models \cite{angluin1987learning}. A major challenge in learning non-deterministic models is the lack of a canonical target acceptor for a given language \cite{denis2001residual}. The problem has been independently approached for different variants of non-determinism, often with the idea of finding a subclass admitting a unique representative \cite{esposito2002learning,berndt2017learning} such as the canonical RFSA, the minimal xor automaton, or the \'atomaton. 

 A more general and unifying perspective on learning automata that may not have a canonical target was given by Van Heerdt \cite{van2020learning, van2016master, van2020phd}. One of the central notions in this work is the concept of a scoop, originally introduced by Arbib and Manes \cite{arbib1975fuzzy} and here referred to as a generator. The main contribution in \cite{van2020learning} is a general procedure to find irreducible sets of generators, which thus restricts the work to the category of sets. In this chapter we generally work over arbitrary categories, although we assume the existence of a minimal set-based generator in \Cref{minimalitytheorem}. Furthermore, the work of Van Heerdt has no size-minimality results.

Closely related to the content of this chapter is the work of Myers et al.\ \cite{MyersAMU15}, who present a coalgebraic construction for canonical non-deterministic automata. They cover the canonical RFSA, the minimal xor automaton, the \'atomaton, and the distromaton. The underlying idea in \cite{MyersAMU15} for finding succinct representations is similar to ours: first they build the minimal DFA for a regular language in a locally finite variety, then they apply an equivalence between the category of finite algebras and a suitable category of finite structured sets and relations.
On the one hand, the category of finite algebras in a locally finite variety can be translated into our setting by considering a category of algebras over a monad preserving finite sets. In fact, modulo this translation, many of the categories considered here already appear in \cite{MyersAMU15}, e.g.\ vector spaces, Boolean algebras, complete join-semi lattices, and distributive lattices. On the other hand, their construction seems to be restricted to the category of sets and non-deterministic automata, while we work over arbitrary monads on arbitrary categories. Their work does not provide a general algorithm to construct a succinct automaton, i.e., the specifics vary with the equivalences considered, while we give a general definition and a soundness argument in \Cref{generatorbialgebrahom}. While Myers et al. \cite{MyersAMU15} give minimality results for a wide range of acceptors, each proof follows case-specific arguments. In \Cref{minimalitytheorem} we provide a unifying minimality result for succinct automata that implies \Cref{minimalityatomaton} (cf. \cite[Theorem 4.9]{MyersAMU15}), \Cref{minimalxor} (cf. \cite[Theorem 4.10]{MyersAMU15}), \Cref{canonicalrfsaminimal} (cf. \cite[Corollary 4.11]{MyersAMU15}), \Cref{distromatonminimal} (cf. \cite[Theorem 4.13]{MyersAMU15}).

\section{Discussion and Future Work}

\label{discussion}

We have presented a general categorical framework based on bialgebras and distributive law homomorphisms for the derivation of canonical automata. The framework instantiates to a wide range of well-known examples from the literature and allowed us to discover a previously unknown canonical acceptor for regular languages. Finally, we presented a theorem that subsumes previously independently proven minimality results for canonical acceptors, implied new characterisations, and allowed us to make size-comparisons between canonical automata. 

In the future, we would like to cover other examples, such as the canonical probabilistic RFSA \cite{esposito2002learning} and the canonical alternating RFSA \cite{berndt2017learning, angluin2015learning}.
Probabilistic automata of the type in \cite{esposito2002learning} are typically modelled as $TF$-coalgebras instead of $FT$-coalgebras \cite{jacobs2012trace}, and thus will need a shift in perspective.
For alternating RFSAs we expect a canonical form can be constructed in the spirit of this work, from generators for algebras over the neighbourhood monad, by interpreting the join-dense atoms of a CABA as a full meet of ground elements. 

Generally, it would be valuable to have a more systematic treatment of the range of available monads and distributive law homomorphisms \cite{zwart2019no}, making use of the fact that distributive law homomorphisms compose.

Further generalisation in another direction could be achieved by distributive laws between monads and endofunctors on different categories. For instance, we expect that operations on automata as the product can be captured by homomorphisms between distributive laws of such more general type. 

Finally, we would like to lift existing double-reversal characterisations of the minimal DFA \cite{brzozowski1962canonical}, the \'atomaton \cite{BrzozowskiT14}, the distromaton \cite{MyersAMU15}, and the minimal xor automaton \cite{VuilleminG210} to general canonical automata. The work in \cite{bonchi2012brzozowski, bonchi2014algebra} gives a coalgebraic generalisation of Brzozowski's algorithm based on dualities between categories, but does not cover the cases we are interested in. The framework in \cite{adamek2012coalgebraic} recovers the \'atomaton as the result of a minimisation procedure, but does not consider other canonical acceptors.

\chapter{Generating Monadic Closures}

  We have seen that every regular language admits a unique size-minimal deterministic acceptor, while there can be several size-minimal non-deterministic acceptors that are not isomorphic (cf. \Cref{example_nonisomorphic_nfa_canonical_automata_chapter}). We also have seen that to tackle this issue, authors have identified a number of sub-classes of non-deterministic automata, all admitting canonical minimal representatives. In \Cref{canonical_automata_chapter} we demonstrated that such representatives can generally be recovered categorically in two steps. First, one constructs the minimal bialgebra accepting a given regular language, by closing the minimal coalgebra with additional algebraic structure over a monad. Second, one identifies canonical generators for the algebraic part of the bialgebra, to derive an equivalent coalgebra with side effects in a monad. In this chapter, we further develop the general theory underlying these two steps. On the one hand, we show that deriving a minimal bialgebra from a minimal coalgebra can be realized by applying a monad on an appropriate category of subobjects. On the other hand, we explore the abstract theory of generators and bases for algebras over a monad.

\section{Introduction}

Recall that the framework for the construction of canonical automata in \Cref{canonical_automata_chapter} adopts the well-known representation of automata as coalgebras (\Cref{def:coalgebra}) and side-effects like non-determinism as monads (\Cref{def:monad}). For instance, an NFA (without initial states) is represented as a coalgebra $k \colon X \to 2 \times \Pow(X)^A$ with side-effects in the powerset monad $\Pow$. 

To derive canonical non-deterministic acceptors, we suggested an idea that is closely related to the (generalised) \emph{powerset construction} of \Cref{determinisationexample}. Under the latter construction, a coalgebra $k: X \rightarrow FTX$ with dynamics in a functor $F$ and side-effects in a monad $T$ is transformed into an equivalent coalgebra $k^{\sharp}: TX \rightarrow FTX$ \cite{silva2010generalizing}. The deterministic automata resulting from such determinisation constructions have \emph{additional algebraic structure}: the state-space $TX$ defines a (free) algebra for the monad $T$, and $k^\sharp$ is a $T$-algebra homomorphism, thus constituting a \emph{bialgebra} over a distributive law relating $F$ and $T$ (\Cref{bialgebradef}). Using the powerset construction, we were able to obtain a canonical succinct acceptor for a regular language $L \subseteq A^*$ by consecutively following two steps.

First, we constructed the minimal\footnote{Minimal in the sense that every state is reachable by an element of $A^*$ and no two different states observe the same language.} pointed coalgebra $\mathsf{M}_L$ for the functor $F = 2 \times (-)^{A}$ accepting $L$. We then equipped the former with additional algebraic structure in a monad $T$ by applying the determinisation procedure to $\textsf{M}_L$ when seen as coalgebra with trivial side-effects in $T$ (cf. \Cref{freeexpfunctor}). By identifying semantically equivalent states, we then derived the minimal\footnote{Minimal in the sense that every state is reachable by an element of $T(A^*)$ and no two different states observe the same language.} (pointed) bialgebra for $L$.

Second, we exploited the additional algebraic structure underlying the minimal bialgebra for $L$ to  ``reverse'' the generalised determinisation procedure. That is, we identified canonical generators (\Cref{generatordefinition}) to derive an equivalent succinct automaton with side-effects in $T$ (cf. \Cref{forgenerator-isharp-is-bialgebra-hom}).

 In this chapter, we further develop the general theory underlying these two steps by making the following contributions, respectively: 
\begin{itemize}
	\item We generalise the closure of a subset of an algebraic structure with respect to the latter as a functor  between categories of subobjects relative to a factorisation system (\Cref{functorprop}). We then equip the functor with the structure of a monad (\Cref{inducedmonad}). 
	 We investigate the closure of a particular subclass of subobjects: the ones that arise by taking the image of a morphism (\Cref{closureofimage}). We show that deriving a minimal bialgebra from a minimal coalgebra can be realized by applying the monad to a subobject in this particular class (\Cref{closure_of_coalg}). 
	 	\item We define a category of algebras with generators (\Cref{generatorcategory}), which is in adjunction with the category of Eilenberg-Moore algebras (\Cref{galgfreeforgetful}), and, under certain assumptions, monoidal (\Cref{monoidalproduct}).
We generalise the matrix representation theory of vector spaces (\Cref{representationtheorysec}) and discuss bases for bialgebras, which are algebras over a particular monad (\Cref{basesforbialgebrasec}).
We compare our ideas with an approach that generalises bases as coalgebras (\Cref{basesascoaglebrassec}). We find that a basis in our sense induces a basis the sense of \cite{jacobs2011bases} (\Cref{impliesjacobsbasis}), and identify assumptions under which the reverse is true, too.
	We characterise generators for finitary varieties in the sense of universal algebra (\Cref{sigmatermlemma}) and relate our work to the theory of locally finitely presentable categories (\Cref{finitarymonadlemma}).
\end{itemize}

\section{Step 1: Closure}
\label{closure_sec}

In this section, we further explore the categorical construction of minimal canonical acceptors given in \Cref{canonical_automata_chapter}. In particular, we show that deriving a minimal bialgebra from a minimal coalgebra by closing the latter with additional algebraic structure has a direct analogue in universal algebra: taking the closure of a subset of an algebra.

\subsection{Factorisation Systems and Subobjects}

In the category of sets and functions, every morphism can be factored into a surjection onto its image followed by an injection into the codomain of the morphism. In this section, we recall an abstraction of this phenomenon for arbitrary categories. The ideas are well established \cite{bousfield1977constructions,riehl2008factorization,maclane1950duality}. We choose to adapt the formalism of \cite{adamek2009abstract}.

\begin{definition}[Factorisation System]
\label{factorisation_system_def}
Let $\mathscr{E}$ and $\mathscr{M}$ be classes of morphisms in a category $\mathscr{C}$. We call the pair $(\mathscr{E}, \mathscr{M})$ a \emph{factorisation system} for $\mathscr{C}$ if the following three conditions hold:
\begin{itemize}
	\item[(F1)] Each of $\mathscr{E}$ and $\mathscr{M}$ is closed under composition with isomorphisms.
	\item[(F2)] Each morphism $f$ in $\mathscr{C}$ can be factored as $f = m \circ e$, with $e \in \mathscr{E}$ and $m \in \mathscr{M}$.
	\item[(F3)] For each commutative square with $e \in \mathscr{E}$ and $m \in \mathscr{M}$ as on the left below
	 \begin{equation*}
			\begin{tikzcd}
				\cdot \arrow{r}{e} \arrow{d}[left]{f} &\cdot \arrow{d}{g} \\
				\cdot \arrow{r}[below]{m} &\cdot
			\end{tikzcd}
			\qquad
			\begin{tikzcd}
							\cdot \arrow{r}{e} \arrow{d}[left]{f} &\cdot \arrow{d}{g} \arrow[dashed]{dl}{d} \\
				\cdot \arrow{r}[below]{m} &\cdot
			\end{tikzcd}
		\end{equation*}
	there exists a unique diagonal $d$ such that the diagram on the right above commutes.
\end{itemize}	
\end{definition}

We will use double headed ($\twoheadrightarrow$) and hooked ($\hookrightarrow$) arrows to indicate that a morphism is in $\mathscr{E}$ and $\mathscr{M}$, respectively. If $f$ factors into $e$ and $m$, we call the codomain of $e$, or equivalently, the domain of $m$, the \emph{image} of $f$ and denote it by $\im(f)$.

One can show that each of $\mathscr{E}$ and $\mathscr{M}$ contains all isomorphisms and is closed under composition \cite[Prop. 14.6]{adamek2009abstract}.
From the uniqueness condition on the diagonal one can deduce that factorisations are unique up to unique isomorphism \cite[Prop. 14.4]{adamek2009abstract}. 
It further follows that $\mathscr{E}$ has the \emph{right cancellation property}, that is $g \circ f \in \mathscr{E}$ and $f \in \mathscr{E}$ implies $g \in \mathscr{E}$. Dually, $\mathscr{M}$ has the \emph{left cancellation property}, that is, $g \circ f \in \mathscr{M}$ and $g \in \mathscr{M}$ implies $f \in \mathscr{M}$ \cite[Prop. 14.9]{adamek2009abstract}.

As intended, in the category of sets and functions, surjective and injective functions, or equivalently, epi- and monomorphisms, constitute a factorisation system \cite[Ex. 14.2]{adamek2009abstract}. More involved examples can be constructed for e.g. the category of topological spaces or the category of categories \cite[Ex. 14.2]{adamek2009abstract}. We are particularly interested in factorisation systems for the category of algebras over a monad.

The naive categorical generalisation of a subset $Y \subseteq X$ is a monomorphism $Y \rightarrow X$. Since in the category of sets epi- and monomorphism constitute a factorisation system, we may generalise subsets to arbitrary categories $\mathscr{C}$ with a factorisation system $(\mathscr{E}, \mathscr{M})$ in the following way:   

\begin{definition}[Subobjects]
\label{subobject_def}
A \emph{subobject} of an object $X \in \mathscr{C}$ is a morphism $m_Y: Y \hookrightarrow X \in \mathscr{M}$. A morphism $f: m_{Y_1} \rightarrow m_{Y_2}$ between subobjects of $X$ consists of a morphism $f: Y_1 \rightarrow Y_2$ such that $m_{Y_2} \circ f = m_{Y_1}$. 
\end{definition}

The category of (isomorphism classes of) subobjects of $X$ is denoted by $\Sub(X)$.

 As $\mathscr{M}$ has the left cancellation property, every morphism between subobjects in fact lies in $\mathscr{M}$. We work with isomorphism classes of subobjects since factorisations of morphisms are only defined up to unique isomorphism. For epi-mono factorisations, there is at most one morphism between any two subobjects, that is, $\Sub(X)$ is simply a partially ordered set.

\subsection{Factorising Algebra Homomorphisms}

 \label{factorisationsystemalgebra}

In this section, we recall that if one is given a category $\mathscr{C}$ with a factorisation system $(\mathscr{E}, \mathscr{M})$ and a monad $T$ on $\mathscr{C}$ that preserves $\mathscr{E}$ (i.e. satisfies $T(e) \in \mathscr{E}$ for all $e \in \mathscr{E}$), it is possible to lift the factorisation system of the base category $\mathscr{C}$ to a factorisation system on the category of Eilenberg-Moore algebras $\EM$. The result appears in e.g. \cite{wissmann2022minimality} and may be extended to algebras over an endofunctor. Alternatively, it can be stated in its dual version: if an endofunctor on $\mathscr{C}$ preserves $\mathscr{M}$, it is possible to lift the factorisation system of $\mathscr{C}$ to the category of coalgebras \cite{kurzlogics, wissmann2022minimality}. 

The factorisation system we propose for $\EM$ consists of those algebra homomorphisms, whose underlying morphism lies in $\mathscr{E}$ or $\mathscr{M}$, respectively. Clearly such a system preserves (F1). The next result shows that it also satisfies (F3). Note that the statement is slightly more general, as it holds not just for Eilenberg-Moore algebras over the monad $T$, but for algebras over the underlying endofunctor.

\begin{lemma}[{\cite[Lem. 3.6]{wissmann2022minimality}}]
\label{diagonal_algebras}
For each commutative square of homomorphisms between algebras for the endofunctor $T$ as on the left below
\begin{equation*}
		\begin{tikzcd}
	(A, h_A) \arrow[twoheadrightarrow]{r}{e} \arrow{d}[left]{f} & (B, h_B) \arrow{d}{g}  \\
	(C, h_C) \arrow[hookrightarrow]{r}[below]{m} &(D, h_D)
\end{tikzcd}
\qquad
	\begin{tikzcd}
	(A, h_A) \arrow[twoheadrightarrow]{r}{e} \arrow{d}[left]{f} & (B, h_B) \arrow{d}{g} \arrow[dashed]{dl}{d} \\
	(C, h_C) \arrow[hookrightarrow]{r}[below]{m} &(D, h_D)
\end{tikzcd}
\end{equation*}
there exists a unique diagonal $d: (B, h_B) \rightarrow (C, h_C)$ such that the diagram on the right above commutes.
\end{lemma}	
\begin{proof}
The proof for {\cite[Lem. 3.6]{wissmann2022minimality}} consists of a corollary to the dual statement for coalgebras \cite[Lem. 3.3]{wissmann2022minimality}. Below we offer an explicit version for algebras.

Since any commuting diagram of algebra homomorphism projects to a commuting diagram in $\mathscr{C}$, the factorisation system of $\mathscr{C}$ implies the existence of a unique diagonal $d$ in $\mathscr{C}$. It remains to show that $d$ is an algebra homomorphism, that is, we need to establish the following identity:
\[
h_C \circ Td = d \circ h_B.
\]
	To this end, we observe that since the following two diagrams commute
		
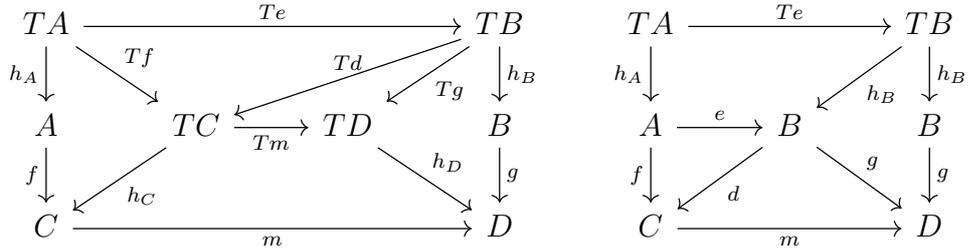
\begin{equation*}
		\begin{tikzcd}
			TA \arrow{dr}{Tf} \arrow{d}[left]{h_A} \arrow{rrr}{Te} & & & TB \arrow{dl}{Tg} \arrow{dll}[above]{Td} \arrow{d}{h_B} \\
			A \arrow{d}[left]{f} & TC \arrow{dl}{h_C} \arrow{r}[below]{Tm} & TD \arrow{dr}{h_D} & B \arrow{d}{g} \\
			C \arrow{rrr}[below]{m} &&& D
		\end{tikzcd}
		\qquad
					\begin{tikzcd}
				TA \arrow{rr}{Te} \arrow{d}[left]{h_A} & & TB \arrow{dl}{h_B} \arrow{d}{h_B} \\
				A \arrow{d}[left]{f} \arrow{r}{e} & B \arrow{dl}{d} \arrow{dr}{g} & B \arrow{d}{g} \\
				C \arrow{rr}[below]{m} & & D		
					\end{tikzcd}
		\end{equation*}
	both $h_C \circ Td$ and $d \circ h_B$ are solutions to the unique diagonal below:
	\begin{equation*}
						\begin{tikzcd}							TA \arrow[twoheadrightarrow]{r}{Te} \arrow{d}[left]{f \circ h_A} & TB \arrow[dashed]{dl}{} \arrow{d}{g \circ h_B} \\
							C \arrow[hookrightarrow]{r}[below]{m} & D
						\end{tikzcd}. 
	\end{equation*}
\end{proof}

Let us now show that the proposed factorisation system satisfies (F2). Assume we are given a homomorphism $f$ as on the left of \Cref{monadpreserveepilift}. Using the factorisation system of the base category $\mathscr{C}$, we can factorise it, as ordinary morphism, into $e \in \mathscr{E}$ and $m \in \mathscr{M}$. In consequence the outer square of the diagram on the right of \Cref{monadpreserveepilift} commutes. Since by assumption the morphism $Te$ is again in $\mathscr{E}$, we thus find a unique diagonal $h_{\im(f)}$ in $\mathscr{C}$ that makes the triangles on the right of \Cref{monadpreserveepilift}  commute. The result below shows that $h_{\im(f)}$ equips $\im(f)$ with the structure of a $T$-algebra.

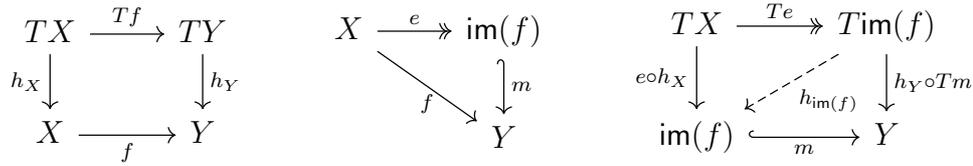
\begin{figure*}
\centering
		\begin{tikzcd}
		TX \arrow{d}[left]{h_X} \arrow{r}{Tf} & TY \arrow{d}{h_Y} \\
		X \arrow{r}[below]{f} & Y
	\end{tikzcd}
	\qquad
	\begin{tikzcd}
		X \arrow[twoheadrightarrow]{r}{e}  \arrow{dr}[below]{f} & \im(f) \arrow[hook']{d}{m} \\
		  & Y 
	\end{tikzcd}
	\qquad
			\begin{tikzcd}
		TX \arrow{d}[left]{e \circ h_X} \arrow[twoheadrightarrow]{r}{Te} & T\im(f) \arrow{d}{h_Y \circ Tm} \arrow[dashed]{dl}{h_{\im(f)}} \\
		\im(f) \arrow[hookrightarrow]{r}[below]{m} & Y
	\end{tikzcd}
\caption{Factorising a $T$-algebra homomorphism via the factorisation system of a base category}
\label{monadpreserveepilift}
\end{figure*}

\begin{lemma}[{\cite[Prop. 3.7]{wissmann2022minimality}}]
	$(\im(f), h_{\im(f)})$ is an Eilenberg-Moore $T$-algebra.
\end{lemma}
\begin{proof}
We need to establish the following two identities
\begin{align*}
	h_{\im(f)} \circ \eta_{\im(f)} &= \id_{\im(f)} \\
	h_{{\im(f)}} \circ \mu_{{\im(f)}} &= h_{\im(f)} \circ Th_{\im(f)}
\end{align*}
where the latter captures that $h_{\im(f)}: (T\im(f), \mu_{{\im(f)}}) \rightarrow (\im(f), h_{\im(f)} )$ is a $T$-algebra homomorphism.

For the first equality, we observe (as in the proof for (\cite[Prop. 3.7]{wissmann2022minimality})) that since the diagram below commutes
\begin{equation*}
	\begin{tikzcd}
					X \arrow{r}{1}  \arrow{dd}[left]{e} & X \arrow{d}[left]{\eta_X} \arrow{rrr}{e} & & & \im(f) \arrow{d}{m} \arrow{dll}{\eta_{\im(f)}} \\
					& TX \arrow{d}[left]{e \circ h_X} \arrow{r}{Te} & T\im(f) \arrow{dl}{h_{\im(f)}}\arrow{r}{Tm} & TY \arrow{dr}{h_Y} & Y \arrow{l}[above]{\eta_Y} \arrow{d}{1} \\
					\im(f) \arrow{r}[below]{1} & \im(f) \arrow{rrr}[below]{m} & & & Y
				\end{tikzcd}
\end{equation*}			
both $h_{\im(f)} \circ \eta_{\im(f)}$ and $\id_{\im(f)}$ are solutions to the unique diagonal in $\mathscr{C}$ below:
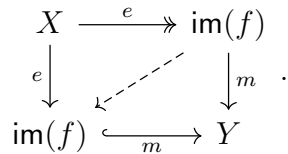
\begin{equation*}
\begin{tikzcd}
				X \arrow{d}[left]{e}	\arrow[twoheadrightarrow]{r}{e}  & \im(f) \arrow{d}{m} \arrow[dashed]{dl}{} \\
				\im(f) \arrow[hookrightarrow]{r}[below]{m} & Y
				\end{tikzcd}.
\end{equation*}
Similarly, for the second equality, we observe that since the following two diagrams commute
\begin{equation*}
		\begin{tikzcd}[column sep = 1em, row sep =1.5em ]
				T^2X \arrow{dd}[left]{\mu_X} \arrow{rrr}{T^2e} & & &	T^2 \im(f) \arrow{ddll}{\mu_{\im(f)}} \arrow{d}{T^2m} \\
				& & & T^2Y \arrow{d}{Th_Y} \arrow{dl}{\mu_Y} \\
		TX \arrow{d}[left]{e \circ h_X} \arrow{r}{Te} & T\im(f) \arrow{d}{h_{\im(f)}} \arrow{r}{Tm} & TY   \arrow{dr}{h_Y} & TY \arrow{d}{h_Y} \\
			\im(f) \arrow{r}[below]{1} & \im(f) \arrow{rr}[below]{m} & &Y	
		\end{tikzcd}
\qquad
	\begin{tikzcd}[column sep = 1em, row sep =1.5em ]
		T^2 X \arrow{r}{1} \arrow{d}[left]{\mu_X} & T^2X \arrow{d}{Th_X} \arrow{r}{T^2e} & T^2 	\im(f) \arrow{ddl}{Th_{\im(f)}} \arrow{d}{T^2m} \\
		TX \arrow{d}[left]{h_X} & TX \arrow{dl}{h_X} \arrow{d}{Te} & T^2Y \arrow{d}{Th_Y} \\
		X \arrow{d}[left]{e} & T\im(f) \arrow{dl}{h_{	\im(f)}} \arrow{r}{Tm} & TY \arrow{d}{h_Y} \\
			\im(f) \arrow{rr}[below]{m} & & Y
	\end{tikzcd}
\end{equation*}
both $h_{_{\im(f)}} \circ \mu_{{\im(f)}}$ and  $h_{\im(f)} \circ Th_{\im(f)}$ are solutions to the unique diagonal below:
\begin{equation*}
\begin{tikzcd}
				T^2X \arrow{d}[left]{e \circ h_X \circ \mu_X}	\arrow[twoheadrightarrow]{r}{T^2e}  & T^2\im(f) \arrow{d}{h_Y \circ Th_Y \circ T^2m} \arrow[dashed]{dl}{} \\
				\im(f) \arrow[hookrightarrow]{r}[below]{m} & Y
				\end{tikzcd}
\end{equation*}
Alternatively (as in the proof for \cite[Prop. 3.7]{wissmann2022minimality}), one may observe that since the following outer square of homomorphisms between algebras for the endofunctor $T$ commutes
\[
\begin{tikzcd}
	(TX, \mu_X) \arrow[twoheadrightarrow]{r}{Te} \arrow{d}[left]{e \circ h_X} & (T\im(f), \mu_{\im(f)}) \arrow{d}{h_Y \circ Tm} \arrow[dashed]{dl}{} \\
	(\im(f), h_{\im(f)}) \arrow[hookrightarrow]{r}[below]{m} & (Y, h_Y)
\end{tikzcd}
\]
\Cref{diagonal_algebras} implies the existence of a unique diagonal algebra homomorphism making the two triangles above commute. As we know that the diagonal coincides with the unique diagonal of the corresponding diagram in $\mathscr{C}$, which is given by  $h_{\im(f)}$, we can deduce that the latter is a $T$-algebra homomorphism.
\end{proof}

We thus obtain the following factorisation of $f$ into Eilenberg-Moore $T$-algebra homomorphisms: 
\[
f = (X,h_X) \overset{e}{\twoheadrightarrow} (\im(f), h_{\im(f)}) \overset{m}{\hookrightarrow} (Y,h_Y).
\]

\subsection{The Subobject Closure Functor}

\label{closureasfunctorsec}

While subobjects in the category of sets generalise subsets, subobjects in the category of algebras generalise subalgebras. By taking the algebraic closure of a subset of an algebra one can thus transition from one category of subobjects to the other. 

In this section, we generalise this phenomenon from the category of sets to more general categories. As before, we assume a base category $\mathscr{C}$ with a factorisation system $(\mathscr{E}, \mathscr{M})$ and a monad $T$ on $\mathscr{C}$ that preserves $\mathscr{E}$. Our aim is to construct, for any $T$-algebra $\mathbb{X}$ with carrier $X$, a functor from the subobjects $\Sub(X)$ in $\mathscr{C}$ to the subobjects $\Sub(\mathbb{X})$ in $\EM$ that assigns to a subobject of $X$ its \emph{closure}, that is, the least $T$-subalgebra of $\mathbb{X}$ containing it.

  Recall the natural isomorphism that witnesses the free Eilenberg-Moore algebra adjunction. For any object $Y$ in $\mathscr{C}$ and $T$-algebra $\mathbb{X} = (X,h)$, it maps a morphism $\varphi: Y \rightarrow X$ to the $T$-algebra homomorphism $\varphi^{\sharp} := h \circ T\varphi : (TY, \mu_Y) \rightarrow (X, h)$. In \Cref{factorisationsystemalgebra} we have seen that the factorisation system of $\mathscr{C}$ naturally lifts to a factorisation system on the category of $T$-algebras. In particular, we know that up to isomorphism the homomorphism $\varphi^{\sharp}$ admits a factorisation of the following form:
  \begin{equation*}
 	\varphi^{\sharp} = (TY, \mu_Y) \overset{e_{\im(\varphi^{\sharp})}}{\twoheadrightarrow} (\im(\varphi^{\sharp}), h_{\im(\varphi^{\sharp})}) \overset{m_{\im(\varphi^{\sharp})}}{\hookrightarrow} (X, h).
 \end{equation*}
 In the case that the morphism $\varphi$ is given by a subobject \[ m_Y: Y \rightarrow X \in \mathscr{M},\] the above construction yields a second subobject \[ m_{\overline{Y}}: \overline{Y} \rightarrow \mathbb{X} \in \mathscr{M},\] where $\overline{Y} := (\im(m_Y^{\sharp}), h_{\im(m_Y^{\sharp})})$.
  Since for a morphism $f: m_{Y_1} \rightarrow m_{Y_2}$ between subobjects of $X$ the diagram on the left of \eqref{subobjectfunctor_onmorphism} below commutes, there further exists a unique diagonal algebra homomorphism $\overline{f}: m_{\overline{Y_1}} \rightarrow m_{\overline{Y_2}}$ between subobjects of $\mathbb{X}$ making the two triangles on the right of \eqref{subobjectfunctor_onmorphism} commute:
  
 \begin{equation}
 \label{subobjectfunctor_onmorphism}
  	 	\begin{tikzcd}[column sep = 1em]  	 	(TY_1, \mu_{Y_1}) \arrow[twoheadrightarrow]{rr}{e_{\overline{Y_1}}}
  	 	\arrow{dr}{Tm_{Y_1}} \arrow{d}[left]{Tf} && (\overline{Y_1}, h_{\overline{Y_1}}) \arrow{dd}[right]{m_{\overline{Y_1}}} \\
  	 	(TY_2, \mu_{Y_2}) \arrow{d}[left]{e_{\overline{Y_2}}} \arrow{r}[below]{Tm_{Y_2}} & (TX, \mu_X) \arrow{dr}{h} \\
  	 	(\overline{Y_2}, h_{\overline{Y_2}})  \arrow[hookrightarrow]{rr}[below]{m_{\overline{Y_2}}} && (X,h)
  	 	\end{tikzcd}   	 
  	 	\quad
  	 	 	\begin{tikzcd}[column sep = 1em]
 	 		(TY_1, \mu_{Y_1}) \arrow[twoheadrightarrow]{r}{e_{\overline{Y_1}}} \arrow{d}[left]{e_{\overline{Y_2}} \circ Tf} & (\overline{Y_1}, h_{\overline{Y}}) \arrow{d}{m_{\overline{Y_1}}} \arrow[dashed]{dl}{\overline{f}} \\
			(\overline{Y_2}, h_{Y_2}) \arrow[hookrightarrow]{r}[below]{m_{\overline{Y_2}}} &(X,h)
 	 	\end{tikzcd}			.
 \end{equation}
 
 The following result shows that the above constructions are compositional.

\begin{proposition}
\label{functorprop}
The assignments $m_Y \mapsto m_{\overline{Y}}$ and $f \mapsto \overline{f}$ yield a functor \\$\overline{(\cdot)}^{\mathbb{X}}: \Sub(X) \rightarrow \Sub(\mathbb{X})$.
\end{proposition} 
\begin{proof}
We need to establish the identities 
\begin{align*}
	\overline{\id_{Y}} = \id_{\overline{Y}} \quad \textnormal{and} \quad 
	\overline{f \circ g} = \overline{f} \circ \overline{g}.
\end{align*}
As before, the equations follow from the uniqueness of diagonals and the commutativity of the left, respectively right, diagram below:
 \begin{equation*}
 	 	 	 	\begin{tikzcd}[row sep=3em, column sep = 3em]
 	 	 	 			(TY, \mu_Y) \arrow[twoheadrightarrow]{r}{e_{\overline{Y}}} \arrow{d}[left]{ e_{\overline{Y}} \circ T(\id_Y) } & (\overline{Y}, h_{\overline{Y}}) \arrow{d}{m_{\overline{Y}}} \arrow{dl}{\id_{\overline{Y}}} \\
		(\overline{Y}, h_{\overline{Y}}) \arrow[hookrightarrow]{r}[below]{m_{\overline{Y}}} & (X,h)
 	 	 	 	\end{tikzcd}
 	 	 	 	\qquad
 	 	 	 	 	 	 	\begin{tikzcd}[row sep=3em, column sep = 3em]
 	 	 	 	 	 	 		(TY_1, \mu_{Y_1}) \arrow[twoheadrightarrow]{rr}{e_{\overline{Y_1}}} \arrow{d}[left]{Tg} && (\overline{Y_1}, h_{\overline{Y_1}}) \arrow{ddd}{m_{\overline{Y_1}}}  \arrow{dl}{\overline{g}} \\
 	 	 	 	 	 	 		(TY_2, \mu_{Y_2}) \arrow{d}[left]{Tf} \arrow{r}{e_{\overline{Y_2}}} & (\overline{Y_2}, h_{\overline{Y_2}}) \arrow{ddl}{\overline{f}} \arrow{ddr}{m_{\overline{Y_2}}} & \\
 	 	 	 	 	 	(TY_3, \mu_{Y_3}) \arrow{d}[left]{e_{\overline{Y_3}}} & & \\
 	 	 	 	 	 	(\overline{Y_3}, h_{\overline{Y_3}}) \arrow[hookrightarrow]{rr}[below]{m_{\overline{Y_3}}} & & (X, h)
 	 	 	 	 	 	 	\end{tikzcd}.
 \end{equation*}
  \end{proof}
  
   Let us instantiate above result for the free vector space monad on the category of sets equipped with its canonical surjective-injective factorisation system.
 
 \begin{example}
 \label{freevectorspaceexamplefunctor}
 Given an injective embedding $m_Y$ of some set $Y$ into a vector space $\mathbb{V}$, one easily verifies that the lifting $(m_Y)^{\sharp}$ maps a formal linear combination $\sum_i \lambda_i \cdot y_i$ to the vector $\sum_i \lambda_i \cdot m_Y(y_i)$ in $V$. The image $\overline{Y}$ is thus given by the vector space that consists of equivalence classes of formal linear combinations, and the injection $m_{\overline{Y}}$ interprets representatives as demonstrated. In particular, if $Y$ is a subset of $\mathbb{V}$, that is, $m_Y(y) = y$ for all $y \in Y$, the closure can be recognised as the sub vector space of $\mathbb{V}$ generated by $Y$.
 \end{example}
   
Mapping an algebra homomorphism with codomain $\mathbb{X}$ to the $\mathscr{M}$-part of its factorisation extends to a functor from the slice category\footnote{The \emph{slice category} $\mathscr{C}/X$ of a category $\mathscr{C}$ over an object $X \in \mathscr{C}$ has as objects (isomorphism classes of) morphisms $g_Y: Y \rightarrow X$ with codomain $X$, and a morphism $f: g_{Y_1} \rightarrow g_{Y_2}$ consists of a morphism $f: Y_1 \rightarrow Y_2$ in $\mathscr{C}$ that satisfies $g_{Y_1} = g_{Y_2} \circ f$.} over $\mathbb{X}$ to the category of subobjects of $\mathbb{X}$.  
Similarly, one observes that the free Eilenberg-Moore algebra adjunction gives rise to a functor from the slice category over $X$ to the slice category over $\mathbb{X}$. Finally, it is clear that there exists a functor from the category of subobjects of $X$ to the slice category over $X$. The functor defined in \Cref{functorprop} can thus be recognised as the following composition:
\begin{equation}
\label{closure_as_composition}
\begin{tikzcd}
 			\mathscr{C}/X \arrow{r}{} & \EM/\mathbb{X} \arrow{d}{} \\
 			\Sub(X) \arrow{u}{} \arrow{r}[above]{\overline{(\cdot)}^{\mathbb{X}}} & \Sub(\mathbb{X})
 		\end{tikzcd}.		
\end{equation}

 \subsection{The Subobject Closure Monad}
 
\label{closureasmonadsec} 
 
 In this section, we show that the functor $\overline{(\cdot)}^{\mathbb{X}}$ in \Cref{functorprop} induces a monad on the category of subobjects $\Sub(X)$. 
 As before, we assume a base category $\mathscr{C}$ with a factorisation system $(\mathscr{E}, \mathscr{M})$ and a monad $T = (T, \eta, \mu)$ on $\mathscr{C}$ that preserves $\mathscr{E}$. 

 We begin by establishing the following two technical identities, which assume a $T$-algebra $\mathbb{X} = (X,h)$ and a subobject $m_Y: Y \hookrightarrow X \in \mathscr{M}$.

\begin{lemma} 
The following two equalities hold:
 \begin{itemize}
	\item $m_{\overline{Y}} \circ e_{\overline{Y}} \circ \eta_Y = m_Y$
	\item $m_{\overline{Y}} \circ e_{\overline{Y}} \circ \mu_Y = m_{\overline{\overline{Y}}} \circ e_{\overline{\overline{Y}}} \circ Te_{\overline{Y}}$
\end{itemize}
  \end{lemma}
\begin{proof}
  	For the first identity we observe:
  	 \begin{align*}
  	 	m_{\overline{Y}} \circ e_{\overline{Y}} \circ \eta_Y &= m_Y^{\sharp} \circ \eta_Y && (m_{\overline{Y}} \circ e_{\overline{Y}} = m_Y^{\sharp}) \\
  	 	&= h \circ Tm_Y \circ \eta_Y && \textnormal{(Definition of } m_Y^{\sharp}) \\
  	 	&= h \circ \eta_X \circ m_Y && \textnormal{(Naturality of } \eta) \\
  	 	&= m_Y && (h \circ \eta_X = \id_X).
  	 \end{align*}
Similarly, for the second identity we deduce:
  	\begin{align*}
  		m_{\overline{Y}} \circ e_{\overline{Y}} \circ \mu_Y &= m_Y^{\sharp} \circ \mu_Y && (m_{\overline{Y}} \circ e_{\overline{Y}} = m_Y^{\sharp}) \\
  		&= h \circ Tm_Y \circ \mu_Y && \textnormal{(Definition of } m_Y^{\sharp}) \\
  		&= h \circ \mu_X \circ T^2m_Y && \textnormal{(Naturality of } \mu)\\
  		 &= h \circ Th \circ T^2 m_Y && (h \circ \mu_X = h \circ Th) \\
  		 &= h \circ Tm_Y^{\sharp} && \textnormal{(Definition of } m_Y^{\sharp}) \\
  		 &= h \circ Tm_{\overline{Y}} \circ Te_{\overline{Y}} && (m_Y^{\sharp} = m_{\overline{Y}} \circ e_{\overline{Y}}) \\
  		 &= m_{\overline{Y}}^{\sharp} \circ Te_{\overline{Y}} && \textnormal{(Definition of } m_{\overline{Y}}^{\sharp}) \\
  		 &= m_{\overline{\overline{Y}}} \circ e_{\overline{\overline{Y}}} \circ Te_{\overline{Y}} && (m_{\overline{Y}}^{\sharp} = m_{\overline{\overline{Y}}} \circ e_{\overline{\overline{Y}}}). 
  	\end{align*} 
  \end{proof}

  In consequence, we can define candidates for the monad unit $\eta^{\mathbb{X}}$ and the monad multiplication $\mu^{\mathbb{X}}$, respectively, as the unique diagonals below:
 \begin{equation}
 \label{unitmultdef}
 		\begin{tikzcd}[row sep=3em, column sep = 4em]
 			Y \arrow[twoheadrightarrow]{r}{1} \arrow{d}[left]{e_{\overline{Y}} \circ \eta_Y} & Y \arrow[dashed]{dl}{\eta^{\mathbb{X}}_{m_Y}} \arrow{d}{m_Y} \\
 			\overline{Y} \arrow[hookrightarrow]{r}[below]{m_{\overline{Y}}} & X
 		\end{tikzcd} 
 		\qquad
\begin{tikzcd}[row sep=3em, column sep = 4em]
 			T^2Y \arrow[twoheadrightarrow]{r}{e_{\overline{\overline{Y}}} \circ Te_{\overline{Y}}} \arrow{d}[left]{e_{\overline{Y}} \circ \mu_Y} & \overline{\overline{Y}} \arrow[dashed]{dl}{\mu^{\mathbb{X}}_{m_Y}} \arrow{d}{m_{\overline{\overline{Y}}}} \\
 			\overline{Y} \arrow[hookrightarrow]{r}[below]{m_{\overline{Y}}} & X
 		\end{tikzcd}.
 \end{equation} 
  
 By construction both morphisms are homomorphisms of subobjects: 
 \[\eta^{\mathbb{X}}_{m_Y}: m_Y \longrightarrow m_{\overline{Y}} \qquad \mu^{\mathbb{X}}_{m_Y}: m_{\overline{\overline{Y}}} \longrightarrow m_{\overline{Y}}. \] 
 
  The remaining proofs of naturality and the monad laws are covered below. By a slight abuse of notation, we write $\overline{(\cdot)}^{\mathbb{X}}$ for the endofunctor on $\Sub(X)$ that arises by post-composition of the functor in \Cref{functorprop} with the canonical forgetful functor from $\Sub(\mathbb{X})$ to $\Sub(X)$. 
 
 \begin{theorem}
 \label{inducedmonad}
 	$(\overline{(\cdot)}^{\mathbb{X}}, \eta^{\mathbb{X}}, \mu^{\mathbb{X}})$ is a monad on $\textnormal{\Sub}(X)$.
 \end{theorem}
 \begin{proof}
 	On the one hand, we have to establish the naturality of  $\eta^{\mathbb{X}}$ and $ \mu^{\mathbb{X}}$, that is, the identities
 	\begin{align}
 	\label{natproof}
 	 	\begin{split}
 		\overline{f} \circ \eta^{\mathbb{X}}_{m_{Y_1}} &= \eta^{\mathbb{X}}_{m_{Y_2}} \circ f \\
 		\overline{f} \circ \mu^{\mathbb{X}}_{m_{Y_1}} &= \mu^{\mathbb{X}}_{m_{Y_2}} \circ \overline{\overline{f}}
 		\end{split}
 	\end{align} for any subobject homomorphism $f: m_{Y_1} \rightarrow m_{Y_2}$. On the other hand, we need to establish the unitality and associativity laws:
 	\begin{align}
 	\label{unitassoclawproof}
 	\begin{split}
 		\mu^{\mathbb{X}}_{m_Y} \circ \eta^{\mathbb{X}}_{m_{\overline{Y}}} &= \id_{\overline{Y}} = \mu^{\mathbb{X}}_{m_Y} \circ \overline{\eta^{\mathbb{X}}_{m_Y}} \\
 		\mu^{\mathbb{X}}_{m_Y} \circ \mu^{\mathbb{X}}_{m_{\overline{Y}}} &= \mu^{\mathbb{X}}_{m_Y} \circ \overline{\mu^{\mathbb{X}}_{m_Y}}.
 	\end{split}	
 	\end{align}
 	The first equation of \eqref{natproof} follows from the commutativity of the diagram below:
 	\begin{equation*}
 				    		\begin{tikzcd}
 				    		Y_1 \arrow{r}{1} \arrow{dd}[left]{\eta^{\mathbb{X}}_{m_{Y_1}}} & Y_1 \arrow{d}{\eta_{Y_1}} \arrow{r}{f} & Y_2 \arrow{d}{\eta_{Y_1}} \arrow{r}{1} & Y_2 \arrow{dd}{\eta^{\mathbb{X}}_{m_{Y_2}}} \\
 				    		& TY_1 \arrow{r}{Tf} \arrow{d}{e_{\overline{Y_1}}} & TY_2 \arrow{dr}{e_{\overline{Y_2}}} & \\
 				    		\overline{Y_1} \arrow{r}[below]{1} & \overline{Y_1} \arrow{rr}[below]{\overline{f}} & & \overline{Y_2}
 				    		\end{tikzcd}.
 	\end{equation*}
 	For the second equation of \eqref{natproof} we observe that since the following two diagrams commute
 	\begin{equation*}
 	 		\begin{tikzcd}
 		 			T^2Y_1 \arrow{d}[left]{\mu_{Y_1}} \arrow{rr}{e_{\overline{\overline{Y_1}}} \circ T e_{\overline{Y_1}}} & & \overline{\overline{Y_1}} \arrow{dl}[left]{\mu^{\mathbb{X}}_{m_{Y_1}}} \arrow{dd}{m_{\overline{\overline{Y_1}}}} \\
 		 			TY_1 \arrow{r}{e_{\overline{Y_1}}} \arrow{d}[left]{e_{\overline{Y_2}} \circ Tf} & \overline{Y_1} \arrow{dl}{\overline{f}} \arrow{dr}{m_{\overline{Y_1}}} & \\
 		 			\overline{Y_2} \arrow{rr}[below]{m_{\overline{Y_2}}} & & X
 		 		\end{tikzcd}
 		 		 		\quad
 		\begin{tikzcd}
 			T^2 Y_1 \arrow{r}{1} \arrow{d}[left]{Tf \circ \mu_{Y_1}} & T^2 Y_1 \arrow{d}{T^2f} \arrow{r}{T e_{\overline{Y_1}}} & T\overline{Y_1} \arrow{d}{T\overline{f}} \arrow{r}{e_{\overline{\overline{Y_1}}}} & \overline{\overline{Y_1}} \arrow{ddl}{\overline{\overline{f}}} \arrow{ddd}{m_{\overline{\overline{Y_1}}}} \\
 			TY_2 \arrow{dd}[left]{e_{\overline{Y_2}}} & T^2 Y_2 \arrow{l}{\mu_{Y_2}} \arrow{r}{Te_{\overline{Y_2}}} & T \overline{Y_2}  \arrow{d}[left]{e_{\overline{\overline{Y_2}}}} & \\
 			& & \overline{\overline{Y_2}} \arrow{lld}{\mu^{\mathbb{X}}_{m_{Y_2}}} \arrow{dr}{m_{\overline{\overline{Y_2}}}} & \\
 			\overline{Y_2} \arrow{rrr}[below]{m_{\overline{Y_2}}} & & & X
 		\end{tikzcd}
 	\end{equation*}
 	both $\overline{f} \circ \mu^{\mathbb{X}}_{m_{Y_1}}$ and $\mu^{\mathbb{X}}_{m_{Y_2}} \circ \overline{\overline{f}}$ are solutions to the unique diagonal below:
 	\begin{equation*}
 		  		\begin{tikzcd}
 		  	 		  		T^2Y_1 \arrow[twoheadrightarrow]{r}{e_{\overline{\overline{Y_1}}} \circ T e_{\overline{Y_1}}} \arrow{d}[left]{e_{\overline{Y_2}} \circ Tf \circ \mu_{Y_1}}  & \overline{\overline{Y_1}} \arrow[dashed]{dl}{}	\arrow{d}{m_{\overline{\overline{Y_1}}}} \\
 		  		\overline{Y_2} \arrow[hookrightarrow]{r}[below]{m_{\overline{Y_2}}} & X
 		  		\end{tikzcd}.
 	\end{equation*}
 	For the first equation of \eqref{unitassoclawproof} we observe that since the following diagrams commute
 	\begin{equation*}
 		    		\begin{tikzcd}[column sep = 2em, row sep = 1em]
 		    TY \arrow{r}{1} \arrow{dd}[left]{1} & TY \arrow{d}[left]{\eta_{TY}} \arrow{r}{e_{\overline{Y}}} &  \overline{Y} \arrow{d}[left]{\eta_{\overline{Y}}}  \arrow{r}{1}  & \overline{Y} \arrow{ddl}{\eta^{\mathbb{X}}_{m_{\overline{Y}}}} \arrow{ddd}{m_{\overline{Y}}} \\
 		    & T^2Y \arrow{r}{T e_{\overline{Y}}} \arrow{d}[left]{\mu_Y} & T\overline{Y} \arrow{d}[left]{e_{\overline{\overline{Y}}}} & & \\
 		    TY \arrow{r}{1} \arrow{d}[left]{e_{\overline{Y}}} & TY  \arrow{d}[left]{e_{\overline{Y}}}  & \overline{\overline{Y}} \arrow{dl}{\mu^{\mathbb{X}}_{m_Y}} \arrow{dr}{m_{\overline{\overline{Y}}}} & \\
 		    \overline{Y} \arrow{r}[below]{1} 	& \overline{Y} \arrow{rr}[below]{m_{\overline{Y}}} & &  X		
 		    		\end{tikzcd}
 \quad
 		    		\begin{tikzcd}[column sep = 2em, row sep = 1em]
 		    TY \arrow{r}{1} \arrow{dd}[left]{1} & TY \arrow{d}[left]{T\eta_{Y}} \arrow{r}{1} &  TY \arrow{d}[left]{T\eta^{\mathbb{X}}_{m_Y}}  \arrow{r}{e_{\overline{Y}}}  & \overline{Y} \arrow{ddl}{\overline{\eta^{\mathbb{X}}_{m_{Y}}}} \arrow{ddd}{m_{\overline{Y}}} \\
 		    & T^2Y \arrow{r}{T e_{\overline{Y}}} \arrow{d}[left]{\mu_Y} & T\overline{Y} \arrow{d}[left]{e_{\overline{\overline{Y}}}} & & \\
 		    TY \arrow{r}{1} \arrow{d}[left]{e_{\overline{Y}}} & TY  \arrow{d}[left]{e_{\overline{Y}}}  & \overline{\overline{Y}} \arrow{d}[left]{\mu^{\mathbb{X}}_{m_Y}} \arrow{dr}{m_{\overline{\overline{Y}}}} & \\
 		    \overline{Y} \arrow{r}[below]{1} 	& \overline{Y}  \arrow{r}[below]{1} & \overline{Y} \arrow{r}[below]{m_{\overline{Y}}} &  X	
 		    		\end{tikzcd}
 	\end{equation*}
 all three morphisms $\mu^{\mathbb{X}}_{m_Y} \circ \eta^{\mathbb{X}}_{m_{\overline{Y}}}$,  $\id_{\overline{Y}}$ and $\mu^{\mathbb{X}}_{m_Y} \circ \overline{\eta^{\mathbb{X}}_{m_Y}}$ are solutions to the unique diagonal:
  	 	\begin{equation*}
 		    		\begin{tikzcd}
 		    		 		    			TY \arrow[twoheadrightarrow]{r}{e_{\overline{Y}}} \arrow{d}[left]{e_{\overline{Y}}} & \overline{Y} \arrow{d}{m_{\overline{Y}}} \arrow[dashed]{dl}{} \\
 		    			\overline{Y} \arrow[hookrightarrow]{r}[below]{m_{\overline{Y}}} & X
 		    		\end{tikzcd}.
\end{equation*}

Similarly, for the second equation of \eqref{unitassoclawproof} we note that since the following two diagrams commute
\begin{equation*}
	 		    		\begin{tikzcd}[column sep = 2.5em, row sep = 1em]
	 	T^3Y \arrow{d}[left]{\mu_{TY}} \arrow{r}{T^2e_{\overline{Y}}} & T^2 \overline{X} \arrow{d}[left]{\mu_{\overline{Y}}} \arrow{r}{e_{\overline{\overline{\overline{Y}}}} \circ T e_{\overline{\overline{Y}}}} & \overline{\overline{\overline{Y}}} \arrow{ddl}{\mu^{\mathbb{X}}_{m_{\overline{Y}}}} \arrow{ddd}{m_{\overline{\overline{\overline{Y}}}}} \\
	 	T^2 Y \arrow{r}{Te_{\overline{Y}}} \arrow{dd}[left]{e_{\overline{Y}} \circ \mu_Y} & T \overline{Y} \arrow{d}[left]{e_{\overline{\overline{Y}}}} & \\
	 	& \overline{\overline{Y}} \arrow{dl}{\mu^{\mathbb{X}}_{m_Y}} \arrow{dr}{m_{\overline{\overline{Y}}}} & \\
	 	\overline{Y} \arrow{rr}[below]{m_{\overline{Y}}} & & X 	    
	 		    		\end{tikzcd}
\quad
	 		    		\begin{tikzcd}[column sep = 3em, row sep = 1em]
	 		    		T^3Y \arrow{r}{1} \arrow{d}[left]{\mu_{TY}} & T^3Y \arrow{d}[left]{T \mu_Y} \arrow{r}{T e_{\overline{\overline{Y}}} \circ T^2 e_{\overline{Y}}} & T \overline{\overline{Y}} \arrow{d}[left]{T \mu^{\mathbb{X}}_{m_Y}} \arrow{r}{e_{\overline{\overline{\overline{Y}}}}} & \overline{\overline{\overline{Y}}} \arrow{ddl}{\overline{\mu^{\mathbb{X}}_{m_Y}}} \arrow{ddd}{m_{\overline{\overline{\overline{Y}}}}} \\
	 		    		T^2 Y \arrow{d}[left]{\mu_Y} & T^2Y \arrow{d}[left]{\mu_Y} \arrow{r}[below]{Te_{\overline{Y}}} & T \overline{Y} \arrow{d}[left]{e_{\overline{\overline{Y}}}} & \\
	 		    		TY \arrow{d}[left]{e_{\overline{Y}}} \arrow{r}{1} & TY \arrow{d}{e_{\overline{Y}}} & \overline{\overline{Y}} \arrow{dl}{\mu^{\mathbb{X}}_{m_Y}} \arrow{dr}{m_{\overline{\overline{Y}}}} & \\
	 		    		\overline{Y} \arrow{r}[below]{1} & \overline{Y} \arrow{rr}[below]{m_{\overline{Y}}} & & X	
	 		    		\end{tikzcd}
\end{equation*}
both morphisms $\mu^{\mathbb{X}}_{m_Y} \circ \mu^{\mathbb{X}}_{m_{\overline{Y}}}$ and  $\mu^{\mathbb{X}}_{m_Y} \circ \overline{\mu^{\mathbb{X}}_{m_Y}}$ are solutions to the unique diagonal below:
\begin{equation*}
		 		    		\begin{tikzcd}[column sep = 5em]
		 		    			T^3Y \arrow{d}[left]{e_{\overline{Y}} \circ \mu_Y \circ \mu_{TY}} \arrow[twoheadrightarrow]{r}{e_{\overline{\overline{\overline{Y}}}} \circ Te_{\overline{\overline{Y}}} \circ T^2 e_{\overline{Y}}} & \overline{\overline{\overline{Y}}} \arrow{d}{m_{\overline{\overline{\overline{Y}}}}} \arrow[dashed]{dl}{} \\
		 		    			\overline{Y} \arrow[hookrightarrow]{r}[below]{m_{\overline{Y}}} & X
		 		    		\end{tikzcd}.
\end{equation*}
 \end{proof}
 
As before, we exemplify \Cref{inducedmonad} for the free vector space monad on the category of sets and functions with its canonical factorisation system.

\begin{example}
We have previously seen that the closure of a subobject $m_Y$ of a vector space $\mathbb{V}$ consists of formal linear combinations that are considered equivalent, if their interpretation in $\mathbb{V}$ via $m_Y$ coincides. By construction (cf. \eqref{unitmultdef}), the monad unit $\eta^{\mathbb{V}}_{m_Y}$ maps an element $y \in Y$ to the equivalence class $\lbrack 1 \cdot y \rbrack$ in $\overline{Y}$. One further verifies that the multiplication $\mu^{\mathbb{V}}_{m_Y}$ assigns to the element $\lbrack \sum_{\lbrack \varphi \rbrack} \lambda_{\lbrack \varphi \rbrack} \cdot \lbrack \varphi \rbrack \rbrack$ in $\overline{\overline{Y}}$  the element $\lbrack \sum_x (\sum_{\lbrack \varphi \rbrack} \lambda_{\lbrack \varphi \rbrack} \cdot \varphi(y)) \cdot y \rbrack$ in $\overline{Y}$. If $Y$ is a subset of $\mathbb{V}$, that is, $m_Y(y) = y$ for all $y \in Y$, the multiplication thus flattens vectors of vectors in the usual way.
\end{example}

We will now show that the mapping of an algebra $\mathbb{X}$ to the monad $\overline{(\cdot)}^{\mathbb{X}}$ in \Cref{inducedmonad} extends to algebra homomorphisms. To this end, for any algebra homomorphism $f: \mathbb{A} \rightarrow \mathbb{B}$ in $\mathscr{M}$, let 
$f_{*}: \Sub(A) \rightarrow \Sub(B)$ be the induced functor defined by $f_{*}(m_X) = f \circ m_X$ and $f_*(g) = g$. The result below shows that $f_*$ can be extended to a morphism between monads. 

\begin{lemma}
\label{morphismmonad1}
	For any $f: \mathbb{A} \rightarrow \mathbb{B} \in \mathscr{M}$, there exists a monad morphism $(f_{*}, \alpha): (\Sub(A), \overline{(\cdot)}^{\mathbb{A}}) \rightarrow (\Sub(B), \overline{(\cdot)}^{\mathbb{B}})$. 
\end{lemma}
\begin{proof}
We need to define a natural transformation $\alpha: \overline{(\cdot)}^{\mathbb{B}} \circ f_{*} \Rightarrow f_{*} \circ \overline{(\cdot)}^{\mathbb{A}}$ between functors of type $\Sub(A) \rightarrow \Sub(B)$. That is, for any subobject $m_X : X \rightarrow A$, we require a homomorphism 
 \[ \alpha_{m_X}: m_{\overline{X}^{\mathbb{B}}} \rightarrow f \circ m_{\overline{X}^{\mathbb{A}}} \]
between subobjects of $B$. Since factorisations are unique up to unique isomorphism, and the diagram on the left below commutes 
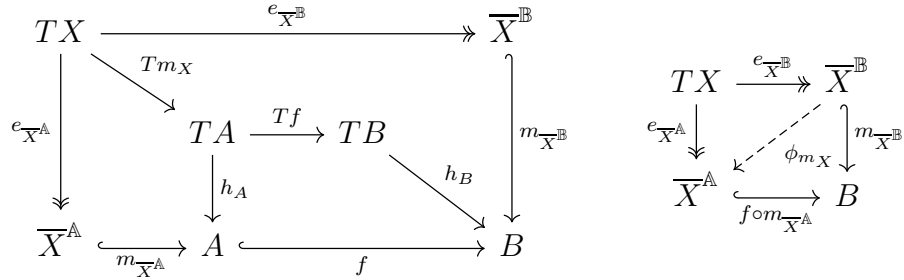
\begin{equation*}
		\begin{tikzcd}
	TX \arrow[twoheadrightarrow]{dd}[left]{e_{\overline{X}^{\mathbb{A}}}}  \arrow[twoheadrightarrow]{rrr}[above]{e_{\overline{X}^{\mathbb{B}}}} \arrow{dr}{Tm_X} & & & \overline{X}^{\mathbb{B}}  \\
		& TA \arrow{d}{h_A} \arrow{r}{Tf} & TB \arrow{dr}{h_B} & \\
		\overline{X}^{\mathbb{A}} \arrow[hookrightarrow]{r}[below]{m_{\overline{X}^{\mathbb{A}}}} & A \arrow[hookrightarrow]{rr}[below]{f} & & B	 \arrow[hookleftarrow]{uu}[right]{m_{\overline{X}^{\mathbb{B}}}}	
		\end{tikzcd}
	\qquad
		\begin{tikzcd}
			TX \arrow[twoheadrightarrow]{d}[left]{e_{\overline{X}^{\mathbb{A}}}} \arrow[twoheadrightarrow]{r}[above]{e_{\overline{X}^{\mathbb{B}}}}  & \overline{X}^{\mathbb{B}} \arrow[dashed]{dl}{\phi_{m_X}}  \\
			\overline{X}^{\mathbb{A}} \arrow[hookrightarrow]{r}[below]{f \circ m_{\overline{X}^{\mathbb{A}}}} &B \arrow[hookleftarrow]{u}[right]{m_{\overline{X}^{\mathbb{B}}}}
		\end{tikzcd}
\end{equation*}
there exists a unique homomorphism $
\phi_{m_X}: m_{\overline{X}^{\mathbb{B}}} \rightarrow f \circ m_{\overline{X}^{\mathbb{A}}}$  of subobjects of $B$ as indicated on the right above.
	We thus propose the definition \[ \alpha_{m_X} := \phi_{m_X}. \]
	We begin by showing that the above proposal turns $\alpha$ into a \emph{natural} transformation. Let $g: m_{X} \rightarrow m_{Y}$ be a morphism of subobjects of $A$ and $f_{*}(g) = g: f_{*}(m_{X}) \rightarrow f_{*}(m_{Y})$ the induced morphism of subobjects of $B$. We need to prove the equality
	\[
	\phi_{m_Y} \circ \overline{f_{*}(g)}^{\mathbb{B}} = \overline{g}^{\mathbb{A}} \circ \phi_{m_X}.
	\]
	To this end, note that, as the two diagrams below commute
	
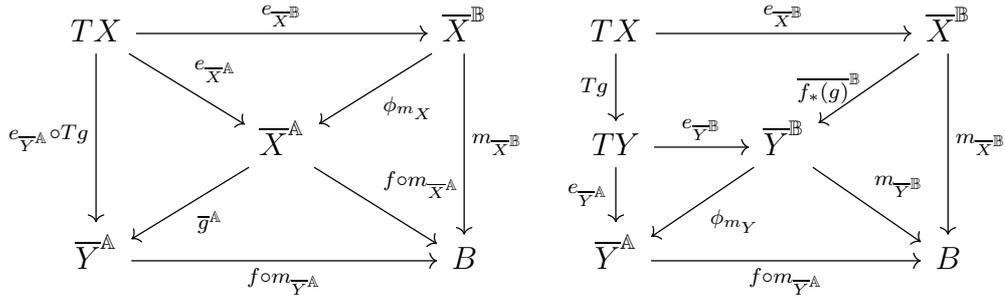
\[
	  \begin{tikzcd}[column sep = 3.5em]
TX \arrow{dr}{e_{\overline{X}^{\mathbb{A}}}} \arrow{rr}{e_{\overline{X}^{\mathbb{B}}}} \arrow{dd}[left]{e_{\overline{Y}^{\mathbb{A}}} \circ Tg} & &  \overline{X}^{\mathbb{B}} \arrow{dd}{m_{\overline{X}^{\mathbb{B}}}} \arrow{dl}{\phi_{m_X}} \\
	& \overline{X}^{\mathbb{A}} \arrow{dl}{\overline{g}^{\mathbb{A}}} \arrow{dr}{f \circ m_{\overline{X}^{\mathbb{A}}}} & \\
	\overline{Y}^{\mathbb{A}} \arrow{rr}[below]{f \circ m_{\overline{Y}^{\mathbb{A}}}} & & B
		  \end{tikzcd}
	\quad
		  \begin{tikzcd}[column sep = 3em]
		  	TX \arrow{d}[left]{Tg} \arrow{rr}{e_{\overline{X}^{\mathbb{B}}}} & & \overline{X}^{\mathbb{B}} \arrow{dd}{m_{\overline{X}^{\mathbb{B}}}} \arrow{dl}[left]{\overline{f_{*}(g)}^{\mathbb{B}}} \\
		  	TY \arrow{d}[left]{e_{\overline{Y}^{\mathbb{A}}}} \arrow{r}{e_{\overline{Y}^{\mathbb{B}}}} & \overline{Y}^{\mathbb{B}} \arrow{dl}{\phi_{m_Y}} \arrow{dr}{m_{\overline{Y}^{\mathbb{B}}}} & \\
		  	\overline{Y}^{\mathbb{A}} \arrow{rr}[below]{f \circ m_{\overline{Y}^{\mathbb{A}}}} & & B
		  \end{tikzcd}
	\]
	both $\phi_{m_Y} \circ \overline{f_{*}(g)}^{\mathbb{B}}$ and $\overline{g}^{\mathbb{A}} \circ \phi_{m_X}$ are solutions to the unique diagonal below:
	\[
	  \begin{tikzcd}
	TX \arrow[twoheadrightarrow]{r}{e_{\overline{X}^{\mathbb{B}}}} \arrow{d}[left]{e_{\overline{Y}^{\mathbb{A}}} \circ Tg} & \overline{X}^{\mathbb{B}} \arrow{d}{m_{\overline{X}^{\mathbb{B}}}}\arrow[dashed]{dl}{} \\
	\overline{Y}^{\mathbb{A}} \arrow[hookrightarrow]{r}[below]{f \circ m_{\overline{Y}^{\mathbb{A}}}} & B
			  \end{tikzcd}.
	\]
Finally, one verifies that the commutative diagrams turning $(f_{*}, \alpha)$ into a morphism between monads correspond to the two equations 
\[ \phi_{m_X} \circ \mu^{\mathbb{B}}_{f_{*}(m_X)} = \mu_{m_X}^{\mathbb{A}} \circ \phi_{m_{\overline{X}^{\mathbb{A}}}} \circ \overline{\phi_{m_X}}^{\mathbb{B}} \qquad \eta^{\mathbb{A}}_{m_X} = \phi_{m_X} \circ \eta^{\mathbb{B}}_{f_{*}(m_X)}. \]
For the first equation we observe that since the two diagrams below commute
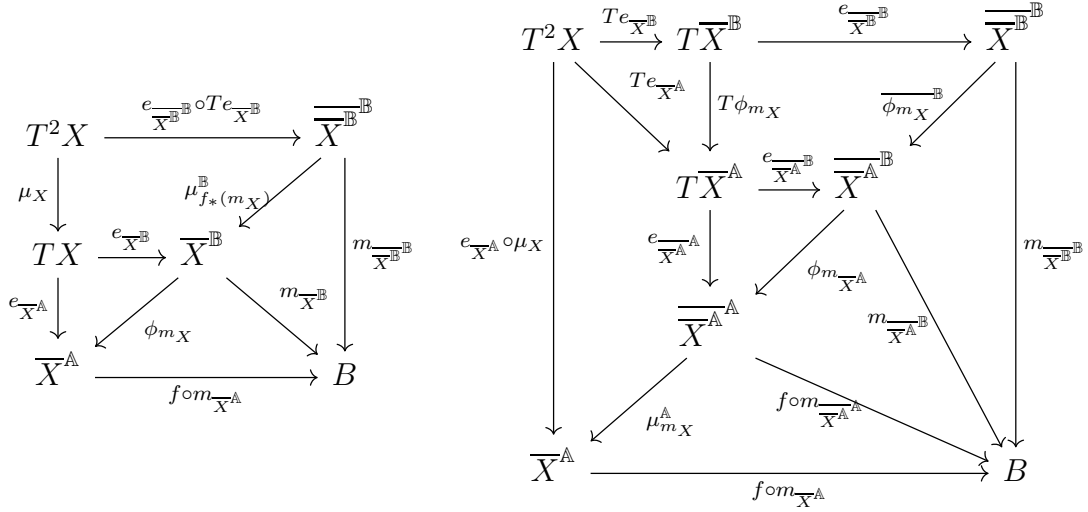
\[
  \begin{tikzcd}[row sep = 2em, column sep = 2em]
T^2X \arrow{rr}{e_{\overline{\overline{X}^{\mathbb{B}}}^{\mathbb{B}}} \circ Te_{\overline{X}^{\mathbb{B}}}} \arrow{d}[left]{\mu_X} && \overline{\overline{X}^{\mathbb{B}}}^{\mathbb{B}} \arrow{dd}{m_{\overline{\overline{X}^{\mathbb{B}}}^{\mathbb{B}}}} \arrow{dl}[left]{\mu^{\mathbb{B}}_{f_{*}(m_X)}} \\ 
TX \arrow{d}[left]{e_{\overline{X}^{\mathbb{A}}}} \arrow{r}{e_{\overline{X}^{\mathbb{B}}}} & \overline{X}^{\mathbb{B}} \arrow{dl}{\phi_{m_X}} \arrow{dr}{m_{\overline{X}^{\mathbb{B}}}} &  \\
\overline{X}^{\mathbb{A}} \arrow{rr}[below]{f \circ m_{\overline{X}^{\mathbb{A}}}} & & B
  \end{tikzcd}
  \quad
    \begin{tikzcd}[row sep = 2.5em, column sep = 2em]
    	T^2X \arrow{ddd}[left]{e_{\overline{X}^{\mathbb{A}}} \circ \mu_X} \arrow{r}{Te_{\overline{X}^{\mathbb{B}}}} \arrow{dr}{Te_{\overline{X}^{\mathbb{A}}}} & T\overline{X}^{\mathbb{B}}  \arrow{d}{T\phi_{m_X}} \arrow{rr}{e_{\overline{\overline{X}^{\mathbb{B}}}^{\mathbb{B}}}} & & \overline{\overline{X}^{\mathbb{B}}}^{\mathbb{B}} \arrow{ddd}{m_{\overline{\overline{X}^{\mathbb{B}}}^{\mathbb{B}}}} \arrow{dl}[left]{\overline{\phi_{m_X}}^{\mathbb{B}}} \\
    	& T \overline{X}^{\mathbb{A}} \arrow{d}[left]{e_{\overline{\overline{X}^{\mathbb{A}}}^{\mathbb{A}}}} \arrow{r}{e_{\overline{\overline{X}^{\mathbb{A}}}^{\mathbb{B}}}} & \overline{\overline{X}^{\mathbb{A}}}^{\mathbb{B}} \arrow{ddr}[left]{m_{\overline{\overline{X}^{\mathbb{A}}}^{\mathbb{B}}}} \arrow{dl}{\phi_{m_{\overline{X}^{\mathbb{A}}}}} & \\
    	& \overline{\overline{X}^{\mathbb{A}}}^{\mathbb{A}} \arrow{dl}{\mu_{m_X}^{\mathbb{A}}} \arrow{drr}[left]{f \circ m_{\overline{\overline{X}^{\mathbb{A}}}^{\mathbb{A}}}} & & \\
    	\overline{X}^{\mathbb{A}} \arrow{rrr}[below]{f \circ m_{\overline{X}^{\mathbb{A}}}} & & & B 
    \end{tikzcd}
\]
both $\phi_{m_X} \circ \mu^{\mathbb{B}}_{f_{*}(m_X)}$ and $\mu_{m_X}^{\mathbb{A}} \circ \phi_{m_{\overline{X}^{\mathbb{A}}}} \circ \overline{\phi_{m_X}}^{\mathbb{B}}$ are solutions to the unique diagonal below:
\[
  \begin{tikzcd}[column sep = 5em]
  	T^2X \arrow{d}[left]{e_{\overline{X}^{\mathbb{A}}} \circ \mu_X} \arrow[twoheadrightarrow]{r}{e_{\overline{\overline{X}^{\mathbb{B}}}^{\mathbb{B}}} \circ Te_{\overline{X}^{\mathbb{B}}}} & \overline{\overline{X}^{\mathbb{B}}}^{\mathbb{B}} \arrow[dashed]{dl} \arrow{d}{m_{\overline{\overline{X}^{\mathbb{B}}}^{\mathbb{B}}}} \\
  	\overline{X}^{\mathbb{A}} \arrow[hookrightarrow]{r}[below]{f \circ m_{\overline{X}^{\mathbb{A}}}} & B
  \end{tikzcd}.
\]
 For the second equation we observe that since the two diagrams below commute
 
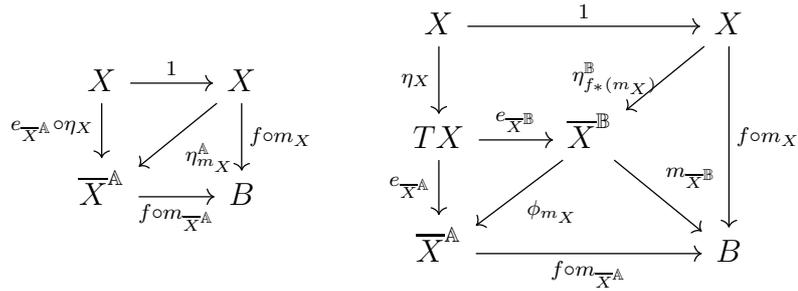
\[
  \begin{tikzcd}
    	X \arrow{r}{1} \arrow{d}[left]{e_{\overline{X}^{\mathbb{A}}} \circ \eta_X} & X \arrow{dl}{\eta^{\mathbb{A}}_{m_X}} \arrow{d}{f \circ m_X} \\
  	\overline{X}^{\mathbb{A}} \arrow{r}[below]{f \circ m_{\overline{X}^{\mathbb{A}}}} & B
  \end{tikzcd}
\qquad
 \begin{tikzcd}
 	X \arrow{rr}{1} \arrow{d}[left]{\eta_X} && X \arrow{dd}{f \circ m_X} \arrow{dl}[left]{\eta_{f_{*}(m_X)}^{\mathbb{B}}} \\
 	TX \arrow{d}[left]{e_{\overline{X}^{\mathbb{A}}}} \arrow{r}{e_{\overline{X}^{\mathbb{B}}}}& \overline{X}^{\mathbb{B}} \arrow{dr}{m_{\overline{X}^{\mathbb{B}}}} \arrow{dl}{\phi_{m_X}} & \\ 
 	\overline{X}^{\mathbb{A}} \arrow{rr}[below]{f \circ m_{\overline{X}^{\mathbb{A}}}} & & B
 \end{tikzcd} 
 \]
 both $\eta^{\mathbb{A}}_{m_X}$ and $\phi_{m_X} \circ \eta^{\mathbb{B}}_{f_{*}(m_X)}$ are solutions to the unique diagonal below:
 \[
   \begin{tikzcd}
  	X \arrow[twoheadrightarrow]{r}{1} \arrow{d}[left]{e_{\overline{X}^{\mathbb{A}}} \circ \eta_X} & X \arrow[dashed]{dl}{} \arrow{d}{f \circ m_X} \\
  	\overline{X}^{\mathbb{A}} \arrow[hookrightarrow]{r}[below]{f \circ m_{\overline{X}^{\mathbb{A}}}} & B
  \end{tikzcd}.
 \]
	\end{proof} 
  
The next statement establishes that the canonical forgetful functor $U: \Sub(X) \rightarrow \mathscr{C}$ defined by $U(m_Y) = Y$ and $U(f) = f$ extends to a morphism between monads.

  \begin{lemma}
  \label{morphismmonad2}
 	There exists a monad morphism $(U, \alpha): (\Sub(X), \overline{(\cdot)}^{\mathbb{X}}) \rightarrow (\mathscr{C}, T)$.
 \end{lemma}
    \begin{proof}
 We propose the following definition:
 \[
 \alpha: T \circ U \Rightarrow U \circ \overline{(\cdot)}^{\mathbb{X}} \qquad \alpha_{m_Y} := e_{\overline{Y}} : TY \rightarrow \overline{Y}.
 \]
  From the definition of $\overline{(\cdot)}^{\mathbb{X}}$ on morphisms it follows that $\alpha$ is a natural transformation.
   	The commutative diagrams turning $(U, \alpha)$ into a morphism between monads correspond to the following two equations:
 	\[
 	 \mu^{\mathbb{X}}_{m_{\overline{Y}}} \circ e_{\overline{\overline{Y}}} \circ T e_{\overline{Y}} = e_{\overline{Y}} \circ \mu_Y \qquad \eta^{\mathbb{X}}_{m_Y} = e_{\overline{Y}} \circ \eta_Y. \] The equalities are satisfied by the definitions of $\eta^{\mathbb{X}}$ and $\mu^{\mathbb{X}}$, respectively.
 	\end{proof}
 
We conclude with the observation that the monad morphism defined in \Cref{morphismmonad1} commutes with the monad morphisms defined in \Cref{morphismmonad2}.

\begin{lemma}
For any algebra homomorphism $f: \mathbb{A} \rightarrow \mathbb{B} \in \mathscr{M}$, the following diagram commutes:
\[
 \begin{tikzcd}[ampersand replacement=\&]
 	(\Sub(A), \overline{(\cdot)}^{\mathbb{A}}) \arrow{dr}[left]{(U_{\mathbb{A}}, \alpha_{\mathbb{A}})} \arrow{rr}{(f_*, \alpha_f)} \& \& (\Sub(B), \overline{(\cdot)}^{\mathbb{B}}) \arrow{dl}{(U_{\mathbb{B}}, \alpha_{\mathbb{B}})} \\
 	\& (\mathscr{C}, T) \&
 \end{tikzcd}
 \]	
\end{lemma}
\begin{proof}
	We have to show that the monad morphism $(U_{\mathbb{B}} \circ f_*, \beta): (\Sub(A), \overline{(\cdot)}^{\mathbb{A}}) \rightarrow (\mathscr{C}, T)$ with the natural transformation
\[
		\beta = T \circ (U_{\mathbb{B}} \circ f_*) \overset{\alpha_{\mathbb{B}} \circ f_*}{\Longrightarrow} U_{\mathbb{B}} \circ \overline{(\cdot)}^{\mathbb{B}} \circ f_* \overset{U_{\mathbb{B}} \circ \alpha_f}{\Longrightarrow} (U_{\mathbb{B}} \circ f_*) \circ \overline{(\cdot)}^{\mathbb{A}}
		\]
	given on a subobject $m_Y : Y \rightarrow A$ in $\Sub(A)$ by the morphism
	\[
		\beta_{m_Y} : TY \overset{e_{\overline{Y}^{\mathbb{B}}}}{\longrightarrow} \overline{Y}^{\mathbb{B}} \overset{\phi_{m_Y}}{\longrightarrow}  \overline{Y}^{\mathbb{A}} \]
	coincides with the morphism $(U_{\mathbb{A}}, \alpha_{\mathbb{A}}): (\Sub(A), \overline{(\cdot)}^{\mathbb{A}}) \rightarrow (\mathscr{C}, T)$ with $(\alpha_{\mathbb{A}})_{m_Y} = e_{\overline{Y}^{\mathbb{A}}}$. 
	
	The equality $U_{\mathbb{B}} \circ f_{*} = U_{\mathbb{A}}$ is immediate from the involved definitions. The identity $\phi_{m_Y} \circ e_{\overline{Y}^{\mathbb{B}}} = e_{\overline{Y}^{\mathbb{A}}}$ follows from the definition of $\phi_{m_Y}$ as unique diagonal.
\end{proof}

 \subsection{Closing an Image}
 
 \label{closinganimagesec}
 
 In this section we investigate the closure of a particular class of subobjects: the ones that arise by taking the image of a morphism. We then show that deriving a minimal bialgebra from a minimal coalgebra by equipping the latter with additional algebraic structure can be realized as the closure of a subobject in this class.
 
  We assume a category $\mathscr{C}$ with a factorisation system $(\mathscr{E}, \mathscr{M})$ and a monad $T$ on $\mathscr{C}$ that preserves $\mathscr{E}$. Suppose that $\mathbb{X} = (X, h_X)$ is a $T$-algebra and $f: Y \rightarrow X$ a morphism in $\mathscr{C}$.
On the one hand, there exists a factorisation of $f$ in $\mathscr{C}$:
 \[
 f = Y \overset{e_{\im(f)}}{\twoheadrightarrow} \im(f) \overset{m_{\im(f)}}{\hookrightarrow} X.
 \] 
 On the other hand, there exists a factorisation of the lifing $f^{\sharp} = h_X \circ Tf$ in the category of Eilenberg-Moore algebras $\EM$:
 \[
 f^{\sharp} = (TY, \mu_Y) \overset{e_{\im(f^{\sharp})}}{\twoheadrightarrow} (\im(f^{\sharp}), h_{\im(f^{\sharp})}) \overset{m_{\im(f^{\sharp})}}{\hookrightarrow} (X,h_X).
 \]

 The next result shows that, up to isomorphism, the closure of the subobject $m_{\im(f)}$ with respect to the algebra $\mathbb{X}$ is given by the subobject $m_{\im(f^{\sharp})}$.
 \begin{lemma}
 \label{closureofimage}
 	$\overline{m_{\textnormal{\im}(f)}}^{\mathbb{X}} = m_{\textnormal{\im}(f^{\sharp})}$ in $\Sub(\mathbb{X})$.
 \end{lemma}
  \begin{proof}
 Using the factorisation of $(m_{\im(f)})^{\sharp} = h_X \circ Tm_{\im(f)}$ in $\EM$,
 \[
 (m_{\im(f)})^{\sharp} = (T\textnormal{\im}(f), \mu_{\textnormal{\im}(f)}) \overset{e_{\overline{\textnormal{\im}(f)}}}{\twoheadrightarrow} (\overline{\textnormal{\im}(f)} , h_{\overline{\textnormal{\im}(f)}}) \overset{m_{\overline{\textnormal{\im}(f)}}}{\hookrightarrow} (X, h_X)
 \]
 one easily verifies that the diagram below commutes: 
 	\begin{equation*}
 		\begin{tikzcd}[row sep = 1em]
 			(TY, \mu_Y) \arrow[twoheadrightarrow]{rr}{e_{\im(f^{\sharp})}} \arrow[twoheadrightarrow]{d}[left]{Te_{\im(f)}} \arrow{dr}{Tf} & & (\im(f^{\sharp}), h_{\im(f^{\sharp})})  \\
 			(T \im(f), \mu_{\im(f)}) \arrow{r}[below]{Tm_{\im(f)}} \arrow[twoheadrightarrow]{d}[left]{e_{\overline{\im(f)}}} & (TX, \mu_X) \arrow{dr}{h_X} & \\
 			(\overline{\im(f)} , h_{\overline{\im(f)}})\arrow[hookrightarrow]{rr}[below]{m_{\overline{\im(f)}}} & & (X,h_X) \arrow[hookleftarrow]{uu}[right]{m_{\im(f^{\sharp})}}
 		\end{tikzcd}.
 	\end{equation*}
 	Since factorisations are unique up to unique isomorphism, there thus exists a unique isomorphism $\phi: m_{\textnormal{\im}(f^{\sharp})}  \simeq m_{\overline{\im(f)}} $ of subobjects of $\mathbb{X}$ as indicated below:
\begin{equation*}
\begin{tikzcd}
 		 			(TY, \mu_Y) \arrow[twoheadrightarrow]{r}{e_{\im(f^{\sharp})}} \arrow[twoheadrightarrow]{d}[left]{e_{\overline{\im(f)}} \circ Te_{\im(f)}} &  (\im(f^{\sharp}), h_{\im(f^{\sharp})})  \arrow[dashed]{dl}{\phi} \\
 		 			(\overline{\im(f)} , h_{\overline{\im(f)}}) \arrow[hookrightarrow]{r}[below]{m_{\overline{\im(f)}}} &  (X,h_X)  \arrow[hookleftarrow]{u}[right]{m_{\im(f^{\sharp})}}
 		 		\end{tikzcd}	.
\end{equation*}
Since by definition $\overline{m_{\textnormal{\im}(f)}}^{\mathbb{X}} \simeq m_{\overline{\im(f)}} $, this shows the claim.
  	 	 \end{proof}

The next example uses \Cref{closureofimage} to show that deriving a minimal bialgebra from a minimal coalgebra can be realised as the closure of a subobject with respect to a monad of the type in \Cref{inducedmonad}.

\begin{example}[Closure of Minimal Moore Automata]
\label{closure_of_coalg}	
Let $F$ be the set endofunctor with $FX = B \times X^A$, for fixed sets $A$ and $B$. Let $T$ be a set monad, $h: TB \rightarrow B$ a $T$-algebra structure for $B$, and let $L: A^* \rightarrow B$ be a generalised language.

As $F$ preserves monomorphisms, the canonical epi-mono factorisation system of the category of sets lifts to the category $\Coalg(F)$, which consists of unpointed Moore automata with input $A$ and output $B$. 

There exists a size-minimal Moore automaton $\M_L$ that accepts $L$. It can be recovered as the epi-mono factorisation of the final $F$-coalgebra homomorphism $\obs: A^* \rightarrow  \Omega$, that is, $\M_L = m_{\im(\obs)}$. In more detail: $\Omega$ is carried by $B^{A^*}$; $\obs$ satisfies $\obs(w)(v) = L(wv)$; and $A^*$ is equipped with the $F$-coalgebra structure $\langle \varepsilon, \delta \rangle: A^* \rightarrow B \times (A^*)^A$ defined by $\varepsilon(w) = L(w)$ and $\delta(w)(a) = wa$ \cite{van2016master}. 

The algebra structure $h$ induces a canonical\footnote{\label{canonicaldistributivelawfootnote}Given an algebra $h: TB \rightarrow B$ for a set monad $T$, one can define a distributive law $\lambda$ between $T$ and $F$ with $FX = B \times X^A$ by $\lambda_X := (h \times \st) \circ \langle T\pi_1, T\pi_2 \rangle: TFX \rightarrow FTX$ \cite{jacobs2006bialgebraic}. (We write $\st$ for the usual strength function $\st: T(X^A) \rightarrow (TX)^A$ defined by $ \st(U)(a) = T(\ev_a)(U)$,
	where $\ev_{a}(f) = f(a)$.)} distributive law $\lambda$ between $T$ and $F$. One can show that $\lambda$-bialgebras are algebras over the monad $T_{\lambda}$ on $\Coalg(F)$ defined by $T_{\lambda}(X,k) = (TX, \lambda_X \circ Tk)$ and $T_{\lambda}f = Tf$ \cite{turi1997towards}. One such $T_{\lambda}$-algebra is the final $F$-coalgebra $\Omega$, when equipped with a canonical $T$-algebra structure \cite[Prop. 3]{jacobs2012trace}.

The functor $T_{\lambda}$ preserves epimorphisms in the category $\Coalg(F)$, if $T$ preserves epimorphisms in the category of sets. The latter is the case for every set endofunctor. By \Cref{inducedmonad}, there thus exists a well-defined monad $\overline{(\cdot)}$ on $\textnormal{\Sub}(\Omega)$. 

By construction, the minimal Moore automaton $\M_L$ lives in $\textnormal{\Sub}(\Omega)$. Reviewing the constructions in \cite{zetzsche2021} shows that the minimal $\lambda$-bialgebra $\mathbb{M}_L$ for $L$ is given by the image of the lifting of $\obs$, that is, $\mathbb{M}_L = m_{\im(\obs^{\sharp})}$. From \Cref{closureofimage} it thus follows $\mathbb{M}_L = \overline{\M_L}$. Hence the minimal $\lambda$-bialgebra for $L$ can be obtained from the minimal $F$-coalgebra for $L$ by closing the latter with respect to the $T_{\lambda}$-algebra structure of $\Omega$. 

For an example of the monad unit, observe how the minimal coalgebra in \Cref{m(l)} embeds into the minimal bialgebra in \Cref{overlineml}. 
\end{example}

The situation can be further generalised. We assume that i) $\mathscr{C}$ is a category with an $(\mathscr{E}, \mathscr{M})$-factorisation system; ii) $\lambda$ is a distributive law between a monad $T$ on $\mathscr{C}$ that preserves $\mathscr{E}$ and an endofunctor $F$ on $\mathscr{C}$ that preserves $\mathscr{M}$; iii) $(\Omega, h_{\Omega}, k_{\Omega})$ is a final $\lambda$-bialgebra.

\begin{theorem}
\label{coalgebratheorem}
There exists a functor $\overline{(\cdot)}: \Sub(\Omega, k_{\Omega}) \rightarrow \Sub(\Omega, h_{\Omega}, k_{\Omega})$ yielding a monad on $\Sub(\Omega, k_{\Omega})$ and satisfying $\overline{m_{\im(\obs_{(X,k)})}} \cong m_{\im(\obs_{\free_T(X,k)})}$ in $\Sub(\Omega, h_{\Omega}, k_{\Omega})$, for any $F$-coalgebra $(X,k)$.
\end{theorem} 
  \begin{proof}
 As $F$ preserves $\mathscr{M}$, the $(\mathscr{E}, \mathscr{M})$-factorisation system of $\mathscr{C}$ lifts to $\Coalg(F)$. The category of $\lambda$-bialgebras is isomorphic to the category of algebras over the monad $T_{\lambda}$ on $\Coalg(F)$ defined by $T_{\lambda}(X,k) = (TX, \lambda_X \circ Tk)$ and $T_{\lambda}f = Tf$ \cite{turi1997towards}. The functor $T_{\lambda}$ preserves the $\mathscr{E}$-part of the lifted factorisation system of $\Coalg(F)$, if $T$ preserves the $\mathscr{E}$-part of the factorisation system of $\mathscr{C}$. In consequence, the factorisation system of $\Coalg(F)$ thus lifts to $\Bialg(\lambda)$. By \Cref{functorprop} and \Cref{inducedmonad} it follows that there exists a functor $\overline{(\cdot)}: \Sub(\Omega, k_{\Omega}) \rightarrow \Sub(\Omega, h_{\Omega}, k_{\Omega})$ that yields a monad on $\Sub(\Omega, k_{\Omega})$. 
 Since $(\Omega, h_{\Omega}, k_{\Omega})$ is a final $\lambda$-bialgebra, $(\Omega, k_{\Omega})$ is a final $F$-coalgebra. By \Cref{closureofimage} we have the equality $\overline{m_{\im(\obs_{(X,k)})}} = m_{\im(\obs_{(X,k)}^{\sharp})}$ in $\Sub(\Omega, h_{\Omega}, k_{\Omega})$, where $\obs_{(X,k)}^{\sharp} = h_{\Omega} \circ T_{\lambda}(\obs_{(X,k)})$ is of type $\free_T(X, k) = (TX, \mu_X, \lambda_X \circ Tk) \rightarrow (\Omega, h_{\Omega}, k_{\Omega})$. By uniqueness it follows $\obs_{(X,k)}^{\sharp} = \obs_{\free_T(X,k)}$, which proves the claim.
 \end{proof}

To recover \Cref{closure_of_coalg} as a special case of \Cref{coalgebratheorem}, one instantiates the latter for $F$ with $FX = B \times X^A$ and the canonical $F$-coalgebra with carrier $A^*$.

Finally, using analogous functors to the ones present in \eqref{closure_as_composition}, we observe that, as a consequence of \Cref{closureofimage}, the diagram below commutes:
 \begin{equation*}
  		\begin{tikzcd}
 			\mathscr{C}/X \arrow{r}{} \arrow{d}[left]{}  & \EM/\mathbb{X} \arrow{d}{} \\
 			\Sub(X) \arrow{r}[above]{\overline{(\cdot)}^{\mathbb{X}}} & \Sub(\mathbb{X})
 		\end{tikzcd}.	
 \end{equation*}

\section{Step 2: Generators and Bases}

\label{generatorsandbases_sec}

One of the central concepts of linear algebra is the notion of a basis for a vector space: a subset of a vector space is called a basis for the former if every vector can be uniquely  written as a finite linear combination of basis elements. Part of the importance of bases stems from the convenient consequences that follow from their existence. For example, linear transformations between vector spaces admit matrix representations relative to pairs of bases \cite{lang2004algebra}, which can be used for efficient calculations. The idea of a basis however is not restricted to the theory of vector spaces: other algebraic theories have analogous notions of bases (and generators, by waiving the uniqueness constraint), for instance modules, semi-lattices, Boolean algebras, convex sets, and many more. In fact, the theory of bases for vector spaces is special only in the sense that every vector space admits a basis, which is not  the case for e.g. modules. 

In this section, we use the compact category-theoretical abstraction of generators and bases given in \Cref{generatordefinition} to lift results from one theory to the others. For example, one may wonder if there exists a matrix representation theory for convex sets that is analogous to the one of vector spaces.

\subsection{Categorification}

This section introduces morphisms between algebras with a generator or a basis.

\begin{definition}
\label{generatorcategory}
The category $\Galg(T)$ of algebras with a generator over a monad $T$ is defined as follows:
\begin{itemize}
	\item Objects are pairs $(\mathbb{X}_{\alpha}, \alpha)$, where $\mathbb{X}_{\alpha} = (X_{\alpha}, h_{\alpha})$ is a $T$-algebra with generator $\alpha = (Y_{\alpha}, i_{\alpha}, d_{\alpha})$.
	\item A morphism $(f, p): (\mathbb{X}_{\alpha}, \alpha) \rightarrow (\mathbb{X}_{\beta}, \beta)$ consists of a $T$-algebra homomorphism $f: \mathbb{X}_{\alpha} \rightarrow \mathbb{X}_{\beta}$ and a Kleisli-morphism $p: Y_{\alpha} \rightarrow TY_{\beta}$, such that the diagram below commutes:
	\begin{equation}
	\label{galgmorph}
	\begin{tikzcd}
		X_{\alpha}  \arrow{r}{d_{\alpha}} \arrow{d}{f} & TY_{\alpha} \arrow{d}{p^{\sharp}}   \arrow{r}{i_{\alpha}^{\sharp}}  & X_{\alpha} \arrow{d}{f} \\
		 X_{\beta}  \arrow{r}{d_{\beta}} & TY_{\beta} \arrow{r}{i_{\beta}^{\sharp}} &  X_{\beta} 
	\end{tikzcd}.
	\end{equation}
	Given $(f, p): (\mathbb{X}_{\alpha}, \alpha) \rightarrow (\mathbb{X}_{\beta}, \beta)$ and $(g, q): (\mathbb{X}_{\beta}, \beta) \rightarrow (\mathbb{X}_{\gamma}, \gamma)$, their composition is defined componentwise as $(g, q) \circ (f, p) := (g \circ f, q \cdot p)$, where $q \cdot p := \mu_{Y_{\gamma}} \circ Tq \circ p$ denotes the usual Kleisli-composition.
\end{itemize}
The category $\Balg(T)$ of algebras with a basis is defined as the obvious full subcategory of $\Galg(T)$.
\end{definition}

Let $F:  \EM \rightarrow \Galg(T)$ be the functor with $F(\mathbb{X})  := (\mathbb{X}, (X, \id_X, \eta_{X}))$ and $F(f: \mathbb{X} \rightarrow \mathbb{Y}) := (f, \eta_Y \circ f)$, and $U: \Galg(T) \rightarrow \EM$ the forgetful functor defined as the projection on the first component. Then $F$ and $U$ are in an adjoint relation:

\begin{lemma}
\label{galgfreeforgetful}
$F \dashv U \,\colon\, \Galg(T) \leftrightarrows \EM$.
\end{lemma}
\begin{proof}
	Since every algebra can be generated by itself, the definition for $F$ is well-defined on objects. For morphisms, one easily establishes \eqref{galgmorph} from the naturality of $\eta$, the monad law $\mu_{Y} \circ T\eta_Y = \id_{TY}$, and the commutativity of $f$ with algebra structures. The compositionality of $F$ follows analogously; preservation of identity is trivial. For the natural isomorphism
	\[
	\Hom_{\Galg(T)}(F(\mathbb{X}), (\mathbb{X}_{\alpha}, \alpha)) \simeq \Hom_{\EM}(\mathbb{X}, U(\mathbb{X}_{\alpha}, \alpha))
	\]
	we propose mapping $(f, p)$ to $f$, and conversely, $f$ to $(f, d_{\alpha} \circ f)$. The latter is well-defined since \[ (d_{\alpha} \circ f)^{\sharp} \circ \eta_X = d_{\alpha} \circ f \quad \textnormal{and} \quad i_{\alpha}^{\sharp} \circ (d_{\alpha} \circ f)^{\sharp} = i_{\alpha}^{\sharp} \circ d_{\alpha} \circ f^{\sharp} = f^{\sharp} = f \circ (\id_X)^{\sharp}. \] Composition in one of the directions trivially yields the identity; for the other direction we note that if $(f, p)$ satisfies \eqref{galgmorph}, then $p = p^{\sharp} \circ \eta_X = d_{\alpha} \circ f$. 	
\end{proof}

\subsection{Products}

\label{productofbasessec}

In this section we show that, under certain assumptions, the monoidal product of a base category naturally extends to a monoidal product of algebras with bases in the base category. As a natural example we obtain the tensor-product of vector spaces with fixed bases.

We assume basic familiarity with monoidal categories. A monoidal monad $T$ on a monoidal category $(\mathscr{C}, \otimes, I)$ is a monad which is equipped with natural transformations $T_{X,Y}: TX \otimes TY \rightarrow T(X \otimes Y)$ and $T_0: I \rightarrow TI$, satisfying certain coherence conditions (see e.g. \cite{seal2013tensors}). One can show that, given such additional data, the monoidal structure of $\mathscr{C}$ induces a monoidal category $(\EM, \boxtimes, (TI, \mu_I))$, if the following two assumptions are given \cite[Corollary 2.5.6]{seal2013tensors}:
\begin{enumerate}
	\item[(A1)] For any two algebras $\mathbb{X}_{\alpha} = (X_{\alpha}, h_{\alpha})$ and $\mathbb{X}_{\beta} = (X_{\beta}, h_{\beta})$ the coequaliser $q_{\mathbb{X}_{\alpha}, \mathbb{X}_{\beta}}$ of the algebra homomorphisms 
$T(h_{\alpha} \otimes h_{\beta})$ and $\mu_{X_{\alpha} \otimes X_{\beta}} \circ T(T_{X_{\alpha}, X_{\beta}})$ of type $(T(TX_{\alpha} \otimes TX_{\beta}), \mu_{TX_{\alpha} \otimes TX_{\beta}}) \rightarrow (T(X_{\alpha} \otimes X_{\beta}), \mu_{X_{\alpha} \otimes X_{\beta}})$ exists (we denote its codomain by $\mathbb{X}_{\alpha} \boxtimes \mathbb{X}_{\beta} := (X_{\alpha} \boxtimes X_{\beta}, h_{\alpha \boxtimes \beta})$);
\item[(A2)] Left and right-tensoring with the induced functor $\boxtimes$ preserves reflexive coequaliser.
\end{enumerate}
The two monoidal products $\otimes$ and $\boxtimes$ are related via the natural embedding $
\iota_{\mathbb{X}_{\alpha}, \mathbb{X}_{\beta}} := q_{\mathbb{X}_{\alpha}, \mathbb{X}_{\beta}} \circ \eta_{X_{\alpha} \otimes X_{\beta}}: X_{\alpha} \otimes X_{\beta} \rightarrow X_{\alpha} \boxtimes X_{\beta}$.
  One can show that the product $(TY_{\alpha}, \mu_{Y_{\alpha}}) \boxtimes (TY_{\beta}, \mu_{Y_{\beta}})$ is given by $(T(Y_{\alpha} \otimes Y_{\beta}), \mu_{Y_{\alpha} \otimes Y_{\beta}})$ and the coequaliser $q_{(TY_{\alpha}, \mu_{Y_{\alpha}}), (TY_{\beta}, \mu_{Y_{\beta}})}$ by $\mu_{Y_{\alpha} \otimes Y_{\beta}} \circ T(T_{Y_{\alpha}, Y_{\beta}})$ \cite{seal2013tensors}. 

With the previous remarks in mind, we are able to claim the following.
\begingroup
 \allowdisplaybreaks
\begin{lemma}
\label{productofbases}
	Let $T$ be a monoidal monad on $(\mathscr{C}, \otimes, I)$ such that (A1) and (A2) are satisfied. Let $\alpha = (Y_{\alpha}, i_{\alpha}, d_{\alpha})$ and $\beta = (Y_{\beta}, i_{\beta}, d_{\beta})$ be generators (bases) for $T$-algebras $\mathbb{X}_{\alpha}$ and $\mathbb{X}_{\beta}$. Then $\alpha \boxtimes \beta = (Y_{\alpha} \otimes Y_{\beta}, \iota_{\mathbb{X}_{\alpha}, \mathbb{X}_{\beta}} \circ (i_{\alpha} \otimes i_{\beta}), (d_{\alpha} \boxtimes d_{\beta}))$ is a generator (basis) for the $T$-algebra $\mathbb{X}_{\alpha} \boxtimes \mathbb{X}_{\beta}$.	
\end{lemma}
\begin{proof}
First, we calculate
\begin{align}
\begin{split}
\label{productfirsteq}
	& h_{\alpha} \boxtimes h_{\beta} \\
	=&\ (\id = \mu_{X_{\alpha} \otimes X_{\beta}} \circ T(\eta_{X_{\alpha} \otimes X_{\beta}})) \\
	& (h_{\alpha} \boxtimes h_{\beta}) \circ \mu_{X_{\alpha} \otimes X_{\beta}} \circ T(\eta_{X_{\alpha} \otimes X_{\beta}}) \\
	=&\ (q_{\mathbb{X}_{\alpha}, \mathbb{X}_{\beta}} = h_{\alpha} \boxtimes h_{\beta}\ \text{\cite{seal2013tensors}}) \\
	& q_{\mathbb{X}_{\alpha}, \mathbb{X}_{\beta}} \circ \mu_{X_{\alpha} \otimes X_{\beta}} \circ T(\eta_{X_{\alpha} \otimes X_{\beta}}) \\	
	=&\ (q_{\mathbb{X}_{\alpha}, \mathbb{X}_{\beta}}\ \textnormal{is algebra homomorphism}) \\
	& h_{\alpha \boxtimes \beta} \circ T(q_{\mathbb{X}_{\alpha}, \mathbb{X}_{\beta}}) \circ T(\eta_{X_{\alpha} \otimes X_{\beta}}) \\	 
	=&\ (\textnormal{Definition of } \iota_{\mathbb{X}_{\alpha}, \mathbb{X}_{\beta}}) \\
	 & h_{\alpha \boxtimes \beta} \circ T(\iota_{\mathbb{X}_{\alpha}, \mathbb{X}_{\beta}}).
\end{split}
\end{align}

If $\alpha$ and $\beta$ are generators, it thus follows
\begin{align*}
	& h_{\alpha \boxtimes \beta} \circ T(\iota_{\mathbb{X}_{\alpha}, \mathbb{X}_{\beta}}) \circ T(i_{\alpha} \otimes i_{\beta}) \circ (d_{\alpha} \boxtimes d_{\beta}) \\
	=&\ \eqref{productfirsteq} \\ 
	& (h_{\alpha} \boxtimes h_{\beta}) \circ T(i_{\alpha} \otimes i_{\beta}) \circ (d_{\alpha} \boxtimes d_{\beta}) \\
	=&\ (T(f \otimes g) = Tf \boxtimes Tg\ \text{\cite{seal2013tensors}}) \\
	& (h_{\alpha} \boxtimes h_{\beta}) \circ (T(i_{\alpha}) \boxtimes T(i_{\beta})) \circ (d_{\alpha} \boxtimes d_{\beta}) \\
	=&\ (\boxtimes \textnormal{ is functorial}) \\
	& (h_{\alpha} \circ T(i_{\alpha}) \circ d_{\alpha}) \boxtimes (h_{\beta} \circ T(i_{\beta}) \circ d_{\beta}) \\
	=&\ (\alpha, \beta \textnormal{ are generators}) \\
	& \id_{X_{\alpha}} \boxtimes \id_{X_{\beta}} \\
	=&\ (\boxtimes \textnormal{ is functorial})  \\
 & \id_{X_{\alpha} \boxtimes X_{\beta}}.
\end{align*} 
The additional equality for the case in which $\alpha$ and $\beta$ are bases follows analogously. 
\end{proof} 
\endgroup
 
\begin{corollary}
\label{monoidalproduct}
Let $T$ be a monoidal monad on $(\mathscr{C}, \otimes, I)$ such that (A1) and (A2) are satisfied. The definitions $(\mathbb{X}_{\alpha}, \alpha) \boxtimes (\mathbb{X}_{\beta}, \beta) := (\mathbb{X}_{\alpha} \boxtimes \mathbb{X}_{\beta}, \alpha \boxtimes \beta)$ and $(f, p) \boxtimes (g,q) := (f \boxtimes g, T_{Y_{\alpha'}, Y_{\beta'}} \circ (p \otimes q))$ yield monoidal structures with unit $((TI, \mu_I), (I, \eta_I, \id_{TI}))$ on $\Galg(T)$ and $\Balg(T)$. 
\end{corollary}
\begin{proof}
	By \Cref{productofbases} the construction is well-defined on objects. That it is well-defined on morphisms, i.e. the commutativity of \eqref{galgmorph}, is a consequence of the equalities $Tf \boxtimes Tg = T(f \otimes g)$ and $q_{\mathbb{X}_{\alpha}, \mathbb{X}_{\beta}} = h_{\alpha} \boxtimes h_{\beta}$ \cite{seal2013tensors}, which imply $(T_{Y_{\alpha}, Y_{\beta}} \circ (p \otimes q))^{\sharp} = (\mu_{Y_{\alpha'}} \boxtimes \mu_{Y_{\beta'}}) \circ (Tp \boxtimes Tq)$. The natural isomorphisms underlying the monoidal structure for $\EM$ can be extended to morphisms in $\Galg(T)$ by associating canonical Kleisli-morphisms between generators as in \eqref{basisrepresentation}.
\end{proof}

We conclude by instantiating above construction to the setting of vector spaces.

\begin{example}[Tensor Product of Vector Spaces]
	Recall the free $\mathbb{K}$-vector space monad $\mathcal{V}_{\mathbb{K}}$ from \Cref{exampleofmonads}. The category of sets is monoidal (in fact, cartesian) with respect to the cartesian product $\times$ and the singleton set $\lbrace \star \rbrace$. The monad $\mathcal{V}_{\mathbb{K}}$ is monoidal when equipped with the morphisms $(\mathcal{V}_{\mathbb{K}})_{X,Y}(\varphi, \psi)(x,y) := \varphi(x) \cdot \psi(y)$ and $(\mathcal{V}_{\mathbb{K}})_0(\star)(\star) := 1_{\mathbb{K}}$ \cite{parlant_et_al:LIPIcs:2020:12746}. The category of $\mathcal{V}_{\mathbb{K}}$-algebras is isomorphic to the category of $\mathbb{K}$-vector spaces, and satisfies (A1) and (A2). The monoidal structure induced by $\mathcal{V}_{\mathbb{K}}$ is the usual tensor product $\otimes$ of vector spaces with the unit field $\mathcal{V}_{\mathbb{K}}(\lbrace \star \rbrace) \simeq \mathbb{K}$. \Cref{productofbases} captures the well-known fact that the dimension of the tensor product of two vector spaces equals the product of the dimensions of each vector space. The structure maps of the product generator map an element $(y_{\alpha}, y_{\beta})$ to the vector $i(y_{\alpha}) \otimes i(y_{\beta})$, and a vector $x$ to the function $(d_{\alpha} \otimes d_{\beta})(x)$, where 
	\begin{align*}
		d_{\alpha} \otimes d_{\beta} = \overline{d_{\alpha} \times d_{\beta}}: \mathbb{X}_{\alpha} \otimes \mathbb{X}_{\beta} &\rightarrow (\mathcal{V}_{\mathbb{K}}(Y_{\alpha}), \mu_{Y_{\alpha}}) \otimes (\mathcal{V}_{\mathbb{K}}(Y_{\beta}), \mu_{Y_{\beta}}) \\ &\simeq (\mathcal{V}_{\mathbb{K}}(Y_{\alpha} \times Y_{\beta}), \mu_{Y_{\alpha} \otimes Y_{\beta}})
	\end{align*}  is the unique linear extension of the bilinear map defined by \[(d_{\alpha} \times d_{\beta})(x_{\alpha}, x_{\beta})(y_{\alpha}, y_{\beta}) = d_{\alpha}(x_{\alpha})(y_{\alpha}) \cdot d_{\beta}(x_{\beta})(y_{\beta}). \] 
\end{example}

\subsection{Kleisli Representation Theory}

\label{representationtheorysec}

In this section we use our category-theoretical definition of a basis to derive a representation theory for homomorphisms between algebras over monads that is analogous to the matrix representation theory for linear transformations between vector spaces.

In more detail, recall that a linear transformation $L: V \rightarrow W$ between $k$-vector spaces with finite bases $\alpha = \lbrace v_1, ... , v_n \rbrace$ and $\beta = \lbrace w_1,  ..., w_m \rbrace$, respectively, admits a matrix representation $L_{\alpha \beta} \in \Mat_{k}(m, n)$ with  $ L(v_j) = \sum_i (L_{\alpha \beta})_{i,j} w_i$, such that for any vector $v$ in $V$ the coordinate vectors $L(v)_{\beta} \in k^m$ and $v_{\alpha} \in k^n$   
satisfy the equality $L(v)_{\beta} = L_{\alpha \beta} v_{\alpha}$. A great amount of linear algebra is concerned with finding bases such that the corresponding matrix representation is in an efficient shape, for instance diagonalised. The following definitions generalise the situation by substituting Kleisli morphisms for matrices. 
  
\begin{definition}
\label{representationdef}
	Let $\alpha = (Y_{\alpha}, i_{\alpha}, d_{\alpha})$ and $\beta = (Y_{\beta}, i_{\beta}, d_{\beta})$ be bases for $T$-algebras $\mathbb{X}_{\alpha} = (X_{\alpha},h_{\alpha})$ and $\mathbb{X}_{\beta} = (X_{\beta},h_{\beta})$, respectively. The basis representation $f_{\alpha \beta}$ of a $T$-algebra homomorphism $f: \mathbb{X}_{\alpha} \rightarrow \mathbb{X}_{\beta}$ with respect to $\alpha$ and $\beta$ is defined by 
	\begin{equation}
	\label{basisrepresentation}
		f_{\alpha \beta} := Y_{\alpha} \overset{i_{\alpha}}{\longrightarrow} X_{\alpha} \overset{f}{\longrightarrow} X_{\beta} \overset{d_{\beta}}{\longrightarrow} TY_{\beta}.
	\end{equation}
	 Conversely, the morphism $p^{\alpha \beta}$ associated with a Kleisli morphism $p: Y_{\alpha} \rightarrow TY_{\beta}$ with respect to $\alpha$ and $\beta$ is defined by \begin{equation}
		\label{associatedmorph}
		p^{\alpha \beta} := X_{\alpha} \overset{d_{\alpha}}{\longrightarrow} TY_{\alpha} \overset{Tp}{\longrightarrow} T^2Y_{\beta} \overset{\mu_{Y_{\beta}}}{\longrightarrow} TY_{\beta} \overset{Ti_{\beta}}{\longrightarrow} TX_{\beta} \overset{h_{\beta}}{\longrightarrow} X_{\beta}.
			\end{equation}
	\end{definition}	
	
	\begin{figure*}[t]
\[ A = L_{\alpha' \alpha'} = 
\begin{pmatrix}
	  0 & -1 \\ 1 & 0
\end{pmatrix}, 
\quad
L_{\alpha \alpha} = \begin{pmatrix}
	3 & 2 \\
	-5 & -3
\end{pmatrix},
\quad
P = \begin{pmatrix}
	-1 & 1 \\
	2 & -1
\end{pmatrix},
\quad
P^{-1} = \begin{pmatrix}
	1 & 1 \\ 2 & 1
\end{pmatrix}
\]
	\caption{The basis representation of the counter-clockwise rotation by 90 degree $L: \mathbb{R}^2 \rightarrow \mathbb{R}^2$, $L(v) = Av$ with respect to $\alpha = \lbrace (1,2), (1,1) \rbrace$ and $\alpha' = \lbrace (1,0), (0,1) \rbrace$ satisfies $L_{\alpha' \alpha'} = P^{-1} L_{\alpha \alpha} P$}
		\label{basisrepex1}
\end{figure*}

The morphism associated with a Kleisli morphism should be understood as the linear transformation between vector spaces induced by some matrix of the right type. The following result confirms this intuition.
	
\begin{lemma}
\label{associsalgebrahom}
	The function \eqref{associatedmorph} is a $T$-algebra homomorphism $p^{\alpha \beta}: \mathbb{X}_{\alpha} \rightarrow \mathbb{X}_{\beta}$.
\end{lemma}
\begin{proof}
Using \Cref{forbasis-d-isalgebrahom} we deduce the commutativity of the following diagram:
\begin{equation*}
	\begin{tikzcd}[row sep=3em, column sep = 3em]
	TX_{\alpha} \arrow{d}[left]{h_{\alpha}} \arrow{r}{Td_{\alpha}} & T^2Y_{\alpha} \arrow{r}{T^2p} \arrow{d}{\mu_{Y_{\alpha}}} & T^3Y_{\beta} \arrow{r}{T\mu_{Y_{\beta}}} \arrow{d}{\mu_{TY_{\beta}}} & T^2Y_{\beta} \arrow{d}{\mu_{Y_{\beta}}} \arrow{r}{T^2i_{\beta}} & T^2X_{\beta} \arrow{r}{Th_{\beta}} \arrow{d}{\mu_{X_{\beta}}} & TX_{\beta} \arrow{d}{h_{\beta}} \\
	X_{\alpha} \arrow{r}{d_{\alpha}} & TY_{\alpha} \arrow{r}{Tp} & T^2Y_{\beta} \arrow{r}{\mu_{Y_{\beta}}} & TY_{\beta} \arrow{r}{Ti_{\beta}} & TX_{\beta} \arrow{r}{h_{\beta}} & X_{\beta}	
	\end{tikzcd}.
\end{equation*}
\end{proof}

The next result establishes a generalisation of the observation that for fixed bases, constructing a matrix representation of a linear transformation on the one hand, and associating a linear transformation to a matrix of the right type on the other hand, are mutually inverse operations.

\begin{lemma}
\label{inverseoperations}
	The operations $\eqref{basisrepresentation}$ and $\eqref{associatedmorph}$ are mutually inverse.	
\end{lemma}
\begin{proof}
Essentially, the statement follows from the observation that, for bases, the functions involved in the composition below are isomorphisms:
\begin{align}
\label{repisos}
\begin{split}
	\Hom_{\EM}(\mathbb{X}_{\alpha}, \mathbb{X}_{\beta})
	 \overset{(d_{\beta})_* \circ (i_{\alpha}^{\sharp})^*}&{\longrightarrow} \Hom_{\EM}((TY_{\alpha}, \mu_{Y_{\alpha
}}), (TY_{\beta}, \mu_{Y_{\beta}}))  \\
 \overset{(\eta_{Y_{\alpha}})^*}&{\longrightarrow} \Hom_{\KL}(Y_{\alpha}, Y_{\beta}).	
\end{split}
\end{align}
More concretely, the definitions imply: 
	\begin{align*}
		(p^{\alpha \beta})_{\alpha \beta} &= d_{\beta} \circ (h_{\beta} \circ Ti_{\beta} \circ \mu_{Y_{\beta}} \circ Tp \circ d_{\alpha}) \circ i_{\alpha} \\
		(f_{\alpha \beta})^{\alpha \beta} &= h_{\beta} \circ Ti_{\beta} \circ \mu_{Y_{\beta}} \circ T(d_{\beta} \circ f \circ i_{\alpha}) \circ d_{\alpha}.
	\end{align*}
Using \Cref{forbasis-d-isalgebrahom} we deduce the commutativity of the diagrams below:
	\[
	\begin{gathered}
		\begin{tikzcd}[row sep=3em, column sep = 3em]
			Y_{\alpha} \arrow{d}[left]{i_{\alpha}} \arrow{rrr}{p} \arrow{dr}{\eta_{Y_{\alpha}}} & & &  TY_{\beta} \arrow{d}{\id_{TY_{\beta}}} \arrow{rr}{\id_{TY_{\beta}}} \arrow{dl}{\eta_{TY_{\beta}}} & & TY_{\beta} \\
			X_{\alpha} \arrow{r}{d_{\alpha}} & TY_{\alpha} \arrow{r}{Tp} & T^2Y_{\beta} \arrow{r}{\mu_{Y_{\beta}}}  & TY_{\beta} \arrow{r}{Ti_{\beta}} & TX_{\beta} \arrow{r}{h_{\beta}} & X_{\beta} \arrow{u}[right]{d_{\beta}}
			\end{tikzcd} \\
		\begin{tikzcd}[row sep=3em, column sep = 3em]
				X_{\alpha} \arrow{r}{\id_{X_{\alpha}}} \arrow{d}[left]{Ti_{\alpha} \circ d_{\alpha}} & X_{\alpha} \arrow{rr}{f} & & X_{\beta} \arrow{d}{d_{\beta}} \arrow{r}{\id_{X_{\beta}}}  & X_{\beta} \\
				TX_{\alpha} \arrow{ur}{h_{\alpha}} \arrow{r}{Tf} & TX_{\beta} \arrow{r}{Td_{\beta}} \arrow{urr}{h_{\beta}} & T^2Y_{\beta} \arrow{r}{\mu_{Y_{\beta}}} & TY_{\beta} \arrow{r}{Ti_{\beta}} & TX_{\beta} \arrow{u}[right]{h_{\beta}}
				\end{tikzcd}	.		
	\end{gathered}
	\]	
\end{proof}

At the beginning of this section we recalled the soundness identity $L(v)_{\beta} = L_{\alpha \beta} v_{\alpha}$ for the matrix representation $L_{\alpha \beta}$ of a linear transformation $L$. The next result is a natural generalisation of this statement.

\begin{lemma}
	\label{uniquehom}
	$f_{\alpha \beta}$ is the unique Kleisli-morphism such that $f_{\alpha \beta} \cdot d_{\alpha} = d_{\beta} \circ f $. Conversely, $p^{\alpha \beta}$ is the unique $T$-algebra homomorphism such that $p \cdot d_{\alpha} = d_{\beta} \circ p^{\alpha \beta}$.  
\end{lemma}
	\begin{proof}
	The definitions imply
	\[
	 f_{\alpha \beta} \cdot d_{\alpha} = \mu_{Y_{\beta}} \circ T(d_{\beta} \circ f \circ i_{\alpha}) \circ d_{\alpha}.
	\]
	Using \Cref{forbasis-d-isalgebrahom} we deduce the commutativity of the diagram below:
	\begin{equation*}
									\begin{tikzcd}
					X_{\alpha} \arrow{d}[left]{\id_{X_{\alpha}}} \arrow{r}{d_{\alpha}} & TY_{\alpha} \arrow{r}{Ti_{\alpha}} & TX_{\alpha} \arrow{dll}{h_{\alpha}} \arrow{r}{Tf} & TX_{\beta} \arrow{dll}{h_{\beta}} \arrow{r}{Td_{\beta}} & T^2Y_{\beta} \arrow{d}{\mu_{Y_{\beta}}} \\
					X_{\alpha} \arrow{r}[below]{f} & X_{\beta} \arrow{rrr}[below]{d_{\beta}} & & & TY_{\beta}				
									\end{tikzcd}.
	\end{equation*}
	Since an equality of the type $p \cdot d_{\alpha} = d_{\beta} \circ f$ implies 	\begin{align*}
		p &= \mu_{Y_{\beta}} \circ \eta_{TY_{\beta}} \circ p = \mu_{Y_{\beta}} \circ Tp \circ \eta_{Y_{\alpha}} = \mu_{Y_{\beta}} \circ Tp \circ d_{\alpha} \circ i_{\alpha} = d_{\beta} \circ f \circ i_{\alpha} = f_{\alpha \beta},
	\end{align*}
	the morphism $f_{\alpha \beta}$ is moreover uniquely determined. For the second part of the claim we observe that by above and \Cref{inverseoperations} it holds
	$
	p \cdot d_{\alpha} = (p^{\alpha \beta})_{\alpha \beta} \cdot d_{\alpha} = d_{\beta} \circ p^{\alpha \beta}$,
	and that an equality of the type $p \cdot d_{\alpha} = d_{\beta} \circ f$ implies $ p^{\alpha \beta} = i_{\beta}^{\sharp} \circ (p \cdot d_{\alpha}) = i^{\sharp}_{\beta} \circ d_{\beta} \circ f = f$.
\end{proof}

The next result establishes the compositionality of basis representations: the matrix representation of the composition of two linear transformations is given by the multiplication of the matrix representations of the individual linear transformations. 

\begin{lemma}
	\label{compone}
$g_{\beta \gamma} \cdot f_{\alpha \beta} = (g \circ f)_{\alpha \gamma}$.
\end{lemma}
\begin{proof}
	The definitions imply
	\begin{align*}
		g_{\beta \gamma} \cdot f_{\alpha \beta} &= \mu_{Y_{\gamma}} \circ T(d_{\gamma} \circ g \circ i_{\beta}) \circ d_{\beta} \circ f \circ i_{\alpha} \\
		(g \circ f)_{\alpha \gamma} &= d_{\gamma} \circ (g \circ f) \circ i_{\alpha}.
	\end{align*}
	We delete common terms and use  \Cref{forbasis-d-isalgebrahom} to deduce the commutativity of the diagram below:
\begin{equation*}
					\begin{tikzcd}
						 X_{\beta} \arrow{d}[left]{\id_{X_{\beta}}} \arrow{r}{d_{\beta}} & TY_{\beta} \arrow{r}{Ti_{\beta}} & TX_{\beta} \arrow{r}{Tg} \arrow{dll}{h_{\beta}} & TX_{\gamma} \arrow{r}{Td_{\gamma}} \arrow{dll}{h_{\gamma}} & T^2 Y_{\gamma} \arrow{d}{\mu_{Y_{\gamma}}} \\
						 X_{\beta} \arrow{r}[below]{g} & X_{\gamma} \arrow{rrr}[below]{d_{\gamma}} & & & TY_{\gamma}
					\end{tikzcd}.
\end{equation*}
\end{proof}

Similarly to the previous result, the next observation captures the compositionality of the operation that assigns to a Kleisli morphism its associated homomorphism.

\begin{lemma}
\label{comptwo}
	$q^{\beta \gamma} \circ p^{\alpha \beta} = (q \cdot p)^{\alpha \gamma}$. 	
\end{lemma}
\begin{proof}
	The definitions imply 
	\begin{align*}
		q^{\beta \gamma} \circ p^{\alpha \beta} &= (h_{\gamma} \circ Ti_{\gamma} \circ \mu_{Y_{\gamma}} \circ Tq \circ d_{\beta}) \circ (h_{\beta} \circ Ti_{\beta} \circ \mu_{Y_{\beta}} \circ Tp \circ d_{\alpha}) \\
		(q \cdot p)^{\alpha \gamma} &= h_{\gamma} \circ Ti_{\gamma} \circ \mu_{Y_{\gamma}} \circ T\mu_{Y_{\gamma}} \circ T^2q \circ Tp \circ d_{\alpha}.
	\end{align*}
	By deleting common terms and using the equality $d_{\beta} \circ h_{\beta} \circ Ti_{\beta} = \id_{TY_{\beta}}$ it is thus sufficient to show
	\[
	\mu_{Y_{\gamma}} \circ Tq \circ \mu_{Y_{\beta}} = \mu_{Y_{\gamma}} \circ T\mu_{Y_{\gamma}} \circ T^2q.
	\]
	Above equation follows from the commutativity of the diagram below:
	\begin{equation*}
							\begin{tikzcd}
								T^2Y_{\beta} \arrow{d}[left]{T^2q} \arrow{r}{\mu_{Y_{\beta}}} & TY_{\beta} \arrow{r}{Tq} & T^2Y_{\gamma} \arrow{d}{\mu_{Y_{\gamma}}} \\
								T^3Y_{\gamma} \arrow{urr}{\mu_{TY_{\gamma}}} \arrow{r}[below]{T\mu_{Y_{\gamma}}} & T^2 Y_{\gamma} \arrow{r}[below]{\mu_{Y_{\gamma}}} & TY_{\gamma}
							\end{tikzcd}.
	\end{equation*}
\end{proof}

The previous statements may be summarised as functors between the following two categories, which arise from the usual Eilenberg-Moore and Kleisli categories.

\begin{definition}[$\Algb(T)$ and $\Klb(T)$]
	Let $\Algb(T)$ be the category defined as follows:
		\begin{itemize}
		\item objects are given by pairs $(\mathbb{X}_{\alpha}, \alpha)$, where $\mathbb{X}_{\alpha}$ is a $T$-algebra with basis $\alpha = (Y_{\alpha}, i_{\alpha}, d_{\alpha})$; and
		\item a morphism $f: (\mathbb{X}_{\alpha}, \alpha) \rightarrow (\mathbb{X}_{\beta}, \beta)$ consists of a $T$-algebra homomorphism $f: \mathbb{X}_{\alpha} \rightarrow \mathbb{X}_{\beta}$.
		\end{itemize}
Similarly, let $\Klb(T)$ be the category defined as follows:
\begin{itemize}
		\item objects are given by pairs $(\mathbb{X}_{\alpha}, \alpha)$, where $\mathbb{X}_{\alpha}$ is a $T$-algebra with basis $\alpha = (Y_{\alpha}, i_{\alpha}, d_{\alpha})$; and
		\item a morphism $p: (\mathbb{X}_{\alpha}, \alpha) \rightarrow (\mathbb{X}_{\beta}, \beta)$ consists of a morphism $p: Y_{\alpha} \rightarrow TY_{\beta}$; the composition is given by the usual Kleisli-composition.
		\end{itemize}	
\end{definition}

\begin{corollary}
\label{isomorphismrep}
	There exist the following isomorphisms of categories:
	\[
	\Balg(T) \simeq \Algb(T) \simeq \Klb(T).
	\]
\end{corollary}
\begin{proof}
For the first isomorphism we define a functor $F: \Balg(T) \rightarrow  \Algb(T)$ by $F(\mathbb{X}_{\alpha}, \alpha) = (\mathbb{X}_{\alpha}, \alpha)$ and $F(f,p) = f$; and a functor $G: \Algb(T) \rightarrow \Balg(T)$ by $G(\mathbb{X}_{\alpha}, \alpha) = (\mathbb{X}_{\alpha}, \alpha)$ and $G(f) := (f, f_{\alpha \beta})$. The functoriality of $G$ is a consequence of \Cref{compone}. The mutual invertibility of $F$ and $G$ is a consequence of \Cref{uniquehom}.
For the second isomorphism we define a functor $F: \Algb(T) \rightarrow \Klb(T)$ by  $F(\mathbb{X}_{\alpha}, \alpha) = (\mathbb{X}_{\alpha}, \alpha)$ and $Ff = f_{\alpha \beta}$; and a functor $G: \Klb(T) \rightarrow \Algb(T)$ by  $G(\mathbb{X}_{\alpha}, \alpha) = (\mathbb{X}_{\alpha}, \alpha)$ and $Gp = p^{\alpha \beta}$. The functoriality of $F$ and $G$ is a consequence of \Cref{compone} and \Cref{comptwo}, respectively. Their mutual invertibility is a consequence \Cref{inverseoperations}.
\end{proof}

Assume we are given bases $\alpha, \alpha'$ and $\beta, \beta'$ for $T$-algebras $(X_{\alpha}, h_{\alpha})$ and $(X_{\beta}, h_{\beta})$, respectively. 
The following result clarifies how the two basis representations $f_{\alpha \beta}$ and $f_{\alpha' \beta'}$ are related.

\begin{proposition}
\label{similaritygeneral}
	There exist Kleisli isomorphisms $p$ and $q$ such that $f_{\alpha' \beta'} = q \cdot f_{\alpha \beta} \cdot p$.
\end{proposition}
\begin{proof}
	The Kleisli morphisms $p$ and $q$ and their respective candidates for inverses $p^{-1}$ and $q^{-1}$ are defined below
	\begin{alignat*}{6}
		& p &&:= d_{\alpha} \circ i_{\alpha'}: Y_{\alpha'} &&\longrightarrow TY_{\alpha} \qquad \qquad && q &&:= d_{\beta'} \circ i_{\beta}: Y_{\beta} && \longrightarrow TY_{\beta'} \\
		& p^{-1} &&:= d_{\alpha'} \circ i_{\alpha}: Y_{\alpha}  && \longrightarrow TY_{\alpha'} \qquad \qquad  &&  q^{-1} &&:= d_{\beta} \circ i_{\beta'} : Y_{\beta'} && \longrightarrow  TY_{\beta}.
	\end{alignat*}
From \Cref{forbasis-d-isalgebrahom} it follows that the diagram below commutes:
 \begin{equation*}
		\begin{tikzcd}
			Y_{\alpha} \arrow{r}{i_{\alpha}} \arrow{dd}[left]{\eta_{Y_{\alpha}}} & X_{\alpha} \arrow{d}{\id_{X_{\alpha}}} \arrow{r}{d_{\alpha'}} & TY_{\alpha'} \arrow{dd}{Ti_{\alpha'}} \\
			& X_{\alpha} \arrow{dl}{d_{\alpha}} & \\
			TY_{\alpha}  & T^2Y_{\alpha}  \arrow{l}{\mu_{Y_{\alpha}}} & TX_{\alpha} \arrow{ul}{h_{\alpha}} \arrow{l}{Td_{\alpha}}
		\end{tikzcd}.
		\end{equation*}
		This shows that $p^{-1}$ is a Kleisli right-inverse of $p$. A symmetric version of above diagram shows that $p^{-1}$ is also a Kleisli left-inverse of $p$. Analogously it follows that $q^{-1}$ is a Kleisli inverse of $q$.
		
		The definitions further imply the equalities
		\begin{align*}
			q \cdot f_{\alpha \beta} \cdot p &= \mu_{Y_{\beta'}} \circ T(d_{\beta'} \circ i_{\beta}) \circ \mu_{Y_{\beta}} \circ T(d_{\beta} \circ f \circ i_{\alpha}) \circ d_{\alpha} \circ i_{\alpha'} \\
			f_{\alpha' \beta'} &= d_{\beta'} \circ f \circ i_{\alpha'}.
		\end{align*}
		We delete common terms and use \Cref{forbasis-d-isalgebrahom} to establish the commutativity of the diagram below:
		\begin{equation*}
			\begin{tikzcd}
			X_{\alpha} \arrow{dr}{\id_{X_{\alpha}}} \arrow{r}{d_{\alpha}} \arrow{dd}[left]{f} & TY_{\alpha} \arrow{r}{Ti_{\alpha}} & TX_{\alpha} \arrow{dl}[above]{h_{\alpha}} \arrow{r}{Tf} & TX_{\beta} \arrow{ddll}{h_{\beta}} \arrow{d}{Td_{\beta}} \\
			 & X_{\alpha} \arrow{d}{f} &  & T^2Y_{\beta} \arrow{dd}{\mu_{Y_{\beta}}} \\
	X_{\beta} \arrow{d}[left]{d_{\beta'}} & X_{\beta} \arrow{l}[above]{\id_{X_{\beta}}} \arrow{drr}{d_{\beta}} & & \\
			TY_{\beta'} & T^2 Y_{\beta'} \arrow{l}{\mu_{Y_{\beta'}}} & TX_{\beta} \arrow{llu}{h_{\beta}} \arrow{l}{Td_{\beta'}} & TY_{\beta} \arrow{l}{Ti_{\beta}}
		\end{tikzcd}.
		\end{equation*}	
\end{proof}

Above result simplifies if one restricts to an endomorphism: the basis representations are \emph{similar}. This generalises the situation for vector spaces, cf. \Cref{basisrepex1}.

\begin{corollary}
	\label{similarity}
There exists a Kleisli isomorphism $p$ with Kleisli inverse $p^{-1}$ such that $f_{\alpha' \alpha'} = p^{-1} \cdot f_{\alpha \alpha} \cdot p$.
\end{corollary}
\begin{proof}
	In \Cref{similaritygeneral} let $\beta = \alpha$ and $\beta' = \alpha'$. One verifies that in the corresponding proof the definitions of the morphisms $p^{-1}$ and $q$ coincide.
\end{proof}

\subsection{Bases for Bialgebras}

\label{basesforbialgebrasec}

This subsection is concerned with generators and bases for a bialgebra. It is well-known \cite{turi1997towards} that an Eilenberg-Moore law $\lambda$ between a monad $T$ and an endofunctor $F$ induces simultaneously
\begin{itemize}
	\item a monad $T_{\lambda} = (T_{\lambda}, \mu, \eta)$ on $\Coalg(F)$ by $T_{\lambda}(X,k) = (TX, \lambda_X \circ Tk)$ and $T_{\lambda}f = Tf$; and
	\item an endofunctor $F_{\lambda}$ on $\EM$ by $F_{\lambda}(X,h) = (FX, Fh \circ \lambda_X)$ and $F_{\lambda}f = Ff$,
\end{itemize}
such that the algebras over $T_{\lambda}$, the coalgebras of $F_{\lambda}$, and $\lambda$-bialgebras coincide.
 We will consider generators and bases for $T_{\lambda}$-algebras, or equivalently, $\lambda$-bialgebras.

By \Cref{generatordefinition}, a generator for a $\lambda$-bialgebra $(X,h,k)$ consists of an $F$-coalgebra $(Y,k_Y)$ and morphisms $i: Y \rightarrow X$ and $d: X \rightarrow TY$, such that the three diagrams on the left below commute:
\begin{equation}
\label{bialgebrageneratorequations}	
	\begin{tikzcd}[column sep = 1.5em]
			Y \arrow{r}{i} \arrow{d}[left]{k_Y} & X \arrow{d}{k} \\
		FY \arrow{r}{Fi} &FX
	\end{tikzcd}
	\quad
		\begin{tikzcd}[column sep = 1.5em]
		X \arrow{r}{d} \arrow{d}[left]{k} & TY \arrow{d}[left]{\lambda_Y \circ Tk_Y} \\
		FX \arrow{r}{Fd} &FTY
	\end{tikzcd}
	\quad
	\begin{tikzcd}[column sep = 1.5em]			TY \arrow{r}{Ti} & TX \arrow{d}{h}\\
			X \arrow{u}{d} \arrow{r}{\id_X}  & X 
		\end{tikzcd}
		\quad
		\begin{tikzcd}[column sep = 1.5em]
			TX \arrow{r}{h} & X \arrow{d}{d}\\
			TY \arrow{u}{Ti} \arrow{r}{\id_{TY}}  & TY 
		\end{tikzcd}.
\end{equation}

A basis for a $\lambda$-bialgebra is a generator, such that in addition the diagram on the right above commutes.

It is easy to verify that by forgetting the $F$-coalgebra structure, every generator for a bialgebra in particular provides a generator for its underlying $T$-algebra. By \Cref{forgenerator-isharp-is-bialgebra-hom} it thus follows that there exists a $\lambda$-bialgebra homomorphism 
$ i^{\sharp}: \expa(Y, Fd \circ k \circ i) \rightarrow (X,h,k)$. The next result establishes that there exists a second equivalent free bialgebra with a different coalgebra structure.

\begin{lemma}
\label{generatorbialgebraisharp}
	Let $(Y,k_Y, i,d)$ be a generator for $(X,h,k)$. Then $i^{\sharp}: TY \rightarrow X$ is a $\lambda$-bialgebra homomorphism $i^{\sharp} : \free_T(Y, k_Y) \rightarrow (X,h,k)$.
\end{lemma}
\begin{proof}
By definition we have $\free_T(Y, k_Y) = (TY, \mu_Y, \lambda_Y \circ Tk_Y)$.
	Clearly the lifting $i^{\sharp} = h \circ Ti$ is a $T$-algebra homomorphism $i^{\sharp}: (TY, \mu_Y) \rightarrow (X,h)$. It is an $F$-coalgebra homomorphism $i^{\sharp}: (TY, \lambda_Y \circ Tk_Y) \rightarrow (X, k)$  since the diagram below commutes:
	\begin{equation*}
		\begin{tikzcd}
			TY \arrow{r}{Ti} \arrow{d}[left]{Tk_Y} & TX \arrow{r}{h} \arrow{d}{Tk} & X \arrow{dd}{k} \\
			TFY \arrow{r}{TFi} \arrow{d}[left]{\lambda_Y} & TFX \arrow{d}{\lambda_X} & \\
			FTY \arrow{r}{FTi} & FTX \arrow{r}{Fh} & FX
		\end{tikzcd}.
	\end{equation*}	
\end{proof}

If one moves from generators to bases for bialgebras, both structures coincide.

\begin{lemma}
\label{samecoalgebrastructures}
	Let $(Y,k_Y, i, d)$ be a basis for $(X,h,k)$, then $\free_T(Y, k_Y) = \expa_T(Y, Fd \circ k \circ i)$.
\end{lemma}
\begin{proof}
	Using \Cref{forbasis-d-isalgebrahom} we establish the commutativity of the diagram below:
	\begin{equation*}
	\begin{tikzcd}[column sep=2em]
		TY \arrow{dd}[left]{\id_{TY}}  \arrow{rr}{Tk_Y}  & & TFY \arrow{rr}{\lambda_Y} \arrow{dd}{TFi} & & FTY  \arrow{d}{FTi} \arrow{rr}{\id_{FTY}} && FTY \arrow{dd}{\id_{FTY}} \\
	& & & & FTX \arrow{d}{FTd} \arrow{r}{Fh} & FX \arrow{dr}{Fd} & \\
	TY \arrow{r}{Ti} & TX \arrow{r}{Tk} & TFX \arrow{r}{TFd} & TFTY \arrow{r}{\lambda_{TY}} & FT^2Y  \arrow{rr}{F\mu_{Y}} & & FTY	
	\end{tikzcd}.
	\end{equation*}
\end{proof}

\begin{example}[Canonical RFSA]
	In \Cref{canonicalrfsaexample} we considered the generator $(J(L), i, d)$ with $i(y) = y$ and $d(x) = \lbrace y \in J(L) \mid y \leq x \rbrace$ for the underlying algebraic part of the minimal pointed bialgebra $(X, h, k)$ to recover the canonical RFSA for $L= (a+b)^*a$ as the coalgebra $Fd \circ k \circ i$. \Cref{crfsadiag1} shows that the coalgebraic part $k$ restricts to the join-irreducibles $J(L)$, suggesting $\alpha = (J(L), k, i, d)$ as a possible generator for the full bialgebra. However, as one easily verifies, the $a$-action on $\lbrack \lbrace y \rbrace \rbrack$ implies the non-commutativity of the second diagram on the left of \eqref{bialgebrageneratorequations}. The issue can be fixed by modifying the definition of $d$ by $d(\lbrack \lbrace y \rbrace \rbrack) := \lbrace \lbrack \lbrace y \rbrace \rbrack \rbrace$. In consequence $\free_{\mathcal{P}}(J(L)), k)$ and $\expa_{\mathcal{P}}(J(L), Fd \circ k \circ i)$ coincide (even though the modification does not yield a basis).	\end{example}

We close this section by observing that a basis for the underlying algebra of a bialgebra is sufficient for constructing a generator for the full bialgebra.

\begin{lemma}
	\label{generatorfullbialgebra}
	Let $(X,h,k)$ be a $\lambda$-bialgebra and $(Y,i,d)$ a basis for the $T$-algebra $(X,h)$. Then $(TY,F\mu_Y \circ \lambda_{TY} \circ T(Fd \circ k \circ i)
,h \circ Ti,\eta_{TY} \circ d)$ is a generator for $(X,h,k)$.
\end{lemma}
\begin{proof}
	In the following we abbreviate $k_{TY} := F\mu_Y \circ \lambda_{TY} \circ T(Fd \circ k \circ i): TY \rightarrow FTY$. By \Cref{forgenerator-isharp-is-bialgebra-hom} the lifting $h \circ Ti$ is an $F$-coalgebra homomorphism $h \circ Ti: (TY, k_{TY}) \rightarrow (X, k)$. This shows the commutativity of the diagram on the left of \eqref{bialgebrageneratorequations}. 
	By \Cref{forbasis-isharp-is-bialgebra-iso} the morphism $d$ is an $F$-coalgebra homomorphism in the reverse direction. Together with the commutativity of the diagram on the left below this implies the commutativity of the second diagram to the left of \eqref{bialgebrageneratorequations}:
	\[
	\begin{tikzcd}	TY \arrow{r}{\eta_{TY}} \arrow{dd}[left]{k_{TY}} & T^2Y \arrow{d}{Tk_{TY}} \\
	& TFTY \arrow{d}{\lambda_{TY}} \\
	FTY \arrow{ur}{\eta_{FTY}} \arrow{r}{F\eta_{TY}} & FT^2Y	
	\end{tikzcd}
	\qquad
	\begin{tikzcd}				T^2Y \arrow{r}{T^2i} & T^2X \arrow{rr}{Th}  \arrow{dr}{\mu_X} & &  TX \arrow{dd}{h} \\
				TY \arrow{u}{\eta_{TY}} \arrow{r}{Ti} & TX \arrow{u}{\eta_{TX}} \arrow{r}{\id_{TX}} & TX \arrow{dr}{h}  \\
				X \arrow{u}{d} \arrow{rrr}{\id_X} & & & X
			\end{tikzcd}.
	\] 		
	Similarly, the commutativity of third diagram to the left of \eqref{bialgebrageneratorequations} follows from the commutativity of the diagram on the right above.
\end{proof}

\subsection{Bases as Coalgebras}

\label{basesascoaglebrassec}

In this section, we compare our approach with an alternative, coalgebraic, perspective on the generalisation of bases. 
More specifically, we are interested in the work of Jacobs \cite{jacobs2011bases}, where a basis for an algebra over a monad is defined as a coalgebra for the comonad on the category of Eilenberg-Moore algebras induced by the free algebra adjunction. Explicitly, a basis for a $T$-algebra $(X,h)$, in the sense of \cite{jacobs2011bases}, consists of a $T$-coalgebra $(X,k)$ such that the following three diagrams commute:

\begin{equation}
\label{jacobsbasis}
\begin{tikzcd}
			TX \arrow{d}[left]{h} \arrow{r}{Tk} & T^2X \arrow{d}{\mu_X} \\
			X \arrow{r}{k} & TX
		\end{tikzcd}
		\qquad
	\begin{tikzcd}			X \arrow{dr}[left]{\id_X} \arrow{r}{k} & TX \arrow{d}{h} \\
			& X
		\end{tikzcd}	
		\qquad
		\begin{tikzcd}
			X \arrow{d}[left]{k} \arrow{r}{k} & TX \arrow{d}{T\eta_X} \\
			TX \arrow{r}{Tk} & T^2X
		\end{tikzcd}.
\end{equation}

Next we show that a basis as in \Cref{generatordefinition} induces a basis as in  \cite{jacobs2011bases}.

\begin{lemma}
\label{impliesjacobsbasis}
	Let $(Y,i,d)$ be a basis for a $T$-algebra $(X,h)$. Then \eqref{jacobsbasis} commutes for $k := Ti \circ d$.
\end{lemma}
\begin{proof}
	The commutativity of the diagram on the left of  \eqref{jacobsbasis} follows from the naturality of $\mu$ and \Cref{forbasis-d-isalgebrahom}. The diagram in the middle of  \eqref{jacobsbasis} commutes by the definition of a generator. The commutativity of the diagram on the right of  \eqref{jacobsbasis} is again a consequence of \Cref{forbasis-d-isalgebrahom}:
	\begin{equation*}
				\begin{tikzcd}
					X \arrow{r}{d} \arrow{d}[left]{d} & TY \arrow{r}{Ti} & TX \arrow{dd}{T \eta_X} \\
					TY \arrow{ur}[right]{\id_{TY}} \arrow{d}[left]{Ti} \arrow{dr}{T \eta_Y} \\
					TX \arrow{r}{Td} & T^2Y \arrow{r}{T^2i} & T^2X
				\end{tikzcd}.
	\end{equation*}
\end{proof}

Conversely, assume $(X,k)$ is a $T$-coalgebra structure satisfying \eqref{jacobsbasis} and $i_k: Y_k \rightarrow X$ an equaliser  of $k$ and $\eta_X$. If the underlying category is the category of sets and functions, the equaliser of any two functions exists. If $Y_k$ is non-empty, one can show that the equaliser is preserved under $T$, that is, $Ti_k$ is an equaliser of $Tk$ and $T\eta_X$ \cite{jacobs2011bases}. By \eqref{jacobsbasis} we have $Tk \circ k = T\eta_X \circ k$. Thus there exists a unique morphism $d_k: X \rightarrow TY_k$ such that $Ti_k \circ d_k = k$, which can be shown to be the inverse of $h \circ Ti_k$ \cite{jacobs2011bases}. In other words, $G(X,k) := (Y_k, i_k, d_k)$ is a basis for $(X,h)$ in the sense of \Cref{generatordefinition}. In the following let $F(Y,i,d) := (X,Ti \circ d)$ for any basis of $(X,h)$.

\begin{lemma}
\label{equaliserlemma}	
	Let $(Y,i,d)$ be a basis for a $T$-algebra $(X,h)$ and $k := Ti \circ d$. Then $\eta_X \circ i = k \circ i$ and $Tk \circ (\eta_X \circ i) = T\eta_{X} \circ (\eta_X \circ i)$.  
\end{lemma}
\begin{proof}
The statement follows from \Cref{forbasis-d-isalgebrahom}:
\begin{equation*}
\begin{tikzcd}Y \arrow{rr}{i} \arrow{d}{i} \arrow{dr}{\eta_Y} && X \arrow{d}{\eta_X} \\
X \arrow{r}{d} & TY \arrow{r}{Ti} & TX 	
\end{tikzcd}
\qquad
	\begin{tikzcd}			Y \arrow{r}{i} \arrow{dd}{i} \arrow{dr}{\eta_Y} & X \arrow{d}{d} \arrow{r}{\eta_X} & TX \arrow{r}{Td} & T^2Y \arrow{dd}{T^2i} \\
		& TY \arrow{d}{Ti} \arrow{rru}{\eta_{TY}} & & \\
		X \arrow{r}{\eta_X} \arrow[bend right]{rrr}{T\eta_X \circ \eta_X} & TX \arrow{rr}{\eta_{TX}} & & T^2X
	\end{tikzcd}.
\end{equation*}
\end{proof}

\begin{corollary}
\label{uniquemorphism}
	Let $\alpha := (Y,i,d)$  be a basis for a set-based $T$-algebra $(X,h)$ and $k := Ti \circ d$. Let $i_k: Y_k \rightarrow X$ be an equaliser of $k$ and $\eta_X$, and $Y_k$ non-empty, then $(\id_{(X,h)})_{\alpha, GF\alpha}: Y \rightarrow TY_{k}$ is the unique morphism $\psi$ making the diagram below commute:
\[
\begin{tikzcd}[ampersand replacement=\&, row sep = 0.75em]
Y \arrow[dashed]{r}[below]{\psi} \arrow[bend left=20]{rr}{\eta_X \circ i} \& TY_k \arrow{r}[below]{Ti_k}	 \& TX \arrow[shift left=2]{r}[above]{Tk} \arrow[shift right=1.5]{r}[below]{T\eta_X} \& T^2X
\end{tikzcd}.
\]	
\end{corollary}
\begin{proof}
Since $i_k$ is an equaliser of $k$ and $\eta_X$, it follows from \Cref{equaliserlemma} that there exists a unique morphism $\varphi: Y \rightarrow Y_k$ such that $i_k \circ \varphi = i$. Since $Y_k$ is non-empty, $Ti_k$ is an equaliser of $Tk$ and $T\eta_X$ \cite{jacobs2011bases}. It follows from \Cref{equaliserlemma} that there exists a unique morphism $\psi: Y \rightarrow TY_k$ such that $Ti_k \circ \psi = \eta_X \circ i$. It is not hard to see that $\psi = \eta_{Y_k} \circ \varphi$. The statement thus follows from $
(\id_{(X,h)})_{\alpha, GF\alpha} = d_k \circ i = d_k \circ i_k \circ \varphi = \eta_{Y_k} \circ \varphi = \psi$.
\end{proof}

\subsection{Signatures, Equations, and Finitary Monads}

\label{varietiessec}

Most of the algebras over set monads one usually considers generators for constitute finitary varieties in the sense of universal algebra. In this section, we will briefly explore the consequences for generators that arise from this observation. The constructions are well-known; we include them for completeness.

 Let $\Sigma$ be a set, whose elements we think of as \emph{operations}, and $\textnormal{ar}: \Sigma \rightarrow \mathbb{N}$ a function that assigns to an operation its \emph{arity}. Any such \emph{signature} induces a set endofunctor $H_{\Sigma}$ defined on a set as the coproduct $H_{\Sigma}X = \coprod_{\sigma \in \Sigma} X^{\textnormal{ar}(\sigma)}$, and consequently, a set monad $\mathbb{S}_{\Sigma}$ that assigns to a set $V$ of variables the initial algebra $S_{\Sigma} V = \mu X.(V + H_{\Sigma} X)$, i.e. the set of $\Sigma$-terms generated by $V$ (see e.g. \cite{turi1996functorial}). One can show that the categories of $H_{\Sigma}$-algebras and $\mathbb{S}_{\Sigma}$-algebras are isomorphic. A $\mathbb{S}_{\Sigma}$-algebra $\mathbb{X}$ \emph{satisfies a set of equations} $E \subseteq S_{\Sigma} V \times S_{\Sigma} V$, if for all $(s,t) \in E$ and valuations $v: V \rightarrow X$ it holds $v^{\sharp}(s) = v^{\sharp}(t)$, where $v^{\sharp}: (S_{\Sigma}V, \mu_V) \rightarrow \mathbb{X}$ is the unique extension of $v$ to a $\mathbb{S}_{\Sigma}$-algebra homomorphism \cite{adamek1994locally}. The set of $\mathbb{S}_{\Sigma}$-algebras that satisfy $E$ is denoted by $\Alg(\Sigma, E)$. As one verifies, the forgetful functor $U: \Alg(\Sigma, E) \rightarrow \Set$ admits a left-adjoint $F: \Set \rightarrow \Alg(\Sigma, E)$, thus resulting in a set monad $T_{\Sigma, E}$ with underlying endofunctor $U \circ F$ that preserves directed colimits. The functor $U$ can be shown to be monadic, that is, the comparison functor $K: \Alg(\Sigma, E) \rightarrow \Set^{T_{\Sigma, E}}$ is an isomorphism \cite{mac2013categories}. In other words, the category of Eilenberg-Moore algebras over $T_{\Sigma, E}$ and the finitary variety of algebras over $\Sigma$ and $E$ coincide. In fact, set monads preserving directed colimits (\emph{finitary} monads \cite{adamek1994locally}) and finitary varieties are in \emph{bijection}.
 
The following result characterises generators for algebras over $T_{\Sigma, E}$. It can be seen as a unifying proof for observations analogous to the one in \Cref{setbasedexamples}. For any $\Sigma$-term $t \in \mathbb{S}_{\Sigma}V$ over variables $V$, let $\llbracket t \rrbracket_E \in S_{\Sigma}V /{\cong_E}$ denote the equivalence class of $t$ w.r.t. the smallest congruence relation $\cong_E$ on $S_{\Sigma}V$ generated by the equations $E$.

\begin{lemma}
\label{sigmatermlemma}
A morphism $i: Y \rightarrow X$ is part of a generator for a $T_{\Sigma, E}$-algebra $\mathbb{X}$ iff every element $x \in X$ can be expressed as a $\Sigma$-term in $i\lbrack Y \rbrack$ modulo $E$, that is, there is a term $d(x) \in S_{\Sigma}Y$ such that $i^{\sharp}(\llbracket d(x) \rrbracket_E) = x$.
\end{lemma}
\begin{proof}
	Let $i: Y \rightarrow X$ be part of a generator $(Y, i, \overline{d})$ for a $T_{\Sigma, E}$-algebra $\mathbb{X}$. Then any $x \in X$ admits some $\overline{d}(x) \in T_{\Sigma, E}Y$ such that $i^{\sharp}(\overline{d}(x)) = x$, where $i^{\sharp}: (T_{\Sigma, E}Y, \mu_Y) \rightarrow \mathbb{X}$. By construction $T_{\Sigma, E}Y = UFY$, where $FY = S_{\Sigma}Y /{\cong_E}$ is the set of $\Sigma$-terms generated by $Y$ modulo the smallest congruence $\cong_E$ generated by $E$. Let $d(x) \in S_{\Sigma}Y$ be any representative of $\overline{d}(x) \in T_{\Sigma, E}Y$, that is, such that $\llbracket d(x) \rrbracket_E = \overline{d}(x)$. Then it follows $i^{\sharp}(\llbracket d(x) \rrbracket_E) = i^{\sharp}(\overline{d}(x)) = x$.
	
	Conversely, assume we have a $T_{\Sigma, E}$-algebra $\mathbb{X}$ and for any $x \in X$ there exists a term $d(x) \in S_{\Sigma}Y$ such that $i^{\sharp}(\llbracket d(x) \rrbracket_E) = x$. Then we can define a function $\overline{d}: X \rightarrow T_{\Sigma, E}Y$ by $\overline{d}(x) = \llbracket d(x) \rrbracket_E$. It immediately follows $i^{\sharp}(\overline{d}(x)) = i^{\sharp}(\llbracket d(x) \rrbracket_E) = x$, which shows that $(Y, i, \overline{d})$ is a generator for $\mathbb{X}$. 
\end{proof}

\subsection{Finitely Generated Objects}

\label{finitelygeneratedsec}

In this section, we relate our abstract definition of a generator to the theory of \emph{locally finitely presentable} categories, in particular, to the notions of \emph{finitely generated} and \emph{finitely presentable} objects, which are categorical abstractions of finitely generated algebraic structures.

For intuition, recall that an element $x \in X$ of a partially ordered set is \emph{compact}, if for each directed set $D \subseteq X$ with $x \leq \bigvee D$, there exists some $d \in D$ satisfying $x \leq d$. An \emph{algebraic lattice} is a partially ordered set that has all joins, and every element is a join of compact elements. The naive categorification of compact elements is equivalent to the following definition: a object $Y$ in $\mathscr{C}$ is \emph{finitely presentable (generated)}, if $\Hom_{\mathscr{C}}(Y, -): \mathscr{C} \rightarrow \Set$ preserves filtered colimits (of monomorphisms). Consequently, one can categorify algebraic lattices as  \emph{locally finitely presentable} (lfp) categories, which are cocomplete and admit a set of finitely presentable objects, such that every object is a filtered colimit of objects from 	that set \cite{adamek1994locally}.

In \cite[Theor. 3.5]{adamek2019finitely} it is shown that an algebra $\mathbb{X}$ over a finitary monad $T$ on an lfp category $\mathscr{C}$ is a finitely generated object of $\EM$ iff there exists a finitely presentable object $Y$ of $\mathscr{C}$ and a morphism $i: Y \rightarrow X$, such that $i^{\sharp}: (TY, \mu_Y) \rightarrow \mathbb{X}$ is a strong\footnote{An epimorphism $e: A \rightarrow B$ is said to be \emph{strong}, if for any monomorphism $m: C \rightarrow D$ and any morphisms $f: A \rightarrow C$ and $g: B \rightarrow D$ such that $g \circ e = m \circ f$, there exists a diagonal monomorphism $d: B \rightarrow C$ such that $f = d \circ e$ and $g = m \circ d$.} epimorphism in $\EM$. Below, we give a variant of this statement where instead the carrier of $i^{\sharp}$ is a split\footnote{A morphism $e: A \rightarrow B$ is called \emph{split}, if there exists a morphism $s: B \rightarrow A$ such that $e \circ s = \id_B$. Any morphism that is split is necessarily a strong epimorphism.} epimorphism in $\mathscr{C}$, which is the case iff $\mathbb{X}$ admits a generator in the sense of \Cref{generatordefinition}.
  
 \begin{proposition}
 \label{finitarymonadlemma}	
 Let $\mathscr{C}$ be a lfp category in which strong epimorphisms split and $T$ a finitary monad on $\mathscr{C}$ preserving epimorphisms.
 Then an algebra $\mathbb{X}$ over $T$ is a finitely generated object of $\EM$ iff it is generated by a finitely presentable object $Y$ in $\mathscr{C}$ in the sense of \Cref{generatordefinition}.
 \end{proposition}
\begin{proof}
Assume that an algebra $\mathbb{X}$ over $T$ is a finitely generated object of $\EM$. From \cite[Theor. 3.5]{adamek2019finitely} it follows that there exists a finitely presentable object $Y$ of $\mathscr{C}$ and a morphism $i: Y \rightarrow X$ such that $i^{\sharp}: (TY, \mu_Y) \rightarrow \mathbb{X}$ is a strong epimorphism in $\EM$. Since $T$ preserves epis, it is sound to assume that the carrier $i^{\sharp}: TY \rightarrow X$ is a strong epimorphism in $\mathscr{C}$. (This is because the proof of \cite[Theor. 3.5]{adamek2019finitely} can be modified by replacing the (strong-epi, mono)-factorisation system of the lfp category $\EM$ (cf. \cite[Remark 2.2.1]{adamek2019finitely} and \cite[Remark 3.1]{adamek2019finitely}) with the factorisation system for $\EM$ induced by lifting the (strong-epi, mono)-factorisation system of $\mathscr{C}$. The lifted factorisation system (cf. \Cref{factorisationsystemalgebra}) consists of those algebra homomorphisms whose carrier is a strong-epi- or monomorphism in $\mathscr{C}$, respectively.) By assumption, $i^{\sharp}: TY \rightarrow X$ thus splits in $\mathscr{C}$, that is, there exists at least one morphism $d: X \rightarrow TY$ in $\mathscr{C}$ such that $i^{\sharp} \circ d = \id_X$. This shows that $(Y, i, d)$ is a generator for $\mathbb{X}$ in the sense of \Cref{generatordefinition}.

Conversely, assume that an algebra $\mathbb{X}$ over $T$ is generated by $(Y, i, d)$, where $Y$ is a finitely presentable object in $\mathscr{C}$. Then $d$ witnesses that $i^{\sharp}: TY \rightarrow X$ splits in $\mathscr{C}$. Since every split epimorphism is necessarily strong, $i^{\sharp}: TY \rightarrow X$ thus is a strong epimorphism in $\mathscr{C}$. It immediately follows that $i^{\sharp}: (TY, \mu_Y) \rightarrow \mathbb{X}$ is an epimorphism in $\EM$. Since $T$ preserves epis, it also is a \emph{strong} epimorphism in $\EM$. From \cite[Theor. 3.5]{adamek2019finitely} it follows that the algebra $\mathbb{X}$ over $T$ is a finitely generated object of $\EM$. 
\end{proof}

\section{Related Work}

\label{relatedworksec_closure}

One of the motivations for this chapter has been our broad interest in active learning algorithms for state-based models \cite{angluin1987learning}, in particular automata for NetKAT \cite{anderson2014netkat}, a formal system for the verification of networks based on Kleene Algebra with Tests \cite{kozen1996kleene}.   
One of the main challenges in learning non-deterministic models such as NetKAT automata is the common lack of a unique minimal acceptor for a given language \cite{denis2001residual}. The problem has been independently approached for different variants of non-determinism, often with the common idea of finding a subclass admitting a unique representative \cite{esposito2002learning, berndt2017learning}. More general and unifying perspectives were given by van Heerdt  \cite{van2020learning, van2016master, van2020phd} and Myers et al. \cite{MyersAMU15}. One of the central notions in the work of van Heerdt is the concept of a scoop, originally introduced by Arbib and Manes \cite{arbib1975fuzzy}.

 In \Cref{canonical_automata_chapter} we have presented a categorical framework that recovers minimal non-deterministic representatives in two steps. The framework is based on ideas closely related to the ones in \cite{MyersAMU15}, adopts scoops under the name generators and strengthens the former to the notion of a basis (\Cref{generatordefinition}). In a first step, the framework constructs the minimal bialgebra accepting a given regular language, by closing the minimal coalgebra with additional algebraic structure over a monad. In a second step, it identifies generators for the algebraic part of the bialgebra, to derive an equivalent coalgebra with side effects in a monad. 
  In this chapter, we generalise the first step as application of a monad on an appropriate category of subobjects with respect to an $(\mathscr{E}, \mathscr{M})$-factorisation system, and explore the second step by further developing the abstract theory of generators and bases.
  
  	Categorical factorisation systems are well-established \cite{bousfield1977constructions,riehl2008factorization,maclane1950duality}. Among others, they have been used for a general view on the minimisation and determinisation of state-based systems \cite{adamek2009abstract, adamek2012coalgebraic, wissmann2022minimality}. In \Cref{closure_sec} we use the formalism of \cite{adamek2009abstract}. In \Cref{factorisationsystemalgebra} we have shown that under certain assumptions factorisation systems can be lifted to the categories of algebras and coalgebras. We later realised that the constructions had recently been published in \cite{wissmann2022minimality}.
  
The notion of a basis for an algebra over an arbitrary monad has been subject of previous interest. Jacobs, for instance, defines a basis as a coalgebra for the comonad on the category of algebras induced by the free algebra adjunction \cite{jacobs2011bases}. In \Cref{basesascoaglebrassec} we have shown that a basis in our sense always induces a basis in their sense, and, conversely, it is possible to recover a basis in our sense from a basis in their sense, if certain assumptions about the existence and preservation of equaliser are given. As equaliser do not necessarily exist and are not necessarily preserved, our approach carries additional data and thus can be seen as finer.

\section{Discussion and Future Work}

We generalised the closure of a subset of an algebraic structure with respect to the latter as a monad between categories of subobjects relative to a factorisation system. We have identified the closure of a minimal coalgebra with additional algebraic structure as an instance of the closure of subobjects that arise by taking the image of a morphism. We have extended the notion of a generator to a category of algebras with generators, and explored its characteristics.
  We have generalised the matrix representation theory of vector spaces and discussed bases for bialgebras. We compared our ideas with a coalgebraic generalisation of bases, explored the case in which a monad is induced by a variety, and briefly related our notion to finitely generated objects in finitely presentable categories. 
  
In \Cref{canonical_automata_chapter} we have shown that generators and bases in the sense of \Cref{generatorsandbases_sec} are central ingredients in the definitions of minimal canonical acceptors. Many such acceptors admit double-reversal characterisations \cite{brzozowski1962canonical,BrzozowskiT14,MyersAMU15,VuilleminG210}. Duality based characterisations as the former have been shown to be closely related to minimisation procedures with respect to factorisation systems \cite{bonchi2012brzozowski,bonchi2014algebra,wissmann2022minimality}. In the future, it would be interesting to further explore the connection between the minimality of generators on the one side, and the notion of minimality of an acceptor with respect to a factorisation system on the other side.

Another interesting question is whether the construction that underlies our definition of a monad in \Cref{inducedmonad} could be introduced at a more general level of an arbitrary adjunction between categories with suitable factorisation systems, such that the adjunction between the base category $\mathscr{C}$ and the category of Eilenberg-Moore algebras $\EM$ is a special case.

Overall, our presentation primarily focused on applications to coalgebra. It would be valuable to also explore possible implications to other fields, for instance the minimisation of logical formulae or proofs. 
\cleardoublepage
\addcontentsline{toc}{chapter}{Bibliography}
\printbibliography

@article{smolka2019guarded,
  title={Guarded Kleene Algebra with Tests: Verification of Uninterpreted Programs in Nearly Linear Time},
  author={Smolka, Steffen and Foster, Nate and Hsu, Justin and Kapp{\'e}, Tobias and Kozen, Dexter and Silva, Alexandra},
  journal={Proceedings of the ACM on Programming Languages},
  volume={4},
  number={POPL},
  pages={1--28},
  year={2019},
publisher = {Association for Computing Machinery}, 
doi = {10.1145/3371129}, 
}

@inproceedings{schmid2021guarded,
  title={Guarded Kleene Algebra with Tests: Coequations, Coinduction, and Completeness},
  author={Schmid, T and Kapp{\'e}, T and Kozen, D and Silva, A},
  booktitle =	{48th International Colloquium on Automata, Languages, and Programming (ICALP 2021)},
   series =	{Leibniz International Proceedings in Informatics (LIPIcs)},
  volume={\\198},
  pages={142},
  year={2021},
  publisher={Schloss Dagstuhl -- Leibniz-Zentrum f{\"u}r Informatik},
    doi =		{10.4230/LIPIcs.ICALP.2021.142},

}

@techreport{kozen2001automata,
  title={Automata on Guarded Strings and Applications},
  author={Kozen, Dexter},
  year={2001},
  institution={Cornell University}
}

@article{angluin1987learning,
  title={Learning Regular Sets from Queries and Counterexamples},
  author={Angluin, Dana},
  journal={Information and Computation},
  volume={75},
  number={2},
  pages={87--106},
  year={1987},
  publisher={Elsevier},
  doi={10.1016/0890-5401(87)90052-6}
}

@inproceedings{kozen1996kleene,
  title={Kleene Algebra with Tests: Completeness and Decidability},
  author={Kozen, Dexter and Smith, Frederick},
  booktitle={International Workshop on Computer Science Logic},
  pages={244--259},
  year={1997},
  organization={Springer},
  doi={10.1007/3-540-63172-0_43}
}

@article{kozen1994completeness,
  title={A Completeness Theorem for {Kleene} Algebras and the Algebra of Regular Events},
  author={Kozen, Dexter},
  journal={Information and Computation},
  volume={110},
  number={2},
  pages={366--390},
  year={1994},
  publisher={Elsevier},
  doi = {10.1006/inco.1994.1037},
}

@inbook{kleene1951representation,
title = {Representation of Events in Nerve Nets and Finite Automata},
booktitle = {Automata Studies},
volume={34},
author = {S. C. Kleene},
publisher = {Princeton University Press},
pages = {3--42},
doi = {10.1515/9781400882618-002},
year = {1956},
}

@inproceedings{smolka2019scalable,
  title={Scalable Verification of Probabilistic Networks},
  author={Smolka, Steffen and Kumar, Praveen and Kahn, David M and Foster, Nate and Hsu, Justin and Kozen, Dexter and Silva, Alexandra},
  booktitle={Proceedings of the 40th ACM SIGPLAN Conference on Programming Language Design and Implementation},
  pages={190--203},
  year={2019},
  publisher = {Association for Computing Machinery}, 
  doi = {10.1145/3314221.3314639}, 
  numpages = {14}, 
  series = {PLDI 2019} 
}

@article{anderson2014netkat,
  title={NetKAT: Semantic Foundations for Networks},
  author={Anderson, Carolyn Jane and Foster, Nate and Guha, Arjun and Jeannin, Jean-Baptiste and Kozen, Dexter and Schlesinger, Cole and Walker, David},
  journal={ACM SIGPLAN Notices},
  volume={49},
  number={1},
  pages={113--126},
  year={2014},
  publisher={Association for Computing Machinery},
  doi={10.1145/2578855.2535862}
}

@article{kozen1997kleene,
  title={Kleene Algebra with Tests},
  author={Kozen, Dexter},
  journal={ACM Transactions on Programming Languages and Systems (TOPLAS)},
  volume={19},
  number={3},
  pages={427--443},
  year={1997},
  publisher={Association for Computing Machinery},
  doi={10.1145/256167.256195}
}

@inproceedings{foster2015coalgebraic,
  title={A Coalgebraic Decision Procedure for NetKAT},
  author={Foster, Nate and Kozen, Dexter and Milano, Matthew and Silva, Alexandra and Thompson, Laure},
  booktitle={Proceedings of the 42nd Annual ACM SIGPLAN-SIGACT Symposium on Principles of Programming Languages},
  pages={343--355},
  year={2015},
  doi = {10.1145/2676726.2677011}, 
  series = {POPL '15},
  publisher = {Association for Computing Machinery}
}

@incollection{kozen2017coalgebraic,
  title={On the Coalgebraic Theory of Kleene Algebra with Tests},
  author={Kozen, Dexter},
  booktitle={Rohit Parikh on Logic, Language and Society},
  pages={279--298},
  year={2017},
  publisher={Springer},
  doi={10.1007/978-3-319-47843-2_15}
}

@inproceedings{dahlqvist2021write,
  title={How to Write a Coequation},
  author={Dahlqvist, Fredrik and Schmid, Todd},
  booktitle={9th Conference on Algebra and Coalgebra in Computer Science (CALCO 2021)},
  year={2021},
    pages =	{13:1--13:25},
  series =	{Leibniz International Proceedings in Informatics (LIPIcs)},
  doi={10.4230/LIPIcs.CALCO.2021.13}
}

@inproceedings{aarts2010learning,
  title={Learning I/O Automata},
  author={Aarts, Fides and Vaandrager, Frits},
  booktitle={International Conference on Concurrency Theory},
  pages={71--85},
  year={2010},
  organization={Springer},
  doi={10.1007/978-3-642-15375-4_6}
}

@inproceedings{moerman2017learning,
  author = {Moerman, Joshua and Sammartino, Matteo and Silva, Alexandra and Klin, Bartek and Szynwelski, Micha\l{}},
 title = {Learning Nominal Automata}, 
 year = {2017}, 
 publisher = {Association for Computing Machinery}, 
 doi = {10.1145/3009837.3009879}, 
 booktitle = {Proceedings of the 44th ACM SIGPLAN Symposium on Principles of Programming Languages}, 
 pages = {613–625}, 
 numpages = {13}, 
 series = {POPL '17}
}

@inproceedings{angluin1997learning,
title = {Learning Markov Chains with Variable Memory Length from Noisy Output},
author = {Angluin, Dana and Csundefinedr\"{o}s, Mikl\'{o}s},   booktitle={Proceedings of the Tenth Annual Conference on Computational Learning Theory},
publisher = {Association for Computing Machinery}, 
  pages={298--308},
  year={1997},
  doi={10.1145/267460.267517}
}

@inproceedings{denis2001residual,
  title={Residual Finite State Automata},
  author={Denis, Fran{\c{c}}ois and Lemay, Aur{\'e}lien and Terlutte, Alain},
  booktitle={Annual Symposium on Theoretical Aspects of Computer Science},
  pages={144--157},
  year={2001},
  organization={Springer},
  doi={10.1007/3-540-44693-1_13}
}

@inproceedings{bollig2009angluin,
  title={Angluin-Style Learning of NFA},
  author={Bollig, Benedikt and Habermehl, Peter and Kern, Carsten and Leucker, Martin},
booktitle = {Proceedings of the 21st International Joint Conference on Artificial Intelligence},
  year={2009},
  pages = {1004–1009}, 
  numpages = {6},
  series = {IJCAI'09},
}

@article{zetzsche2021,
   title={Canonical Automata via Distributive Law Homomorphisms},
   volume={351},
   DOI={10.4204/eptcs.351.18},
   journal={Electronic Proceedings in Theoretical Computer Science},
   publisher={Open Publishing Association},
   author={Zetzsche, Stefan and van Heerdt, Gerco and Sammartino, Matteo and Silva, Alexandra},
   year={2021},
   pages={296–313}
}

@article{zetzsche2022guarded,
  TITLE = {{Guarded Kleene Algebra with Tests: Automata Learning}},
  AUTHOR = {Stefan Zetzsche and Alexandra Silva and Matteo Sammartino},
  DOI = {10.46298/entics.10505},
  JOURNAL = {{Electronic Notes in Theoretical Informatics and Computer Science}},
  VOLUME = {{Volume 1 - Proceedings of MFPS XXXVIII}},
  YEAR = {2023},
}

@article{zetzsche2020generators,
  title={Generators and Bases for Monadic Closures},
  author={Zetzsche, Stefan and Silva, Alexandra and Sammartino, Matteo},
  journal={arXiv preprint arXiv:2010.10223},
  year={2023}
}

@inproceedings{van2020learning,
  title={Learning Automata with Side-Effects},
  author={van Heerdt, Gerco and Sammartino, Matteo and Silva, Alexandra},
  booktitle={Coalgebraic Methods in Computer Science},
  pages={68--89},
  year={2020},
  publisher={\\Springer},
  doi={10.1007/978-3-030-57201-3_5}
}

@article{thompson1968programming,
  title={Programming Techniques: Regular Expression Search Algorithm},
  author={Thompson, Ken},
  journal={Communications of the ACM},
  volume={11},
  number={6},
  pages={419--422},
  year={1968},
publisher = {Association for Computing Machinery}, 
numpages = {4}, 
doi = {10.1145/363347.363387}, 
}

@inproceedings{kozen2008bohm,
  title={The B{\"o}hm--Jacopini Theorem is False, Propositionally},
  author={Kozen, Dexter and Tseng, Wei-Lung Dustin},
  booktitle={International Conference on Mathematics of Program Construction},
  pages={177--192},
  year={2008},
  publisher={Springer},
  doi={10.1007/978-3-540-70594-9_11}
}

@inproceedings{levy2011similarity,
  title={Similarity Quotients as Final Coalgebras},
  author={Levy, Paul Blain},
booktitle={Foundations of Software Science and Computational Structures},
  pages={27--41},
  year={2011},
  organization={Springer},
  doi={10.1007/978-3-642-19805-2_3}
}

@article{baltag2000logic,
  title={A Logic for Coalgebraic Simulation},
  author={Baltag, Alexandru},
  journal={Electronic Notes in Theoretical Computer Science},
  volume={33},
  pages={42--60},
  year={2000},
  publisher={Elsevier},
  doi={10.1016/S1571-0661(05)80343-3}
}

@article{hughes2004simulations,
title = {Simulations in Coalgebra},
journal = {Electronic Notes in Theoretical Computer Science},
volume = {82},
number = {1},
pages = {128-149},
year = {2003},
issn = {1571-0661},
doi = {10.1016/S1571-0661(04)80636-4},
author = {Bart Jacobs and Jesse Hughes},
}

@article{hesselink2000fixpoint,
  title={Fixpoint Semantics and Simulation},
  author={Hesselink, Wim H. and Thijs, Albert},
  journal={Theoretical Computer Science},
  volume={238},
  number={1},
  pages={275--311},
  year={2000},
  publisher={Elsevier},
  doi={10.1016/S0304-3975(98)00176-5}
}

@inproceedings{chalupar2014automated,
  title={Automated Reverse Engineering using Lego{\textregistered}},
  author={Chalupar, Georg and Peherstorfer, Stefan and Poll, Erik and De Ruiter, Joeri},
  booktitle={8th USENIX Workshop on Offensive Technologies (WOOT 14)},
  year={2014},
  publisher = {USENIX Association},

}

@inproceedings{hagerer2002model,
  title={Model Generation by Moderated Regular Extrapolation},
  author={Hagerer, Andreas and Hungar, Hardi and Niese, Oliver and Steffen, Bernhard},
  booktitle={International Conference on Fundamental Approaches to Software Engineering},
  pages={80--95},
  year={2002},
  organization={Springer},
  doi={10.1007/3-540-45923-5_6}
}

@article{vaandrager2017model,
  title={Model Learning},
  author={Vaandrager, Frits},
  journal={Communications of the ACM},
  volume={60},
  number={2},
  pages={86--95},
  year={2017},
publisher = {Association for Computing Machinery}, 
doi={10.1145/2967606}
}

@article{feamster2014road,
  title={The Road to SDN: An Intellectual History of Programmable Networks},
  author={Feamster, Nick and Rexford, Jennifer and Zegura, Ellen},
  journal={ACM SIGCOMM Computer Communication Review},
  volume={44},
  number={2},
  pages={87--98},
  year={2014},
  publisher={Association for Computing Machinery},
  doi={10.1145/2602204.2602219}
}

@inproceedings{moore1956gedanken,
  title={Gedanken--experiments on Sequential Machines},
  author={Moore, Tyler},
  booktitle={Automata Studies, Annals of Mathematical Studies, no. 34},
  year={1956}, 	
  organization={Citeseer}
}

@inproceedings{pous2015symbolic,
  title={Symbolic Algorithms for Language Equivalence and Kleene Algebra with Tests},
  author={Pous, Damien},
booktitle = {Proceedings of the 42nd Annual ACM SIGPLAN-SIGACT Symposium on Principles of Programming Languages},
  pages={357--368},
  year={2015},	
  numpages = {12}, 
  series = {POPL '15},
  publisher = {Association for Computing Machinery}, 
doi = {10.1145/2676726.2677007}, 
}

@inproceedings{jacobs2011bases,
  title={Bases as Coalgebras},
  author={Jacobs, Bart},
  booktitle={Algebra and Coalgebra in Computer Science},
  pages={237--252},
  year={2011},
  organization={Springer},
  doi={10.1007/978-3-642-22944-2_17}
}

@inproceedings{jacobs2015recipe,
  title={A Recipe for State-and-Effect Triangles},
  author={Jacobs, Bart},
  booktitle={6th Conference on Algebra and Coalgebra in Computer Science (CALCO 2015)},
  year={2015},
  organization={Schloss Dagstuhl-Leibniz-Zentrum fuer Informatik}
}

@inproceedings{moerman2019residual,
  author    = {Joshua Moerman and
               Matteo Sammartino},
  title     = {Residual Nominal Automata},
  booktitle = {{CONCUR}},
  volume    = {171},
  pages     = {44:1--44:21},
  year      = {2020},
  doi       = {10.4230/LIPIcs.CONCUR.2020.44}
}

@book{pitts2013nominal,
  title={Nominal Sets: Names and Symmetry in Computer Science},
  author={Pitts, Andrew M.},
  volume={57},
  year={2013},
  publisher={Cambridge University Press}
}

@article{MyersAMU15,
  author    = {Robert S. R. Myers and
               Jiri Adamek and
               Stefan Milius and
               Henning Urbat},
  title     = {Coalgebraic Constructions of Canonical Nondeterministic Automata},
journal = {Theoretical Computer Science},
  volume    = {604},
  pages     = {81--101},
  year      = {2015},
  doi       = {10.1016/j.tcs.2015.03.035}
 }

@article{BrzozowskiT14,
  author    = {Janusz A. Brzozowski and
               Hellis Tamm},
  title     = {Theory of {\'{A}}tomata},
  journal   = {Theor. Comput. Sci.},
  volume    = {539},
  pages     = {13--27},
  year      = {2014},
  doi       = {10.1016/j.tcs.2014.04.016}}

@article{HeerdtMSS19,
  author    = {Gerco van Heerdt and
               Joshua Moerman and
               Matteo Sammartino and
               Alexandra Silva},
  title     = {A (Co)Algebraic Theory of Succinct Automata},
journal = {Journal of Logical and Algebraic Methods in Programming},
  volume    = {105},
  pages     = {112--125},
  year      = {2019},
  doi       = {10.1016/j.jlamp.2019.02.008}
 }

@article{arbib1975fuzzy,
  title={Fuzzy Machines in a Category},
  author={Arbib, Michael A and Manes, Ernest G},
  journal={Bulletin of the Australian Mathematical Society},
  volume={13},
  number={2},
  pages={169--210},
  year={1975},
  publisher={Cambridge University Press},
  doi={10.1017/S0004972700024412}
}

@article{power2002combining,
  title={Combining a Monad and a Comonad},
  author={Power, John and Watanabe, Hiroshi},
  journal={Theoretical Computer Science},
  volume={280},
  number={1},
  pages={137--162},
  year={2002},
  publisher={Elsevier},
  doi = {https://doi.org/10.1016/S0304-3975(01)00024-X},
}

@article{watanabe2002well,
  title={Well-Behaved Translations Between Structural Operational Semantics},
  author={Watanabe, Hiroshi},
  journal={Electronic Notes in Theoretical Computer Science},
  volume={65},
  number={1},
  pages={337--357},
  year={2002},
  publisher={Elsevier},
  doi = {10.1016/S1571-0661(04)80372-4},
}

@inproceedings{bonsangue2013presenting,
  title={Presenting Distributive Laws},
  author={Bonsangue, Marcello M and Hansen, Helle Hvid and Kurz, Alexander and Rot, Jurriaan},
  booktitle={International Conference on Algebra and Coalgebra in Computer Science},
  pages={95--109},
  year={2013},
  organization={Springer},
  doi={10.1007/978-3-642-40206-7_9}
}

@inproceedings{klin2015presenting,
  title={Presenting Morphisms of Distributive Laws},
  author={Klin, Bartek and Nachyla, Beata},
  booktitle={6th Conference on Algebra and Coalgebra in Computer Science (CALCO 2015)},
  year={2015},
  organization={Schloss Dagstuhl-Leibniz-Zentrum fuer Informatik},
  doi = {10.4230/LIPIcs.CALCO.2015.190},
  volume={35},
    pages =	{\\190--204},
}

@article{klin2004coalgebraic,
author = {Klin, Bartek},
year = {2004},
pages = {201-218},
title = {A Coalgebraic Approach to Process Equivalence and a Coinduction Principle for Traces},
volume = {106},
journal = {Electronic Notes in Theoretical Computer Science},
doi = {10.1016/j.entcs.2004.02.029},
  publisher={Elsevier}
}

@article{schroder2008expressivity,
  title={Expressivity of Coalgebraic Modal Logic: The Limits and Beyond},
  author={Schr{\"o}der, Lutz},
  journal={Theoretical Computer Science},
  volume={390},
  number={2},
  pages={230--247},
  year={2008},
  publisher={Elsevier},
  doi = {10.1016/j.tcs.2007.09.023},
}

@inproceedings{hansen2014strong,
  title={Strong Completeness for Iteration-Free Coalgebraic Dynamic Logics},
  author={Hansen, Helle Hvid and Kupke, Clemens and Leal, Raul Andres},
  booktitle={IFIP International Conference on Theoretical Computer Science},
  pages={281--295},
  year={2014},
  organization={Springer},
  doi={10.1007/978-3-662-44602-7_22}
}

@article{taylor2002subspaces,
  title={Subspaces in Abstract Stone Duality},
  author={Taylor, Paul},
  journal={Theory and Applications of Categories},
  volume={10},
  number={13},
  pages={301--368},
  year={2002},
}

@article{eilenberg1965adjoint,
  title={Adjoint Functors and Triples},
  author={Eilenberg, Samuel and Moore, John C. and others},
  journal={Illinois Journal of Mathematics},
  volume={9},
  number={3},
  pages={381--398},
  year={1965},
  publisher={University of Illinois at Urbana-Champaign},
  doi={10.1215/ijm/1256068141}
}

@inproceedings{linton1966some,
  title={Some Aspects of Equational Categories},
  author={Linton, Fred EJ},
  booktitle={Proceedings of the Conference on Categorical Algebra},
  pages={84--94},
  year={1966},
  organization={Springer},
  doi={10.1007/978-3-642-99902-4_3}
}

@book{moggi1988computational,
  title={Computational Lambda-Calculus and Monads},
  author={Moggi, Eugenio},
  year={1988},
  publisher={University of Edinburgh, Department of Computer Science, Laboratory for Foundations of Computer Science}
}

@book{moggi1990abstract,
  title={An Abstract View of Programming Languages},
  author={Moggi, Eugenio},
  year={1990},
  publisher={University of Edinburgh, Department of Computer Science, Laboratory for Foundations of Computer Science}
}

@article{moggi1991notions,
  title={Notions of Computation and Monads},
  author={Moggi, Eugenio},
  journal={Information and Computation},
  volume={93},
  number={1},
  pages={55--92},
  year={1991},
  publisher={Elsevier},
  doi={10.1016/0890-5401(91)90052-4}
}

@article{rutten2000universal,
  title={Universal Coalgebra: A Theory of Systems},
  author={Rutten, Jan},
  journal={Theoretical Computer Science},
  volume={249},
  number={1},
  pages={3--80},
  year={2000},
  publisher={Elsevier},
  doi={10.1016/S0304-3975(00)00056-6}
}

@inproceedings{beck1969distributive,
  title={Distributive Laws},
  author={Beck, Jon},
  booktitle={Seminar on Triples and Categorical Homology Theory},
  pages={119--140},
  year={1969},
  organization={Springer},
  doi={10.1007/BFb0083084}
}

@article{Street2009,
author = {Street, Ross},
journal = {Theory and Applications of Categories},
pages = {313-320},
publisher = {Mount Allison University, Department of Mathematics and Computer Science, Sackville},
title = {Weak Distributive Laws},
volume = {22},
year = {2009},
}

@article{rutten2013generalizing,
  title={Generalizing Determinization from Automata to Coalgebras},
  author={Rutten, Jan and Bonsangue, Marcello and Bonchi, Filippo and Silva, Alexandra},
  journal={Logical Methods in Computer Science},
  volume={9},
  year={2013},
  number={1},
  doi={10.2168/LMCS-9(1:9)2013}
}

@inproceedings{jacobs2012trace,
  title={Trace Semantics via Determinization},
  author={Jacobs, Bart and Silva, Alexandra and Sokolova, Ana},
  booktitle={International Workshop on Coalgebraic Methods in Computer Science},
  pages={109--129},
  year={2012},
  organization={Springer},
  doi={10.1007/978-3-642-32784-1_7}
}

@incollection{mohri2009weighted,
  title={Weighted Automata Algorithms},
  author={Mohri, Mehryar},
booktitle="Handbook of Weighted Automata",
  pages={213--254},
  year={2009},
  publisher={Springer},
  doi={10.1007/978-3-642-01492-5_6},
}

@incollection{mohri2008speech,
  title={Speech Recognition with Weighted Finite-State Transducers},
  author={Mohri, Mehryar and Pereira, Fernando and Riley, Michael},
booktitle={Springer Handbook of Speech Processing},
  pages={559--584},
  year={2008},
  publisher={Springer},
  doi={10.1007/978-3-540-49127-9_28},
}

@article{bonchi2014algebra,
  title={Algebra-Coalgebra Duality in Brzozowski's Minimization Algorithm},
  author={Bonchi, Filippo and Bonsangue, Marcello M and Hansen, Helle H and Panangaden, Prakash and Rutten, Jan and Silva, Alexandra},
  journal={ACM Transactions on Computational Logic (TOCL)},
  volume={15},
  number={1},
  pages={1--29},
  year={2014},
  publisher={Association for Computing Machinery},
  doi={10.1145/2490818}
}

@article{arnold1992note,
  title={A Note About Minimal Non-Deterministic Automata},
  author={Arnold, Andr{\'e} and Dicky, Anne and Nivat, Maurice},
  journal={Bulletin of the EATCS},
  volume={47},
  pages={166--169},
  year={1992}
}

@inproceedings{silva2010generalizing,
  title={Generalizing the Powerset Construction, Coalgebraically},
  author={Silva, Alexandra and Bonchi, Filippo and Bonsangue, Marcello M and Rutten, Jan},
  booktitle =	{IARCS Annual Conference on Foundations of Software Technology and Theoretical Computer Science (FSTTCS 2010)},
  year={2010},
    volume =	{8},
      pages =	{272--283},
  publisher =	{Schloss Dagstuhl -- Leibniz-Zentrum fuer Informatik},
    doi =		{10.4230/LIPIcs.FSTTCS.2010.272},
}

@article{nerode1958linear,
  title={Linear Automaton Transformations},
  author={Nerode, Anil},
  journal={Proceedings of the American Mathematical Society},
  volume={9},
  number={4},
  pages={541--544},
  year={1958},
  publisher={JSTOR},
  doi={10.2307/2033204}
}

@book{awodey2010category,
  title={Category Theory},
  author={Awodey, Steve},
  year={2010},
publisher = {Oxford University Press, Inc.}, 
}

@techreport{VuilleminG210,
  title={Efficient Equivalence and Minimization for Non Deterministic Xor Automata},
  author={Vuillemin, Jean and Gama, Nicolas},
    institution = {{Ecole Normale Sup{\'e}rieure}},
  year={2010}
}

@mastersthesis{van2016master,
  title={An Abstract Automata Learning Framework},
  author={van Heerdt, Gerco},
  school={Radboud University Nijmegen},
  year={2016}
}

@phdthesis{van2020phd,
  title={{CALF}: Categorical Automata Learning Framework},
  author={van Heerdt, Gerco},
  school={University College London},
  year={2020}
}

@inproceedings{esposito2002learning,
  title={Learning Probabilistic Residual Finite State Automata},
  author={Esposito, Yann and Lemay, Aur{\'e}lien and Denis, Fran{\c{c}}ois and Dupont, Pierre},
  booktitle={International Colloquium on Grammatical Inference},
  pages={77--91},
  year={2002},
  organization={Springer},
  doi={10.1007/3-540-45790-9_7}
}

@inproceedings{berndt2017learning,
  title={Learning Residual Alternating Automata},
  author={Berndt, Sebastian and Li{\'s}kiewicz, Maciej and Lutter, Matthias and Reischuk, R{\"u}diger},
  booktitle={Thirty-First AAAI Conference on Artificial Intelligence},
  year={2017},
  doi={10.1609/aaai.v31i1.10891}
}

@online{stackex,
title={Number of join-irreducible elements of a lattice: is it monotonic?},
	url={https://math.stackexchange.com/questions/801833/number-of-join-irreducible-elements-of-a-lattice-is-it-monotonic},
	urldate={2021-03-29}
	year={2014}
	}

@inproceedings{angluin2015learning,
  title={Learning Regular Languages via Alternating Automata},
  author={Angluin, Dana and Eisenstat, Sarah and Fisman, Dana},
booktitle = {Proceedings of the 24th International Conference on Artificial Intelligence},
  pages={3308--3314},
  year={2015},
  series = {IJCAI'15}
}

@inproceedings{brzozowski1962canonical,
  title={Canonical Regular Expressions and Minimal State Graphs for Definite Events},
  author={Brzozowski, Janusz A},
  booktitle={Proc. Symposium of Mathematical Theory of Automata},
  pages={529--561},
  volume={12},
  year={1962}
}

@incollection{bonchi2012brzozowski,
  title={Brzozowski’s Algorithm (Co)Algebraically},
  author={Bonchi, Filippo and Bonsangue, Marcello M and Rutten, Jan and Silva, Alexandra},
  booktitle={Logic and Program Semantics},
  pages={12--23},
  year={2012},
  publisher={Springer},
  doi={10.1007/978-3-642-29485-3_2}
}

@inproceedings{adamek2012coalgebraic,
  title={A Coalgebraic Perspective on Minimization and Determinization},
  author={Adamek, Jiri and Bonchi, Filippo and H{\"u}lsbusch, Mathias and K{\"o}nig, Barbara and Milius, Stefan and Silva, Alexandra},
  booktitle={International Conference on Foundations of Software Science and Computational Structures},
  pages={58--73},
  year={2012},
  organization={Springer}, 
  doi={10.1007/978-3-642-28729-9_4}
}

@inproceedings{zwart2019no,
  title={No-Go Theorems for Distributive Laws},
  author={Zwart, Maaike and Marsden, Dan},
  booktitle={2019 34th Annual ACM/IEEE Symposium on Logic in Computer Science (LICS)},
  pages={1--13},
  year={2019},
  publisher={IEEE Computer Society},
  doi={10.1109/LICS.2019.8785707}
}

@article{rivest1993inference,
  title={Inference of Finite Automata Using Homing Sequences},
  author={Rivest, Ronald L and Schapire, Robert E},
  journal={Information and Computation},
  volume={103},
  number={2},
  pages={299--347},
  year={1993},
  publisher={Elsevier},
  doi = {10.1006/inco.1993.1021},
}

@phdthesis{howar2012active,
  title={Active Learning of Interface Programs},
  author={Howar, Falk},
  year={2012},
  school={Dortmund University of Technology},
  doi={10.17877/DE290R-4817}
}

@book{kearns1994introduction,
  title={An Introduction to Computational Learning Theory},
  author={Kearns, Michael J and Vazirani, Umesh Virkumar and Vazirani, Umesh},
  year={1994},
  publisher={MIT Press},
  doi={10.7551/mitpress/3897.001.0001}
}

@inproceedings{isberner2014ttt,
  title={The TTT Algorithm: A Redundancy-Free Approach to Active Automata Learning},
  author={Isberner, Malte and Howar, Falk and Steffen, Bernhard},
  booktitle={International Conference on Runtime Verification},
  pages={307--322},
  year={2014},
  organization={Springer},
  doi={10.1007/978-3-319-11164-3_26}
}

@InProceedings{vaandrager2021new,
author={Vaandrager, Frits
and Garhewal, Bharat
and Rot, Jurriaan
and Wi{\ss}mann, Thorsten},
title={A New Approach for Active Automata Learning Based on Apartness},
booktitle={Tools and Algorithms for the Construction and Analysis of Systems},
year={2022},
publisher={Springer},
pages={223--243},
doi={10.1007/978-3-030-99524-9_12}
}

@inproceedings{foster2016probabilistic,
  title={Probabilistic NetKAT},
  author={Foster, Nate and Kozen, Dexter and Mamouras, Konstantinos and Reitblatt, Mark and Silva, Alexandra},
booktitle = {Proceedings of the 25th European Symposium on Programming Languages and Systems},
  pages={282--309},
  numpages = {28},
  year={2016},
  publisher={Springer},
  doi={10.1007/978-3-662-49498-1_12}
}

@inproceedings{grathwohl2014kat+,
  title={KAT + B!},
  author={Grathwohl, Niels Bj{\o}rn Bugge and Kozen, Dexter and Mamouras, Konstantinos},
booktitle = {Proceedings of the Joint Meeting of the Twenty-Third EACSL Annual Conference on Computer Science Logic and the Twenty-Ninth Annual ACM/IEEE Symposium on Logic in Computer Science},
  pages={1--10},
  year={2014},
  publisher = {Association for Computing Machinery}, 
  doi = {10.1145/2603088.2603095}, 
  series = {CSL-LICS '14}
}

@article{maler1995learnability,
  title={On the Learnability of Infinitary Regular Sets},
  author={Maler, Oded and Pnueli, Amir},
  journal={Information and Computation},
  volume={118},
  number={2},
  pages={316--326},
  year={1995},
  publisher={Elsevier},
  doi={10.1006/inco.1995.1070}
}

@book{mac2013categories,
  title={Categories for the Working Mathematician},
  author={Mac Lane, Saunders},
  volume={5},
  year={2013},
  publisher={\\Springer},
  doi={10.1007/978-1-4757-4721-8}
}

@book{leinster2014basic,
  title={Basic Category Theory},
  series={Cambridge Studies in Advanced Mathematics},
  author={Leinster, Tom},
  year={2014},
  publisher={Cambridge University Press},
  doi={10.1017/CBO9781107360068},
  collection={Cambridge Studies in Advanced Mathematics}
}

@inproceedings{drews2017learning,
  title={Learning Symbolic Automata},
  author={Drews, Samuel and D’Antoni, Loris},
  booktitle={International Conference on Tools and Algorithms for the Construction and Analysis of Systems},
  pages={173--189},
  year={2017},
  organization={Springer},
  doi={10.1007/978-3-662-54577-5_10}
}

@inproceedings{maler2014learning,
  title={Learning Regular Languages over Large Alphabets},
  author={Maler, Oded and Mens, Irini-Eleftheria},
booktitle={Tools and Algorithms for the Construction and Analysis of Systems},
  pages={485--499},
  year={2014},
  publisher={Springer},
  doi={10.1007/978-3-642-54862-8_41}
}

@inproceedings{Farzan2008Extending,
  title={Extending Automated Compositional Verification to the Full Class of Omega-Regular Languages},
  author={Farzan, Azadeh and Chen, Yu Fang and Clarke, Edmund M. and Tsay, Yih Kuen and Wang, Bow Yaw},
  booktitle={Tools and Algorithms for the Construction and Analysis of Systems},
  pages={2-17},
  year={2008},
  publisher={Springer}
  doi={10.1007/978-3-540-78800-3_2}
}

@article{angluin2016learning,
  title = {Learning Regular Omega Languages},
  journal = {Theoretical Computer Science},
  volume = {650},
  pages = {57-72},
  year = {2016},
  publisher={Elsevier},
  doi = {10.1016/j.tcs.2016.07.031},
  author = {Dana Angluin and Dana Fisman},
}

@InProceedings{li2021novel,
  title={A Novel Learning Algorithm for B{\"u}chi Automata Based on Family of DFAs and Classification Trees},
  author={Li, Yong and Chen, Yu-Fang and Zhang, Lijun and Liu, Depeng},
booktitle={Tools and Algorithms for the Construction and Analysis of Systems},  volume={281},
pages={208--226},
  year={2017},
  publisher={Springer},
  doi={10.1007/978-3-662-54577-5_12}
}

@inproceedings{calbrix1993ultimately,
  title={Ultimately Periodic Words of Rational $\omega$-Languages},
  author={Calbrix, Hugues and Nivat, Maurice and Podelski, Andreas},
  booktitle={International Conference on Mathematical Foundations of Programming Semantics},
  pages={554--566},
  year={1993},
  organization={Springer},
  doi={10.1007/3-540-58027-1_27}
}

@inproceedings{li2017novel,
  title={A Novel Learning Algorithm for B{\"u}chi Automata Based on Family of DFAs and Classification Trees},
  author={Li, Yong and Chen, Yu-Fang and Zhang, Lijun and Liu, Depeng},
  booktitle={International Conference on Tools and Algorithms for the Construction and Analysis of Systems},
  pages={208--226},
  year={2017},
  organization={Springer},

}

@article{alpern1987recognizing,
  title={Recognizing Safety and Liveness},
  author={Alpern, Bowen and Schneider, Fred B},
  journal={Distributed Computing},
  volume={2},
  number={3},
  pages={117--126},
  year={1987},
  publisher={Springer}, 
  doi={10.1007/BF01782772}
}

@inproceedings{howar2012inferring,
  title={Inferring Canonical Register Automata},
  author={Howar, Falk and Steffen, Bernhard and Jonsson, Bengt and Cassel, Sofia},
  booktitle={International Workshop on Verification, Model Checking, and Abstract Interpretation},
  pages={251--266},
  year={2012},
  organization={Springer}
}

@article{isberner2014learning,
  title={Learning Register Automata: From Languages to Program Structures},
  author={Isberner, Malte and Howar, Falk and Steffen, Bernhard},
  journal={Machine Learning},
  volume={96},
  number={1},
  pages={65--98},
  year={2014},
  publisher={Springer},
  doi={10.1007/s10994-013-5419}
}

@inproceedings{vardi1986automata,
  title={An Automata-Theoretic Approach to Automatic Program Verification},
  author={Vardi, Moshe Y and Wolper, Pierre},
  booktitle={Proceedings of the Symposium on Logic in Computer Science (LICS)},
  year={1986},
  pages     = {332--344},
  publisher={IEEE Computer Society}
}

@inproceedings{pnueli1989synthesis,
  title={On the Synthesis of a Reactive Module},
  author={Pnueli, Amir and Rosner, Roni},
  booktitle = {Proceedings of the 16th ACM SIGPLAN-SIGACT Symposium on Principles of Programming Languages},
  pages={179--190},
  year={1989},
  publisher = {Association for Computing Machinery}, 
  doi = {10.1145/75277.75293}, 
  numpages = {12}, 
  series = {POPL '89} 
}

@inproceedings{lee2001size,
  title={The Size-Change Principle for Program Termination},
  author={Lee, Chin Soon and Jones, Neil D and Ben-Amram, Amir M},
  booktitle = {Proceedings of the 28th ACM \\SIGPLAN-SIGACT Symposium on Principles of Programming Languages},
  pages={81--92},
  publisher = {Association for Computing Machinery}, 
  doi = {10.1145/360204.360210},
  numpages = {12}, 
  series = {\\POPL '01},
  year={2001}
}

@article{bergadano1996learning,
  title={Learning Behaviors of Automata from Multiplicity and Equivalence Queries},
  author={Bergadano, Francesco and Varricchio, Stefano},
  journal={SIAM Journal on Computing},
  volume={25},
  number={6},
  pages={1268--1280},
  year={1996},
  publisher={SIAM},
  doi = {10.1137/S009753979326091X},
}

@InProceedings{bergadano1994learning,
 author={Bergadano, Francesco and Varricchio, Stefano},
title="Learning Behaviors of Automata from Multiplicity and Equivalence Queries",
booktitle="Algorithms and Complexity",
year="1994",
publisher="Springer",
pages="54--62",
doi={10.1007/3-540-57811-0_6},
}

@inproceedings{balle2015learning,
  title={Learning Weighted Automata},
  author={Balle, Borja and Mohri, Mehryar},
  booktitle={International Conference on Algebraic Informatics},
  pages={1--21},
  year={2015},
  organization={Springer},
  doi={10.1007/978-3-319-23021-4_1}
}

@inproceedings{van2020learningweighted,
  title={Learning Weighted Automata over Principal Ideal Domains},
  author={van Heerdt, Gerco and Kupke, Clemens and Rot, Jurriaan and Silva, Alexandra},
booktitle={Foundations of Software Science and Computation Structures},
doi={10.1007/978-3-030-45231-5_31},
  pages={602--621},
  publisher={Springer},
  year={2020}
}

@inproceedings{aarts2015learning,
  title={Learning Register Automata with Fresh Value Generation},
  author={Aarts, Fides and Fiterau-Brostean, Paul and Kuppens, Harco and Vaandrager, Frits},
  booktitle={International Colloquium on Theoretical Aspects of Computing},
  pages={165--183},
  year={2015},
  organization={Springer},
  doi = {10.1007/978-3-319-25150-9_11}
}

@inproceedings{balle2012spectral,
 author = {Balle, Borja and Mohri, Mehryar},
 booktitle = {Advances in Neural Information Processing Systems},
 editor = {F. Pereira and C.J. Burges and L. Bottou and K.Q. Weinberger},
 pages = {},
 publisher = {Curran Associates, Inc.},
 title = {Spectral Learning of General Weighted Automata via Constrained Matrix Completion},
 volume = {25},
 year = {2012}
}

@inproceedings{salomaa2012automata,
  title={Automata-Theoretic Aspects of Formal Power Series},
  author={Salomaa, Arto and Soittola, Matti},
  year={1978},
  publisher={Springer},
  doi={10.1007/978-1-4612-6264-0},
    booktitle={Texts and Monographs in Computer Science},
}

@inproceedings{breuel2008ocropus,
  title={The OCRopus open source OCR system},
  author={Breuel, Thomas M},
  booktitle = {Document Recognition and Retrieval XV},
  volume={6815},
  pages={120--134},
  year={2008},
  organization = {International Society for Optics and Photonics},
  publisher = {SPIE},
  doi = {10.1117/12.783598}
}

@incollection{albert2009digital,
  title={Digital Image Compression},
  author={Albert, J{\"u}rgen and Kari, Jarkko},
  booktitle={Handbook of Weighted Automata},
  pages={453--479},
  year={2009},
  publisher={Springer},
  doi={10.1007/978-3-642-01492-5_11}
}

@article{culik1993image,
  title={Image Compression Using Weighted Finite Automata},
  author={Culik II, Karel and Kari, Jarkko},
  journal={Computers \& Graphics},
  volume={17},
  number={3},
  pages={305--313},
  year={1993},
  publisher={Elsevier},
  doi={10.1016/0097-8493(93)90079-O}
}

@inproceedings{allauzen2008sequence,
  title={Sequence Kernels for Predicting Protein Essentiality},
  author={Allauzen, Cyril and Mohri, Mehryar and Talwalkar, Ameet},
  booktitle={Proceedings of the 25th International Conference on Machine Learning},
  pages={9--16},
  year={2008},
  doi={10.1145/1390156.1390158}
}

@inproceedings{aminof2011formal,
  title={Formal Analysis of Online Algorithms},
  author={Aminof, Benjamin and Kupferman, Orna and Lampert, Robby},
  booktitle={International Symposium on Automated Technology for Verification and Analysis},
  pages={213--227},
  year={2011},
  organization={Springer},
  doi={10.1007/978-3-642-24372-1_16}
}

@inproceedings{shahbaz2009inferring,
  title={Inferring Mealy Machines},
  author={Shahbaz, Muzammil and Groz, Roland},
  booktitle={International Symposium on Formal Methods},
  pages={207--222},
  year={2009},
  organization={Springer},
  doi={10.1007/978-3-642-05089-3_14}
}

@inproceedings{aarts2014algorithms,
  title={Algorithms for Inferring Register Automata},
  author={Aarts, Fides and Howar, Falk and Kuppens, Harco and Vaandrager, Frits},
  booktitle={International Symposium On Leveraging Applications of Formal Methods, Verification and Validation},
  pages={202--219},
  year={2014},
  organization={Springer}, 
  doi={10.1007/978-3-662-45234-9_15}
}

@inproceedings{steffen2012active,
  title={Active Automata Learning: From DFAs to Interface Programs and Beyond},
  author={Steffen, Bernhard and Howar, Falk and Isberner, Malte},
  booktitle = 	 {Proceedings of the Eleventh International Conference on Grammatical Inference},
  pages={195--209},
  year={2012},
  publisher={PMLR},
    volume = 	 {21},
    series = 	 {Proceedings of Machine Learning Research},
}

@inproceedings{steffen2011introduction,
  title={Introduction to Active Automata Learning From a Practical Perspective},
  author={Steffen, Bernhard and Howar, Falk and Merten, Maik},
  booktitle={International School on Formal Methods for the Design of Computer, Communication and Software Systems},
  pages={256--296},
  year={2011},
  organization={Springer}
}

@inproceedings{pena1998new,
  title={A New Algorithm for the Reduction of Incompletely Specified Finite State Machines},
  author={Pena, Jorge M and Oliveira, Arlindo L},
  booktitle={1998 IEEE/ACM International Conference on Computer-Aided Design}, 
  pages={482--489},
  year={1998},
    doi={10.1145/288548.289075}
}

@article{mealy1955method,
  author={Mealy, George H.},
  journal={The Bell System Technical Journal}, 
  title={A Method for Synthesizing Sequential Circuits}, 
  year={1955},
  volume={34},
  number={5},
  pages={1045-1079},
  doi={10.1002/j.1538-7305.1955.tb03788.x}
}

@article{lasota2014automata,
  title={Automata Theory in Nominal Sets},
  author={Lasota, S{\l}awomir and Klin, Bartek and Boja{\'n}czyk, Miko{\l}aj},
  journal={Logical Methods in Computer Science},
  volume={10},
  number={3},
  year={2014},
  doi={10.2168/LMCS-10(3:4)2014}
}

@inproceedings{d2014minimization,
  title={Minimization of Symbolic Automata},
  author={D'Antoni, Loris and Veanes, Margus},
booktitle = {Proceedings of the 41st ACM SIGPLAN-SIGACT Symposium on Principles of Programming Languages}, 
  pages={541--553},
  publisher = {Association for Computing Machinery}, 
  series = {POPL '14},
  numpages = {13},
  year={2014},
  doi={10.1145/2535838.2535849}
}

@article{cassel2016active,
  title={Active Learning for Extended Finite State Machines},
  author={Cassel, Sofia and Howar, Falk and Jonsson, Bengt and Steffen, Bernhard},
  journal={Formal Aspects of Computing},
  volume={28},
  number={2},
  pages={233--263},
  year={2016},
  publisher={Springer},
  doi={10.1007/s00165-016-0355-5}
}

@inproceedings{bollig2013fresh,
  title={A Fresh Approach to Learning Register Automata},
  author={Bollig, Benedikt and Habermehl, Peter and Leucker, Martin and Monmege, Benjamin},
  booktitle={International Conference on Developments in Language Theory},
  pages={118--130},
  year={2013},
  organization={Springer},
  doi={10.1007/978-3-642-38771-5_12}
}

@inproceedings{tamm2018theoretical,
  title={Theoretical Aspects of Symbolic Automata},
  author={Tamm, Hellis and Veanes, Margus},
booktitle={SOFSEM 2018: Theory and Practice of Computer Science},
  pages={428--441},
  year={2018},
  publisher={Springer},
  doi={10.1007/978-3-319-73117-9_30}
}

@inproceedings{fisman2022inferring,
  title={Inferring Symbolic Automata},
  author={Fisman, Dana and Frenkel, Hadar and Zilles, Sandra},
  booktitle={30th EACSL Annual Conference on Computer Science Logic (CSL 2022)},
  year={2022},
    volume =	{216},
  publisher =	{Schloss Dagstuhl -- Leibniz-Zentrum f{\"u}r Informatik},
    pages =	{21:1--21:19},
  doi={10.4230/LIPIcs.CSL.2022.21}
}

@inproceedings{schroder2017nominal,
  title={Nominal Automata with Name Binding},
  author={Schr{\"o}der, Lutz and Kozen, Dexter and Milius, Stefan and Wi{\ss}mann, Thorsten},
booktitle={Foundations of Software Science and Computation Structures},
  pages={124--142},
  year={2017},
  publisher={Springer},
  doi={10.1007/978-3-662-54458-7_8}
}

@article{kaminski1994finite,
  title={Finite-Memory Automata},
  author={Kaminski, Michael and Francez, Nissim},
  journal={Theoretical Computer Science},
  volume={134},
  number={2},
  pages={329--363},
  year={1994},
  publisher={Elsevier},
  doi={10.1016/0304-3975(94)90242-9}
}

@inproceedings{d2019symbolic,
  title={Symbolic Register Automata},
  author={D’Antoni, Loris and Ferreira, Tiago and Sammartino, Matteo and Silva, Alexandra},
  booktitle={International Conference on Computer Aided Verification},
  pages={3--21},
  year={2019},
  organization={Springer},
  doi={10.1007/978-3-030-25540-4_1}
}

@InProceedings{giantamidis2021learning,
author={Giantamidis, Georgios
and Tripakis, Stavros},
title={Learning Moore Machines from Input-Output Traces},
booktitle={FM 2016: Formal Methods},
year={2016},
publisher={Springer},
pages={291--309},
doi={10.1007/978-3-319-48989-6_18}
}

@inproceedings{moerman2019learning,
  title={Learning Product Automata},
  author={Moerman, Joshua},
  booktitle = 	 {Proceedings of The 14th International Conference on Grammatical Inference 2018},
  pages={54--66},
  year={2019},
volume = 	 {93},
  series = 	 {Proceedings of Machine Learning Research},
   publisher =    {PMLR},
}

@inproceedings{de2015protocol,
  title={Protocol State Fuzzing of TLS Implementations},
  author={De Ruiter, Joeri and Poll, Erik},
booktitle = {Proceedings of the 24th USENIX Conference on Security Symposium},  pages={193--206},
  publisher = {USENIX Association}, 
  numpages = {14}, 
  series = {SEC'15},
  year={2015}
}

@article{angluin1981note,
  title={A Note on the Number of Queries Needed to Identify Regular Languages},
  author={Angluin, Dana},
  journal={Information and Control},
  volume={51},
  number={1},
  pages={76--87},
  year={1981},
  publisher={Elsevier},
  doi={10.1016/S0019-9958(81)90090-5}
}

@article{angluin1988queries,
  title={Queries and Concept Learning},
  author={Angluin, Dana},
  journal={Machine Learning},
  volume={2},
  number={4},
  pages={319--342},
  year={1988},
  publisher={Springer},
  doi={10.1023/A:1022821128753}
}

@article{seal2013tensors,
  title={Tensors, Monads and Actions},
  author={Seal, Gavin J},
  journal={Theory and Applications of Categories},
  volume={28},
  number={15},
  pages={403--433},
  year={2013}
}

@InProceedings{parlant_et_al:LIPIcs:2020:12746,
  author =	{Louis Parlant and Jurriaan Rot and Alexandra Silva and Bas Westerbaan},
  title =	{{Preservation of Equations by Monoidal Monads}},
  booktitle =	{45th International Symposium on Mathematical Foundations of Computer Science (MFCS 2020)},
  pages =	{77:1--77:14},
  series =	{Leibniz International Proceedings in Informatics (LIPIcs)},
  year =	{2020},
  volume =	{170},
  publisher =	{Schloss Dagstuhl--Leibniz-Zentrum f{\"u}r Informatik},
  doi =		{10.4230/LIPIcs.MFCS.2020.77},
}

@inproceedings{turi1997towards,
  title={Towards a Mathematical Operational Semantics},
  author={Turi, Daniele and Plotkin, Gordon},
  booktitle={Proceedings of Twelfth Annual IEEE Symposium on Logic in Computer Science},
  pages={280--291},
  year={1997},
  publisher={IEEE},
  doi={10.1109/LICS.1997.614955}
}

@book{adamek1994locally,
  title={Locally Presentable and Accessible Categories},
  author={Adamek, Jiri and Rosicky, Jiri},
  volume={189},
  year={1994},
  publisher={Cambridge University Press},
  doi={10.1017/CBO9780511600579}
}

@phdthesis{turi1996functorial,
  title={Functorial Operational Semantics},
  author={Turi, Daniele},
  year={1996},
  school={Vrije Universiteit Amsterdam}
}

@article{adamek2019finitely,
  title={Finitely Presentable Algebras For Finitary Monads},
  author={Adamek, Jiri and Milius, Stefan and Sousa, Lurdes and Wi{\ss}mann, Thorsten},
  journal={Theory and Applications of Categories},
  volume={34},
  number={37},
  pages={1179--1195},
  year={2019}
}

@incollection{coumans2010scalars,
    author = {Coumans, Dion and Jacobs, Bart},
    title = {Scalars, Monads, and Categories},
    booktitle = {Quantum Physics and Linguistics: A Compositional, Diagrammatic Discourse},
    publisher = {Oxford University Press},
    year = {2013},
    doi = {10.1093/acprof:oso/9780199646296.003.0007}
}

@book{jacobs2017introduction, 
  series={Cambridge Tracts in Theoretical Computer Science}, 
  title={Introduction to Coalgebra: Towards Mathematics of States and Observation}, 
  DOI={10.1017/CBO9781316823187}, 
  publisher={Cambridge University Press}, 
  author={Jacobs, Bart},
  year={2016}, 
  collection={Cambridge Tracts in Theoretical Computer Science}
}

@article{rutten2019method,
  title={The Method of Coalgebra: Exercises in Coinduction},
  author={Rutten, Jan},
  year={2019},
  publisher={CWI Amsterdam}
}

@inbook{jacobs2014automata,
  title={Automata Learning: A Categorical Perspective},
  author={Jacobs, Bart and Silva, Alexandra},
  booktitle={Horizons of the Mind. A Tribute to Prakash Panangaden},
  pages={384--406},
  year={2014},
  publisher={Springer},
  doi={10.1007/978-3-319-06880-0_20}
}

@article{klin2011bialgebras,
  title={Bialgebras for Structural Operational Semantics: An Introduction},
  author={Klin, Bartek},
  journal={Theoretical Computer Science},
  volume={412},
  number={38},
  pages={5043--5069},
  year={2011},
  publisher={Elsevier},
  doi={10.1016/j.tcs.2011.03.023}
}

@incollection{jacobs2006bialgebraic,
  title={A Bialgebraic Review of Deterministic Automata, Regular Expressions and Languages},
  author={Jacobs, Bart},
  booktitle={Algebra, Meaning, and Computation},
  pages={375--404},
  year={2006},
  publisher={Springer},
  doi={10.1007/11780274_20}
}

@article{lenisa2000distributivity,
  title={Distributivity for Endofunctors, Pointed and Co-pointed Endofunctors, Monads and Comonads},
  author={Lenisa, Marina and Power, John and Watanabe, Hiroshi},
  journal={Electronic Notes in Theoretical Computer Science},
  volume={33},
  pages={230--260},
  year={2000},
  publisher={Elsevier},
  doi = {10.1016/S1571-0661(05)80350-0},
}

@InProceedings{van2017calf,
  author =	{Gerco van Heerdt and Matteo Sammartino and Alexandra Silva},
  title =	{{CALF: Categorical Automata Learning Framework}},
  booktitle =	{26th EACSL Annual Conference on Computer Science Logic (CSL 2017)},
  pages =	{29:1--29:24},
  series =	{Leibniz International Proceedings in Informatics (LIPIcs)},
  year =	{2017},
  volume =	{82},
  publisher =	{Schloss Dagstuhl--Leibniz-Zentrum fuer Informatik},
  doi =		{10.4230/LIPIcs.CSL.2017.29},
}

@article{lang2004algebra,
  title={Algebra},
  author={Lang, Serge},
  journal={Graduate Texts in Mathematics},
  year={2002},
  publisher={Springer}
  doi={10.1007/978-1-4613-0041-0}
}

@article{adamek2009abstract,
  title={Abstract and Concrete Categories: The Joy of Cats},
  author={Adamek, Jiri and Herrlich, Horst and Strecker, George E},
  journal={Reprints in Theory and Applications of Categories},
  year={2009}
}

@article{riehl2008factorization,
  title={Factorization Systems},
  author={Riehl, Emily},
  year={2008},
  url={https://math.jhu.edu/~eriehl/factorization.pdf},
  urldate={2022-11-28}
}

@article{bousfield1977constructions,
  title={Constructions of Factorization Systems in Categories},
  author={Bousfield, Aldridge K},
  journal={Journal of Pure and Applied Algebra},
  volume={9},
  number={2-3},
  pages={207--220},
  year={1977},
  publisher={Elsevier},
  doi={10.1016/0022-4049(77)90067-6}
}

@article{maclane1950duality,
  title={Duality for Groups},
  author={MacLane, Saunders},
  journal={Bulletin of the American Mathematical Society},
  volume={56},
  number={6},
  pages={485--516},
  year={1950},
  publisher={American Mathematical Society}
}

@phdthesis{kurzlogics,
  title={Logics for Coalgebras and Applications to Computer Science},
  author={Kurz, Alexander},
  school={Ludwig-Maximilians-Universität München},
  year={2000}
}

@article{wissmann2022minimality,
  author       = {Thorsten Wißmann},
  title        = {Minimality Notions via Factorization Systems and Examples},
  year         = {2022},
  journal      = {Logical Methods in Computer Science},
  number= {3},
  doi          = {10.46298/lmcs-18(3:31)2022},
  volume       = {18},
}

@article{STREET1972149,
title = {The Formal Theory of Monads},
journal = {Journal of Pure and Applied Algebra},
volume = {2},
number = {2},
pages = {149-168},
year = {1972},
doi = {10.1016/0022-4049(72)90019-9},
author = {Ross Street}
}

@book{kozen1997automata, author = {Kozen, Dexter C.}, title = {Automata and Computability}, year = {1997}, publisher = {Springer}
, edition = {1st}, doi={10.1007/978-1-4612-1844-9}
 }

@misc{Ocaml,
 title={OCaml},
 url={https://ocaml.org}
 }

@ARTICLE{1702519,
  author={Chow, T.S.},
  journal={IEEE Transactions on Software Engineering}, 
  title={Testing Software Design Modeled by Finite-State Machines}, 
  year={1978},
  volume={SE-4},
  number={3},
  pages={178-187},
  doi={10.1109/TSE.1978.231496}}

@book{hopcroft1971linear,
  title={A Linear Algorithm for Testing Equivalence of Finite Automata},
  author={Hopcroft, John E and Karp, Richard M},
  volume={114},
  year={1971},
  publisher={Defense Technical Information Center}
}

@article{bonchi2013checking,
  title={Checking NFA Equivalence With Bisimulations up to Congruence},
  author={Bonchi, Filippo and Pous, Damien},
  journal={ACM SIGPLAN Notices},
  volume={48},
  number={1},
  pages={457--468},
  year={2013},
  publisher={ACM New York, NY, USA}
}

@inproceedings{birkhoff1935structure,
  title={On the Structure of Abstract Algebras},
  author={Birkhoff, Garrett},
  booktitle={Mathematical Proceedings of the Cambridge Philosophical Society},
  volume={31},
  number={4},
  pages={433--454},
  year={1935},
  organization={Cambridge University Press}
}

@inproceedings{smolka2015fast,
  title={A Fast Compiler for NetKAT},
  author={Smolka, Steffen and Eliopoulos, Spiridon and Foster, Nate and Guha, Arjun},
  booktitle={Proceedings of the 20th ACM SIGPLAN International Conference on Functional Programming},
  pages={328--341},
  year={2015}
}
\cleardoublepage
\addcontentsline{toc}{chapter}{\listfigurename}
\listoffigures


%
%

\end{document}